\documentclass[letterpaper,10pt,defended]{cit_thesis}

\title{Lattice quantum codes and \\ exotic topological phases of matter}
\author{Jeongwan Haah}
\date{May 15, 2013}

\usepackage{color}
\definecolor{darkblue}{rgb}{0.,0.,0.4}
\definecolor{darkred}{rgb}{0.5,0.,0.}
\usepackage[pdftex,colorlinks=true,linkcolor=darkblue,citecolor=darkred,urlcolor=blue]{hyperref}

\usepackage{amsmath}
\usepackage{amssymb}
\usepackage{amsfonts}
\usepackage{amsthm}
\usepackage[all]{xy}

\newtheorem{theorem}{Theorem}[chapter]
\newtheorem{lem}{Lemma}[section]
\newtheorem{cor}[lem]{Corollary}
\newtheorem{prop}[lem]{Proposition}
\theoremstyle{definition} \newtheorem{defn}{Definition}[chapter]
\theoremstyle{definition} \newtheorem{rem}{Remark}[chapter]
\theoremstyle{definition} \newtheorem{example}{Example}[chapter]

\newcommand{\bra}[1]{\left\langle #1 \right|}
\newcommand{\ket}[1]{\left| {#1} \right\rangle}

\newcommand{\trace}{\mathop{\mathrm{Tr}}\nolimits}
\newcommand{\tr}{\mathop{\mathrm{tr}}\nolimits}
\newcommand{\im}{\mathop{\mathrm{im}}\nolimits}
\newcommand{\coker}{\mathop{\mathrm{coker}}\nolimits}
\newcommand{\ann}{\mathop{\mathrm{ann}}\nolimits}
\newcommand{\rank}{\mathop{\mathrm{rank}}\nolimits}
\newcommand{\codim}{\mathop{\mathrm{codim}}\nolimits}
\newcommand{\rad}{\mathop{\mathrm{rad}}\nolimits}
\newcommand{\Tor}{\mathop{\mathrm{Tor}}\nolimits}
\newcommand{\lspan}{\mathop{\mathrm{span}}\nolimits}

\newcommand{\lt}{\mathop{\mathrm{lt}}\nolimits} 

\newcommand{\FF}{\mathbb{F}}
\newcommand{\ZZ}{\mathbb{Z}}
\newcommand{\QQ}{\mathbb{Q}}
\newcommand{\RR}{\mathbb{R}}
\newcommand{\calZ}{\mathcal{Z}}
\newcommand{\mm}{\mathfrak{m}}
\newcommand{\pp}{\mathfrak{p}}
\newcommand{\bb}{\mathfrak{b}}
\newcommand{\id}{\mathrm{id}}
\newcommand{\dd}{\mathrm{d}} 

\newcommand{\half}{\frac{1}{2}}
\newcommand{\ltqo}{{L_{tqo}}}

\newcommand{\drawgenerator}[8]{%
\xymatrix@!0{%
& #8 \ar@{-}[ld]\ar@{.}[dd] \ar@{-}[rr] & & #7 \ar@{-}[ld]  \\%
#1 \ar@{-}[rr] \ar@{-}[dd] &  & #2 \ar@{-}[dd] &            \\%
& #6 \ar@{.}[ld] &  & #5 \ar@{-}[uu] \ar@{.}[ll]       \\%
#3 \ar@{-}[rr] &  & #4 \ar@{-}[ru]                       %
}%
}

\begin{document}

\maketitle

\begin{acknowledgments}
I would like to thank my advisor John Preskill for his inspiring comments and advice, 
and forgivingness and encouragement during my graduate study.
The thesis would not have been possible without reverence to his standard of clear understanding.
I should add that my enjoyable last four years are deeply rooted in him.

A special thanks to Sergey Bravyi for the fruitful collaboration and numerous insights.
It is not too much to say that the algebraic formulation of this thesis
grew out from his note on the degeneracy computation.
The interaction with him gave me an opportunity to do a research internship
at IBM Watson Research Center, which was a valuable experience.
I thank Charles Bennett, John Smolin, and Graeme Smith.

L. Yong-uck Chung, who had been my roommate for years, helped me to study algebra.
He was always willing to explain concepts and related theorems however easy or hard they were.
Informal discussions with him were absolutely important for me.
I also thank Eric Rains and Tom Graber for giving insights and answers for mathematical questions.

I appreciate the inspiring discussion with Guifre Vidal and Daniel Gottesman while I was visiting Perimeter Institute,
which led my attention to the entanglement renormalization group flow.

I also thank Alexei Kitaev for his encouragement and numerous useful comments.
I was honored to have discussions with him.

I thank my friend I. Han-young Kim, who shared moments of inspiration.

Thanks to the members of IQI and more broadly IQIM for creating an exciting environment,
including Ann Harvey,
Michael Beverland,
Ersen Bilgin,
Peter Brooks,
Bill Fefferman,
Alex Kubica,
Shaun Maguire,
Prabha Mandayam,
Sujeet Shukla,
Gorjan Alagic,
Salman Beigi,
Sergio Boixo,
Darrick Chang,
Glen Evenbly,
Steve Flammia,
Lukasz Fidkowski,
Alexey Gorshkov,
Zhengcheng Gu,
Liang Jiang,
Stephen Jordan,
Robert K\"onig,
Yi-Kai Liu,
Nate Lindner,
Spiros Michalakis,
Fernando Pastawski,
Norbert Schuch,
Stephanie Wehner,
and
Beni Yoshida.

I am grateful for comments from Soonwon Choi and Alex Kubica on drafts of these thesis chapters.
\end{acknowledgments}

\begin{abstract}
This thesis addresses whether it is possible to build a robust memory device for quantum information. Many schemes for fault-tolerant quantum information processing have been developed so far, one of which, called topological quantum computation, makes use of degrees of freedom that are inherently insensitive to local errors. However, this scheme is not so reliable against thermal errors. Other fault-tolerant schemes achieve better reliability through active error correction, but incur a substantial overhead cost. Thus, it is of practical importance and theoretical interest to design and assess fault-tolerant schemes that work well at finite temperature without active error correction.

In this thesis, a three-dimensional gapped lattice spin model is found which demonstrates for the first time that a reliable quantum memory at finite temperature is possible, at least to some extent. When quantum information is encoded into a highly entangled ground state of this model and subjected to thermal errors, the errors remain easily correctable for a long time without any active intervention, because a macroscopic energy barrier keeps the errors well localized. As a result, stored quantum information can be retrieved faithfully for a memory time which grows exponentially with the square of the inverse temperature. In contrast, for previously known types of topological quantum storage in three or fewer spatial dimensions the memory time scales exponentially with the inverse temperature, rather than its square.

This spin model exhibits a previously unexpected topological quantum order, in which ground states are locally indistinguishable, pointlike excitations are immobile, and the immobility is not affected by small perturbations of the Hamiltonian. The degeneracy of the ground state, though also insensitive to perturbations, is a complicated number-theoretic function of the system size, and the system bifurcates into multiple noninteracting copies of itself under real-space renormalization group transformations. The degeneracy, the excitations, and the renormalization group flow can be analyzed using a framework that exploits the spin model's symmetry and some associated free resolutions of modules over polynomial algebras.
\end{abstract}

\tableofcontents

\mainmatter

\chapter{Introduction}

The idea of quantum computer dates back at least to Feynman~\cite{Feynman1982Simulating}, who speculated a possibility for exploiting the computational power that Nature allows. The idea raises a deep question. The computation is a manipulation of symbols and numbers according to our logical system. If we can simulate the time evolution of Nature in a controllable and mechanical way, then it is unavoidable to conclude that the present Nature is really computing the future, and that the way she does is essentially the same as our arithmetic. If we cannot simulate what she does, then it means there is a fundamental difference between the time evolution and its artificial simulation by our logic and numbers. Either conclusion must have profound consequences.

An important problem in proving the possibility of a quantum computer is how to suppress decoherence. Shor discovered a scheme in which as long as elementary operations have a low enough error rate one can perform an arbitrarily long computation~\cite{Shor1996Fault-tolerant}. He showed there is a positive constant $\delta$ such that any ideal computation can be simulated by faulty elementary operations if they are close to ideal ones up to precision $\delta$; the scheme effectively reduces the error rate. Thus, the problem of decoherence is solved at least theoretically.

Kitaev proposed yet different scheme, called topological quantum computation, in which elementary operations are physically protected and hence are ideal for all practical purposes~\cite{Kitaev2003Fault-tolerant}. He pointed out that there is a naturally protected subspace in topologically ordered systems in two spatial dimensions. The subspace is accessible by braiding excitations, so-called anyons. As long as the anyons are geometrically well separated at any time step of the computation, the subspace remains decoherence-free. A classical analog is easy to understand. In a magnetic storage medium a bit is encoded into one of two polarizations of little ferromagnets. Each little ferromagnet consists of billions or more electrons whose spins are aligned together. An electron spin may be flipped by some error, but it is energetically unfavorable. Many errors require high energy, and it is very unlikely that the average magnetization would change the sign. The main idea of the encoding quantum information in the topologically ordered system is similar. The information is carried not by local degrees of freedom, but by collective degrees of freedom, where local errors are suppressed by natural means.

A requirement for a system to be useful for the topological quantum computation is that the system must have eigenstates of the same energy that are locally indistinguishable. The states only look different when one has a full description of them; the local reduced density matrices are identical. This property has no classical analog. Systems whose ground states are locally indistinguishable are already found. The fractional quantum Hall systems with a filling fraction $\nu = p/q$ have $q$ degenerate ground states that are locally indistinguishable~\cite{Haldane1985MagneticTranslation,WenNiu1990GroundState}. The local indistinguishability accompanied by a finite energy gap above the degenerate ground states, has a significant consequence that the degeneracy does not split in the thermodynamic limit, even under general local perturbations~\cite{WenNiu1990GroundState,Kitaev2003Fault-tolerant,BravyiHastingsMichalakis2010stability}.
This intrinsic stability underlies the idea to build a quantum computer on topologically ordered systems~\cite{
DennisKitaevLandahlEtAl2002Topological,
JiangBrennenGorshkovEtAl2008Anyonic, 
SarmaFreedmanNayak2005Topologically}.

However, a closer analysis reveals that the topologically ordered system is vulnerable to thermal fluctuations~\cite{AlickiFannesHorodecki2009thermalization}. The collective degrees of freedom are well protected as long as the anyons are far separated, but the thermal fluctuations cause the anyons to propagate randomly throughout the system. The random motions are not a priori suppressed by, for example, energetics. Anyonic systems do not function as protected media as the magnetic media do for classical information storage. That any topologically ordered system has a naturally protected subspace is not entirely true.

We need to separate the problem for further concrete discussions. A computer is loosely divided into two parts: reliable storage and fault-tolerant processing of information. The division is not too fundamental since the storage may require some sort of ancillary information processing, and vice versa. It is a convenient division for the sake of analysis. The storage problem for quantum information concerns a possibility of a quantum analog of classical hard disk drive, which we call a \emph{self-correcting quantum memory}. One asks if there is a system where a subspace is maintained coherently as collective phenomena~\cite{DennisKitaevLandahlEtAl2002Topological,Bacon2006Operator}. The processing problem concerns a wise choice of an elementary operation set and its implementation, and is thus contingent upon the storage scheme. In this thesis we focus on the storage problem.

The result of the present thesis can be summarized as follows. We find a spin model on a three-dimensional cubic lattice, whose excitations are immobile and point-like. Under any perturbation the excitations do not acquire a kinetic term --- any hopping amplitude  vanishes exactly. Moreover, the ground-state subspace is degenerate, exactly in the thermodynamic limit, and no local order parameter can be defined. These properties do not fit into intuitive pictures people had about the topological order. The reason why it is unconventional will be explained below. Given the model, we devise a scheme to use it as quantum memory and compute the storage time. We prove a rigorous lower bound on the storage time that grows with system size up to an optimal value $T_{mem} = e^{\Omega(\beta^2)}$ where $\beta$ is an inverse temperature of a heat reservoir. It should contrast with a conjectured storage time $e^{O(\beta)}$ of any two-dimensional topologically ordered system.

We briefly review how the concept of topological order has emerged, and discuss our results.

\section{Topological order}

The quantum Hall effects are phenomena in which the transversal conductance becomes a locally constant function (plateau) for ranges of perpendicular magnetic field strength. In the integer quantum Hall effect, as the magnetic field is increased, the Hall conductance develops plateaus at quantized values of $n e^2/h$ in the vicinity of field $B = \rho_0 hc/ne$ where $n$ is a small positive integer and $\rho_0$ is the electron number density~\cite{KlitzingDordaPepper1980IQHE}. The quantization is very accurate and universal. The measured conductances from various experiments all agree with one another within a relative error less than a part in a million~\cite{PrangeGirvin1990HallEffect}. The agreement is so remarkable because different experiments do not fine-tune every aspect of experimental setups up to precision $10^{-6}$. Due to its simple reproducibility and high accuracy, the integer quantum Hall effect is now used to define an international standard of resistance~\cite{SI2008}.

The exact quantization can be explained by Laughlin's argument~\cite{Laughlin1981QHE}, refined by Halperin~\cite{Halperin1982QHE}. They consider an adiabatic insertion of magnetic flux near the boundary of the Hall sample. They conclude that the Hall conductance is quantized because extended (delocalized) electronic states near the sample edge whose energies are far from Fermi level are only responsible for the conductance. Thouless {\em et al.}~\cite{TKNN1982} showed that those extended states actually define a topological vector bundle whose invariant, now called TKNN invariant or Chern number, is directly related to the quantized Hall conductance.

The awe of the quantum Hall effects does not end there. Soon after the discovery of the integer quantum Hall effects, another kind of quantization was measured --- the fractional quantum Hall effects~\cite{TsuiStoermerGossard1982FQHE,Laughlin1983FQHE}. The transversal conductance displays many plateaus at fractional multiples of $e^2/h$, not only at integral multiples of $e^2/h$.
An interesting feature is the structure of the ground-state subspace. There are $q$-fold degenerate ground states for a fractional quantum Hall system of filling fraction $\nu = p/q$ defined on a torus, where $p$ and $q$ are co-prime integers~\cite{NiuThoulessWu1985Topological,Haldane1985MagneticTranslation}. By taking a detour through an effective theory, the degeneracy is argued to be a function of topology~\cite{WenNiu1990GroundState}. More specifically, if a fractional quantum Hall Hamiltonian at $\nu = p/q$ is defined on a Riemann surface of genus $g$, the degeneracy is $q^g$.

It is much more interesting that quasi-particles of the fractional quantum Hall system are thought to be {\em anyons}~\cite{Wilczek1982Anyon}. They obey neither bosonic nor fermionic statistics under exchange. Rather, the wave function of two-anyon state may be transformed by a unitary operator if one anyon is transported around the other. It has been conjectured that the nonabelian case, where the unitary is not just a phase factor, could be realized in a fractional quantum Hall system at certain filling fractions $\nu = p/q$~\cite{WillettEtAl1987EvenFQHE,MooreRead1991Nonabelions}. Unfortunately, the nonabelian statistics has not been verified experimentally.

Meanwhile, the concept of topological order had emerged. It was introduced as an abstract notion to describe the fractional quantum Hall systems and spin liquid states~\cite{WenWilczekZee1989Chiral}. A gapped system is generally said to be topologically ordered if the ground-state subspace is degenerate but no symmetry is spontaneously broken, and the degeneracy is robust under any perturbation in the thermodynamic limit~\cite{Wen1991SpinLiquid,ReadSachdev1991LargeN}. Also, the degeneracy as a function of physical space topology and the anyonic quasi-particle statistics are taken as defining characteristics of topological order~\cite{WenNiu1990GroundState,Kitaev2003Fault-tolerant,MooreRead1991Nonabelions}.

Another important yet different characteristic is topological entanglement entropy~\cite{KitaevPreskill2006Topological, LevinWen2006Detecting}. A ground-state wave function of a gapped Hamiltonian is believed to obey an area law. Namely, the von Neumann entropy $S(\rho)$ of the reduced density matrix $\rho$ for a disk region is bounded from above by a constant times the area of the boundary. In our two-dimensional situation, the area is the perimeter of the boundary, so $S(\rho) \le \alpha L$. The topological entanglement entropy is a negative constant correction to the area law; $S(\rho) = \alpha L -\gamma$. Remarkably, the $\gamma$ is insensitive to microscopic details and constant under deformation of the Hamiltonian as long as the deformation does not close the energy gap. This quantity is quite different in nature compared with other characteristics of topological order, since it is computed from a single wave function whereas the others are defined for Hamiltonians.

The topological order can be better understood by studying lattice gauge theories~\cite{Kogut1979LGT} and the toric code model~\cite{Kitaev2003Fault-tolerant}. One of the purposes to introduce lattice gauge theories is to contrast how our new model, called {\em cubic code}, is different from models in a conventional picture. The simplest possible lattice gauge theory is the Ising gauge theory due to Wegner~\cite{Wegner1971IsingGauge}. Ising gauge theory can be defined in any dimensions, but let us focus on the two-dimensional square lattice first. Later we will see a direct relation to the toric code model.

Consider the two-dimensional square lattice with a Ising variable (spin) $Z = \pm 1$ at each edge. For each vertex $v$, let $A(v)$ be an operator that flips four spins around $v$; $A(v) : \pm 1 \mapsto \mp 1$. $A(v)$ is called a {\em gauge transformation}. We can express it as an operator.
\[
 A(v) = X(N,v) X(W,v) X(S,v) X(E,v)
\]
where $X(N,v)$ means the matrix $\begin{pmatrix} 0 & 1 \\ 1 & 0 \end{pmatrix} \begin{matrix} \ket{+1} \\ \ket{-1} \end{matrix}$ 
acting on the north edge of the vertex $v$, and similarly for others.
We identify all states that related by the gauge transformations.
We look for a ``local'' Hamiltonian that is invariant under the action of $A(v)$.
Here, the locality may be ambiguous since we have identified states that differ by the gauge transformations
and formed a Hilbert space that is not a tensor product of local constituents.
However, the locality is still a proper notion with respect to the square lattice.

The gauge invariance allows restricted possibilities of terms in the Hamiltonian.
It is not hard to see that a term in the Hamiltonian must be a product of $Z$'s along a closed loop.
The simplest closed loop is a single plaquette. Thus, the simplest Hamiltonian would be given by the sum over all plaquettes $p$
\begin{equation}
 H = -J \sum_p B(p) = - J \sum_p Z(b,p) Z(r,p) Z(t,p) Z(l,p)
\label{eq:2D-Ising-gauge}
\end{equation}
where $J > 0$ and $B(p)$ is the product of four Ising variables on the bottom, right, top, and left of the plaquette $p$. It is important that the ground-state subspace does not spontaneously break any gauge symmetry. It is a general statement (Elitzur's theorem~\cite{Elitzur1975}) whose proof is not difficult~\cite{Kogut1979LGT,ChesiLossBravyiEtAl2010Thermodynamic}.

If $H$ is defined on a torus with periodic boundary conditions, there are four ground states that are not equivalent under the gauge transformations. This can be understood by visualizing configurations by loops. The configuration $C_0$ with all spins taking $+1$ values is a ground state of $H$. Any equivalent state under the gauge transformations is obtained by applying $A(v)$ at vertices. A single $A(v)$  flips four spins. Imagine connecting those flipped spins by straight lines --- they will form a loop. Applying more gauge transformations $A$ at various vertices, one adds more loops or deforms the loops. What is always true is that the loops are homologically trivial. Indeed, $A$ can be interpreted as a boundary operator in the cellular homology with $\ZZ/2$ coefficient acting on the dual 2-cells. What if we start with a configuration $C_1$ with all spins $+1$ except those $-1$ along a homologically nontrivial loop of the torus? Since $A$ does not alter the homology class of the configuration, we conclude that $C_0$ and $C_1$ cannot be equivalent under the action of $A$. Since there are 4 distinct homology classes of the torus including the trivial one $C_0$, the degeneracy is therefore 4. We emphasize that the homology classes are not locally distinguishable.

One can generalize the model so that it is defined on an arbitrary surface with a triangulation. The gauge transformations will be defined for each vertex, and the Hamiltonian will be a sum of all terms $B$, the product of Ising variables along the perimeter of elementary triangles. The homology description will be valid as well. The degeneracy is a function of homology of the surface.

Returning to the square lattice, we ask how an excited state looks like. It is described by unhappy terms $B(p) = -1$ in $H$. If we flip a spin from a ground state, then the two adjacent plaquette terms will become unhappy. If we flip two spins, say, one on the left and another on the right of a plaquette $p_0$, $B(p_0)$ will remain happy but those on the left plaquette and on the right will not. Generally, if we flip spins from a ground state such that flipped spins form a string on the dual lattice, then only two plaquette terms positioned at the end of the string will be unhappy. Those unsatisfied terms can be isolated at constant energy cost, and appear as the end points of the strings. If the string is extended to the infinity in one direction, there will be a single excitation. This is a topological excitation, in the sense that it cannot be created by a local operator, unlike a pair of nearby excitations. Note that the string that creates excitation has no gauge-invariant meaning. Applying $A(v)$ in the middle of an extended string will deform the string. The homological argument above shows that if two strings with the same end points differ by a homologically trivial cycle, the two strings describe exactly the same excited states.

The properties of the Ising gauge theory we have reviewed here
satisfy several criteria for topological order.
Its degeneracy is a function of topology, and no symmetry is spontaneously broken.
The four-dimensional ground space is stable under gauge-invariant perturbations.
Can we obtain a similar model without the gauge symmetry?
A prescription is to promote gauge transformations to be dynamical~\cite{Kitaev2003Fault-tolerant},
and take the Hilbert space as the tensor product of individual spins.
In other words, one adds gauge fixing terms to the Hamiltonian.
\begin{equation}
 H' = - J \sum_p B(p) - g \sum_v A(v)
\label{eq:H'}
\end{equation}
Since the Hamiltonian $H$ is constructed to be invariant under the gauge transformations, the new quantum Hamiltonian $H'$ is exactly solvable. The ground state is an equal-weight superposition of all equivalent spin configurations under the old gauge transformations. 
Our analysis on the absence of local order parameters and the degeneracy 
using the cellular homology is still valid. 
As the homology classes of spin configuration in the Ising gauge theory were not locally distinguished, 
the quantum ground states of $H'$ are not locally distinguished.
It follows that the degeneracy is not lifted under any local perturbations~\cite{BravyiHastingsMichalakis2010stability}.
Indeed, it is easily checked that the first order degenerate perturbation theory 
is vacuous because any local operator sandwiched between two ground states $\ket{a}$ and $\ket{b}$ 
is proportional to $\langle a | b \rangle$.

The gauge symmetry is not strictly imposed any more, but a particular gauge choice is energetically preferred.
Accordingly, there is one more type of excitations given by unsatisfied $A(v)$ terms, known as $e$-particles.
A violated $B(p)$-term is known as $m$-particle. 
Note that there is a duality between $A$ and $B$ terms. 
$A$ consists of Pauli $X = \sigma^x$ acting on a plaquette on the dual lattice, 
and $B$ consists of Pauli $Z = \sigma^z$ acting on a plaquette on the primary lattice. 
A pair of $m$-particles can be created from a ground state 
by applying ``spin flip'' operators $X$ along a string on the dual lattice.
A pair of $e$-particles can be created from a ground state by applying ``phase flip'' operators $Z$ along a string on the primary lattice. The $e$- and $m$-particles display mutual anyonic statistics. When one makes a complete circle around the other, the wave function acquires a phase factor $-1$. The topological entanglement entropy is nonzero~\cite{KitaevPreskill2006Topological, LevinWen2006Detecting}. We obtained an exactly solvable simple model with topological order, the {\em toric code}.

\section{String operators}

The string operators in the toric code deserve much attention. They are topological objects in that only the homological classes they represent matter. The fact that the strings are extended objects makes it clear how two particles at a distance may interact by braiding. The nontrivial braiding, in turn, implies a nontrivial ground-state subspace. The argument is basically the same as the proof of the degeneracy of the fractional quantum Hall system using magnetic translations~\cite{Haldane1985MagneticTranslation}.

The string picture seems to be correct for all two-dimensional topologically ordered systems. The robust degeneracy implies that there is no local observable that can distinguish different ground states; otherwise, perturbing the system with that local observable will lift the degeneracy. A minimal ``global'' operator whose support is not local would be a string operator stretched across the system. Another conceivable argument is as follows. Consider a region $R$ as large as possible on which no observable can resolve the ground-state subspace. If we assume translation invariance so we can unambiguously speak of thermodynamic limit, then a local operator is any operator with a bounded support. Hence, on a torus geometry, $R$ can be taken as the whole system minus two narrow strips, ensuring $R$ to be a contractible region. Since any operator in $R$ cannot resolve the ground-state subspace, some operator in the strips must be able to resolve it, which suggests the existence of string operators. As we will see in Chapter~\ref{chap:additive-codes}, the argument here can be made rigorous for a class of models. Also, there is an attempt to understand every possible model in two dimensions with the string picture~\cite{LevinWen2003Fermions,LevinWen2005String-net}.

What will happen if we go to three spatial dimensions? Consider a three-dimensional Ising gauge theory. The gauge transformation $A(v)$ is defined for each vertex $v$, and the Hamiltonian will be the sum of all plaquette terms $B(p)$. For the 3D simple cubic lattice, a single spin-flip at an edge will violate four plaquette terms attached to that edge. Many spin flips that form a surface on the dual lattice will violate plaquette terms along its boundary. In general, excited states are caused by \emph{surface} operators and described by loops of unhappy plaquette terms, analogous to domain walls of the 2D Ising model.%
\footnote{An exact mapping actually relates 3D Ising gauge theory with 3D Ising model~\cite{Wegner1971IsingGauge}.}
Now we add gauge fixing terms $A(v)$ into the Hamiltonian to obtain a quantum model, called three-dimensional toric code. As before, there is a new type of excitation, an unsatisfied $A(v)$ term. Let us call $A(v)$ a \emph{star} term. Unlike the plaquette excitations that form loops, the star excitation is attached to a \emph{string} operator along a path in the primary lattice. The star excitations can therefore be isolated.

There are two types of operators acting on the degenerate ground-state subspace. A different ground state is obtained by flipping all spins on a plane that wraps around the system. When a finite system with periodic boundary conditions are considered, the locations of the flipped spins form a nontrivial homological 2-cycle of 3-torus. The other type of operator on the ground-state subspace is the string operator. If we apply the string operator along a nontrivial homological 1-cycle of the 3-torus, then the star excitation does not appear since there is no end points, and a ground state is mapped to a different one. There are three closed surface operators and three closed string operators, as the homology group of 1-cycles and 2-cycles are generated by three elements, respectively.

The closed surface operator $\bar X$ and the closed string operator $\bar Z$ generate a nontrivial algebra, because, when a pair of $\bar X$ and $\bar Z$ share a common support, the intersection is just a single spin on which the pair anti-commute. Of course, $\bar X$ and $\bar Z$ are topological objects and there are many equivalent representatives. One can easily check that in for any equivalent representatives the algebraic structure does not change.

The (closed) string and surface operators seem to be universal for any topologically ordered models. If the ground-state subspace is robust under any perturbations, the algebra generated by the operators acting on the ground-state subspace must also be robust. Note that operators of extended overlapping support may not define a robust algebra. 
For example, $\bigotimes_{i=1}^n \sigma^x$ and $\bigotimes_{i=1}^n \sigma^z $ commute if $n$ is even, but anti-commute if $n$ is odd.
In order to make a well-defined algebra, it is the most straightforward for the operators of extended support to intersect at a point-like (zero-dimensional) support. In other words, an $n$-dimensional operator would have a nontrivial commutation relation with a $(D-n)$-dimensional operator, where $D$ is the total spatial dimension. If there were a three-dimensional model with all surface-like operators acting on the ground-state subspace, they must intersect along a one-dimensional line. The length of the one-dimensional line is not topological, and it would be hard for them to form a robust algebra. 
Therefore, in three dimensions, it seems necessary for the string operators to exist.
In four spatial dimensions it is possible to construct a model that has a robust ground-state subspace such that all operators on it have two-dimensional support~\cite{DennisKitaevLandahlEtAl2002Topological}.

We remark that the existence of string operators implies that point-like excitations at the ends of the string are {\em mobile} generically. Although the $e$- and $m$-particles of the two-dimensional toric code, being eigenstates of $H'$, are stationary, the excitations will acquire a kinetic or hopping term under generic perturbations. This is the reason why the topological order at finite temperature is said to be not stable. At nonzero temperature there must be a nonzero density of the mobile excitations. Their spatial fluctuation induced by the thermal interaction is strong enough to disorder the ground state completely. Topological entanglement entropy calculation at nonzero temperatures supports this intuition~\cite{CastelnovoChamon2007Entanglement}. In terms of a measure how difficult it is to generate a state, the Gibbs state of some topologically ordered model with string operators is not too different from a trivial product state~\cite{Hastings2011warmTQO}. Moreover, the relaxation time towards the Gibbs state of the two-dimensional toric code is only a constant independent of the system size~\cite{AlickiFannesHorodecki2009thermalization}.

\section{Quantum codes and a new model}

In our discussion so far, the topological order is characterized by a collection of very compelling properties of Hamiltonians or ground-state subspaces. However, it is not too clear which one is more fundamental. Even, it is not clear whether all those characteristics should appear all together. The quantum Hall effects~\cite{PrangeGirvin1990HallEffect} and topological insulators~\cite{HasanKane2010TIReview} should serve as reference systems. It appears that the local indistinguishability of ground-state subspace is the most mathematically tractable definition of the topological order. We adhere to this definition. To obtain enough intuition, we would like to study toy models arising from quantum error correcting codes.

Shor's fault-tolerant scheme~\cite{Shor1996Fault-tolerant} is based on the discovery of quantum error correcting codes. He demonstrated that there exists a subspace in a many-qubit Hilbert space such that local errors can be detected and corrected with a high probability without disturbing encoded (logical) quantum state. Measurements inevitably disturb the system, but a trick is that the logical state is encoded in the entanglement of the many qubits. It is essential that different encoded states look exactly the same to the local errors; they are locally indistinguishable.

It did not take too long for people to realize that the Shor's error correcting code can be generalized and studied in an analogous way that one studies classical error correcting codes~\cite{CalderbankShor1996Good, Steane1996Multiple, Gottesman1996Saturating, CalderbankRainsShorEtAl1997Quantum}. The connection is due to the observation that the encoded state can be described as the eigenstate of pairwise commuting tensor products of Pauli spin-$\half$ matrices; the state is stabilized by commuting operators. The commutativity is important because, otherwise, we cannot speak of common eigenstate. The encoded state is required to be highly entangled, and hence avoids in general an efficient classical description. However, the stabilizer operator language tells us that some class of highly entangled states has a simple description, which is enough to ensure the local indistinguishability. We explain this in detail in Chapter~\ref{chap:additive-codes}.

The toric code Hamiltonian $H'$ in Eq.~\eqref{eq:H'} can be viewed as a quantum code. The terms $A(v)$ and $B(p)$ consist of Pauli matrices and commute with each other. If the coupling constants $J$ and $g$ are both positive, then a ground state is a common eigenstate of eigenvalue $+1$. (Having $+1$ eigenvalue ensures $H'$ is not frustrated.) 
In general, if a code is defined by a local commuting Pauli operators, we can write a corresponding Hamiltonian given by the negative sum of all the commuting Pauli operators. If the code has a good error correcting capability, the corresponding Hamiltonian must have a locally indistinguishable ground-state subspace. 

As we have seen earlier, the toric code Hamiltonian captures essential characteristics of the topological order. It is one of good motivations to study Hamiltonians with commuting Pauli operators in order to understand the topological order. One might be uneasy with the fact that the Hamiltonian $H'$ is a four-body interaction. This ``unrealistic'' interaction may be effective at low energies, if, for example, a realistic Hamiltonian is highly frustrated from which the effective $H'$ is derived~\cite{ReadSachdev1991LargeN, Kitaev2006Anyons}. We will be not concerned about the order of the interaction as long as it is local.

One of our guiding problems is the possibility of self-correcting quantum memory. We wish to have a system whose ground-state subspace is locally indistinguishable so that local errors do not corrupt an encoded state. In addition, we demand that there be a physical mechanism such that local errors do not accumulate. The string operators of the previous section are against our goal. Since they make excitations mobile, the physical mechanism to prohibit the error accumulation cannot be achieved. Hence, in two dimensions it seems impossible to have a self-correcting quantum memory. Indeed, one can rigorously prove the existence of the strings operators for a class of models in two dimensions, starting from the local indistinguishability assumption~\cite{BravyiTerhal2009no-go, HaahPreskill2012tradeoff, LandoncardinalPoulin2013nogo}. We noted in the previous section that there is a four-dimensional topologically ordered model lacking the string operators~\cite{DennisKitaevLandahlEtAl2002Topological}. This model indeed functions as a self-correcting quantum memory below a critical nonzero temperature~\cite{AlickiHorodeckiHorodeckiEtAl2010thermal}.

A more realistic three-dimensional case is hence interesting. Based on the dimensional duality of the operators acting on the degenerate ground-state subspace, a no-go theorem seemed plausible~\cite{Yoshida2011feasibility}. In this thesis we present a concrete counterexample, cubic code, to this intuition.
\[
\drawgenerator{ZI}{ZZ}{IZ}{ZI}{IZ}{II}{ZI}{IZ} 
\quad
\drawgenerator{XI}{II}{IX}{XI}{IX}{XX}{XI}{IX}
\]
The cubic code is a quantum error correcting code on the simple cubic lattice defined by local Pauli operators acting on elementary cubes. There are two qubits (spin-$\half$) at each site. The two-letter symbol such as $XI$ means a tensor product $\sigma^x \otimes I$, and $ZZ$ means $\sigma^z \otimes \sigma^z$, etc. Each diagram represents a tensor product of eight Pauli matrices. The Hamiltonian is the negative sum of two diagrams over all elementary cubes. As required, the model has a degenerate ground-state subspace that is locally indistinguishable. Most importantly we can prove that there does not exist any string operator. However, it admits isolated point-like topological excitations. They do not appear as the end points of the string, but as the vertices of \emph{fractal} operators. The fractal operators is supported on a self-similar structure. The absence of string operators implies that the point-like topological excitations are immobile; any hopping amplitude vanishes exactly. It is also interesting that the degeneracy depends sensitively on the number-theoretic property of the lattice size, illustrated in Figure~\ref{fig:k-numerics} on page~\pageref{fig:k-numerics}. The cubic code is a conceptually new phase of matter.

It is necessary to explain more how the vanishing hopping amplitude and the isolated point-like topological excitation are simultaneously possible. By a topological excitation, or a {\em charge} for short, we mean a localized excitation that cannot be created from a ground state by any finitely supported operator. A kink in the 1D Ising model is the simplest example. The isolation of a charge in the cubic code is possible under a self-similar construction as follows. In Figure~\ref{fig:pyramid}, the first figure shows how excitations look like when an $X$ error occurs on a ground state. The second figure show the excitations given four $X$ errors. Note that when two charges are at the same position, they cancel with each other. The configurations are similar by ratio 2, so one can construct even larger configuration where there are only four excitations that are far separated. In the limit of this process there is an isolated charge. The reason why the hopping amplitude vanishes can be heuristically understood by the construction. It is impossible to have a configuration where only one charge is moved away while the other three are held fixed. Such a configuration is forbidden in the spectrum of the Hamiltonian. The transition amplitude to a non-existing state from any state must be vacuous. We later prove it rigorously. A classical model where this happens was discovered by Newman and Moore~\cite{NewmanMoore1999Glassy}. Our model is inherently quantum with a locally indistinguishable ground-state subspace and has a gapped energy spectrum that is robust under generic perturbations~\cite{BravyiHastingsMichalakis2010stability}.

\begin{figure}[t]
\centering
\begin{tabular}{cc}
 \includegraphics[width=2cm]{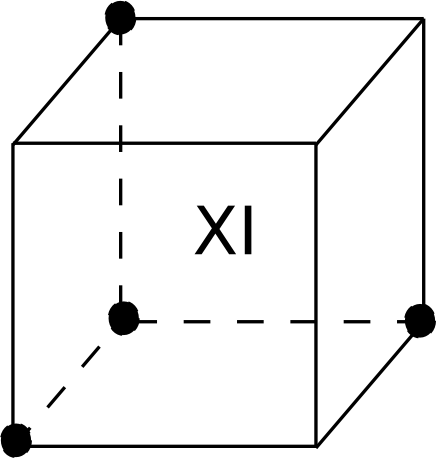} &
 \includegraphics[width=3cm]{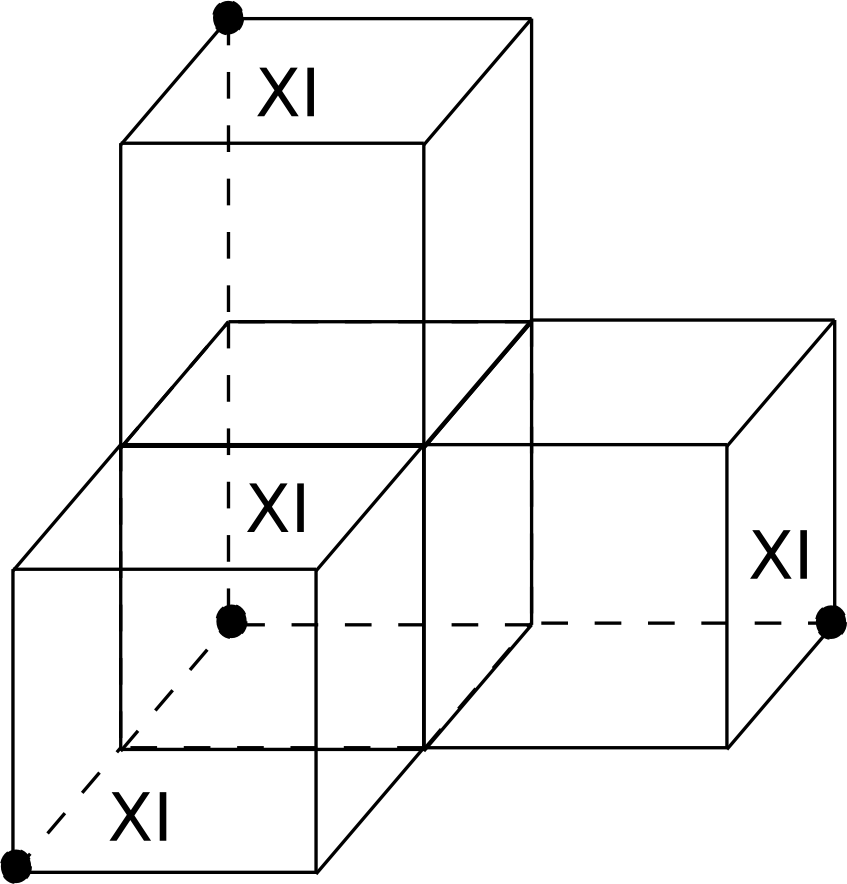}
\end{tabular}
\caption{Isolating a topological excitation. The cubes are in the dual lattice; each vertex represents an elementary cube in the primary lattice.}
\label{fig:pyramid}
\end{figure}

We show that the cubic code model can be used as a quantum memory at nonzero temperatures in the following sense. We develop a well-defined read-out procedure (decoder) with which a ground state maintains coherence for time proportional to $L^{c \beta}$ where $L$ is a linear system size and $\beta$ is an inverse temperature. This statement is valid only if $L \le e^{c' \beta}$. Roughly speaking, this is because of a large entropic contribution from the point-like excitations. For small system size, the entropic contribution is small. At an optimal system size, the memory time is $e^{c c' \beta^2}$. A natural question is then whether we could remove all such point-like excitations. We answer this question negatively: There must be an isolated point-like excitation in any topologically ordered three-dimensional quantum code model if it is translationally invariant. The proof is based on a formalism using translation group algebras and free resolutions of finitely generated modules.

\section{Summary of chapters}

In Chapter~\ref{chap:additive-codes}, we explain a general structure of a class of quantum error correcting codes, called additive or stabilizer codes. It is emphasized that the additive code is described by a binary vector space. A simple counting of vector space dimensions yields {\em cleaning lemma}. It states that the number of independent logical operators on complementary regions add up to the total number of independent logical operators of the code. Combining the geometric locality of two-dimensional Hamiltonian with the cleaning lemma, we prove a trade-off theorem about the support of the logical operators. In particular, this implies that only string operators act on the ground-state subspace in two-dimensional quantum codes. 
The content is published in \cite{HaahPreskill2012tradeoff}.

In Chapter~\ref{chap:alg-theory}, we develop a formulation of translationally invariant quantum codes. We observe that they have a succinct description by a matrix $\sigma$ over the translation group algebra which is commutative. The local indistinguishability is interpreted as a vanishing homology of a complex defined by $\sigma$. Local unitary transformations and coarse-graining are described by simple matrix operations. The ground-state degeneracy is approximated as the number of points on an algebraic variety defined by $\sigma$ over finite fields. Fractal operators and point-like excitations are defined, and the set of all point-like excitations modulo locally created ones is identified with a specific module.
The content is published in \cite{Haah2012PauliModule}.

In Chapter~\ref{chap:lowD-codes}, we use the formalism of Chapter~\ref{chap:alg-theory} to derive consequences of physical dimensionality. The translation invariance is imposed. One-dimensional codes are classified completely. Up to local unitary transformations, any code decomposes into finitely many 1D Ising models. In two dimensions, we prove that there are finitely many types of topological excitations and they are all appear as end points of string operators.
In three dimensions, we prove that there must exist a point-like topological excitation. If the ground-state degeneracy is constant independent of system size, we show that the point-like topological excitations are attached to string operators. Examples are presented for which our formalism is useful.
The content is published in \cite{Haah2012PauliModule}.

In Chapter~\ref{chap:cubic-code}, we explain how the cubic code is found. It is a result of an exhaustive but systematic search. We prove important properties of the cubic code such as the topological order, ground-state degeneracy, and the absence of string operators. The cubic code's thermal partition function is computed to show that the free energy density is smooth in the thermodynamic limit as a function of temperature. An entanglement renormalization group flow is presented. We find that the cubic code (A) bifurcates into itself (A) and another model B. The new model B shares all important properties of A, but is different from A. Under a further real-space renormalization, B bifurcates into two copies of itself. 
A part of the content in this chapter is published in \cite{Haah2011Local}. Our presentation in this thesis is more succinct than that in the published paper, due to algebraic methods of Chapter~\ref{chap:alg-theory} developed more recently.

In Chapter~\ref{chap:consq-no-strings}, we prove theorems implied by the absence of string operators. The theorems quantifies the energy landscape of the Hamiltonian on the ``land'' of energy eigenstates. Two states are considered to be ``close'' in the land if one state can be transformed to the other by a single spin operator; if it is necessary to apply many spin operators, the two states are considered to be far apart in the land. On this land of states, imagine hills whose height is given by the energy of the state. The landscape of this land is analogous to a potential energy barrier in a quantum mechanical tunneling problem. We define the {\em energy barrier} to be the height of the lowest hill on any path in the land connecting two states. We prove that between two ground states there exists an energy barrier larger than the logarithm of the system size $L$, and the distance in the land is at least $L^\gamma$ with $\gamma > 1$. A closely related statement reads that between a ground state and a state where a topological excitation is isolated from others by distance $R$, there exists an energy barrier $\ge \log R$, and the distance in the land is $\ge R^\gamma$ with $\gamma >1$.
The content is published in \cite{BravyiHaah2011Energy}.

In Chapter~\ref{chap:RG-decoder}, we design a decoding algorithm. Any physical memory cannot be pristine after contact with an error source. In order to retrieve correct information for the next step of processing, one needs to map the affected state to the ground state. For a ferromagnet this decoding process amounts to measuring average magnetization. For quantum codes, the first step is to detect errors. This is done by measuring the terms in the commuting Hamiltonian. The next step is to guess probable types and locations of errors, which we specifically call ``decoding algorithm.'' Once error locations and types are identified, the final step is to undo the errors. A good decoder should map corrupt states to its original pristine state with high probability. Our algorithm uses a hierarchy of subroutines which borrows ideas from renormalization group. It is widely applicable and is efficient with running time $O(V\log V)$ where $V$ is the volume of the system. Moreover, we prove that there is a positive critical error probability, called a threshold, below which the decoding algorithm succeeds asymptotically perfectly in the limit of large system size. Our decoder is the first decoder that admits an efficient implementation and a rigorous threshold theorem.
The content is published in \cite{BravyiHaah2011Memory}.

In Chapter~\ref{chap:q-mem}, we directly assess the performance of the cubic code as a robust quantum memory at nonzero temperature. One ingredient is to show that the decoding algorithm of Chapter~\ref{chap:RG-decoder} corrects errors of low energy barriers. Here, the energy barrier of an error is defined in the same way as above --- the height of lowest hill in the energy landscape on any path connecting a ground state and the state affected by the error. Another ingredient is an analysis of Markovian master equation exploiting the fact that there exists a good decoder $\Phi_{ec}$, a trace preserving completely positive quantum operation. An analytic bound on the trace distance between the initial state $\rho(0)$ and the error-corrected time-evolved state $\Phi_{ec}(\rho(t))$ is given as
\[
\| \rho(0) - \Phi_{ec}(\rho(t)) \| \le t \frac{(1+e^{-\beta})^V}{V^{\beta}}
\]
where $V$ is the system volume and $\beta$ is the inverse temperature, neglecting all unimportant constant coefficients. If $V \le e^\beta$, the bound says the fidelity of error-corrected $\Phi_{ec}(\rho(t))$ remains close to $\rho(0)$ until $t \sim V^\beta$. We complement the bound with a numerical simulation, which suggests that our bound is optimal up to constant coefficients.
The content is published in \cite{BravyiHaah2011Memory}.

\chapter{Additive quantum codes}
\label{chap:additive-codes}

The theory of quantum error correcting code
is an important cornerstone for fault-tolerant quantum computers~\cite{Shor1996Fault-tolerant,Preskill1998Fault-tolerant}.
One encodes states one wish to compute about into logical many-qubit states
in such a way that errors that are likely to happen will be correctable.
Thus, at the logical level an effective error rate is much smaller than the physical error rate.
Computation can be carried on this logical level,
and the overhead of error correction can be controlled, not to overwhelm the promised computation.
After the discovery of the very first quantum error correcting code by Shor~\cite{Shor1995nine},
people have quickly realized analogy between classical error correcting codes and a class of quantum error correcting codes%
~\cite{CalderbankRainsShorEtAl1997Quantum,Gottesman1996Saturating},
which are now known as additive or stabilizer codes.
Arguably, it is the most studied class of codes, and is the main object of the present thesis.
The connection between the classical codes and quantum additive codes is provided
by the parameterization of basis operators (Pauli matrices) by binary numbers.
In this chapter we review and exploit this correspondence.

The chapter is organized as follows.
Section~\ref{sec:pauli-group-symplectic-vector-space} provides a convenient and important viewpoint
under which the multiplicative group of all Pauli matrices are described by a vector space over the binary field.
Section~\ref{sec:additive-stabilizer-codes} builds on this viewpoint and explains how to choose a subspace
of a many qubit Hilbert space, which we hope to have a capability to correct errors.
Section~\ref{sec:cleaning-lemma} is devoted to derive an equation, called a cleaning lemma,
that relates the numbers of logical operators 
that can be supported on complementary regions and the total number of logical qubits.
It is remarked that the cleaning lemma implies 
that any error occurring within a region, where no nontrivial logical operator can be supported,
can actually be corrected by a physical operation.

The last section~\ref{sec:subsystem-tradeoff}
presents an application of the cleaning lemma 
which gives a constraint on the geometric shape of logical operators of local codes on lattices.
Most importantly, it is proved that in two-dimensional lattice codes with geometrically local generators with large code distance,
all logical operators have representatives supported on narrow strips.
It had been known by Bravyi and Terhal~\cite{BravyiTerhal2009no-go} 
that there exists a nontrivial logical operator supported on a strip.
Our conclusion extends this result by finding the geometric shape of \emph{all} logical operators.
If the minimal weight logical operator has weight proportional to the linear dimension of a two-dimensional system,
then all logical operator can be found on narrow strips~\cite{HaahPreskill2012tradeoff}.

Note that there are many interesting quantum codes that are not necessarily additive~\cite{
RainsHardinShorSloane1997Nonadditive,
Kitaev2003Fault-tolerant,
LevinWen2005String-net},
but they are out of the scope of this thesis.

\section{Pauli group as a symplectic vector space}
\label{sec:pauli-group-symplectic-vector-space}

The Pauli matrices
\[
 \sigma^x = \begin{pmatrix} 0 & 1 \\ 1 & 0 \end{pmatrix}, \quad
 \sigma^y = \begin{pmatrix} 0 & -i \\ i & 0 \end{pmatrix}, \quad
 \sigma^z = \begin{pmatrix} 1 & 0 \\ 0 & -1 \end{pmatrix}
\]
satisfy
\[
 \sigma_a \sigma_b = i \varepsilon_{abc} \sigma_c, \quad
 \{ \sigma_a , \sigma_b \} = 2 \delta_{ab}.
\]
Thus, the Pauli matrices together with scalars $\pm 1, \pm i$
form a group under multiplication.
Given a system of qubits,
the set of all possible tensor products of the Pauli matrices form a group,
where the group operation is the multiplication of operators.
If the system is infinite, physically meaningful operators
are those of finite support, i.e.,
acting on all but finitely many qubits by the identity.
We shall only consider this Pauli group of finite support,
and call it simply the {\bf Pauli group}.
An element of the Pauli group is called a {\bf Pauli operator}.
When finitely many qubits are considered,
a Pauli operator is of course an \emph{arbitrary} tensor product of Pauli matrices.

Since any two elements of the Pauli group either commute or anti-commute,
ignoring the phase factor altogether, one obtains an {\em abelian} group.
Moreover, since any element $O$ of the Pauli group satisfies $O^2 = \pm I$,
an action of $\mathbb{Z}/2\mathbb{Z}$ on Pauli group modulo phase factors
$P / \{\pm 1, \pm i\}$ is well-defined,
by the rule $ n \cdot O = O^n$ where $n \in \mathbb{Z}/2\mathbb{Z}$.
For $\FF_2 = \mathbb{Z}/2\mathbb{Z}$ being a field,
$P / \{\pm 1, \pm i\}$ becomes a vector space over $\FF_2$.
The group of single-qubit Pauli operators up to phase factors
is identified with the two-dimensional $\FF_2$-vector space.
If $\Lambda$ is the index set of all qubits in the system,
the whole Pauli group up to phase factors is the direct sum $\bigoplus_{i \in \Lambda} V_i$,
which we call {\bf Pauli space},
where $V_i$ is the vector space of the Pauli operators for the qubit at $i$.
Explicitly, $I = (00), \sigma^x = (10), \sigma^z = (01), \sigma^y = (11)$.
A multi-qubit Pauli operator
is written as a finite product of the single-qubit Pauli operators,
and hence is written as a binary string in which all but finitely many entries are zero.
A pair of entries of the binary string describes a single-qubit component in the tensor product expression.
The multiplication of two Pauli operators corresponds to entry-wise addition of the two binary strings modulo 2.

The commutation relation may seem at first lost,
but one can recover it by introducing a symplectic form~\cite{CalderbankRainsShorEtAl1997Quantum}.
Let
\[
 \lambda_1 = \begin{pmatrix} 0 & 1 \\ -1 & 0 \end{pmatrix}
\]
be a symplectic form on the vector space $(\FF_2)^2$ of single-qubit Pauli operators.%
\footnote{The minus sign is not necessary for qubits, but is for qudits of prime dimensions.}
One can easily check that the commutation relation of two Pauli matrices $O_1, O_2$
is precisely the value of this symplectic form
evaluated on the pair of vectors representing $O_1$ and $O_2$.
Two multi-qubit Pauli operator (anti-)commutes, if and only if
there are (odd) even number of pairs of the anticommuting single-qubit Pauli operators
in their tensor product expression.
Therefore, the two Pauli operator \mbox{(anti-)}commutes
precisely when the value of the direct sum of symplectic form $\bigoplus_{q \in \Lambda} \lambda_1$ is \mbox{(non-)}zero. 
$\Lambda$ could be infinite but the form is well-defined since any vector representing a Pauli operator is of finite support.
We shall call the value of the symplectic form the {\bf commutation value}.

\begin{rem}
For systems of qudits, a group corresponding to the Pauli group in the qubit case
is the so-called generalized Pauli group%
~\cite{Knill1996Non-binary,
Knill1996Group,
Rains1999Nonbinary,
Gottesman1999qudit}.
It is the set of all tensor products of powers of $d \times d$ matrices
\[
 X_d = 
\begin{pmatrix}
 0 & 0 & \cdots & 0 & 1     \\
 1 & 0 & \cdots & 0 & 0     \\
 0 & 1 & \cdots & 0 & 0     \\
   &   & \ddots &   & \vdots\\
   &   & \cdots & 1 & 0 
\end{pmatrix} 
\quad \text{and} \quad
Z_d = 
\begin{pmatrix}
 1 &        &         &        & \\
   & \omega &         &        & \\
   &        &\omega^2 &        & \\
   &        &         & \ddots & \\
   &        &         &        & \omega^{d-1}
\end{pmatrix} ~\left( \omega = e^{2\pi i/d} \right),
\]
which satisfy
\[
 X_d Z_d = \omega^{-1} Z_d X_d .
\]
Hence, any generalized Pauli operator on a single qudit is a product $X_d^n Z_d^m$.
The abelianized generalized Pauli group is identified with $P = (\ZZ / d \ZZ)^2$, a module over $\ZZ/d\ZZ$.
If $d$ is a prime number, $\ZZ/ d \ZZ \cong \FF_d$ is a field and the $P$ is a vector space over $\FF_d$.
The commutation relation can be recovered by the symplectic form $\lambda_1$.
The generalization to a system of qudits is straightforward.

Most statements in this thesis are true or easily extended for qudits with prime dimensions.
An exception is Levin-Wen fermion model of Example~\ref{eg:Levin-Wen-fermion-model}.
\end{rem}

Since the Pauli group can be effectively described 
by a vector space equipped with a symplectic form,
it is worth studying symplectic vector spaces in general.
Any vector space in this section is with respect to some fixed field $\FF$.

Let $V$ be a (finite dimensional) vector space. A bilinear form $\lambda : V \times V \to \FF$ 
is {\bf symplectic} or {\bf alternating} if
\[
 \lambda( v, v ) = 0
\]
for any $v \in V$. If follows that
\[
 \lambda(v,w) = -\lambda(w,v)
\]
since $\lambda(v+w,v+w)= \lambda(v,v) + \lambda(v,w) + \lambda(w,v) + \lambda(w,w) = 0$.
Two vectors $v,w$ are said to be {\bf orthogonal} if $\lambda(v,w) = 0$.
If any two vectors are orthogonal to each other, the symplectic space is said to be {\bf null}.
If for any vector $v$ there exists $w$ such that $\lambda(v,w) \neq 0$,
the symplectic space is said to be {\bf hyperbolic} and $\lambda$ {\bf non-degenerate}~\cite{Lang}.

Given any basis of a finite dimensional symplectic space $V$,
one can find a {\bf canonical} basis $\{ v_1, w_1, v_2, w_2, \ldots, v_n, w_n, u_1, \ldots, u_{n'} \}$ such that
\begin{align*}
 \lambda(v_i, w_j ) = \begin{cases} 1 & \text{if } i=j , \\ 0 & \text{otherwise,} \end{cases}
  \quad \text{and} \quad \lambda( u_i, t ) = 0 \text{ for any } t \in V .
\end{align*}
The canonical basis depends on the order of the basis one starts with, and is \emph{not} unique.
Under the canonical basis the symplectic form has a matrix representation
\[
 \lambda = 
\begin{pmatrix}
 0  & 1 &   &   &   & \\
 -1 & 0 &   &   &   & \\
    &   & 0 & 1 &   & \\
    &   & -1& 0 &   & \\
    &   &   &   & 0 & \\
    &   &   &   &   & \ddots
\end{pmatrix}
\]
whose rank is $2n$.
A Gram-Schmidt process for a usual Hermitian inner product space yields a constructive proof of this claim.
The process is inductive:
\begin{enumerate}
 \item If the given basis is $\{b_i\}$ of $V$, set $v_1 := b_1$.
 \item Choose any basis element $b_j$ such that $\lambda(v_1, b_j) \neq 0$.
        (If one cannot find such $b_j$, start over with a different choice of $v_1$.
        If one cannot eventually find an appropriate $v_1$, then declare $u_i = b_i$ for all $i$; the space is null.)
 \item Set $w_1 := b_j$ and normalize $v_1$ in order to have $\lambda(v_1,b_j) = 1$.
      Reorder the index of $b_j$, so $w_1 = b_2$.
      Now suppose, we have a canonical basis $\{ v_i, w_i \}_{i=1}^{m}$ for a hyperbolic subspace $W_m$ of $V$.
 \item Replace $b_{2m+j}$ ($j\ge 1$) with
 \begin{align*}
  b_{2m+j}' := b_{2m+j} + \sum_{i=1}^m \left( \lambda( b_{2m+1}, v_i ) w_i - \lambda( b_{2m+1}, w_i ) v_i \right) .
 \end{align*}
       One sees that the $b_{2m+j}'$ are orthogonal to $W_m$ and still linearly independent.
 \item Iterate 1-4 with $\lspan _\FF \{ b'_{2m+j} | j \ge 1 \}$.
\end{enumerate}
From the algorithm, we have a structure theorem for finite dimensional symplectic vector spaces.
\begin{prop}
Let $V$ be a finite dimensional vector space equipped with a symplectic form over any field.
Then,
$V$ is a direct sum of a hyperbolic subspace and a null subspace.
In particular, if the symplectic form is non-degenerate,
then $V$ must be even dimensional.
\label{prop:structure-symplectic-space}
\end{prop}
\noindent
It is important that \emph{the abelianized Pauli group is non-degenerate symplectic} over $\FF_2$,
since any Pauli operator anticommutes with some Pauli operator.

As noted earlier, there are many canonical bases. The linear transformations that connect different bases
are called {\bf symplectic transformations}, i.e., $T$ is a symplectic transformation if
\[
 T^T \lambda_q T = \lambda_q = \begin{pmatrix} 0 & \id_q \\ -\id_q & 0 \end{pmatrix} .
\]
The symplectic transformation decomposes into a composition of three elementary transformations%
~\cite{CalderbankRainsShorEtAl1997Quantum,KitaevShenVyalyi2002CQC},
as any general linear transformation decomposes into a composition of row operations and scalar multiplications
by Gauss elimination.
For notational clarity,
define $E_{i,j}(a)\ (i \neq j)$ to be the row-addition elementary $2q \times 2q$ matrix
\[
 \left[ E_{i,j}(a) \right]_{\mu \nu} = \delta_{\mu \nu} + \delta_{\mu i} \delta_{\nu j} a
\]
where $\delta_{\mu \nu}$ is the Kronecker delta and $a \in \FF_2$ is a scalar.
The following are {\bf elementary symplectic transformations}:
\begin{itemize}
 \item (Hadamard) $E_{i,i+q}(-1) E_{i+q,i}(1) E_{i,i+q}(-1)$ where $1 \le i \le q$,
 \item (controlled-Phase) $E_{i+q,i}(a)$ and $1 \le i \le q$,
 \item (controlled-NOT) $E_{i,j}(a) E_{j+q,i+q}(-a)$ where $1 \le i \ne j \le q$.
 \item (controlled-NOT-Hadamard) $E_{i+q,j}(a) E_{j+q,i}(a)$ where $ 1 \le i \ne j \le q$.
\end{itemize}
The fourth one is a combination of the first and the third.

\begin{prop}
\cite{CalderbankRainsShorEtAl1997Quantum}
The elementary symplectic transformations generate the group of all symplectic transformations.
\end{prop}
\label{prop:elem-sym-transformation-generates-whole-BlockCodeCase}
\begin{proof}
It suffices to prove that an arbitrary symplectic transformation $T$ is
a finite composition of elementary ones. $T$ is a $2q \times 2q$ matrix over $\FF_2$.
Let us write $T \cong T'$ if two matrices are transformed by elementary symplectic transformations.
Since any row operation in the upper half block of $T$
can be compensated by an appropriate row operation in the lower half block
by controlled-NOT (CNOT).
One can then transform the upper half block into the reduced row echelon form.
Since $T$ has rank $2q$, the upper half block must have rank $q$.
Therefore,
\[
 T \cong \begin{pmatrix} \id & * \\ M & * \end{pmatrix}
\]
where $M$ and $*$ are all $q \times q$.
Using CNOT-Hadamard, one can eliminate the first column of $M$.
Then, since $T^T \lambda_q T = \lambda_q$, the first row of $M$ must be zero.
Inductively, one can completely eliminate all entries of $M$.
\[
 T \cong \begin{pmatrix} \id & L \\ 0 & N \end{pmatrix}
\]
The equation $T^T \lambda_q T = \lambda_q$ now implies that $N = \id_q$.
$L$ can be made zero by a similar transformations as $M$ was made zero.
Thus, we have
\[
 T \cong \begin{pmatrix} \id & 0 \\ 0 & \id \end{pmatrix}.
\]
An arbitrary symplectic transformation $T$ is equivalent to the trivial transformation by elementary symplectic transformations.
\end{proof}

As one can easily see, the elementary symplectic transformations are induced by the following unitary operators.
\[
 \mathrm{CNOT} =
 \begin{pmatrix}
  1 & 0 & 0 & 0 \\
  0 & 1 & 0 & 0 \\
  0 & 0 & 0 & 1 \\
  0 & 0 & 1 & 0
 \end{pmatrix}
 \begin{matrix}
  \ket{00} \\ \ket{01} \\ \ket{10} \\ \ket{11}
 \end{matrix} ,
\quad
 \mathrm{Hadamard} = 
 \frac{1}{\sqrt 2} \begin{pmatrix} 1 & 1 \\ 1 & -1 \end{pmatrix}
 \begin{matrix} \ket 0 \\ \ket 1 \end{matrix} ,
\quad 
 \mathrm{Phase} = 
 \begin{pmatrix} 1 & 0 \\ 0 & \sqrt{-1} \end{pmatrix}
 \begin{matrix} \ket 0 \\ \ket 1 \end{matrix} 
\]

\section{Additive/stabilizer codes}
\label{sec:additive-stabilizer-codes}

An {\bf additive}~\cite{CalderbankRainsShorSloane1998GF4} 
or {\bf stabilizer}~\cite{Gottesman1996Saturating} {\bf code}
is a quantum code defined by the common eigenspace, called {\bf code space},
of \emph{commuting Pauli operators}, called {\bf stabilizers}, of eigenvalue $+1$
acting on many physical qubits.
It is required that the group of stabilizers should not contain $-I$
in order for the code space not to be zero.
The number of physical qubits $q$ is called {\bf code length}.
From the relation of the group of Pauli matrices and the binary linear space,
each additive or stabilizer code corresponds to a unique null-symplectic subspace of 
the abelianized Pauli group of $q$ qubits.
(The requirement that the stabilizer group should not contain $-I$ is however an independent condition.)
If the null-symplectic subspace is spanned by columns of a matrix $\sigma$ over $\FF_2$,
then the nullity is expressed by a matrix equation
\begin{equation}
 \sigma^T \lambda_q \sigma = 0
\label{eq:nullity}
\end{equation}
where
\[
 \lambda_q = \bigoplus_{i=1}^q \lambda_1 = \begin{pmatrix} 0 & \id \\ -\id & 0 \end{pmatrix}
\]
By fixing the form of $\lambda_q$,
we also fix a convention for the abelianized Pauli group.
The first $q$ components represents Pauli matrices $\sigma^x$ and the second $q$ components $\sigma^z$.

In view of Eq.~\eqref{eq:nullity},
an additive code is a classical linear code with an additional nullity condition.
Calderbank and Shor~\cite{CalderbankShor1996Good}, and Steane~\cite{Steane1996Multiple}
(CSS) proposed a solution to this matrix equation of form
\[
 \sigma = \begin{pmatrix} C_1 & 0 \\ 0 & C_2 \end{pmatrix}
\]
where $C_1^T C_2 = 0$, which we now call CSS code.
The history actually goes backwards.
CSS first constructed quantum codes.
Later, Calderbank, Rains, Shor, and Sloane~\cite{CalderbankRainsShorEtAl1997Quantum},
and Gottesman~\cite{Gottesman1996Saturating} 
formulated a version with the full matrix equation Eq.~\eqref{eq:nullity}.
The matrix $\sigma$ is called {\bf generating matrix} of the code.

We may extend the basis for the null space $V$ spanned by the columns of $\sigma$
to a canonical basis $B$ of whole Pauli space $P$
by the Gram-Schmidt process.
Another canonical basis $C$ of $P$ is induced by
\[
 \{ \sigma^x \otimes I \otimes \cdots, I \otimes \sigma^x \otimes \cdots, \ldots, 
  \sigma^z \otimes I \otimes \cdots, I \otimes \sigma^z \otimes \cdots \}
\]
Two bases are interchanged by a symplectic transformation,
and due to Proposition~\ref{prop:elem-sym-transformation-generates-whole-BlockCodeCase},
such a symplectic transformation is induced by \emph{unitary} operators CNOT, Hadamard, and Phase.
Under the basis $C$, the code space ({trivial code}) is the common eigenspace of
\[
 \{ \sigma^z \otimes I \otimes \cdots, I \otimes \sigma^z \otimes \cdots \}
\]
of eigenvalue $+1$.
If $s = \dim_{\FF_2} V \le q$, then the code space consists of all states of $q-s$ qubits.
That is, the {\bf trivial code} is
\[
 \ket{0}^{\otimes s} \otimes \lspan _\mathbb{C} \{ \ket{i_1 \cdots i_{q-s} } | i_a \in \FF_2 \} .
\]
It is clear that the orthogonal complement $V^\perp$ with respect to the symplectic form
is a direct sum of $V$ itself and a hyperbolic subspace of $P$.
In conclusion, if there are $s$ independent stabilizers,
the code space dimension is $2^{q-s}$ and there are $q-s$ pairs of Pauli operators
commuting with every stabilizer
that give rise to a hyperbolic subspace in the symplectic Pauli space.
These $q-s$ pairs of Pauli operators are logical operators.
The {\bf logical operators} are by definition operators that preserve the code space.
If they are Pauli operators, they must belong to $V^\perp$.
Focusing on the action on the code space by the logical operators,
one realizes that different logical operators might have the same action.
Indeed, if $U$ is a logical operator,
then $US$ is also a logical operator for any stabilizer $S$ of the same action on the code space:
\[
 U S \ket \psi = U \ket \psi
\]
where $\ket \psi$ is any code vector.
Conversely, if two logical Pauli operators $U_1$ and $U_2$ have the same action,
then $U_1^\dagger U_2$ have the trivial action so it must be proportional to a stabilizer.
Therefore, an equivalence class of logical \emph{Pauli} operators 
precisely corresponds to a coset of $V^\perp / V$.
Sometimes, stabilizers are called {\bf trivial logical operators}.

The {\bf weight of a Pauli operator} is the number of its non-identity tensor factors.
Likewise, the {\bf weight of a vector} is the number of its nonzero components.
A particularly interesting logical operator is one that corresponds to
the \emph{minimal weight} representative of nonzero elements of $V^\perp / V$.
This minimal weight is called {\bf code distance}, {\bf minimal distance} or just distance for short.
The importance of the code distance in relation to error correcting capability
will become clear once we establish cleaning lemma in the next section.

\section{Cleaning lemma}
\label{sec:cleaning-lemma}

Although this thesis mainly discusses additive codes,
the content of this section is stated in terms of \emph{subsystem} codes.
We continue to use the symplectic viewpoint for the Pauli group/space,
and the terms group and space will be used interchangeably.

Let $G$ be an arbitrary subspace of the Pauli space $P$.
We denote by $[G]$ the dimension of $G$ as a vector space over $\FF_2$.
By Proposition~\ref{prop:structure-symplectic-space}, we have a decomposition
$G = S \oplus W$ where $S$ is null and $W$ is hyperbolic.
The orthogonal complement $G^\perp$ consists of $S$ and a hyperbolic subspace $W'$ disjoint from $W$.
Mapping to trivial code, we see that $W'$ describes all Pauli operators acting on some qubits.
A {\bf subsystem code} is a selection of code space
given by the states of qubits acted upon by $W'$~\cite{Bacon2006Operator,Poulin2005Stabilizer}.
The common eigenspace of $S$
has a nontrivial tensor decomposition $\mathcal H _\text{gauge} \otimes \mathcal H _\text{logical}$.
The Pauli operators acting on $\mathcal H_\text{gauge}$ are represented by $W$,
and those on the \emph{subsystem} $\mathcal H _\text{logical}$ we wish to make use of are represented by $W'$.
The code space of the subsystem code is identified with $\mathcal H_\text{logical}$.
One may think of a subsystem code as a stabilizer code for which some logical operators are discarded.
The group/space $S$ is still called \emph{stabilizer group/space}.
$G$ is called {\bf gauge group}.
In this section \emph{logical operators} only refer to Pauli operators.
For subsystem codes, a {\bf bare} logical operator is one that belongs to $G^\perp$,
and a {\bf dressed} logical operator is one that belongs to $S^\perp$.
Since logical operators' action on $\mathcal H_\text{gauge}$ is ignored,
the set of all equivalence classes of bare logical operators is $G^\perp / S$,
and the set of all equivalence classes of dressed logical operators is $S^\perp / G$.
It is clear that $G^\perp / S = W' = S^\perp / G$.

The cleaning lemma for subsystem codes 
relates the number of independent bare logical operators supported on a set of qubits $M$ 
to the number of independent dressed logical operators supported on the complementary set $M^c$. 
The concept of the cleaning lemma was introduced in \cite{BravyiTerhal2009no-go}, 
then generalized in \cite{YoshidaChuang2010Framework} and  \cite{Bravyi2010Subsystem}. 
Here we use ideas from \cite{YoshidaChuang2010Framework} 
to prove a version stated in \cite{Bravyi2010Subsystem}.
(See also \cite{WildeFattal2009Nonlocal}.)

We use $P_A$ to denote the subgroup of the Pauli group $P$ supported on a set $A$ of qubits; likewise for any subgroup $G$ of the Pauli group $G_A= G \cap P_A$, is the subgroup of $G$ supported on $A$.
We denote by $\Pi_A : P \to P_A$ the restriction map that maps a Pauli operator to its restriction supported on the set $A$, 
and we use $|A|$ to denote the number of qubits contained in $A$; thus $[P_A] = 2|A|$.

If we divide $n$ qubits into two complementary sets $A$ and $B$, then a subgroup $G$ of $P$ can be decomposed into $G_A$, $G_B$, and a ``remainder,'' as follows:

\begin{lem}
Suppose that $A$ and $B$ are complementary sets of qubits. Then for any subgroup $G$ of the Pauli group,
\[
 G = G_A \oplus G_B \oplus G'
\]
for some $G'$, where
\begin{align*}
 [ (G^\perp)_A ] &= 2|A| - [G_A] - [G'] ,\\
 [ (G^\perp)_B ] &= 2|B| - [G_B] - [G']
\end{align*}
\end{lem}
\begin{proof}
If $V$ is a vector space and $W$ is a subspace of $V$, then there is a vector space $V'$ such that $V=W\oplus V'$; we may choose $V'$ to be the span of the basis vectors that extend a basis for $W$ to a basis for $V$. 
Since $G_A$ and $G_B$ are disjoint, i.e., $G_A \cap G_B = \{0\}$, $G_A\oplus G_B$ is a subspace of $G$, and thus there exists an auxiliary vector space $G' \leq G$ such that
\[
 G = G_A \oplus G_B \oplus G'.
\]
The choice of $G'$ is not canonical, but we need only its existence. Since the restriction map $\Pi_A$ obviously annihilates $G_B$, we may regard it as a map from $G_A\oplus G'$ onto $\Pi_A G$. In fact this map is injective. Note that if $\Pi_A x = 0$ for some $x \in G_A\oplus G'$, then since $P=P_A\oplus P_B$ it must be that $x \in G_B$. But because the sum is direct, i.e., $G_B \cap (G_A\oplus G') = \{0\}$, it follows that $x = 0$, which proves injectivity. Hence $\Pi_A: G_A\oplus G'\to \Pi_A G$ is an isomorphism. Now, we may calculate $(G^\perp)_A$ by solving a system of linear equations. Noting that $x \in P_A$ is contained in $G^\perp$ if and only if $x$ commutes with the restriction to $A$ of each element of $G$, we see that the number of independent linear constraints is $[\Pi_A G] = [G_A] + [G']$; hence $[(G^\perp)_A]=[P_A] - [G_A] - [G']= 2|A| - [G_A] - [G']$. Likewise, $\Pi_B: G_B\oplus G'\to \Pi_B G$ is also an isomorphism, and hence $[(G^\perp)_B]=[P_B] - [G_B] - [G']= 2|B| - [G_B] - [G']$.
\end{proof}

Now we are ready to state and prove the cleaning lemma. For a subsystem code, let $g_{\rm bare}(M)$ be the number of independent nontrivial bare logical operators supported on $M$, and let $g(M)$ be the number of independent nontrivial dressed logical operators supported on $M$, i.e.,
\begin{align*}
 g_{\rm bare}(M) &= [G^\perp \cap P_M / S_M ] = [(G^\perp)_M/S_M], \\
 g(M)        &= [S^\perp \cap P_M / G_M ]= [(S^\perp)_M/G_M].
\end{align*}
Likewise, for a CSS subsystem code, let $g_{\rm bare}^X(M)$ be the number of independent nontrivial bare $X$-type logical operators supported on $M$, and let $g^X(M)$ be the number of independent nontrivial dressed $X$-type logical operators supported on $M$, i.e.,
\begin{align*}
 g_{\rm bare}^X(M) &= [(G^Z)^\perp \cap P^X_M / S^X_M], \\
 g^X(M)        &= [(S^Z)^\perp \cap P^X_M / G^X_M],
\end{align*}
and similarly for the $Z$-type logical operators.

\begin{lem}\emph{(Cleaning lemma for subsystem codes)}
Let $k$ be the number of encoded qubits.
For any subsystem code, we have
\[
 g_{\rm bare}(M) + g(M^c) = 2k ,
\]
where $M$ is any set of qubits and $M^c$ is its complement. 
Moreover, for a CSS subsystem code
\[
g_{\rm bare}^X(M) + g^Z(M^c) = k = g_{\rm bare}^Z(M) + g^X(M^c).
\]
\label{lem:counting_op}
\end{lem}
\begin{proof}
We use Lemma 1 to prove the cleaning lemma by a direct calculation:
\[
g_{\rm bare}(M) = [(G^\perp)_M  / S_M] = 2|M| - [G_M] - [G'] - [S_M] ,
\]
and
\[
g(M^c) = [(S^\perp)_{M^c} / G_{M^c}] = 2|M^c| - [S_{M^c}] - [S'] - [G_{M^c}] .
\]
Summing, we find
\[
g_{\rm bare}(M) + g(M_c) = 2|M| + 2|M_c| -([G_M] + [G_{M_c}] + [G']) -([S_M] + [S_{M_c}] + [S'])
\]
and invoking Lemma 1 once again,
\[
g_{\rm bare}(M) + g(M_c) = 2n -[G] - [S] = 2k ,
\]
which proves the claim for general subsystem codes.
For the CSS case, we apply the analogue of Lemma 1 to the $X$-type and $Z$-type Pauli operators, finding
\[
g^Z_{\rm bare}(M) = [ (G^X)^\perp \cap P^Z_{M} / S^Z_{M} ] = |M| - [G^X_{M}] - [(G^X)'] - [S^Z_{M}]
\]
and also 
\[
g^X(M^c) = [ (S^Z)^\perp \cap P^X_{M^c} / G^X_{M^c} ] = |M^c| - [S^Z_{M^c}] - [(S^Z)'] - [G^X_{M^c}]. 
\]
Summing and using Lemma 1 we have
\[
g_{\rm bare}^Z(M) + g^X(M^c) = n - [G^X]-[S^Z] =k ;
\]
a similar calculation yields 
\[
g_{\rm bare}^X(M) + g^Z(M^c) = n - [G^Z]-[S^X] =k ,
\]
proving the claim for CSS subsystem codes.
\end{proof}

Of course, for a stabilizer code there is no distinction 
between bare and dressed logical operators; 
the statement of the cleaning lemma becomes
\[
g(M) + g(M^c) = 2k
\]
for general stabilizer codes, and
\[
g^X(M) + g^Z(M^c) = k
\]
for CSS stabilizer codes.

To understand how the cleaning lemma gets its name, 
note that it implies that 
if no bare logical operator can be supported on the set $M$ 
then all dressed logical operators can be supported on its complement $M^c$. 
That is, any of the code's dressed logical Pauli operators can be ``cleaned up'' 
by applying elements of the gauge group $G$. 
The cleaned operator acts the same way on the protected qubits 
as the original operator (though it might act differently on the gauge qubits),
and acts trivially on $M$.

We say that a region $M$ is \emph{correctable} if there are no nontrivial dressed logical operators supported on $M$.
If $M$ is correctable 
then $g(M)=0$ and thus $g_{\rm bare}(M)=0$.
The cleaning lemma is then rephrased as follows.
\begin{lem}
\label{lem:clean-region}
For any subsystem code, if $M$ is a correctable region and $x$ is a dressed logical operator,
then there is a dressed logical operator $y$ supported on $M^c$ that is equivalent to $x$.
\end{lem}

\begin{rem}
Given a correctable region $M$, we have a complete set of logical Pauli operators $\{ y_a \}$ supported on $M^c$.
Let $U$ be the unitary transformation that maps the code space to that of the trivial code of the previous section;
$U$ is a composition of elementary symplectic transformations.
Since any error $e$ on $M$ was (trivially) commuting with any $y_a$,
it follows that $UeU^\dagger$ acts by identity on the logical qubits of the trivial code.
Replacing the non-logical qubits with fresh qubits and applying $U^\dagger$,
we can map the damaged code vector to its original state.
In conclusion, any error on the correctable region can be corrected.
The code distance is the upper bound on the number of qubits in any correctable regions.

A more general error correcting criterion
can be found in \cite[Chapter 15]{KitaevShenVyalyi2002CQC}.
\end{rem}

\section{Operator trade-off for local subsystem codes}
\label{sec:subsystem-tradeoff}

In this section we consider local subsystem codes 
with qubits residing at the sites of a $D$-dimensional hypercubic lattice $\Lambda$.
The code has {\bf interaction range} $w$, meaning that 
the generators of the gauge group $G$ can be chosen 
so that each generator has support on a hypercube containing $w^D$ sites.

\begin{defn}
\label{defn:boundary}
Given a set of gauge generators for a subsystem code, and a set of qubits $M$, 
let $M'$ denote the support of all the gauge generators that act nontrivially on $M$. 
The {\bf external boundary} of $M$ is $\partial_+ M = M' \cap M^c$, 
where $M^c$ is the complement of $M$, 
and the {\bf internal boundary} of $M$ is  $\partial_- M = \left(M^c\right)' \cap M$.
The {\bf boundary} of $M$ is $\partial M=\partial_+M\cup\partial_-M$,
and the {\bf interior} of $M$ is $M^\circ = M \setminus \partial_- M$.
\end{defn}

Recall that a region (i.e., a set of qubits) $M$ is said to be \emph{correctable}
if no nontrivial dressed logical operation is supported on $M$, 
in which case erasure of $M$ can be corrected. 
Since the code distance $d$ is defined as the minimum weight of a dressed logical operator,
$M$ is certainly correctable if $|M| < d$.
But in fact much larger regions are also correctable, as follows from this lemma: 

\begin{lem}
For a local subsystem code, if $M$ and $A$ are both correctable,
where $A$ contains $\partial M$, then $M\cup A$ is correctable.
\label{lem:subsystem-extend}
\end{lem}
\begin{proof}
Given a subsystem code $\mathcal{C}$ with gauge group $G$,
we may define a subsystem code $\mathcal{C}_{M^c}$ on $M^c$ with gauge group $\Pi_{M^c}G$, 
where $\Pi_{M^c}$ maps a Pauli operator to its restriction supported on $M^c$. 
We note that a Pauli operator $x$ supported on $M^c$ is a bare logical operator for $\mathcal{C}$ 
if and only if $x$ is a bare logical operator for $\mathcal{C}_{M^c}$; 
that is, $x$ commutes with all elements of $G$ 
if and only if it commutes with all elements of the restriction of $G$ to $M^c$. 

Furthermore, if $x$ is a dressed logical operator for $\mathcal{C}_{M^c}$ supported on $\partial_+M$, 
then $x$ can be extended to a dressed logical operator $\bar x$ for $\mathcal{C}$ supported on $\partial M$. 
Indeed, suppose $x=yz$, where $y$ is a bare logical operator for $\mathcal{C}_{M^c}$ 
(and hence also a bare logical operator for $\mathcal{C}$ supported on $M^c$), 
while $z$ is an element of the gauge group $\Pi_{M^c} G$ of $\mathcal{C}_{M^c}$. 
Then $z$ can be written as a product $z=\prod_i g_i$ of generators of $\Pi_{M^c} G$, 
each of which can be expressed as $g_i = \Pi_{M^c} \bar g_i$, 
where $\bar g_i$ is a generator of $G$ supported on $M^c\cup  \partial_-M$.
Thus $\bar x = y\prod_i\bar g_i$ is a dressed logical operator for $\mathcal{C}$ supported on $\partial M$. 
It follows that if $\partial M$ is correctable for the code $\mathcal{C}$ 
(i.e., code $\mathcal{C}$ has no nontrivial dressed logical operators supported on $\partial M$), 
then $\partial_+ M$ is correctable for the code $\mathcal{C}_{M^c}$ 
($\mathcal{C}_{M^c}$ has no nontrivial dressed logical operators supported on $\partial_+ M$). 
By similar logic, if $A$ is correctable for $\mathcal{C}$ and contains $\partial M$, 
then $A\cap M^c$ is correctable for $\mathcal{C}_{M^c}$.

Suppose now that the code $\mathcal{C}$ has $k$ encoded qubits and that $M$ is correctable, 
i.e., $g^{(\mathcal{C})}(M)=0$. 
Therefore, applying Lemma~\ref{lem:counting_op} to the code $\mathcal{C}$, 
$g_{\rm bare}^{(\mathcal{C})}(M^c)= 2k$. 
Suppose further that the set $A$ containing $\partial M$ is correctable for $\mathcal{C}$, 
implying that $A\cap M^c$ is correctable for $\mathcal{C}_{M^c}$,
i.e., $g^{(\mathcal{C}_{M^c})}(A\cap M^c)=0$. 
Then applying Lemma~\ref{lem:counting_op} to the code $\mathcal{C}_{M^c}$, 
we conclude that $g_{\rm bare}^{(\mathcal{C}_{M^c})}(M^c\setminus A)=2k$.
Since each bare logical operator for $\mathcal{C}_{M^c}$, supported on $M^c\setminus A$, 
is also a bare logical operator for $\mathcal{C}$, supported on $M^c\setminus A$, 
we can now apply Lemma \ref{lem:counting_op} once again to the code $\mathcal{C}$, 
using the partition into $M^c\setminus A$ and $M\cup A$, 
finding $g^{(\mathcal{C})}(M \cup A)=0$.
Thus $M \cup A$ is correctable. 
\end{proof}

If the interaction range is $w$, and $M$ is a correctable hypercube with linear size $l-2(w-1)$, 
then we may choose $A\supseteq \partial M$ 
so that $M\cup A$ is a hypercube with linear size $l$ and $M\setminus A$ is a hypercube with linear size $l - 4(w-1)$. 
Then $A$ contains
\[
|A| = l^D - \left[l-4(w-1)\right]^D \le 4(w-1)Dl^{D-1}
\]
qubits, and $A$ is surely correctable provided $|A|<d$, where $d$ is the code distance. 
Suppose that $d>1$, so a single site is correctable. 
Applying Lemma~\ref{lem:subsystem-extend} repeatedly,
we can build up larger and larger correctable hypercubes,
with linear size $1 + 2(w-1), 1+ 4(w-1), 1+ 6(w-1), \dots$. 
This process continues as long as $|A|< d$. We conclude:

\begin{lem}
\label{lem:subsystem-hypercube}
For a $D$-dimensional local subsystem code with interaction range $w>1$ and distance $d>1 $, 
a hypercube with linear size $l$ is correctable if 
\begin{equation}\label{eq:hypercube-size}
4(w-1)Dl^{D-1} < d.
\end{equation}
\end{lem}

\noindent Thus (roughly speaking) for the hypercube to be correctable 
it suffices for its $\left[2(w-1)\right]$-thickened \emph{boundary}, 
rather than its volume, to be smaller than the code distance. 
Bravyi~\cite{Bravyi2010Subsystem} calls this property ``the holographic principle for error correction,''
because the \emph{absence} of information encoded at the boundary of a region 
ensures that no information is encoded in the ``bulk.'' 

For local stabilizer codes, the criterion for correctability is slightly weaker than for local subsystem codes.
We say that a local stabilizer code has interaction range $w$
if each stabilizer generator has support on a hypercube containing $w^D$ sites.
For this case, we can improve the criterion for correctability of a hypercube,
found for local subsystem codes in Lemma~\ref{lem:subsystem-hypercube}.

\begin{lem}
For a local stabilizer code, 
suppose that $\partial_+M$, $A$, and $M\setminus A$  are all correctable, 
where $\partial_- M \subseteq A \subseteq M$. 
Then $M$ is also correctable. 
\end{lem}
\begin{proof}
Suppose, contrary to the claim, that there is a nontrivial logical operator $x$ supported on $M$.
Then, because $A$ is correctable, Lemma \ref{lem:clean-region} implies that
there is a stabilizer generator $y$ such that $xy$ acts trivially on $A$. 
Furthermore, $y$ can be expressed as a product of local stabilizer generators, 
each supported on $M'=M\cup\partial_+M$. 
Thus $xy$ is a product of two factors, 
one supported on $M\setminus A$ and the other supported on $\partial_+M$. 
Because $\partial_-M\subseteq A$, no local stabilizer generator acts nontrivially on both $M\setminus A$ and $\partial_+M$;
therefore, each factor commutes with all stabilizer generators and hence is a logical operator.
Because $M\setminus A$ and $\partial_+M$ are both correctable, 
each factor is a trivial logical operator and therefore $xy$ is also trivial.
It follows that $x$ is trivial, a contradiction. 
\end{proof}

Now, if the interaction range is $w$ and $M$ is a hypercube with linear size $l$, 
we choose $A$ so that $M\setminus A$ is a hypercube with linear size $l-2(w-1)$, 
and we notice that $\partial_+M$ is contained in a hypercube with linear size $l+2(w-1)$. 
Thus both $M\setminus A$ and $\partial_+ M$ are correctable provided that
\begin{align*}
|\partial_+ M| &\le \left[l +2(w-1)\right]^D - l^D \nonumber\\
&\le 2(w-1)D\left[l +2(w-1)\right]^{D-1} < d.
\end{align*}
Reasoning as in the proof of Lemma \ref{lem:subsystem-hypercube}, we conclude that:

\begin{lem}
\label{lem:stabilizer-hypercube}
For a $D$-dimensional local stabilizer code with interaction range $w>1$ and distance $d>1 $, a hypercube with linear size $l$ is correctable if 
\begin{equation}\label{eq:hypercube-size-stabilizer}
2(w-1)D\left[l+ 2(w-1)\right]^{D-1} < d.
\end{equation}
\end{lem}

\noindent To ensure that the hypercube $M$ is correctable, it suffices for its $(w-1)$-thickened boundary, rather than its $\left[2(w-1)\right]$-thickened boundary, to be smaller than the code distance.

Now we are ready to prove our first trade-off theorem.

\begin{figure}
\centering
\includegraphics[width=0.42\textwidth]{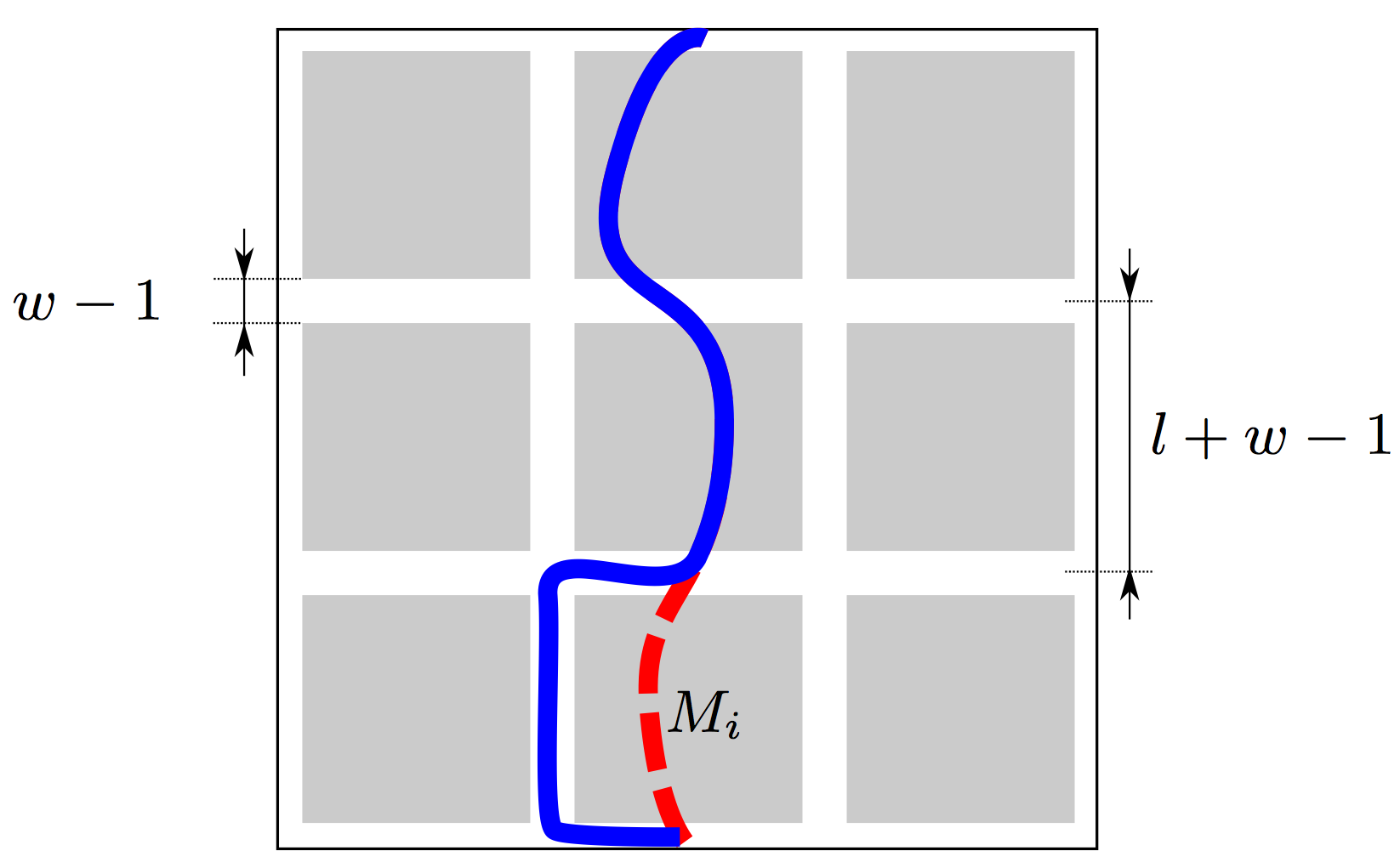}
\caption{Lattice covering used in the proof of Theorem 1, shown in two dimensions. 
Each gray square is $l\times l$ and the white gap between squares has width $w-1$. 
The solid blue curve represents the support of a nontrivial logical operator; 
because the square $M_i$ is correctable, this square can be ``cleaned.'' 
We can find an equivalent logical operator supported on $M_i^c$, the complement of $M_i$. 
When all squares are cleaned, the logical operator is supported on the narrow strips between the squares.}
\label{fig:cleaning}
\end{figure}

\begin{theorem}\emph{(Trade-off theorem for subsystem codes)}
For a local subsystem code in $D\ge 2$ dimensions with interaction range $w>1$ and distance $d\gg w$, 
defined on a hypercubic lattice with linear size $L$, 
every dressed logical operator is equivalent to an operator with weight $\tilde d$ satisfying
\begin{equation}\label{eq:tradeoff-bound}
  \tilde d {d}^{1/(D-1)} < c L^D,
\end{equation}
where $c$ is a constant depending on $w$ and $D$.
\label{thm:subsystem-tradeoff}
\end{theorem}

\begin{proof}
As shown in Fig.~\ref{fig:cleaning}, we fill the lattice with hypercubes, separated by distance $w-1$, 
such that each hypercube has linear size $l$ satisfying Eq.~\eqref{eq:hypercube-size}. 
(By ``distance'' we mean the number of sites in between --- e.g., we say that adjacent sites are ``distance zero'' apart.)
Thus no gauge generator acts nontrivially on more than one hypercube, 
and each hypercube is correctable by Lemma~\ref{lem:subsystem-hypercube}. 
Consider any nontrivial dressed logical operator $x$, 
and label the hypercubes $\{M_1, M_2, M_3, \dots\}$. 
By Lemma 3 there exists a gauge operator $y_i$ that ``cleans'' the logical operator in the hypercube $M_i$, 
i.e., such that $xy_i$ acts trivially in $M_i$. 
Furthermore, since no gauge generator acts nontrivially on more than one hypercube, 
we can choose $y_i$ so that it acts trivially in all other hypercubes. 
Taking the product of all the $y_i$'s we construct 
a gauge operator that cleans all hypercubes simultaneously; 
thus $\tilde x= x\prod_i y_i$ is equivalent to $x$ and supported on the complement of the union of hypercubes $M=\cup_i M_i$. 
Therefore, the weight $\tilde d$ of $\tilde x$ is upper bounded by $|M^c|$. 

The lattice is covered by hypercubes of linear size $l+(w-1)$, 
each centered about one of the $M_i$'s. 
There are $L^D/\left[l+(w-1)\right]^D$ such hypercubes in this union, 
each containing no more than $\left[l+(w-1)\right]^D - l^D \le (w-1)D\left[l+(w-1)\right]^{D-1}$ elements of $M^c$. 
Thus 
\[
\tilde d \le |M^c| 
\le (w-1)D\left[l+(w-1)\right]^{D-1}\frac{L^D}{\left[l+(w-1)\right]^D} 
= \frac{(w-1)D}{l+(w-1)}L^D.
\]
We optimize this upper bound on $\tilde d$ by choosing $l$ to be the largest integer 
such that a hypercube with linear size $l$ is known to be correctable,  
i.e., satisfying
\[
l < \left(\frac{d}{4(w-1)D}\right)^{1/(D-1)},
\]
thus obtaining Eq.~\eqref{eq:tradeoff-bound}.
Note that Eq.~\eqref{eq:tradeoff-bound} is trivial 
if $d$ is a constant independent of $L$, 
since the weight $\tilde d$ cannot be larger than $L^D$.
\end{proof}

\chapter{Algebraic formulation for translationally invariant codes}
\label{chap:alg-theory}

If an additive code is defined on a lattice with periodic boundary conditions
by a translationally invariant set $\{g_i\}$ of local generators,
then we can define a \emph{code Hamiltonian}
\[
 H = - \sum_i g_i
\]
whose ground space is identified with the code space.
Prototypical is the toric code model~\cite{Kitaev2003Fault-tolerant}.
If the stabilizer generators are given as translations of a finite collection of Pauli operators,
we have a family of code Hamiltonians parameterized by the system size or boundary conditions.
In this chapter we present a framework to study such translationally invariant code Hamiltonians.
We will be particularly interested in the \emph{phases of matter} represented by the code Hamiltonians.
Hence, we should allow local unitary transformations and deformation of Hamiltonians 
as long as the deformation does not close the energy gap between the ground and first excited state,
and we study properties of Hamiltonians that are invariant under these transformations.

If one does not make use of the translation structure, 
but insists on the use of usual symplectic vector space description,
then one should deal with infinitely many qubits and the generating matrix of infinite size.
However, the translation symmetry tells us 
that there is only a finite amount of data describing the family of code Hamiltonians.
It is certainly uneconomical to study a general infinite matrix.
We must look for a succinct description.
The starting observation is that translation symmetry is mostly well expressed
with the group algebra of the underlying translation group.
Fortunately, the lattice is an abelian group isomorphic to $\ZZ^D$ for some $D \ge 1$;
the group algebra is a commutative ring, or actually, a Laurent polynomial ring.
We will see that, inasmuch as the additive codes are described by null spaces in symplectic vector spaces,
the translation-invariant code Hamiltonians have corresponding algebraic descriptions
by certain submodules of finitely generated free modules.

We build an effective dictionary between lattice codes and commutative algebra in Section~\ref{sec:Pauli-space-on-group}.
The equivalence of phases of matter is re-expressed in an algebraic form in Section~\ref{sec:equivalent-hamiltonians}.
Perhaps the most important entry in this dictionary is that
the topological order condition can be expressed by 
the vanishing homology of a certain chain complex, presented in Section~\ref{sec:topological-order}.
As we will see in the next chapter,
the physical dimension is essentially the length of the chain complex,
which is a reminiscence of Hilbert syzygy theorem.
The ground-state subspace can be analyzed by studying an algebraic set defined by the chain complex.
In particular, the degeneracy, or the number of encoded qubits, can be approximated by counting points in the algebraic set,
for which some bounds are given in Section~\ref{sec:degeneracy}.
The last section~\ref{sec:fractal-operators} introduces fractal operators and 
establishes a precise algebraic description of topological charges in connection to the fractal operators.
Our discussion assumes some familiarity with commutative algebra.
In Appendix~\ref{app:algebra}, we include basic materials relevant to analysis of codes.

Note that in the classical coding theory the use of multivariate polynomial 
in multidimensional cyclic codes at least dates back to Imai~\cite{Imai1977TDC}.
Imai realized the importance of zero-locus of defining polynomials of the code,
which had been emphasized in the one-dimensional cyclic code~\cite{MacWilliamsSloane1977}.
See also \cite{SaintsHeegard1995MultidimensionalCyclicCodes,GueneriOezbudak2008} and references therein.
In the classical coding theory,
each lattice sites carries one bit $\{0,1\}$,
whereas in our quantum codes each lattice sites may have several qubits.
For this generalization	 we use modules instead of ideals.
Moreover, quantum additive codes require the nullity equation
$ \sigma^T \lambda \sigma = 0$.
(See Section~\ref{sec:additive-stabilizer-codes}.)
We decompose this equation into two parts $\epsilon = \sigma^T\lambda$ and $\sigma$,
and view them as connecting maps of a chain complex
$ G \xrightarrow{\sigma} P \xrightarrow{\epsilon} E $.
The chain complex is relevant only to quantum codes.

We will treat code Hamiltonians on sets of qubits, two-dimensional (two-level) local Hilbert spaces.
However, our language naturally allows an extension to prime-dimensional \emph{qudits}.
This is because we only use the fact that the set of coefficients $\FF_2$ is a field.
The generalization is achieved simply by replacing the ground field $\FF_2$ 
with $\FF_p$ for any prime number $p$.

\begin{table}[t]
\centering
\begin{tabular}{c|l}
\hline
$\FF_2$ & binary field $\{0,1\}$\\
$D$ & spatial dimension \\
$R$ & $\FF_2[x_1, x_1^{-1}, \ldots, x_D, x_D^{-1}]$ \\
$\bb_L$ & ideal $(x_1^L-1,\ldots,x_D^L -1)$\\
$q$ & number of qubits per site \\
$t$ & number of interaction types \\
$G$ & free $R$-module of the interaction labels (rank $t$) \\
$P$ & free $R$-module of Pauli operators (rank $2q$)\\
$E$ & free $R$-module of excitations (rank $t$)\\
$\sigma$ & $G \to P$, generating matrix or map for the stabilizer module \\
$\epsilon$ & $P \to E$, generating matrix or map for excitations \\
$r \mapsto \bar r$ & antipode map of the group algebra $R$.\\
$\dagger$ & transpose followed by antipode map \\
$\lambda_q$ & anti-symmetric $2q \times 2q$ matrix $\begin{pmatrix} 0 & \id \\ -\id & 0 \end{pmatrix}$ \\
\hline
\end{tabular}
\caption{Reserved symbols in Chapter~\ref{chap:alg-theory}. Any ring in this thesis is commutative with 1.}
\end{table}

\section{Algebraic structure of code Hamiltonians}

\subsection{Pauli space on a group}
\label{sec:Pauli-space-on-group}

Let $\Lambda$ be the index set of all qubits,
and suppose now that $\Lambda$ itself is an abelian group.
There is a natural action of $\Lambda$ on the Pauli group modulo phase factors
induced from the group action of $\Lambda$ on itself by multiplication.
For example, if $\Lambda = \mathbb{Z}$,
the action of $\Lambda$ is the translation on the one-dimensional chain of qubits.
If $R=\FF_2[\Lambda]$ is the group algebra with the multiplicative identity denoted by $1$,
the Pauli group modulo phase factors acquires a structure of an $R$-module.
We shall call it the {\em Pauli module}.
The Pauli module is free and has rank 2.

Let $r \mapsto \bar r$ be the antipode map of $R$, 
i.e., the $\FF_2$-linear map into itself 
such that each group element is mapped to its inverse.
Since $\Lambda$ is abelian, the antipode map is an algebra-automorphism.
Let the coefficient of $a \in R$ at $g \in \Lambda$ be denoted by $a_g$.
Hence, $a = \sum_{g \in \Lambda} a_g g$ for any $a \in R$.
One may write $a_g = (a \bar g)_1$.

Define
\[
 \tr(a) = a_1
\]
for any $a \in R$.
\begin{prop}
[\cite{CalderbankRainsShorEtAl1997Quantum}]
Let $(a,b), (c,d) \in R^2$ be two vectors representing Pauli operators $O_1, O_2$ up to phase factors:
\begin{align*}
 O_1 &= \left( \bigotimes_{g \in \Lambda} (\sigma_x^{(g)})^{a_g} \right) 
       \left( \bigotimes_{g \in \Lambda} (\sigma_z^{(g)})^{b_g} \right) ,
\\
 O_2 &= \left( \bigotimes_{g \in \Lambda} (\sigma_x^{(g)})^{c_g} \right) 
       \left( \bigotimes_{g \in \Lambda} (\sigma_z^{(g)})^{d_g} \right) 
\end{align*}
where $\sigma^{(g)}$ denotes the single-qubit Pauli operator at $g \in \Lambda$.
Then, $O_1$ and $O_2$ commute if and only if  
\[
 \tr \left(
\begin{pmatrix} \bar a & \bar b \end{pmatrix}
\begin{pmatrix} 0 & 1 \\ -1 & 0 \end{pmatrix}
\begin{pmatrix} c \\ d \end{pmatrix}
\right) = 0 .
\]
\label{prop:trace-formula-commutation-value}
\end{prop}
\begin{proof}
The commutation value of $(\sigma_x^{(g)})^{n} (\sigma_z^{(g)})^{m}$
and $(\sigma_x^{(g)})^{n'} (\sigma_z^{(g)})^{m'}$
is $nm' - mn' \in \FF_2$.
Viewed as pairs of group algebra elements,
$(\sigma_x^{(g)})^{n} (\sigma_z^{(g)})^{m}$
and $(\sigma_x^{(g)})^{n'} (\sigma_z^{(g)})^{m'}$
are $(n g, m g)$ and $(n' g, m' g)$, respectively.
We see that
\[
 nm' - mn' =  \tr \left(
\begin{pmatrix} n g^{-1} &  m g^{-1} \end{pmatrix}
\begin{pmatrix} 0 & 1 \\ -1 & 0 \end{pmatrix}
\begin{pmatrix} n' g \\ m' g \end{pmatrix}
\right) .
\]
Since any Pauli operator is a finite product of these,
the result follows by linearity.
\end{proof}

We wish to characterize a $\FF_2$-subspace $S$ of the Pauli module
invariant under the action of $\Lambda$, i.e., a submodule,
on which the commutation value is always zero.
As we will see in the next subsection,
this particular subspace yields a local Hamiltonian whose energy spectrum is exactly solvable,
which is the main object of this thesis.
Let $(a,b)$ be an element of $S \subseteq R^2 = (\FF_2[\Lambda])^2$.
For any $r \in R$, $(ra, rb)$ must be a member of $S$.
Demanding that the symplectic form on $S$ vanish,
by Proposition~\ref{prop:trace-formula-commutation-value} we have
\[
 \tr (ra \bar b - rb \bar a) = 0 .
\]
Since $r$ was arbitrary, we must have $a \bar b - b \bar a = 0$.%
\footnote{{A symmetric bilinear form $ \langle r, s \rangle = \tr(r \bar{s}) $ on $R$ is nondegenerate.}}
Let us denote $\begin{pmatrix}\bar a & \bar b \end{pmatrix}$ 
as $\begin{pmatrix} a \\ b \end{pmatrix}^\dagger$,
and write any element of $R^2$ as a $2 \times 1$ matrix.
We conclude that $S$ is a submodule of $R^2$ over $R$ generated by $s_1, \ldots, s_t$
such that any commutation value always vanishes, if and only if
\[
 s_i^\dagger \lambda_1 s_j = 0
\]
for all $i,j = 1,\ldots,t$.

The requirement that $\Lambda$ be a group might be too restrictive.
One may have a coarse group structure on $\Lambda$, the index set of all qubits.
We consider the case that the index set is a product of a finite set and a group.
By abuse of notation, we still write $\Lambda$ to denote the group part,
and insist that to each group element are associated $q$ qubits ($q \ge 1$).
Thus obtained Pauli module should now be identified with $R^{2q}$,
where $R = \FF_2[\Lambda]$ is the group algebra
that encodes the notion of translation.
We write an element $v$ of $R^{2q}$ by a $2q \times 1$ matrix,
and denote by $v^\dagger$ the transpose matrix of $v$ whose each entry is applied by the antipode map.
We always order the entries of $v$ such that the upper $q$ entries
describes the $\sigma_x$-part and the lower the $\sigma_z$-part.
Since the commutation value on $R^{2q}$ is the sum of commutation values on $R^2$,
we have the following:
If $S$ is a submodule of $R^{2q}$ over $R$ generated by $s_1, \ldots, s_t$,
the commutation value always vanishes on $S$, if and only if for all $i,j = 1, \ldots, t$
\[
 s_i^\dagger \lambda_q s_j = 0
\]
where $\lambda_q = \begin{pmatrix} 0 & \id_q \\ -\id_q & 0 \end{pmatrix}$ is a $2q \times 2q$ matrix.

Let us summarize our discussion so far.
\begin{prop}
On a set of qubits $\Lambda \times \{1,\ldots,q\}$ where $\Lambda$ is an abelian group,
the group of all Pauli operators of finite support up to phase factors,
form a free module $P=R^{2q}$ over the group algebra $R = \FF_2[\Lambda]$.
The commutation value
\[ 
 \langle a, b \rangle = \tr( a^\dagger \lambda_q b )
\]
for $a,b \in P$ is zero
if and only if the Pauli operators corresponding to $a$ and $b$ commute.
If $\sigma$ is a $2q \times t$ matrix whose columns generate a submodule $S \subseteq P$,
then the commutation value on $S$ always vanishes if and only if
\[
 \sigma^\dagger \lambda_q \sigma = 0 .
\]
\label{prop:pauli-module}
\end{prop}
Proposition~\ref{prop:trace-formula-commutation-value}~\cite{CalderbankRainsShorEtAl1997Quantum}
is a special case of Proposition~\ref{prop:pauli-module} when $\Lambda$ is a trivial group.
When $\Lambda \cong \ZZ$, a similar equation appears in quantum convolutional 
codes~\cite{OllivierTillich2004Convolutional}.

\subsection{Local Hamiltonians on groups}

Recall that we place $q$ qubits on each {\em site} of $\Lambda$.
The total system of the qubits is $\Lambda \times \{1,\ldots,q\}$.
\begin{defn}
Let
\[
 H = -\sum_{g \in \Lambda} h_{1,g} + \cdots + h_{t,g}
\]
be a local Hamiltonian consisted of Pauli operators that is (i)commuting, (ii) translation-invariant up to signs, and (iii) frustration-free.
We call $H$ a {\bf code Hamiltonian} (also known as {\bf stabilizer Hamiltonian}).
The {\bf stabilizer module} of $H$ is the submodule of the Pauli module $P$
generated by the images of $h_1,\ldots,h_t$ in $P$.
The number of {\bf interaction types} is $t$.
\end{defn}
The energy spectrum of the code Hamiltonian is trivial;
it is discrete and equally spaced.

\begin{example}
One-dimensional Ising model is the Hamiltonian
\[
 H = - \sum_{i \in \mathbb{Z}} \sigma_z^{(i)} \otimes \sigma_z^{(i+1)} .
\]
The lattice is the additive group $\mathbb{Z}$,
and the group algebra is $R=\FF_2[x,\bar x]$.
The Pauli module is $R^2$ and the stabilizer module $S$ is generated by
\[
\begin{pmatrix}
 0 \\
 1+x
\end{pmatrix} .
\]
One can view this as the matrix $\sigma$ of Proposition~\ref{prop:pauli-module}.
$H$ is commuting; $\sigma^\dagger \lambda_1 \sigma = 0$.
\hfill $\Diamond$
\label{eg:1d-ising}
\end{example}

\subsection{Excitations}

For a code Hamiltonian $H$,
an excited state is described by the terms in the Hamiltonian that have eigenvalues $-1$.
Each of the flipped terms is interpreted as an {\bf excitation}.
Although the actual set of all possible configurations of excitations
that are obtained by applying some operator to a ground state, may be quite restricted,
it shall be convenient to think of a larger set.
Let $E$ be the set of all configurations of finite number of excitations without asking physical relevance.
Since an excitation is by definition a flipped term in $H$,
the set $E$ is equal to the collection of all finite sets consisted of the terms in $H$.

If Pauli operators $U_1, U_2$ acting on a ground state creates excitations $e_1, e_2 \in E$,
their product $U_1 U_2$ creates excitations $(e_1 \cup e_2) \setminus (e_1 \cap e_2)$.
Here, we had to remove the intersection because each excitation is its own annihilator;
any term in the $H$ squares to the identity.
Exploiting this fact, we make $E$ into a vector space over $\FF_2$.
Namely, we take formal linear combinations of terms in $H$
with the coefficient $1 \in \FF_2$ when the terms has $-1$ eigenvalue,
and the coefficient $0 \in \FF_2$ when the term has $+1$ eigenvalue.
The symmetric difference is now expressed as the sum of two vectors $e_1 + e_2$ over $\FF_2$.
In view of Pauli group as a vector space,
$U_1 U_2$ is the sum of the two vectors $v_1 + v_2$ that, respectively, represent $U_1$ and $U_2$.
Therefore, the association $U_i \mapsto e_i$ induces 
a linear map from the Pauli space to the space of virtual excitations $E$.

The set of all excited states obeys the translation invariance as the code Hamiltonian $H$ does.
So, $E$ is a module over the group algebra $R=\FF_2[\Lambda]$.
The association $U_i \mapsto e_i$ clearly respects this translation structure.
Our discussion is summarized by saying that the excitations are described by an $R$-linear map
\[
 \epsilon : P \to E
\]
from the Pauli module $P$ to the {\bf module of virtual excitations $E$}.

As the excitation module is the collection of all finite sets of the terms in $H$,
we can speak of the {\bf module of generator labels} $G$,
which is equal to $E$ as an $R$-module.
$G$ is a free module of rank $t$ if there are $t$ types of interaction.
The matrix $\sigma$ introduced in Section~\ref{sec:Pauli-space-on-group} can be viewed as
\[
 \sigma : G \to P
\]
from the module of generator labels to the Pauli module.
\begin{prop}
If $\sigma$ is the generating map for the stabilizer module of a code Hamiltonian,
then 
\[
 \epsilon = \sigma^\dagger \lambda_q .
\]
\end{prop}
The matrix $\epsilon$ can be viewed as a generalization of
the parity check matrix of the standard theory of classical or quantum error correcting codes%
~\cite{MacWilliamsSloane1977,CalderbankShor1996Good,Steane1996Multiple,Gottesman1996Saturating},
when a translation structure is given.

\begin{proof}
This is a simple corollary of Proposition~\ref{prop:pauli-module}.
Let $h_{i,g}$ be the terms in the Hamiltonian where $i = 1,\ldots, t$, and $ g \in \Lambda $.
In the Pauli module, they are expressed as $g h_i$ where $h_i$ is the $i$-th column of $\sigma$.
For any $u \in P$, let $\epsilon(u)_i$ be the $i$-th component of $\epsilon(u)$. By definition,
\[
 \epsilon( u )_i 
= \sum_{g \in \Lambda} g~ \tr \left( (g h_i)^\dagger \lambda_q u \right)
= \sum_{g \in \Lambda} g~ \tr \left( \bar g h_i^\dagger \lambda_q u  \right)
= h_i^\dagger \lambda_q u
\]
Thus, $h_i^\dagger \lambda_q$ is the $i$-th row of $\epsilon$.
\end{proof}

\begin{rem}
The commutativity condition in Proposition~\ref{prop:pauli-module} of the code Hamiltonian
is recast into the condition that
\[
 G \xrightarrow{\sigma} P \xrightarrow{\epsilon} E
\]
be a complex, i.e., $\epsilon \circ \sigma = 0$.
Equivalently,
\[
 \im \sigma \subseteq (\im \sigma)^\perp = \ker \epsilon
\]
where $\perp$ is with respect to the symplectic form.
\end{rem}

\section{Equivalent Hamiltonians}
\label{sec:equivalent-hamiltonians}

The stabilizer module entirely determines
the physical phase of the code Hamiltonian
in the following sense.
\begin{prop}
Let $H$ and $H'$ be code Hamiltonians on a system of qubits,
and suppose their stabilizer modules are the same.
Then, there exists a unitary
\[
 U = \bigotimes_{g \in \Lambda} U_{g}
\]
mapping the ground space of $H$ onto that of $H'$.
Moreover, there exist a continuous one-parameter family of gapped Hamiltonians
connecting $U H U^\dagger$ and $H'$.
\label{prop:associated-module-invariance}
\end{prop}
\begin{proof}
Let $\{ p_\alpha \}$ be a maximal set of $\FF_2$-linearly independent Pauli operators of finite support
that generates the common stabilizer module $S$.
$\{p_\alpha\}$ is not necessarily translation-invariant.
Any ground state $\ket \psi$ of $H$ is a common eigenspace of $\{p_\alpha\}$
with eigenvalues $p_\alpha \ket \psi = e_\alpha \ket \psi$, $e_\alpha = \pm 1$.
Similarly, the ground space of $H'$ gives the eigenvalues $e'_\alpha = \pm 1$ for each $p_\alpha$.

The abelian group generated by $\{ p_\alpha \}$ is precisely the vector space $S$,
and the assignment $p_\alpha \mapsto e_\alpha$ defines a dual vector on $S$.
If $U$ is a Pauli operator of possibly infinite support,
then $ p_\alpha U \ket \psi = e''_\alpha e_\alpha U \ket \psi$ for some $e''_\alpha = \pm 1$,
where $e''_\alpha$ is determined by the commutation relation between $U$ and $p_\alpha$.
Thus, the first statement follows if we can find $U$
such that the commutation value between $U$ and $p_\alpha$ is precisely $e''_\alpha$.
This is always possible since the dual space of the vector space $P$
is isomorphic to the direct product $\prod _{\Lambda \times \{1,\ldots, q\}} \FF_2^2$,
which is vector space isomorphic to the Pauli group of arbitrary support up to phase factors.%
\footnote{{If $V$ is a finite dimensional vector space over some field,
the dual vector space of $\bigoplus_I V$ 
is isomorphic to $\prod_{I} V$ where $I$ is an arbitrary index set.}}

Now, $U H U^\dagger$ and $H'$ have the same eigenspaces, and in particular, the same ground space.
Consider a continuous family of Hamiltonians
\[
H(u,u') = u U H U^\dagger + u' H'
\]
where $u,u' \in \mathbb{R}$.
It is clear that
\[
 H = H(1,0) \to H(1,1) \to H(0,1) = H'
\]
is a desired path.
\end{proof}

The criterion of Proposition~\ref{prop:associated-module-invariance}
to classify the physical phases is too narrow.
Physically meaningful universal properties should be invariant
under simple and local changes of the system. More concretely,
\begin{defn}
Two code Hamiltonians $H$ and $H'$ are {\bf equivalent} if
their stabilizer modules become the same under a finite composition
of symplectic transformations, coarse-graining, and tensoring ancillas.
\label{defn:equiv-H}
\end{defn}
We shall define the symplectic transformations,
the coarse-graining, and the tensoring ancillas shortly.

\subsection{Symplectic transformations}
\label{sec:symplectic-transformations}

\begin{defn}
A {\bf symplectic transformation} $T$ is an automorphism of the Pauli module
induced by a unitary operator on the system of qubits
such that 
\[
  T^\dagger \lambda_q T = \lambda_q
\]
where $\dagger$ is the transposition followed by the entry-wise antipode map.
\label{defn:symplectic-transformation}
\end{defn}
When the translation group is trivial these transformations are given by so-called {\bf Clifford operators}.
Compare Section~\ref{sec:additive-stabilizer-codes} and
see \cite[Chapter 15]{KitaevShenVyalyi2002CQC}.

Only the unitary operator on the physical Hilbert space 
that respects the translation can induce a symplectic transformation.
By definition, a symplectic transformation maps each local Pauli operator to a local Pauli operator, and 
preserves the commutation value for any pair of Pauli operators.
\begin{prop}
Any two unitary operators $U_1, U_2$ that induce the same symplectic transformation
differ by a Pauli operator (of possibly infinite support).
\end{prop}
If the translation group is trivial,
the proposition reduces to Theorem~15.6 of \cite{KitaevShenVyalyi2002CQC}
\begin{proof}
The symplectic transformation induced by $U = U_1^\dagger U_2$ is the identity.
Hence, $U$ maps each single-qubit Pauli operator $\sigma_{x,z}^{(g,i)}$ to $\pm \sigma_{x,z}^{(g,i)}$.
By the argument as in the proof of Proposition~\ref{prop:associated-module-invariance},
there exists a Pauli operator $O$ of possibly infinite support
that acts the same as $U$ on the system of qubits.
Since Pauli operators form a basis of the operator algebra of qubits,
we have $O=U$.
\end{proof}

The effect of a symplectic transformation on the generating map $\sigma$
 is a matrix multiplication on the left.
\[
 \sigma \to U \sigma
\]
For example, the following is induced
by uniform Hadamard, controlled-Phase, and controlled-NOT gates.
For notational clarity,
define $E_{i,j}(a)\ (i \neq j)$ as the row-addition elementary $2q \times 2q$ matrix
\[
 \left[ E_{i,j}(a) \right]_{\mu \nu} = \delta_{\mu \nu} + \delta_{\mu i} \delta_{\nu j} a
\]
where $\delta_{\mu \nu}$ is the Kronecker delta and $a \in R = \FF_2[\Lambda]$.
Recall that we order the components of $P$ such that the first half components are for $\sigma_x$-part, 
and the second half components are for $\sigma_z$-part.
\begin{defn}
The following are {\bf elementary symplectic transformations}:
\begin{itemize}
 \item (Hadamard) $E_{i,i+q}(-1) E_{i+q,i}(1) E_{i,i+q}(-1)$ where $1 \le i \le q$,
 \item (controlled-Phase) $E_{i+q,i}(f)$ where $f = \bar f$ and $1 \le i \le q$,
 \item (controlled-NOT) $E_{i,j}(a) E_{j+q,i+q}(-\bar a)$ where $1 \le i \ne j \le q$.
\end{itemize}
\end{defn}
For the case of a trivial translation group,
these transformations explicitly appear in \cite{CalderbankRainsShorEtAl1997Quantum}
and \cite[Chapter~15]{KitaevShenVyalyi2002CQC}.
The one-dimensional case appears in the context of
quantum convolutional codes~\cite{GrasslRotteler2006Convolutional}.

Recall that the Hadamard gate is a unitary transformation on a qubit given by
\[
U_H = \frac{1}{\sqrt 2}
 \begin{pmatrix}
  1 & 1 \\
  1 & -1
 \end{pmatrix}
\]
with respect to basis $\{ \ket 0 , \ket 1 \}$.
At operator level,
\[
 U_H X U_H^\dagger = Z, \quad U_H Z U_H^\dagger = X 
\]
where $X$ and $Z$ are the Pauli matrices $\sigma_x$ and $\sigma_z$, respectively.
Thus, the application of Hadamard gate on every $i$-th qubit of each site of $\Lambda$
swaps the corresponding $X$ and $Z$ components of $P$.

The controlled phase gate is a two-qubit unitary operator whose matrix is
\[
U_P =
 \begin{pmatrix}
  1 & 0 & 0 & 0 \\
  0 & 1 & 0 & 0 \\
  0 & 0 & 1 & 0 \\
  0 & 0 & 0 & -1
 \end{pmatrix}
\]
with respect to basis $\{ \ket{00}, \ket{01}, \ket{10}, \ket{11} \}$.
At operator level,
\begin{align*}
U_P (X \otimes I ) U_P^\dagger = X \otimes Z, & &
U_P (Z \otimes I ) U_P^\dagger = Z \otimes I, \\
U_P (I \otimes X ) U_P^\dagger = Z \otimes X, & &
U_P (I \otimes Z ) U_P^\dagger = I \otimes Z.
\end{align*}
Note that since $U_P$ is diagonal, any two $U_P$ on different pairs of qubits commute.
Let $(g,i)$ denote the $i$-th qubit at $g \in \Lambda$.
The uniform application
\[
U^{(i)}_g = \prod_{h \in \Lambda} U_P( (h,i), (h+g,i) )
\]
of $U_P$ throughout the lattice $\Lambda$
such that each $U_P( (h,i), (h+g,i) )$ acts on the pair of qubits $(h,i)$ and $(h+g,i)$
is well-defined. From the operator level calculation of $U_P$, we see that $U^{(i)}_g$
induces
\[
P \ni (\ldots,x_i,\ldots, z_i, \ldots) \mapsto ( \ldots, x_i,\ldots, z_i + (g+ \bar g)x_i, \ldots ) \in P
\]
on the Pauli module,
which is represented as $E_{i+q,i}(g+\bar g)$. 
The composition
\[
 U^{(i)}_{g_1} U^{(i)}_{g_2} \cdots U^{(i)}_{g_n}
\]
of finitely many controlled-Phase gates $U^{(i)}_g$ with different $g$
is represented as $E_{i+q,i}(f)$ where $f = \bar f = \sum_{k=1}^{n} g_k + \bar g_k$.
The single-qubit phase gate
\[
 \begin{pmatrix}
  1 & 0 \\
  0 & i
 \end{pmatrix}
\]
maps $X \leftrightarrow Y$ and $Z \mapsto Z$. On the Pauli module $P$, it is
\[
P \ni (\ldots,x_i,\ldots, z_i, \ldots)^T \mapsto ( \ldots, x_i,\ldots, z_i + x_i, \ldots )^T \in P .
\]
which is $E_{i+q,i}(1)$.
Note that any $f \in R$ such that $f = \bar f$ is always of form
$f = \sum g_k + \bar g_k$ or $f = 1 + \sum g_k + \bar g_k$
where $g_k$ are monomials.
Thus, the Phase gate and the controlled-Phase gate induce
transformations $E_{i+q,i}(f)$ where $f = \bar f$. 

The controlled-NOT gate is a two-qubit unitary operator whose matrix is
\[
U_N =
 \begin{pmatrix}
  1 & 0 & 0 & 0 \\
  0 & 1 & 0 & 0 \\
  0 & 0 & 0 & 1 \\
  0 & 0 & 1 & 0
 \end{pmatrix}
\]
with respect to basis $\{ \ket{00}, \ket{01}, \ket{10}, \ket{11} \}$.
That is, it flips the {\em target} qubit conditioned on the {\em control} qubit.
At operator level,
\begin{align*}
U_N (X \otimes I ) U_N^\dagger = X \otimes X, & &
U_N (Z \otimes I ) U_N^\dagger = Z \otimes I, \\
U_N (I \otimes X ) U_N^\dagger = I \otimes X, & &
U_N (I \otimes Z ) U_N^\dagger = Z \otimes Z.
\end{align*}
If $i < j$, the uniform application
\[
U^{(i,j)}_g = \bigotimes_{h \in \Lambda} U_P( (h,i), (h+g,j) )
\]
such that each $U_N( (h,i), (h+g,j) )$ acts on the pair of qubits $(h,i)$ and $(h+g,j)$
with one at $(h,i)$ being the control
induces
\begin{align*}
P \ni & (\ldots,x_i,\ldots,x_j,\ldots, z_i, \ldots, z_j, \ldots)^T \\
& \mapsto ( \ldots, x_i, \ldots, x_j + g x_i, \ldots, z_i + \bar g z_j, \ldots, z_j, \ldots )^T \in P .
\end{align*}
Thus, any finite composition of controlled-NOT gates with various $g$ is of form $E_{i,j}(a) E_{j+q,i+q}(\bar a)$.
It might be useful to note that the controlled-NOT and the Hadamard combined,
induces a symplectic transformation
\begin{itemize}
\item (controlled-NOT-Hadamard) $E_{i+q,j}(a) E_{j+q,i}(\bar a)$ where $a \in R$ and $ 1 \le i \ne j \le q$.
\end{itemize}

Remark that an arbitrary row operation on the upper $q$ components
can be compensated by a suitable row operation on the lower $q$ components
so as to be a symplectic transformation.

\subsection{Coarse-graining}

Not all unitary operators conform with the lattice translation.
In Example~\ref{eg:1d-ising} the lattice translation has period 1.
Then, for example, the Hadamard gate on every second qubit 
does not respect this translation structure;
it only respects a coarse version of the original translation.
We need to shrink the translation group to treat such unitary operators.

Let $\Lambda$ be the original translation group of the lattice with $q$ qubits per site,
and $\Lambda'$ be its subgroup of finite index: $|\Lambda/\Lambda'| = c < \infty$.
The total set of qubits $\Lambda \times \{1,\ldots,q\}$ is set-theoretically the same as
$\Lambda' \times \{ 1, \ldots, c \} \times \{1,\ldots,q\} = \Lambda' \times \{ 1, \ldots, cq\}$.
We take $\Lambda'$ as our new translation group under coarse-graining.
The Pauli group modulo phase factors remains the same as a $\FF_2$-vector space
for it depends only on the total index set of qubits.
We shall say that the system is {\bf coarse-grained by $R'=\FF_2[\Lambda']$}
if we restrict the scalar ring $R$ to $R'$ for all modules pertaining to the system.

For example, suppose $\Lambda = \mathbb{Z}^2$,
so the original base ring is $R = \FF_2[x,y,\bar x,\bar y]$.
If we coarse-grain by $R' = \FF_2[x',y', \bar x', \bar y']$ where $x' = x^2, y' = y^2$,
we are taking the sites $1,x,y,xy$ of the original lattice as a single new site.

Abstractly, the original translation group algebra $R$ is a finitely generated free module
over the coarse translation group algebra $R'$.
Thus, the coarse-graining can be regarded as an exact functor from the category of $R$-modules
to the category of $R'$-modules.

The one-dimensional case appears in the context of 
quantum convolutional codes~\cite{WildeBrun2010Convolutional}.

\subsection{Tensoring ancillas}

We have considered possible transformations
on the stabilizer modules of code Hamiltonians,
and kept the underlying index set of qubits invariant.
It is quite natural to allow tensoring ancilla qubits in trivial states.
In terms of the stabilizer module $S \subseteq P=R^{2q}$,
it amounts to embed $S$ into the larger module $R^{2q'}$ where $q' > q$.
Concretely,
let $\sigma = \begin{pmatrix} \sigma_X \\ \sigma_Z \end{pmatrix}$
be the generating matrix of $S$ as in Proposition~\ref{prop:pauli-module}.
By {\bf tensoring ancilla}, we embed $S$ as
\[
\begin{pmatrix} \sigma_X \\ \sigma_Z \end{pmatrix}
 \to 
\begin{pmatrix}
 \sigma_X & 0 \\
 0        & 0 \\
 \sigma_Z & 0 \\
 0        & 1
\end{pmatrix} .
\]
This amounts to taking the direct sum of the original complex
\[
 G \xrightarrow{\sigma} P \xrightarrow{\epsilon} E
\]
and the trivial complex
\[
0
\to
 R
 \xrightarrow{ \begin{pmatrix} 0 \\ 1 \end{pmatrix} }
 R^2
 \xrightarrow{ \begin{pmatrix} 1 & 0 \end{pmatrix} }
 R
\to 0
\]
to form
\[
  G \oplus R \xrightarrow{} P \oplus R^2 \xrightarrow{} E \oplus R .
\]

\section{Topological order}
\label{sec:topological-order}
From now on we assume that $\Lambda$ is isomorphic to $\mathbb{Z}^D$ as an additive group.
$D$ shall be called the {\bf spatial dimension} of $\Lambda$.

\begin{defn}
Let $\sigma: G \to P$ be the generating map for the stabilizer module of
a code Hamiltonian $H$. We say $H$ is {\bf exact}
if $(\im \sigma )^\perp = \im \sigma$, or equivalently
\[
 G \xrightarrow{\sigma} P \xrightarrow{\epsilon=\sigma^\dagger \lambda_q} E
\]
is exact, i.e., $\ker \epsilon = \im \sigma$.
\end{defn}
\noindent It follows that the exactness condition is
a property of the equivalence class of code Hamiltonians in the sense of Definition~\ref{defn:equiv-H}.

By imposing periodic boundary conditions,
a translation-invariant Hamiltonian yields
a family of Hamiltonians $\{ H(L) \}$ defined on a finite system consisted of $L^D$ sites.
One might be concerned that some $H(L)$ would be frustrated.
We intentionally exclude such a situation.
The frustration might indeed occur, but it can easily be resolved
by choosing the signs of terms in the Hamiltonian.
In this way, one might lose the translation invariance in a strict sense.
However, we retain the physical phase regardless of the sign choice
because different sign choices are related by a Pauli operator acting on
the whole system which is a product unitary operator.
Hence, the entanglement property of the ground state 
and the all properties of excitations do not change.

\begin{defn}
Let $H(L)$ be Hamiltonians on a finite system of linear size $L$
in $D$-dimensional physical space,
and $\Pi_L$ be the corresponding ground space projector.
$H(L)$ is called {\bf topologically ordered} if
for any $O$ supported inside a hypercube of size $(L/2)^D$
one has
\begin{equation}
 \Pi_L O \Pi_L \propto \Pi_L .
\label{eq:tqo}
\end{equation}
\label{defn:tqo}
\end{defn}
This means that no local operator is capable of distinguishing different ground states.
This condition is trivially satisfied if $H(L)$ has a unique ground state.
A technical condition that is used in the proof of the stability of topological order
against small perturbations is the following `local topological order'
condition~\cite{MichalakisPytel2011stability,
BravyiHastingsMichalakis2010stability,
BravyiHastings2011short}.
We say a {\bf diamond region $A(r)$ of radius $r$ at $o \in \ZZ^D$} for the set
\[
 A(r)_o = \left\{ (i_1,\ldots,i_D) + o \in \ZZ^D ~\middle| \sum_\mu |i_\mu| \le r \right\} .
\]
\begin{defn}
Let $H(L)$ be code Hamiltonians on a finite system of linear size $L$
in $D$-dimensional physical space.
For any diamond region $A=A(r)$ of radius $r$,
let $\Pi_A$ be the projector onto the common eigenspace of the most negative eigenvalues of terms
in the Hamiltonian $H(L)$ that are supported in $A$.
For $b > 0$, denote by $A^b$ the distance $b$ neighborhood of $A$.
$H(L)$ is called {\bf locally topologically ordered} if
there exists a constant $b > 0$ such that
for any operator $O$ supported on a diamond region $A$ of radius $r < L/2$
one has
\begin{equation}
 \Pi_{A^b} O \Pi_{A^b} \propto \Pi_{A^b} .
\label{eq:local-tqo}
\end{equation}
\label{defn:local-tqo}
\end{defn}
Since any operator is a $\mathbb{C}$-linear combination of Pauli operators,
if Eq.~\eqref{eq:tqo},\eqref{eq:local-tqo} are satisfied for Pauli operators,
then the (local) topological order condition follows.
If a Pauli operator $O$ is anticommuting with a term in a code Hamiltonian $H(L)$,
The left-hand side of Eq.~\eqref{eq:tqo},\eqref{eq:local-tqo} are identically zero.
In this case, there is nothing to be checked.
If $O$ acting on $A$ is commuting with every term in $H(L)$ supported inside $A^b$,
Eq.~\eqref{eq:tqo} demands that it act as identity on the ground space,
i.e., $O$ must be a product of terms in $H(L)$ up to $\pm i,\pm 1$.
Eq.~\eqref{eq:local-tqo} further demands that
$O$ must be a product of terms in $H(L)$ supported inside $A^b$ up to $\pm i,\pm 1$.

\begin{lem}
A code Hamiltonian $H$ is exact
if and only if $H(L)$ is locally topologically ordered for
all sufficiently large $L$.
\label{lem:local-tqo=exact}
\end{lem}
In order to see this, it will be important to use {\em Laurent polynomials}
to express elements of the group algebra 
$R = \FF_2[\mathbb{Z}^D] \cong \FF_2[x_1,x_1^{-1},\ldots,x_D,x_D^{-1}]$.
The reader might want to see \cite{Imai1977TDC,GueneriOezbudak2008} for classical multidimensional cyclic codes.
For example, 
\[
 x y^2 z^2 + x y^{-1} \quad \Longleftrightarrow \quad 1(1,2,2)+1(1,-1,0) .
\]
The sum of the absolute values of exponents of a monomial will be referred to as {\bf absolute degree}.
The absolute degree of a Laurent polynomial is defined to be the maximum absolute degree of its terms.
The degree measures the distance or size in the lattice.

The Laurent polynomial viewpoint enables us to apply Gr\"obner basis techniques.
The long division algorithm for polynomials in one variable yields an effective and efficient
test whether a given polynomial is divisible by another.
When two or more but finitely many variables are involved,
a more general question is how to test
whether a given polynomial is a member of an ideal.
For instance, $f=xy-1$ is a member of an ideal $J = (x-1,y-1)$
because $xy -1 = y(x-1) + (y-1)$.
But, $g=xy$ is not a member of $J$ because $g = y(x-1) + (y-1) + 1$ and the `remainder' 1 cannot be removed.
Here, the first term is obtained by looking at the initial term $xy$ of $f$
and comparing with the initial terms $x$ and $y$ of the generators of $J$.
While one tries to eliminate the initial term of $f$ and to eventually reach zero,
if one cannot reach zero as for $g$, then the membership question is answered negatively.

Systematically, an well-ordering on the monomials, i.e., a \emph{term order}, is defined
such that the order is preserved by multiplications.
And a set of generators $\{ g_i \}$ for the ideal is given with a special property 
that any element in the ideal has an initial term (leading term)
divisible by an initial term of some $g_i$.
A Gr\"obner basis is precisely such a generating set.
This notion generalizes to free modules over polynomial ring by refining the term order with the basis of the modules.
An example is as follows. Let
\[
 \sigma_1 = 
 \begin{pmatrix}
  \mathbf{x^2} - y \\
  x^2 + 1
 \end{pmatrix} \quad
 \sigma_2 = 
 \begin{pmatrix}
  1 \\
  \mathbf{y}
 \end{pmatrix}
\]
generate a submodule $M$ of $S^2$ where $S = \FF[x,y]$ is a polynomial ring.
They form a Gr\"obner basis, and the initial terms are marked as bold.
A member of $S^2$
\[
 \begin{pmatrix}
 x^2 + x^2 y - y^2 \\
 y+2 x^2 y
 \end{pmatrix}
\]
is in $M$ because the following ``division'' results in zero.
\[
  \begin{pmatrix}
 x^2+\mathbf{x^2 y}-y^2 \\
 y+2 x^2 y
 \end{pmatrix}
 \xrightarrow{-y \sigma_1}
 \begin{pmatrix}
 x^2 \\
 \mathbf{x^2 y}
 \end{pmatrix}
 \xrightarrow{-x^2 \sigma_2} 0
\]
A comprehensive material can be found in \cite[Chapter~15]{Eisenbud}.

The situation for Laurent polynomial ring is less discussed,
but is not too different.
A direct treatment is due to Pauer and Unterkircher~\cite{PauerUnterkircher1999}.
One introduces a well-order on monomials,
that is preserved by multiplications with respect to a so-called \emph{cone decomposition}.
An ideal $J$ over a Laurent polynomial ring can be thought of
as a collection of configurations of coefficient scalars written on the sites of the integral lattice $\ZZ^D$.
If we take a cone, say,
\[
 C = \{ (i_1,i_2,i_3) \in \ZZ^3 | i_1 \le 0, i_2 \ge 0, i_3 \ge 0 \} ,
\]
then $J_C = J \cap \FF[C]$ looks very similar to an ideal $I$ over a polynomial ring $\FF[x,y,z]$.
Concretely, $I$ can be obtained by applying $x^{-1} \mapsto x, y \mapsto y, z \mapsto z$ to $J_C$.
The initial terms of $J_C$ should be treated similarly as those in $I$.
This is where the cone decomposition plays a role.
The lattice $\ZZ^D$ decomposes into $2^D$ cones,
and the initial terms of $J$ is considered in each of the cones.
Correspondingly, a Gr\"obner basis is defined 
to generate the initial terms of a given module in each of the cones.
An intuitive picture for the division algorithm is 
to consider the support of a Laurent polynomial as a finite subset of $\ZZ^D$ around the origin (the least element of $\ZZ^D$),
and to eliminate outmost points so as to finally reach the origin.
If $m$ is a column matrix of Laurent polynomials,
each step in the division algorithm by a Gr\"obener basis $\{g\}$
replaces $m$ with $m' = m - c g$, where $c$ is a monomial,
such that the initial term of $m'$ is strictly smaller than that of $m$.
Note that the absolute degree of $c$ does not exceed that of $m$.%
\footnote{Strictly speaking, one can introduce a term order such that this is true.}

\begin{proof}[Proof of \ref{lem:local-tqo=exact}]
We have to show that 
if $v \in \ker \epsilon = \im \sigma$ is supported in the diamond of radius $r$ centered at the origin,
then $v$ can be expressed as a linear combination
\[
 v = \sum_i c_i \sigma_i
\]
of the columns $\sigma_i$ of $\sigma$ 
such that the coefficients $c_i \in R$ have absolute degree not exceeding $w+r$.
for some fixed $w$.
A Gr\"obner basis~\cite{PauerUnterkircher1999} is computed solely from the matrix $\sigma$,
and the division algorithm yields desired $c_i$.

Conversely,
suppose $v \in \ker \epsilon$. We have to show $v \in \im \sigma$.
Choose so large $L$ that the Pauli operator $O$ representing $v$ is contained
in a pyramid region far from the boundary.
The local topological order condition implies
that $O$ is a product of terms near the pyramid region.
Since this product expression is independent of the boundary,
we see $v \in \im \sigma$.
\end{proof}

The Buchsbaum-Eisenbud theorem~\cite{BuchsbaumEisenbud1973Exact}
below characterizes an exact sequence
from the properties of connecting maps.
(See also \cite[Theorem~20.9, Proposition~18.2]{Eisenbud},\cite[Chapter~6 Theorem~15]{Northcott}.)
A few notions should be recalled.
Let $\mathbf M$ be a matrix, not necessarily square, over a ring.
A minor is the determinant of a square submatrix of $\mathbf M$.
{\bf $k$-th determinantal ideal} $I_k(\mathbf M)$ is the ideal generated by all $k \times k$ minors of $\mathbf M$.
It is not hard to see that the determinantal ideal is invariant under any invertible
matrix multiplication on either side. 
The {\bf rank} of $\mathbf M$ is the largest $k$ such that $k$-th determinantal ideal is nonzero.
Thus, the rank of a matrix over an arbitrary ring is defined,
although the dimension of the image in general is not defined or is infinite.
The $0$-th determinantal ideal is taken to be the unit ideal by convention.
For a map $\phi$ between free modules,
we write $I(\phi)$ to denote the $k$-th determinantal
ideal of the matrix of $\phi$ where $k$ is the rank of that matrix.
Fitting Lemma~\cite[Corollary-Definition~20.4]{Eisenbud}
states that determinantal ideals only depend on $\coker \phi$.

The {\bf (Krull) dimension} of a ring is the supremum of lengths of chains of prime ideals.
Here, the length of a chain of prime ideals
\[
 \pp_0 \subsetneq \pp_1 \subsetneq \cdots \subsetneq \pp_n
\]
is defined to be $n$.
Most importantly, the dimension of $\FF[x_1,\ldots,x_n]$ is $n$ where $\FF$ is a field, as
\[
 (0) \subset (x_1) \subset (x_1,x_2) \subset \cdots \subset (x_1,\ldots,x_n) .
\]
Dimensions are in general very subtle,
but intuitively, it counts the number of independent `variables.'
Geometrically, a ring is a function space of a geometric space,
and the independent variables define a coordinate system on it.
So the Krull dimension correctly captures the intuitive dimension.
For instance, $y-x^2=0$ defines a parabola in a plane,
and the functions that vanish on the parabola form an ideal $(y-x^2) \subset \FF[x,y]$.
Thus, the function space is identified with $\FF[x,y]/(y-x^2) \cong \FF[x]$,
whose Krull dimension is, as expected, 1.

Facts we need are quite simple:
\begin{itemize}
 \item In a zero-dimensional ring, every prime ideal is maximal.
 \item $\dim R = \dim \FF_2[x_1^{\pm 1},\ldots,x_D^{\pm 1}] = D$
 \item When $I$ is an ideal of $R$, $\dim R/I + \codim I = D$.%
\footnote{
The {\bf codimension} or {\bf height} of a prime ideal $\pp$ is 
the supremum of the lengths of chains of prime ideals contained in $\pp$.
That is, the codimension of $\pp$ is the Krull dimension of the local ring $R_\pp$.
The codimension of an arbitrary ideal $I$ is the minimum of codimensions of primes that contain $I$.
If $S$ is an affine domain, i.e., 
a homomorphic image of a polynomial ring over a field with finitely many variables
such that $S$ has no zero-divisors,
it holds that $\codim I + \dim R/I = \dim S$~\cite[Chapter~13]{Eisenbud}.
}
\end{itemize}

We shall be dealing with three different kinds of `dimensions':
The first one is the spatial dimension $D$,
which has an obvious physical meaning.
The second one is the Krull dimension of a ring, just introduced.
The Krull dimension is upper bounded by the spatial dimension in any case.
The last one is the dimension of some module as a vector space.
Recall that all of our base ring contains a field --- $\FF_2$ for qubits.
The vector space dimension arises naturally
when we actually count the number of orthogonal ground states.
The dimension as a vector space will always be denoted with a subscript like $\dim_{\FF_2}$.

\begin{prop}
[\cite{BuchsbaumEisenbud1973Exact}%
\footnote{The original result is stronger than what is presented here.
It is stated with the \emph{depth}s of the determinantal ideals.}]
If a complex of free modules over a ring
\[
 0 \to F_n \xrightarrow{\phi_n} F_{n-1} \to \cdots \to F_1 \xrightarrow{\phi_1} F_0
\]
is exact, then
\begin{itemize}
\item $\rank F_k = \rank \phi_k + \rank \phi_{k+1}$ for $k=1,\ldots,n-1$
\item $\rank F_n = \rank \phi_n$.
\item $I(\phi_k)=(1)$ or else $\codim I(\phi_{k}) \ge k$ for $k=1,\ldots,n$.
\end{itemize}
\label{prop:exact-sequence}
\end{prop}

\begin{rem}
For an exact code Hamiltonian,
we have an exact sequence $G \xrightarrow{\sigma} P \xrightarrow{\epsilon=\sigma^\dagger \lambda} E$.
As we will see in Lemma~\ref{lem:coker-epsilon-resolution-length-D},
$\coker \sigma$ has a finite free resolution,
and we may apply the Proposition~\ref{prop:exact-sequence}.
Since $\overline{ I_k(\sigma) } = I_k(\epsilon)$ for any $k \ge 0$, we have
\[
 2q = \rank P = \rank \sigma + \rank \epsilon = 2~ \rank \sigma.
\]
The size $2q \times t$ of the matrix $\sigma$ satisfies $t \ge q$.
If $I_q(\sigma) \ne R$, then $\codim I_q(\sigma) \ge 2$.
\label{rem:rank-sigma-m}
\end{rem}

\section{Ground-state degeneracy}
\label{sec:degeneracy}

Let $H(L)$ be the Hamiltonians on finite systems
obtained by imposing periodic boundary conditions 
as in Section~\ref{sec:topological-order}.
A symmetry operator of $H(L)$ is 
a $\mathbb{C}$-linear combination of Pauli operator that commutes with $H(L)$.
In order for a Pauli symmetry operator
to have a nontrivial action on the ground space,
it must not be a product of terms in $H(L)$.
In addition, since $H(L)$ is a sum of Pauli operators,
a symmetry Pauli operator must commute with each term in $H(L)$.
Hence, a symmetry Pauli operator $O$ with nontrivial action on the ground space
must have image $v$ in the Pauli module such that
\[
 v(O) \in \ker \epsilon_L \setminus \im \sigma_L
\]
where
\[
 G / \bb_L G \xrightarrow{ \sigma_L } P / \bb_L P \xrightarrow{ \epsilon_L } E / \bb_L E 
\]
and
\[
 \bb_L = (x^L_1 -1,\ldots, x^L_D -1) \subseteq R, 
\]
which effectively imposes the periodic boundary conditions.
Since each term in $H(L)$ acts as an identity on the ground space,
if $O'$ is a term in $H(L)$, the symmetry operator $O$ and the product $OO'$
has the same action on the ground space. 
$OO'$ is expressed in the Pauli module as $v(O) + v'(O')$ for some $v' \in \im \sigma_L$.
Therefore, the set of Pauli operators of distinct actions on the ground space is
in one-to-one correspondence with the factor module
\[
 K(L) = \ker \epsilon_L ~/~ \im \sigma_L .
\] 
The vector space dimension $\dim_{\FF_2} K(L)$ is precisely
the number of independent Pauli operators that have nontrivial action on the ground space.
Since $\ker \epsilon_L = (\im \sigma_L)^\perp$ by definition of $\epsilon$,
and $\im \sigma_L$ as an $\FF_2$-vector space is a null space of the symplectic vector space $P/\bb_L P$,
it follows that $\ker \epsilon_L = \im \sigma_L \oplus W$ for some hyperbolic subspace $W$.
The quotient space $K(L) \cong W$ is thus hyperbolic and has even vector space dimension $2k$.
Choosing a symplectic basis for $K(L)$, 
it is clear that $K(L)$ represents the tensor product of $k$ qubit-algebras.
Therefore, the ground space degeneracy is exactly $2^k$~\cite{Gottesman1996Saturating,CalderbankRainsShorEtAl1997Quantum}.
In the theory of quantum error correcting codes,
$k$ is called the number of logical qubits,
and the elements of $K(L)$ are called the logical operators.
In this section, $k$ will always denote $\frac{1}{2} \dim_{\FF_2} K$.

\begin{defn}
The {\bf associated ideal} for a code Hamiltonian
is the $q$-th determinantal ideal $I_q(\sigma) \subseteq R$
of the generating map $\sigma$.
Here, $q$ is the number of qubits per site.
The {\bf characteristic dimension} is the Krull dimension $\dim R / I_q(\sigma)$.
\end{defn}

The associated ideals appears in Buchsbaum-Eisenbud theorem~(Proposition~\ref{prop:exact-sequence}),
which says that the homology $K(L)$ is intimately related to the associated ideal.
Imposing boundary conditions such as $x^L=1$
amounts to treating $x$ not as variables any more,
but as a `solution' of the equation $x^L-1=0$.
In order for $K(L)$ to be nonzero,
the `solution' $x$ should make the associated ideal to vanish.
Hence, by investigating the solutions of $I_q(\sigma)$
one can learn about the relation between the degeneracy and the boundary conditions.
Roughly, a large number of solutions of $I_q(\sigma)$ 
compatible with the boundary conditions
means a large degeneracy.
As $d = \dim R/I_q(\sigma)$ is the geometric dimension of the algebraic set defined by $I_q(\sigma)$,
a larger $d$ means a larger number of solutions.
Hence, the characteristic dimension $d$
controls the growth of the degeneracy as a function of the system size.

For example, consider a chain complex over $R = \FF[x^{\pm 1},y^{\pm 1}]$.
\[
0 \to
 R^1 
\xrightarrow{ \partial_2 = \begin{pmatrix} x-1 \\ y-1 \end{pmatrix} } 
 R^2
\xrightarrow{ \partial_1 = \begin{pmatrix} y-1 & -x +1 \end{pmatrix} }
 R^1
\]
It is exact at $R^2$.
The smallest nonzero determinantal ideal $I$ for either $\partial_1$ or $\partial_2$ is $I=(x-1,y-1)$.
If we impose `boundary conditions' such that $x=1$ and $y=1$,
then $I$ becomes zero, and according to Buchsbaum-Eisenbud theorem,
the homology $K$ at $R^2$ should be nontrivial.
Since the solution of $I$ consists of a single point $(1,1)$ on a 2-plane,
it is conceivable that `boundary conditions' of form $\bb_L$
would always give $K(L)$ of a constant $\FF$-dimension,
which is true in this case.
If we insist that the complex is over $R' = \FF[x^{\pm 1},y^{\pm 1},z^{\pm 1}]$,
then the zero set of $I$ is a line $(1,1,z)$ in 3-space;
there are many `solutions.'
In this case, $K^{R'}(L)$ has $\FF$-dimension $2L$.

An obvious example where the homology $K$ is always zero regardless of the boundary conditions
is this:
\[
0 \to
 R^1 
\xrightarrow{ \begin{pmatrix} 1 \\ 0 \end{pmatrix} } 
 R^2
\xrightarrow{ \begin{pmatrix} 0 & 1 \end{pmatrix} }
 R^1
\]
Here, the determinantal ideal is $(1)=R$, and thus has no solution.

The intuition from these examples are made rigorous below.

\subsection{Condition for degenerate Hamiltonians}

A routine yet very important tool is \emph{localization}.
The origin of all difficulties in dealing with general rings
is that nonzero elements do not always have multiplicative inverse;
one cannot easily solve linear equations.
The localization is a powerful technique to get around this problem.
As we build rational numbers from integers by \emph{declaring} that
nonzero numbers have multiplicative inverse,
the localization enlarges a given ring
and \emph{formally allows} certain elements to be invertible.
It is necessary and sometimes desirable not to invert all nonzero elements,
in order for the localization to be useful.
For a consistent definition, we need a multiplicatively closed subset $S$
containing 1, but not containing 0,
of a ring $R$ and declare that the elements of $S$ is invertible.
The new ring is written as $S^{-1}R$,
in which a usual formula $\frac{r_1}{s_1} + \frac{r_2}{s_2} = \frac{r_1 s_2 + r_2 s_1}{s_1 s_2}$ holds.
The original ring naturally maps into $S^{-1}R$ as $\phi : r \mapsto \frac{r}{1}$.
The localization means that one views all data as defined over $S^{-1}R$ via the natural map $\phi$.%
\footnote{It is a functor from the category of $R$-modules to that of $S^{-1}R$-modules.}

A localized ring, by definition, has more invertible elements, and hence has less nontrivial ideals.
In fact, our Laurent polynomial ring is a localized ring of the polynomial ring by inverting monomials,
e.g., $\{ x^i y^j | i,j \ge 0 \}$.
Nontrivial ideals such as $(x)$ or $(x,y)$ in the polynomial ring
become the unit ideal $(1)$ in the Laurent polynomial ring.
Further localizations in this thesis are with respect to prime ideals.
In this case, we say the ring is {\bf localized at a prime ideal $\pp$}.
A prime ideal $\pp$ has a defining property that
$a b \notin \pp$ whenever $a \notin \pp$ and $b \notin \pp$.
Thus, the set-theoretic complement of $\pp$ is a multiplicatively closed set containing 1.
In $(R\setminus \pp)^{-1}R$, denoted by $R_\pp$, any element outside $\pp$ is invertible,
and therefore $\pp$ becomes a unique maximal ideal of $R_\pp$.
Moreover, the localization sometimes simplifies the generators of an ideal.
For instance, if $R=\FF[x,x^{-1}]$ and $\pp = (x-1)$,
the ideal $((x -1)(x^5-x+1)) \subseteq R$ localizes to $(x-1)_\pp \subseteq R_\pp$
since $x^5-x+1$ is an invertible element of $R_\pp$.

An important fact about the localization is that a module is zero
if and only if its localization at every prime ideal is zero.
Further, the localization preserves exact sequences.
So we can analyze a complex by localizing at various prime ideals.
For a thorough treatment about localizations,
see Chapter 3 of \cite{AtiyahMacDonald}.
The term `localization' is from geometric considerations
where a ring is viewed as a function space on a geometric space.

\begin{lem}
Let $I$ be the associated ideal of an exact code Hamiltonian,
and $\mm$ be a prime ideal of $R$.
Then, $I \not\subseteq \mm$ implies that the localized homology
\[
K(L)_\mm = \ker (\epsilon_L)_\mm ~/~ \im (\sigma_L)_\mm
\]
is zero for all $L \ge 1$.
\label{lem:associated-maximal-localized-homology}
\end{lem}
It is a simple variant of a well-known fact that
a module over a local ring is free if its first non-vanishing Fitting ideal is the unit ideal~\cite[Chapter~1 Theorem~12]{Northcott}.
\begin{proof}
Recall that the localization and the factoring commute.
By assumption, 
\[
(I_q(\epsilon))_\mm = \overline{ (I_q(\sigma))_\mm } = (1) = R_\mm =: S.
\]6
Recall that the local ring $S$ has the unique maximal ideal $\mm$,
and any element outside the maximal ideal is a unit.
If every entry of $\epsilon$ is in $\mm$, then $I_q(\epsilon) \subseteq \mm \ne S$.
Therefore, there is a unit entry, and by column and row operations,
$\epsilon$ is brought to
\[
 \epsilon \cong 
\begin{pmatrix}
 1 & 0 \\
 0 & \epsilon'
\end{pmatrix}
\]
where $\epsilon'$ is a submatrix.
It is clear that $I_{q-1}(\epsilon') \subseteq I_q(\epsilon)$
since any $q-1 \times q-1$ submatrix of $\epsilon'$ can be thought of
as a $q \times q$ submatrix of $\epsilon$ where the first column and first row
have the unique nonzero entry 1 at $(1,1)$.
It is also clear that $I_{q-1}(\epsilon') \supseteq I_q(\epsilon)$
since any $q \times q$ submatrix of $\epsilon$ contains
either zero row or column, or the $(1,1)$ entry $1$ of $\epsilon$.
Hence, $I_{q-1}(\epsilon') = (1)$, and
we can keep extracting unit elements
into the diagonal by row and column operations~\cite[Chapter~1 Theorem~12]{Northcott}.
After $q$ steps,
$t \times 2q$ matrix $\epsilon$ becomes precisely
\[
\epsilon \cong 
\begin{pmatrix}
 \id_q & 0 \\
 0 & 0
\end{pmatrix}
\]
where $\id_q$ is the $q \times q$ identity matrix.
Since localization preserves the exact sequence $G \to P \to E$,
$\sigma$ maps to the lower $q$ components of $P$ with respect to the basis
where $\epsilon$ is in the above form.
Since $I_q(\sigma) = (1)$, we must have (after basis change)
\[
\sigma \cong 
\begin{pmatrix}
 0 & 0 \\
 \id_q & 0
\end{pmatrix}.
\]
Therefore, even after factoring by the proper ideal $\bb_L$, 
the homology $K(L) = \ker \epsilon_L ~/~ \im \sigma_L$ is still zero.
\end{proof}
\begin{cor}
The associated ideal of an exact code Hamiltonian
is the unit ideal, i.e., $I_q(\sigma) = R$,
if and only if
\[
K(L) = \ker \epsilon_L ~/~ \im \sigma_L = 0
\]
for all $L \ge 1$.
\label{cor:unit-characteristic-ideal-means-nondegeneracy}
\end{cor}
\begin{proof}
If $I(\sigma) = R$,
$I(\sigma)$ is not contained in any prime ideal $\mm$.
The above lemma says $K(L)_\mm = 0$.
Since a module is zero 
if and only if its localization at every prime ideal is zero,
$K(L) = 0$ for all $L \ge 1$.

For the converse, observe that
if $\FF$ is any extension field of $\FF_2$,
for any $\FF_2$-vector space $W$,
we have $\dim_\FF \FF \otimes_{\FF_2} W = \dim_{\FF_2} W$.
We replace the ground field $\FF_2$ with its algebraic closure $\FF^a$ to test whether $K(L) \ne 0$.
If $I_q(\sigma)$ is not the unit ideal, then it is contained in a maximal ideal $\mm \subsetneq R$.
By Nullstellensatz, $\mm = (x_1 - a_1,\ldots,x_D - a_D)$ for some $a_i \in \FF^a$.
Since in $R$ any monomial is a unit, we have $a_i \ne 0$.
Therefore, there exists $L \ge 1$ such that $a_i^L = 1$ and $2 \nmid L$.
The equation $x^L-1=0$ has no multiple root.

We claim that $K(L) \ne 0$. It is enough to verify this for the localization at $\mm$.
Since anything outside $\mm$ is a unit in $R_\mm$
and each $x_i^L-1$ contains exactly one $x_i - a_i$ factor,
we see $(\bb_L)_\mm = \mm_\mm$.
Therefore, $(\epsilon_L)_\mm = \epsilon_\mm / (\bb_L)_\mm$ and
$(\sigma_L)_\mm = \sigma_\mm / (\bb_L)_\mm$ is a matrix over the field $R/\mm = \FF^a$.
Since $I_q(\sigma) \subseteq \mm$, we have $I_q(\sigma_L)_\mm = 0$.
That is, $\rank_{\FF^a} (\sigma_L)_\mm < q$.
It is clear that $\dim_{\FF^a} K(L)_\mm = \dim_{\FF^a} \ker (\epsilon_L)_\mm / \im (\sigma_L)_\mm \ge 1$.
\end{proof}
This corollary says that in order to have a {\em degenerate} Hamiltonian $H(L)$,
one must have a proper associated ideal.
We shall simply speak of a {\bf degenerate} code Hamiltonian
if its associated ideal is proper.

\subsection{Counting points in algebraic varieties}
It is important that the factor ring
\[
 R/\bb_L = \FF_2 [x_1, \ldots, x_D]~/~(x^L_1 -1,\ldots,x^L_D -1)
\]
is finite dimensional as a vector space over $\FF_2$,
and hence is Artinian. In fact, $\dim_{\FF_2} R/\bb_L = L^D$.
This ring appears also in \cite{GueneriOezbudak2008}.
Due to the following structure theorem of Artinian rings,
$K(L)$ can be explicitly analyzed by the localizations.
\begin{prop}
[Chapter~8 of \cite{AtiyahMacDonald}, Section~2.4 of \cite{Eisenbud}]
Let $S$ be an Artinian ring.
(For example, $S$ is a homomorphic image of a polynomial ring 
over finitely many variables with coefficients in a field $\FF$,
and is finite dimensional as a vector space over $\FF$.)
Then, there are only finitely many maximal ideals of $S$, and
\[
 S \cong \bigoplus_\mm S_\mm
\]
where the sum is over all maximal ideals $\mm$ of $S$ and $S_\mm$ is the localization of $S$ at $\mm$. 
\label{prop:Artin-ring}
\end{prop}
The following calculation tool is sometimes useful.
Recall that a group algebra is equipped with a non-degenerate 
scalar product $\langle v,w \rangle = \tr (v \bar w)$.
This scalar product naturally extends to a direct sum of group algebras.
\begin{lem}
Let $\FF$ be a field,
and $S = \FF[\Lambda]$ be the group algebra of a finite abelian group $\Lambda$.
If $N$ is a submodule of $S^n$, then the dual vector space $N^*$
is vector-space isomorphic to $S^n / N^\perp$,
where $\perp$ is with respect to the scalar product $\langle \cdot , \cdot \rangle$.
\end{lem}
\begin{proof}
Consider $\phi : S^n \ni x \mapsto \langle \cdot, x \rangle \in N^*$.
The map $\phi$ is surjective since the scalar product is non-degenerate
and $S^n$ is a finite dimensional vector space.
The kernel of $\phi$ is precisely $N^\perp$.
\end{proof}
\begin{cor}
Put $2k = \dim_{\FF_2} K(L)$. Then,
\[
 k = qL^D - \dim_{\FF_2} \im \sigma_L = \dim_{\FF_2} \ker \epsilon_L - qL^D.
\]
Further, if $q=t$, then
\[
 k = \dim_{\FF_2} \ker \sigma_L = \dim_{\FF_2} \coker \epsilon_L .
\]
\label{cor:k-formulas}
\end{cor}
\noindent
The first formula is a rephrasing of the fact that the number of encoded qubits
is the total number of qubits minus the number of independent stabilizer 
generators~\cite{Gottesman1996Saturating,CalderbankRainsShorEtAl1997Quantum}.
\begin{proof}
Put $S = R/\bb_L$. 
If $v_1,\ldots, v_t$ denote the columns of $\sigma_L$, we have
\begin{equation}
\ker \sigma_L^\dagger = \lambda_q \ker \epsilon_L 
= \bigcap_i v_i^\perp = \left( \sum_i S v_i \right)^\perp = \left( \im \sigma_L \right)^\perp.
\label{eq:orthogonal-to-dual}
\end{equation}
Hence, $\dim_{\FF_2} \ker \epsilon_L = \dim_{\FF_2} S^{2q} - \dim_{\FF_2} \im \sigma_L.$
Since $\dim_{\FF_2} S = L^D$ and $K(L) = \ker \epsilon_L / \im \sigma_L$,
the first claim follows.

Since $\im \sigma_L \cong S^t / \ker \sigma_L$, if $t=q$,
we have $k = \dim_{\FF_2} \ker \sigma_L$ by the first claim.
From Eq.~\eqref{eq:orthogonal-to-dual}, we conclude that 
$k = \dim_{\FF_2} S^t / \im \sigma_L^\dagger = \dim_{\FF_2} \coker \epsilon_L$.
\end{proof}
\noindent We will apply these formulas in Section~\ref{sec:cubic-code} and Example~\ref{eg:ChamonModel}.

The characteristic dimension is related to the rate at which
the degeneracy increases as the system size increases
in the following sense.
Recall that $2k = \dim_{\FF_2} K(L)$ and the ground-state degeneracy is $2^k$.
\begin{lem}
Suppose $2 \nmid L$. Let $\FF^a$ be the algebraic closure of $\FF_2$.
If $N$ is the number of maximal ideals in $\FF^a \otimes_{\FF_2} R$ 
that contain $\bb_L + I_q(\sigma)$, 
then
\[
N \le \dim_{\FF_2} K(L) \le 2q N.
\]
\label{lem:K-L-number-points}
\end{lem}
\begin{proof}
We replace the ground field $\FF_2$ with $\FF^a$.
Any maximal ideal of an Artinian ring $\FF^a[x_i^{\pm 1}]/\bb_L$ is of form
$\mm = (x_1 - a_1, \ldots, x_D - a_D)$ where $a_i^L = 1$ by Nullstellensatz.
Since $2 \nmid L$, we see that $(\bb_L)_\mm = \mm_\mm$ 
and that $(R/\bb_L)_\mm \cong \FF^a$ is the ground field.
(See the proof of Corollary~\ref{cor:unit-characteristic-ideal-means-nondegeneracy}.)

Now, $I_q(\sigma) + \bb_L \subseteq \mm$
iff $I_q(\sigma)_\mm + (\bb_L)_\mm \subseteq \mm_\mm = (\bb_L)_\mm$
iff $I_q(\sigma)$ becomes zero over $R_\mm / (\bb_L)_\mm \cong \FF^a$
iff $1 \le \dim_{\FF^a} K(L)_\mm \le 2q$.
Since by Proposition~\ref{prop:Artin-ring},
$K(L)$ is a finite direct sum of localized ones,
we are done.
\end{proof}

\begin{lem}
Let $I$ be an ideal such that $\dim R/I = d$. We have
\[
\dim_{\FF_2} R / (I+\bb_L) \le cL^d
\]
for all $L \ge 1$ and some constant $c$ independent of $L$.
\end{lem}
\begin{proof}
We replace the ground field with its algebraic closure $\FF^a$.
Write $\tilde x_i$ for the image of $x_i$ in $R/I$.
By Noether normalization theorem~\cite[Theorem~13.3]{Eisenbud},
there exist $y_1,\ldots, y_d \in R/I$ such that
$R/I$ is a finitely generated module over $\FF^a[y_1,\ldots,y_d]$.
Moreover, one can choose $y_i = \sum_{j=1}^D M_{ij} \tilde x_j$
for some rank $d$ matrix $M$ whose entries are in $\FF^a$.
Making $M$ into the reduced row echelon form,
we may assume $y_i = \tilde x_i + \sum_{j>d} a_{ij} \tilde x_j$
for each $1 \le i \le d$.

Let $S=\FF^a[z_1,\ldots,z_D]$ be a polynomial ring in $D$ variables.
Let $\phi : S \to R/( I + \bb_L )$ be the ring homomorphism
such that $z_i \mapsto y_i$ for $1 \le i \le d$
and $z_j \mapsto \tilde x_j$ for $ d < j \le D$.
By the choice of $y_i$, $\phi$ is clearly surjective.
Consider the ideal $J$ of $S$ generated by the initial terms of $\ker \phi$ 
with respect to the lexicographical monomial order in which $z_1 \prec \cdots \prec z_D$.
Since $\tilde x_j$ is integral over $\FF[y_1,\ldots,y_d]$,
the monomial ideal $J$ contains $z_j^{n_j}$ for some positive $n_j$ for $d < j \le D$.
Here, $n_j$ is independent of $L$.
Since $z_i^L \in J$ for $1 \le i \le d$, we conclude that
\[
 \dim_{\FF^a} R/(I+\bb_L) = \dim_{\FF^a} S/J \le L^d \cdot n_{d+1} n_{d+2} \cdots n_{D}
\]
by Macaulay theorem \cite[Theorem~15.3]{Eisenbud}.
\end{proof}

\begin{cor}
If $2 \nmid L$, and $d = \dim R/I_q(\sigma)$ is the characteristic dimension of a code Hamiltonian,
then 
\[
 \dim_{\FF_2} K(L) \le c L^d
\]
for some constant $c$ independent of $L$.
\label{cor:upper-bound-k}
\end{cor}
\begin{proof}
If $J  = \bb_L + I(\sigma)$,
$N$ in Lemma~\ref{lem:K-L-number-points} is equal to $\dim_{\FF^a} \FF^a \otimes R /\mathop{\mathrm{rad}} J$.
This is at most $\dim_{\FF^a} \FF^a \otimes R / J = \dim_{\FF_2} R/J$.
\end{proof}

\begin{lem}
Let $d$ be the characteristic dimension.
There exists an infinite set of integers $\{ L_i \}$ such that
\[
 \dim_{\FF_2} K(L_i) \ge {L_i}^d /2
\]
\label{lem:k-growing-sequence-L}
\end{lem}
\begin{proof}
We replace the ground field with its algebraic closure $\FF^a$.
Let $\pp' \supseteq I(\sigma)$ be a prime of $R$ of codimension $D-d$.
Let $\pp$ be the contraction (pull-back) of $\pp'$ in the polynomial ring
$S = \FF^a[x_1,\ldots,x_D]$. Since the set of all primes of $R$ is in one-to-one
correspondence with the set of primes in $S$ that does not include monomials,
it follows that $\pp$ has codimension $D-d$ and does not contain any monomials.
Let $V$ denote the affine variety defined by $\pp=(g_1,\ldots,g_n)$.
Since $\pp$ contains no monomials, $V$ is not contained in any hyperplanes $x_i = 0$
($i = 1,\ldots, D$).

Let $A_1$ be a finite subfield of $\FF^a$ that contains all the coefficients of $g_i$,
so $V$ can be defined over $A_1$.
Let $A_n \subseteq \FF^a$ be the finite extension fields of $A_1$ of extension degree $n$.
Put $L_n = |A_n|-1$.
For any subfield $A$ of $\FF^a$, let us say a point of $V$ is \emph{rational} over $A$
if its coordinates are in $A$.
The number $N'(L_n)$ of points $(a_i) \in V$ satisfying $a_i^{L_n} = 1$
is precisely the number of the rational points of $V$ over $A_n$
that are not contained in the hyperplanes $x_i=0$.
Since $I(\sigma) \subseteq \pp'$,
the number $N$ in Lemma~\ref{lem:K-L-number-points} is at least $N'(L_n)$.
It remains to show $N'(L_n) \ge L_n^d /2$ for all sufficiently large $n$.

This follows from the result by Lang and Weil~\cite{LangWeil1954},
which states that the number of points of a projective variety
of dimension $d$ that are rational over a finite field of $m$ elements
is $m^d + O\left(m^{d-\frac{1}{2}} \right)$ asymptotically in $m$.
Since Lang-Weil theorem is for projective variety
and we are with an affine variety $V$,
we need to subtract the number of points in the hyperplanes $x_i = 0$
($i = 0,1,\ldots,D$)
from the Zariski closure of $V$.
The subvarieties in the hyperplanes, being closed,
have strictly smaller dimensions, and we are done.
\end{proof}

\section{Fractal operators and topological charges}
\label{sec:fractal-operators}

This section is to provide a characterization of topological charges,
and their dynamical properties.
Before we turn to a general characterization and define fractal operators,
let us review familiar examples.
Note that for two dimensions the base ring is $R = \FF_2[x,\bar x, y, \bar y]$.

\begin{example}[Toric Code]
\label{eg:2d-toric}
Although the original two-dimensional toric code~\cite{Kitaev2003Fault-tolerant}
has qubits on edges,
we put two qubits per site of the square lattice to fit it into our setting.
Concretely, the first qubit to each site represents the one on its east edge,
and the second qubit the one on its north edge. With this convention,
the Hamiltonian is the negative sum of the following two types of interactions:
\[
\xymatrix@!0{
XI \ar@{-}[r] & XX \ar@{-}[d] \\
II \ar@{-}[u] & IX \ar@{-}[l]
} \quad
\xymatrix@!0{
ZI \ar@{-}[r] & II \ar@{-}[d] \\
ZZ \ar@{-}[u] & IZ \ar@{-}[l]
} \quad \quad
\xymatrix@!0{
y \ar@{-}[r] & xy \ar@{-}[d] \\
1 \ar@{-}[u] & x \ar@{-}[l]
}
\]
where we used $X,Z$ to abbreviate $\sigma_x,\sigma_z$, and omitted the tensor product symbol.
Here, the third square specifies the coordinate system of the square lattice.
Since there are $q=2$ qubits per site, the Pauli module is of rank 4.
The corresponding generating map $\sigma : R^2 \to R^4$ is given by the matrix
\[
 \sigma_{\text{2D-toric}} = 
\begin{pmatrix}
 y+xy & 0 \\
 x+xy & 0 \\
\hline
 0    & 1 + y \\
 0    & 1 + x
\end{pmatrix}
\cong
\begin{pmatrix}
 1 + \bar x & 0 \\
 1 + \bar y & 0 \\
\hline
 0          & 1 + y \\
 0          & 1 + x
\end{pmatrix} .
\]
Here, the each column expresses each type of interaction.
It is clear that
\[
\epsilon_{\text{2D-toric}} = \sigma^\dagger \lambda_2 =
\begin{pmatrix}
 0         & 0         & 1 + x & 1 + y \\
 1+ \bar y & 1+ \bar x & 0     & 0
\end{pmatrix}
\]
and $\ker \epsilon = \im \sigma$;
the two-dimensional toric code satisfies our exactness condition.
The associated ideal is $I(\sigma) = ( (1+x)^2, (1+x)(1+y), (1+y)^2 )$.
The characteristic dimension is $\dim R / I(\sigma) = 0$.
Note also that $\ann \coker \epsilon = (x-1,y-1)$.
The electric and magnetic charge are represented by
$\begin{pmatrix} 1 \\ 0 \end{pmatrix}, \begin{pmatrix} 0 \\ 1 \end{pmatrix} \in E \setminus \im \epsilon$,
respectively.

The connection with cellular homology should be mentioned.
$\sigma$ can be viewed as the boundary map
from the free module of all 2-cells with $\mathbb{Z}_2$ coefficients
of the cell structure of 2-torus
induced from the tessellation by the square lattice.
Then, $\epsilon$ is interpreted as the boundary map from the free module of all 1-cells
to that of all 0-cells.
$\sigma$ or $\epsilon$ is actually the direct sum of two boundary maps.
Indeed, the space 
$K(L) = \ker \epsilon_L / \im \sigma_L$
of operators acting on the ground space (logical operators)
has four generators
\begin{align*}
 l_y(X) = \begin{pmatrix} 1+y+\cdots+y^{L-1} \\ 0 \\ 0 \\ 0 \end{pmatrix}, & & 
 l_x(X) = \begin{pmatrix} 0 \\ 1+ x+ \cdots + x^{L-1} \\ 0 \\ 0 \end{pmatrix}, \\
 l_x(Z) = \begin{pmatrix} 0 \\ 0 \\ 1+x+\cdots+x^{L-1} \\ 0 \end{pmatrix}, & &
 l_y(Z) = \begin{pmatrix} 0 \\ 0 \\ 0 \\ 1+y+\cdots+y^{L-1} \end{pmatrix},
\end{align*}
which correspond to the usual nontrivial first homology classes of 2-torus.

The description by the cellular homology might be advantageous for the toric code
over our description with pure Laurent polynomials;
in this way, it is clear that the toric code can be defined on
an arbitrary tessellation of compact orientable surfaces.
However, it is unclear whether this cellular homology description
is possible after all for other topologically ordered code Hamiltonians.
\hfill $\Diamond$
\end{example}

\begin{example}[2D Ising model on square lattice]
The Ising model has nearest neighbor interactions that are horizontal and vertical.
In our formalism, they are represented as $1+x$ and $1+y$. Thus,
\[
 \sigma_\text{2D Ising} =
\begin{pmatrix}
 0   & 0 \\
 1+x & 1+y
\end{pmatrix} .
\]
As it is not topologically ordered, the complex $G \to P \to E$ is not exact.
Moreover, $\sigma$ is not injective.
\[
 \sigma_\text{2D Ising;1} =
\begin{pmatrix}
 1+y \\
 1+x
\end{pmatrix}
\]
generates the kernel of $\sigma$. That is, the complex 
$0 \to G_1 \xrightarrow{ \sigma_\text{2D Ising;1}} G \xrightarrow{ \sigma_\text{2D Ising} } P$ is exact.
\hfill $\Diamond$
\end{example}

In both examples, there exist isolated excitations.
In the toric code, the isolated excitation can be
(topologically) nontrivial since the electric charge is not in $\im \epsilon$.
On the contrary, in 2D Ising model, any isolated excitation is 
actually created by an operator of finite support
because any excitation created by some Pauli operator appears as several connected loops.
This difference motivates the following definition for charges.

Let $\tilde R$ be the set of all $\FF_2$-valued functions on the translation group $\Lambda$,
not necessarily finitely supported. For instance, if $\Lambda = \ZZ$,
\[
\tilde f = \cdots + x^{-4} + x^{-2} + 1 + x^2 + x^4 + \cdots \in \tilde R
\]
represents a function whose value is 1 at even lattice points, and 0 at odd points.
Note that $\tilde R$ is a $R$-module,
since the multiplication is a convolution between an arbitrary function and a finitely supported function.
For example,
\begin{align*}
 (1+x) \cdot \tilde f &= \cdots + x^{-2} + x^{-1} + 1 + x + x^2 + \cdots , \\
 (1+x)^2 \cdot \tilde f &= 0 .
\end{align*}
Let $\tilde P = \tilde R^{2q}$ be the module of Pauli operators of possibly infinite support.
Similarly, let $\tilde E$ be the module of virtual excitations of possibly infinitely many terms.
Formally, $\tilde P$ is the module of all $2q$-tuples of functions on the translation group,
and $\tilde E$ is that of all $t$-tuples.
Clearly, $P \subseteq \tilde P$ and $E \subseteq \tilde E$.
The containment is strict if and only if the translation group is infinite.
Since the matrix $\epsilon$ consists of Laurent polynomials with finitely many terms,
$\epsilon:P \to E$ extends to a map from $\tilde P$ to $\tilde E$.
\begin{defn}
A {\bf topological charge}, or {\bf charge} for short, $e = \epsilon( \tilde p ) \in E$ 
is an excitation of finite energy (an element of the virtual excitation module)
created by a Pauli operator $\tilde p \in \tilde P$ of possibly infinite support.
A charge $e$ is called {\bf trivial} if $e \in \epsilon(P)$.
\end{defn}
\noindent By definition, the set of all charges modulo trivial ones is
in one-to-one correspondence with the superselection sectors.
According to the definition, any charge of 2D Ising model is trivial.
A nontrivial charge may appear due to the following fractal generators.
\begin{defn}
We call zero-divisors on $\coker \epsilon$ as {\bf fractal generators}.
In other words, 
an element $f \in R \setminus \{0\}$ is a {\em fractal generator} if
there exists $v \in E \setminus \im \epsilon$ 
such that $f v \in \im \epsilon$.
\end{defn}

There is a natural reason the fractal generator deserves its name.
Consider a code Hamiltonian with a single type of interaction: $t=1$.
So each configuration of excitations is described by one Laurent polynomial.
For example, in two dimensions, $f = 1 + x + y = \epsilon (p)$ represents
three excitations, one at the origin of the lattice and the others at $(1,0)$ and $(0,1)$
created by a Pauli operator represented by $p$.
(This example is adopted from \cite{NewmanMoore1999Glassy}.)
In order to avoid repeating phrase,
let us call each element of the Pauli module a Pauli operator,
and instead of using multiplicative notation
we use module operation $+$ to mean the product of the corresponding Pauli operators.

Consider the Pauli operator $fp = p + x p + y p \in P$. 
It describes the Pauli operator $p$ at the origin 
multiplied by the translations of $p$ at $(1,0)$ and at $(0,1)$.
So $fp$ consists of three copies of $p$.
This Pauli operator maps the ground state to the excited state $f^2 = 1 + x^2 + y^2$.
The number of excitations is still three, but the excitations at $(1,0), (0,1)$ 
have been replaced by those at $(2,0),(0,2)$.
Similarly, the Pauli operator $f^{2+1} p = f^2 (fp)$ 
consists of three copies of $fp$, or $3^2$ copies of $p$.
The excited state created by $f^3 p$ is $f^4 = (f^2)^2 = 1 + x^{2^2} + y^{2^2}$.
Still it has three excitations, but they are further apart.
The Pauli operator $f^{2^n -1} p$ consists of $3^n$ copies of $p$ in a self-similar way,
and the excited state caused by $f^{2^n -1} p$ consists of a constant number of excitations.
More generally, if there are $t > 1$ types of terms in the Hamiltonian,
the excitations are described by a $t \times 1$ matrix.
If it happens to be of form $f v$ for some $f \in R$ consisted of two or more terms,
there is a family of Pauli operators $f^{2^n -1} p$ with self-similar support
such that it only creates a bounded number of excitations.
An obvious but uninteresting way to have such a situation is to put
$f v = \epsilon (f p')$ for a Pauli operator $p'$ where $v = \epsilon( p' )$.
Our definition avoids this triviality by requiring $v \notin \im \epsilon$.
The reader may wish to compare the fractals with finite cellular automata~\cite{MartinOdlyzkoAndrewWolfram1984}.

\begin{prop}
\cite[16.33]{BrunsVetter}
Suppose $\coker \epsilon \ne 0$.
Then, the following are equivalent:
\begin{itemize}
 \item There does not exist a fractal generator.
 \item $\coker \epsilon$ is torsion-free.
 \item There exists a free $R$-module $E'$ of finite rank such that
\[
 P \xrightarrow{\epsilon} E \to E'
\]
is exact.
\label{prop:nofractal-torsionless}
\end{itemize}
\end{prop}
\begin{proof}
The first two are equivalent by definition.
The sequence above is exact if and only if
$ 0 \to \coker{\epsilon} \to E' $
is exact.
Since $\coker \epsilon$ has a finite free resolution,
the second is equivalent to the third.
\end{proof}
The following theorem states that the fractal operators produces all nontrivial charges.
\begin{theorem}
Suppose $\Lambda = \mathbb{Z}^D$ is the translation group of the underlying lattice.
The set of all charges modulo trivial ones is in one-to-one correspondence
with the torsion submodule of $\coker \epsilon$.
\label{thm:charge-equals-torsion}
\end{theorem}

To illustrate the idea of the proof,
consider a (classical) excitation map%
\footnote{
It is classical because it is not derived from an interesting quantum commuting Pauli Hamiltonian.
For a classical Hamiltonian where all terms are tensor products of $\sigma_z$,
there is no need to keep a $t \times 2q$ matrix $\epsilon$
since the right half $\epsilon$ is zero.
Just the left half suffices, which can be arbitrary
since the commutativity equation $\epsilon \lambda \epsilon^\dagger =0$ is automatic.
Nevertheless, the excitations and fractal operators are relevant.
Our proof of the theorem is not contingent on the commutativity equation.
}
\[
 \phi =
 \begin{pmatrix}
  1+x+y & 0   \\
     0  & 1+x \\
     0  & 1+y \\
 \end{pmatrix}
 : R^2 \to R^3 .
\]
A nonzero element $f=1+x+y \in R$ is a fractal generator
since $\begin{pmatrix} 1 & 0 & 0 \end{pmatrix}^T \notin \im \phi$ 
and $(1+x+y)\begin{pmatrix} 1 & 0 & 0 \end{pmatrix}^T \in \im \phi$;
$f$ is a zero-divisor on a torsion element $\begin{pmatrix} 1 & 0 & 0 \end{pmatrix}^T \in \coker \epsilon$.
It is indeed a charge since 
$\phi(\tilde f \begin{pmatrix} 1 & 0 \end{pmatrix}^T) = \begin{pmatrix} 1 & 0 & 0 \end{pmatrix}^T$ where
\[
 \tilde f = \lim_{n \to \infty} f^{2^n-1} \in \FF[[x,y]]
\]
is a formal power series, which can be viewed as an element of $\tilde R$.
The limit is well-defined since $f^{2^{n+1}-1} - f^{2^n -1}$ only contains terms of degree $2^n$ or higher.
That is to say, only higher order `corrections' are added and lower order terms are not affected.
Of course, there is no natural notion of smallness in the ring $\FF[x,y]$.
But one can formally call the members of the ideal power $(x,y)^n \subseteq \FF[x,y]$ \emph{small}.
It is legitimate to introduce a topology in $R$ defined by the \emph{ever shrinking} ideal powers $(x,y)^n$.
They play a role analogous to the ball of radius $1/n$ in a metric topological space.
The \emph{completion} of $\FF[x,y]$
where every Cauchy sequence with respect to this topology is promoted to a convergent sequence,
is nothing but the formal power series ring $\FF[[x,y]]$.
For a detailed treatment, see Chapter 10 of \cite{AtiyahMacDonald}.

The completion and the limit only make sense in the polynomial ring $\FF[x,y]$.
The reason $\tilde f$ is well-defined is that $f \in \FF[x^{\pm 1},y^{\pm 1}]$ is accidentally expressed
as a usual polynomial with lowest order term $1$.
In the proof below we show that every fractal generator can be expressed in this way.
Hence, a torsion element of $\coker \epsilon$ is really a charge.

\begin{proof}
For a module $M$, let $T(M)$ denote the torsion submodule of $M$:
\[
 T(M) = \{ m \in M \ |\  \exists \, r \in R \setminus \{0\} \text{ such that } rm = 0 \}
\]

Suppose first that $T(\coker \epsilon) = 0$.
We claim that in this case there is no nontrivial charge.
Let $e = \epsilon( \tilde p ) \in E$ be a charge, where $\tilde p \in \tilde P$.
By Proposition~\ref{prop:nofractal-torsionless}
we have an exact sequence of finitely generated free modules
$ P \xrightarrow{\epsilon} E \xrightarrow{\epsilon_1} E_1$.
Since the matrix $\epsilon_1$ is over $R$,
the complex extends to a complex of modules of tuples of functions on the translation group.
\[
 \tilde P \xrightarrow{\epsilon} \tilde E \xrightarrow{\epsilon_1} \tilde E_1
\]
(This extended sequence may not be exact.)
Then, $\epsilon_1(e) = \epsilon_1( \epsilon( \tilde p ) ) = 0$
since $\epsilon_1 \circ \epsilon = 0$ identically.
But, $e \in E$, and therefore, $e \in \ker \epsilon_1 \cap E = \epsilon(P)$.
It means that $e$ is a trivial charge, i.e., $e$ maps to zero in $\coker \epsilon$,
and proves the claim.%
\footnote{
One may wish to consider $\epsilon$ to consist of the second column of $\phi$ above.
Then $\epsilon_1 = \begin{pmatrix} 1+y & -1-x \end{pmatrix}$.
}

Now, allow $\coker (P \xrightarrow{\epsilon} E)$ to contain torsion elements.
$Q = (\coker \epsilon) / T(\coker \epsilon)$ is torsion-free,
and is finitely presented as $Q = \coker( \epsilon' : P' \to E )$
where $P'$ is a finitely generated free module.
In fact, we may choose $\epsilon'$ by adding more columns representing the generators of
the torsion submodule of $\coker \epsilon$ to the matrix $\epsilon$.
\[
 \epsilon = \begin{pmatrix} \# & \# \\ \# & \# \end{pmatrix} \quad \quad
 \epsilon'= \begin{pmatrix} \# & \# & * & * \\ \# & \# & * & * \end{pmatrix}
\]
Then, $P$ can be regarded as a direct summand of $P'$.%
\footnote{
If we take $\epsilon = \phi$ above,
then
\[
 \epsilon' =  \begin{pmatrix}
  1+x+y & 0   & 1 \\
     0  & 1+x & 0 \\
     0  & 1+y & 0 \\
 \end{pmatrix}.
\]
Note that $P' = P \oplus R$.
}

Let $e = \epsilon( \tilde p ) \in E$ be any charge.
Since the matrix $\epsilon'$ contains $\epsilon$ as submatrix,
we may write $e = \epsilon'(\tilde p) \in E$.
Since $T(\coker \epsilon')=0$,
we see by the first part of the proof that $e = \epsilon'( p' )$ for some $p' \in P'$.
Then, $e$ maps to zero in $Q$, and it follows that $e$ maps into $T(\coker \epsilon)$
in $\coker \epsilon$. 
In other words, the equivalence class of $e$ modulo trivial charges 
is a torsion element of $\coker \epsilon$.

Conversely,
we have to prove that for every element $e \in E$ 
such that $fe = \epsilon(p)$ for some $f \in R \setminus \{0\}$ and $p \in P$,
there exists $\tilde p \in \tilde P$ such that $ e = \epsilon( \tilde p )$.
Here, $\tilde P$ is the module of all $2q$-tuples of 
$\FF_2$-valued functions on the translation group.
Consider the lexicographic total order on $\mathbb{Z}^D$ in which $x_1 \succ x_2 \succ \cdots \succ x_D$.
It induces a total order on the monomials of $R$.
Choose the least term $f_0$ of $f$.
By multiplying $f_0^{-1}$, we may assume $f_0 = 1$.%
\footnote{
If $D=1$, $f$ would be a polynomial of nonnegative exponents with the lowest order term being 1.
If $D=2$ and $f = y + y^2 + x$, then the least term is $y$.
After multiplying $f_0^{-1}$, it becomes $1+y+xy^{-1}$.
}

We claim that the sequence
\begin{equation}
 f,\ \  f^2 f,\ \  f^4 f^2 f,\ \  \ldots,\ \  f^{2^n}f^{2^{n-1}} \cdots f^2 f,\ \  \ldots
\label{eq:ftn-series}
\end{equation}
converges to $\tilde f \in \tilde R$,
where $\tilde R$ is the set of all $\FF_2$-valued functions on $\Lambda$.
Given the claim,
since $f^{2^n} e = e + (f-1)^{2^n}e = \epsilon(f^{2^{n-1}} \cdots f^2 f p)$ where $p \in P$,
we conclude that $e = \epsilon( \tilde f p )$ is a charge.

If $f$ is of nonnegative exponents, and hence $f \in S=\FF_2[x_1,\ldots,x_D]$,
then the claim is clearly true.
Indeed, the positive degree terms of $f^{2^n} = 1 + (f-1)^{2^n}$ 
are in the ideal power $(x_1,\ldots,x_D)^{2^p} \subset S$.
Therefore, the sequence Eq.~\eqref{eq:ftn-series} converges
in the formal power series ring $\FF_2[[x_1,\ldots,x_D]]$,
which can be regarded as a subset of $\tilde R$.
If $f$ is not of nonnegative exponents, one can introduce the following change of basis
of the lattice $\mathbb{Z}^D$ such that $f$ becomes of nonnegative exponents.
In other words, the sequence Eq.~\eqref{eq:ftn-series} is in fact contained in a ring 
that is isomorphic to the formal power series ring, where the convergence is clear.

For any nonnegative integers $m_1,\ldots,m_{D-1}$, define a linear transformation
\[
\zeta_m =
\zeta_{(m_1,m_2,\ldots,m_{D-1})} :
\begin{pmatrix}
 a_1 \\ a_2 \\ \vdots \\ a_D
\end{pmatrix}
\mapsto
 \begin{pmatrix}
  a'_1 \\ a'_2 \\ \vdots \\ a'_D
 \end{pmatrix}
=
\begin{pmatrix}
 1      & 0      & 0      & \cdots & 0       \\
 m_1    & 1      &  0     &        & 0       \\
 m_1    & m_2    &  1     &        & 0       \\
 \vdots &        &        & \ddots & \vdots  \\
 m_1    & m_2    & \cdots & m_{D-1}& 1       \\
\end{pmatrix}
\begin{pmatrix}
 a_1 \\ a_2 \\ \vdots \\ a_D
\end{pmatrix}
\text{ on } \mathbb{Z}^D .
\]
$\zeta_m$ induces the map $x_1^{a_1} \cdots x_D^{a_D} \mapsto x_1^{a'_1} \cdots x_D^{a'_D}$ on $R$.
Let $u=x_1^{a_1} \cdots x_D^{a_D}$ be an arbitrary term of $f$ other than $1$, so $u \succ 1$.
For the smallest $i \in \{1,\ldots,D\}$ such that $a_i \neq 0$,
one has $a_i > 0$ due to the lexicographic order.
Hence, if we choose $m_i$ large enough and set $m_j=0\,(j\neq i)$,
then $\zeta_m (u)$ has nonnegative exponents.
Since any $\zeta_{m}$ maps a nonnegative exponent term to a nonnegative exponent term,
and there are only finitely many terms in $f$,
it follows that
there is a finite composition $\zeta$ of $\zeta_m$'s 
which maps $f$ to a polynomial of nonnegative exponents.%
\footnote{
For our previous example $f = 1+y+xy^{-1}$,
one takes $\zeta : x^i y^j \mapsto x^i y^{i+j}$, so $\zeta(f) = 1 + y + x$.
}
\end{proof}

Since a nontrivial charge $v$ has finite {\bf size} anyway
(the maximum exponent minus the minimum exponent of the Laurent polynomials 
in the $t \times 1$ matrix $v$),
we can say that the charge $v$ is {\bf point-like}.
Moreover, we shall have a description how the point-like charge
can be separated from the other by a local process.
By the {\em local process} we mean a sequence of Pauli operators $[[o_1, \ldots, o_n]]$
such that $o_{i+1} - o_i$ is a monomial.
The number of excitations, i.e., {\em energy}, at an instant $i$ 
will be the number of terms in $\epsilon(o_i)$.

\begin{theorem}
\cite{NewmanMoore1999Glassy}
If there is a fractal generator of a code Hamiltonian,
then for all sufficiently large $r$,
there is a local process starting from the identity
by which a point-like charge is separated from the other excitations by distance at least $2^r$.
One can choose the local process in such a way that at any intermediate step
there are at most $c r$ excitations for some constant $c$ independent of $r$.
\label{thm:logarithmic-energy-barrier-for-fractal-operator}
\end{theorem}
For notational simplicity, we denote the local process $[[o_1,\ldots,o_n]]$
by
\[
s = [o_1,~ o_2-o_1,~ o_3 - o_2, \ldots, o_n - o_{n-1}].
\]
It is a {\em recipe} to construct $o_n$, consisted of single-qubit operators.
$o_n$ can be expressed as ``$o_n = \int s$'', the sum of all elements in the recipe.
\begin{proof}
Let $f$ be a fractal generator, 
and put $fv = \epsilon (p)$
where $v \notin \im \epsilon$.
We already know $v$ is a point-like nontrivial charge.
Write
\[
 p = \sum_{i=1}^n p_i, \quad f = \sum_{i=1}^l f_i
\]
where each of $p_i$ and $f_i$ is a monomial.
Let $s_0 = [ 0, p_1, p_2, \ldots, p_{n}]$ be a recipe for constructing $p$; $\int s_0 = p$.
Given $s_i$, define inductively
\[
 s_{i+1} = (f_1^{2^i} \cdot s_i) \circ (f_2^{2^i} \cdot s_i) \circ \cdots \circ (f_l^{2^i} \cdot s_i) 
\]
where $\circ$ denotes the concatenation and $f_i \cdot [u_1,\ldots,u_{n'}] = [f_i u_1, \ldots, f_i u_{n'}]$.
It is clear that $s_{i+1}$ constructs the Pauli operator
\begin{align*}
\int s_{r} 
= f^{2^{r-1}} \int s_{r-1} = f^{2^{r-1}} f^{2^{r-2}} \int s_{r-2} 
= f^{2^{r-1} + 2^{r-2} + \cdots + 1} \int s_0 = f^{2^r-1} p 
\end{align*}
whose image under $\epsilon$ is $f^{2^r}v$. 
Thus, if $r$ is large enough so that $2^r$ is greater than the size of $v$,
the configuration of excitations is precisely $l$ copies of $v$.
The distance between $v$'s is at least $2^r$ minus twice the size of $v$.

Therefore, there is a constant $e > 0$ such that for any $r \ge 0$
the energy of $f^{2^r}v \in E$ is $\le e$.
Let $\Delta(r)$ be the maximum energy during the process $s_r$.
We prove by induction on $r$ that 
\[
\Delta(r) \le el (r+1).
\]
When $r = 0$, it is trivial.
In $s_{r+1}$, the energy is $\le \Delta(r)$ until $f_1^{2^r} s_r$ is finished.
At the end of $f_1^{2^r} s_r$, the energy is $\le e $.
During the subsequent $f_2^{2^r} s_r$, the energy is $\le \Delta(r) + e$,
and at the end of $(f_1^{2^r} s_r) \circ (f_2^{2^r} s_r)$, the energy is $\le 2 e$.
During the subsequent $f_j^{2^r} s_r$, the energy is $\le \Delta(r) + je$.
Therefore, 
\[
\Delta(r+1) \le \Delta(r) + el \le el(r+2)
\]
by the induction hypothesis.
\end{proof}

Fractal operators appear in Newman-Moore model~\cite{NewmanMoore1999Glassy} where classical spin glass is discussed.
Their model has generating matrix $\sigma = \begin{pmatrix} 0 & 1+x+y \end{pmatrix}^T$.
The theorem is a simple generalization of Newman and Moore's construction.
Another explicit example of fractal operators 
in a quantum model can be found in Section~\ref{sec:numerics}.

Note that the notion of fractal generators includes that of ``string operators.''
In fact, a fractal generator that contains exactly two terms
gives a family of nontrivial {\em string segments} of unbounded length,
as defined in Chapter~\ref{chap:cubic-code}.

Below, we point out a couple of sufficient conditions 
for nontrivial charges, or equivalently, fractal generators to exist.
\begin{prop}
For code Hamiltonians,
the existence of a fractal generator is a property of an equivalence class of Hamiltonians
in the sense of Definition~\ref{defn:equiv-H}.
\label{prop:fractal-property-of-eq-class}
\end{prop}
\begin{proof}
Suppose $\im \sigma = \im \sigma'$.
Each column of $\sigma'$ is a $R$-linear combination of those of $\sigma$, and vice versa.
Thus, there is a matrix $B$ and $B'$ such that $\epsilon' = B\epsilon$ and $\epsilon = B' \epsilon'$.
$BB'$ and $B'B$ are identity on $\im \epsilon'$ and $\im \epsilon$, respectively.
In particular, $B'$ and $B$ are injective on $\im \epsilon'$ and $\im \epsilon$, respectively.
Suppose $f$ is a fractal generator for $\epsilon$, i.e., $fv = \epsilon p \ne 0$.
Then, $0 \ne Bfv = f B v = B\epsilon(p) = \epsilon'(p)$.
If $Bv \in \im \epsilon'$, then $v = B'Bv \in \im \epsilon$, a contradiction.
Therefore, $f$ is also a fractal generator for $\epsilon'$.
By symmetry, a fractal generator for $\epsilon'$ is a fractal generator for $\epsilon$, too.

Suppose $R' \subseteq R$ is a coarse-grained base ring.
If $\coker \epsilon$ is torsion-free as an $R$-module,
then so it is as an $R'$-module.
If $f \in R$ is a fractal generator,
the determinant of $f$ as a matrix over $R'$ is a fractal generator.

A symplectic transformation or tensoring ancillas
does not change $\coker \epsilon$.
\end{proof}

\begin{prop}
For any ring $S$ and $t \ge 1$,
if $0 \to S^t \to S^{2t} \xrightarrow{\phi} S^t$ is exact and $I(\phi) \ne S$,
then $\coker \phi$ is not torsion-free.
In particular,
for a degenerate exact code Hamiltonian,
if $\sigma$ is injective, then there exists a fractal generator.
\label{prop:injective-sigma-implies-fractal}
\end{prop}
\begin{proof}
By Proposition~\ref{prop:exact-sequence}, $\rank \phi = t$.
Since $0 \subsetneq I_t(\phi) \subsetneq S$ is the initial Fitting ideal,
we have $0 \ne \ann \coker \phi \ne S$. That is, $\coker \phi$ is not torsion-free.

For the second statement, set $S=R$.
If $\sigma$ is injective, we have an exact sequence
\[
 0 \to G \xrightarrow{\sigma} P \xrightarrow{\epsilon} E .
\]
By Remark~\ref{rem:rank-sigma-m}, $t = \rank G = \rank \sigma = \rank \epsilon = q$.
\end{proof}

\begin{prop}
Suppose the characteristic dimension is $D-2$
for a degenerate exact code Hamiltonian.
Then, there exists a fractal generator.
\label{prop:maximal-characteristic-dimension-means-fractal}
\end{prop}
\begin{proof}
Suppose on the contrary there are no fractal generators.
Then, by Proposition~\ref{prop:nofractal-torsionless},
\[
 G \xrightarrow{\sigma} P \xrightarrow{\epsilon} E \to E'
\]
is exact for some finitely generated free module $E'$.
Together with Hilbert syzygy theorem~\cite[Corollary~15.11]{Eisenbud},%
\footnote{We will prove a stronger version in Lemma~\ref{lem:coker-epsilon-resolution-length-D}.}
Proposition~\ref{prop:exact-sequence} implies
$\codim I(\sigma) \ge 3$ unless $I(\sigma) = R$.
But, $\codim I(\sigma) = 2$ and $I(\sigma) \ne R$ by 
Corollary~\ref{cor:unit-characteristic-ideal-means-nondegeneracy}.
This is a contradiction.
\end{proof}

\chapter{Structure of exact Hamiltonians on lattices}
\label{chap:lowD-codes}

One of the goals of the formalism in the previous chapter
is to classify all code Hamiltonians that are translationally invariant.
A satisfactory classification should be 
comprised of a complete set of invariants,
and a table of inequivalent code Hamiltonians or a machinery to obtain them.
In this chapter we present preliminary results towards this goal.

In one dimension, the classification problem is completely solved.
A particularly nice property of one dimension is
that its translation group algebra $\FF_2[x,x^{-1}]$ is a Euclidean domain,
over which any matrix is diagonal 
up to multiplications on the left and right by invertible matrices.
We show that this diagonalization is possible
even if the left multiplications are restricted to symplectic transformations.
Exploiting the finiteness of the ground field $\FF_2$,
we conclude that there are only Ising models~\cite{Yoshida2011Classification}.
An almost identical treatment appears in quantum convolutional
codes~\cite{OllivierTillich2004Convolutional,GrasslRotteler2006Convolutional}.
(A more general statement is proved by Beigi~\cite{Beigi2012OneDimension}.)

In two dimensions, we find that 
any excitation of an exact code Hamiltonian 
is described by a collection of point-like ones,
each of which is attached to a string operator.
Furthermore, it seems that one can prove a stronger statement that there are only toric codes%
~\cite{Yoshida2011Classification,BombinDuclosCianciPoulin2012}, 
although we do not complete the proof.
In three dimensions, we show that there must be point-like excitations for exact code Hamiltonians.
This implies that it is unavoidable in the translation-invariant case
to have a topological charge penetrating through the system,
for which the energy penalty is only logarithmic in the displacement.
If we assume that the degeneracy under periodic boundary conditions should be constant independent of system size,
then the energy penalty is upper bounded by a constant,
reproducing the result of Yoshida~\cite{Yoshida2011feasibility}.


Note that all the lemmas and theorems in this chapter are valid over qudits with prime dimensions.
We conclude this chapter with several examples.

\section{One dimension}

The group algebra $R = \FF_2[x,\bar x]$ for the one-dimensional lattice $\mathbb{Z}$
is a Euclidean domain
where the degree of a polynomial is defined to be the maximum exponent minus the minimum exponent.
(In particular, any monomial has degree $0$.)
Given two polynomials $f,g$ in $R$, one can find their $\gcd$ by the Euclid's algorithm.
It can be viewed as a column operation on the $1 \times 2$ matrix $\begin{pmatrix} f & g \end{pmatrix}$.
Similarly, one can find $\gcd$ of $n$ polynomials by column operations on $1 \times n$ matrix
\[
 \begin{pmatrix}
  f_1 & f_2 & \cdots & f_n
 \end{pmatrix} .
\]
The resulting matrix after the Euclid's algorithm will be
\[
\begin{pmatrix}
\gcd(f_1,\ldots,f_n) & 0 & \cdots & 0
\end{pmatrix} .
\]

Given a matrix $\mathbf M$ of univariate polynomials, we can apply Euclid's algorithm
to the first row and first column by elementary row and column operations
in such a way that the degree of $(1,1)$-entry $\mathbf M_{11}$ decreases unless all other
entries in the first row and column are divisible by $\mathbf M_{11}$.
Since the degree cannot decrease forever, this process must end
with all entries in the first row and column being zero except $\mathbf M_{11}$.
By induction on the number of rows or columns,
we conclude that $\mathbf M$ can be transformed to a diagonal matrix
by the elementary row and column operations.
This is known as the Smith's algorithm.

The following is a consequence of the finiteness of the ground field.
\begin{lem}
Let $\FF$ be a finite field and $S = \FF[x]$ be a polynomial ring.
Let $\phi : S \xrightarrow{f(x) \times} S$ be a $1 \times 1$ matrix such that $f(0) \ne 0$.
$\phi$ can be viewed as an $n \times n$ matrix acting on the free $S'$-module $S$
where $S' = \FF[x']$ and $x' =x^n$.
Then, for some $n \ge 1$, the matrix $\phi$ is transformed by elementary row and column operations
into a diagonal matrix with entries $1$ or $x'-1$.
The number of $x'-1$ entries in the transformed $\phi$ is equal to the degree of $f$.
\label{lem:coarse-graining-single-polynomial-in-1D}
\end{lem}
\begin{proof}
The splitting field $\tilde \FF$ of $f(x)$ is a finite extension of $\FF$.
Since $\tilde \FF$ is finite, every root of $f(x)$ is a root of $x^{n'} -1$ for some $n' \ge 1$.
Choose an integer $p \ge 1$ such that $2^p$ is greater than any multiplicity of the roots of $f(x)$.
Then, clearly $f(x)$ divides $(x^{n'}-1)^{2^p} = x^{2^p n'} - 1$.
Let $n$ be the smallest positive integer such that $f(x)$ divides $x^n-1$.%
\footnote{This part is well known, at least in the linear cyclic coding theory~\cite{MacWilliamsSloane1977}.}

Consider the coarse-graining by $S' = \FF[x']$ where $x'=x^n$.
$S$ is a free $S'$-module of rank $n$, and $(f)$ is now an endomorphism of the module $S$
represented as an $n \times n$ matrix.
Since $f(x)g(x) = x^n -1$ for some $g(x) \in \FF[x]$,
we have
\[
 A B = (x'-1) \id_n
\]
where $x' = x^n$, and $A, B$ are the matrix representation of $f(x)$ and $g(x)$, respectively, as endomorphisms.
$A$ and $B$ have polynomial entries in variable $x'$.
The determinants of $A,B$ are nonzero for their product is $(x'-1)^n \ne 0$.
Let $E_1$ and $E_2$ be the products of elementary matrices such that $A' = E_1 A E_2$ is diagonal.
Such matrices exist by the Smith's algorithm.
Put $B' = E_2^{-1} B E_1^{-1}$.
Then,
\[
 A' B' = E_1 A E_2 E_2^{-1} B E_1^{-1} = E_1 A B E_1^{-1} = (x'-1)\id_n .
\]
Since $A'$ and $I_n$ are diagonal of non-vanishing entries, $B'$ must be diagonal, too.
It follows that the diagonal entries of $A'$ divides $(x'-1)$; that is, they are $1$ or $x'-1$.

The number of entries $x'-1$ can be counted by considering $S/(f(x))$ as an $\FF$-vector space.
It is clear that $\dim_{\FF} S/(f(x)) = \deg f(x)$.
$S/(f(x)) = \coker \phi$ viewed as a $S'$-module is isomorphic to $S'^n / \im A'$,
the vector space dimension of which is precisely the number of $x'-1$ entries in $A'$.
\end{proof}

For example, consider $f(x) = x^2 + x + 1 \in S = \FF_2[x]$.
It is the primitive polynomial of the field $\FF_4$ of four elements over $\FF_2$.
Any element in $\FF_4$ is a solution of $x^4-x=0$.
Since $f(0) = 1$, we see that $n=3$ is the smallest integer such that $f(x)$ divides $x^n-1$.
As a module over $S' = \FF_2[x^3]$, the original ring $S$ is free with (ordered) basis $\{1, x, x^2\}$.
The multiplication by $x$ on $S$ viewed as an endomorphism has a matrix representation
\[
x = 
 \begin{pmatrix}
  0 & 0 & x^3 \\
  1 & 0 & 0   \\
  0 & 1 & 0   
 \end{pmatrix} .
\]
Thus, $f(x)$ as an endomorphism of $S'$-module $S$ has a matrix representation as follows.
\[
 f(x) =
 \begin{pmatrix}
  1 & x^3 & x^3 \\
  1 & 1   & x^3 \\
  1 & 1   & 1  
 \end{pmatrix} 
\cong
 \begin{pmatrix}
  1 & 0      & 0 \\
  0 & x^3 +1 & 0 \\
  0 & 0      & x^3 +1
 \end{pmatrix}
\]
Here, the second matrix is obtained by row and column operations.
There are 2 diagonal entries $x^3+1$ as $f(x)$ is of degree 2.

\begin{theorem}
If $\Lambda = \mathbb{Z}$,
any system governed by a code Hamiltonian
is equivalent to finitely many copies of Ising models,
plus some non-interacting qubits.
In particular, the topological order condition is never satisfied.
\label{thm:1D-classification}
\end{theorem}
\noindent
Yoshida~\cite{Yoshida2011Classification} arrived at a similar conclusion
assuming that the ground space degeneracy when the Hamiltonian is defined on a ring 
should be independent of the length of the ring.
If translation group is trivial,
the proof below reduces to a well-known fact that 
the Clifford group is generated by controlled-NOT, Hadamard, and Phase
gates~\cite[Proposition~15.7]{KitaevShenVyalyi2002CQC}.
The proof in fact implies that the group of all symplectic transformations
in one dimension is generated by elementary symplectic transformations
of Section~\ref{sec:symplectic-transformations}.

We will make use of the elementary symplectic transformations
and coarse-graining to deform $\sigma$ to a familiar form.
Recall that for any elementary row-addition $E$ on the upper block of $\sigma$
there is a unique symplectic transformation that restricts to $E$.

\begin{proof}
Applying Smith's algorithm to the first row and the first column of $2q \times t$ matrix $\sigma$,
one gets
\[
\begin{pmatrix}
  f_1 & 0 \\
  0   & A \\
  \hline
  g_1 & g_2 \\
  \vdots   & B  
\end{pmatrix}
\]
by elementary symplectic transformations.
Let $1 \le i < j \le q$ be integers.
If some $(1,q+j)$-entry is not divisible by $f_1$,
apply Hadamard on $j^\text{th}$ qubit to bring $(q+j)^\text{th}$ row to the upper block,
and then run Euclid's algorithm again to reduce the degree of $(1,1)$-entry.
The degree is a positive integer, so this process must end after a finite number of iteration.
Now every $(q+j,1)$-entry is divisible by $f_1$ and hence can be made to be $0$
by the controlled-NOT-Hadamard:
\[
\begin{pmatrix}
  f_1 & 0 \\
  0   & A \\
  \hline
  g_1 & g_2 \\
  0   & B  
\end{pmatrix} .
\]
Further we may assume $\deg f_1 \le \deg g_1$.
Since $\sigma^\dagger \lambda_q \sigma = 0$,
we have a commutativity condition
\[
 \bar f_1 g_1 - \bar g_1 f_1 = 0 .
\]
Write $f_1 = \alpha x^a + \cdots + \beta x^b $
and $ g_1 = \gamma x^c + \cdots + \delta x^d $
where $ a \le b$ and $c \le d$ and $\alpha,\beta,\gamma,\delta \neq 0$.
Then, $\bar f_1 g_1 = \beta \gamma x^{c-b} + \cdots + \alpha \delta x^{d-a}$.
Since $f_1 \bar g_1 = \bar f_1 g_1$, it must hold that $-(c-b)=d-a$
and $\alpha \delta = \beta \gamma$.
Since $\deg f_1 \le \deg g_1$, we have $d-b = -(c-a) \ge 0$.
The controlled-Phase $E_{1+q,1}( -(x^{d-b} + x^{c-a})\delta / \beta )$
will decrease the degree of $g_1$ by two,
which eventually becomes smaller than $\deg f_1$.
One may then apply Hadamard to swap $f_1$ and $g_1$.
Since the degree of $(1,1)$-entry cannot decrease forever,
the process must end with $g_1 = 0$.

The commutativity condition between $i^\text{th}$ ($i>1$) column and the first
is $ f_1 \bar g_i = 0 $. Since $f_1 \ne 0$, we get $g_i = 0$:
\[
\begin{pmatrix}
  f_1 & 0 \\
  0   & A \\
  \hline
  0 & 0 \\
  0 & B  
\end{pmatrix} .
\]
Continuing, we transform $\sigma$ into a diagonal matrix.
(We have shown that $\sigma$ can be transformed 
via elementary symplectic transformations to the Smith normal form.)

Now the Hamiltonian is a sum of non-interacting purely classical spin chains
plus some non-interacting qubits ($f_i = 0$).
It remains to classify classical spin chains
whose stabilizer module is generated by
\[
\begin{pmatrix}
f
\end{pmatrix}
\]
where we omitted the lower half block.
We can always choose $f=f(x)$ such that $f(x)$ has only nonnegative exponents and
$f(0)\ne 0$ since $x$ is a unit in $R$.
Lemma~\ref{lem:coarse-graining-single-polynomial-in-1D} says that
$(f)$ becomes a diagonal matrix of entries $1$ or $x'-1$
after a suitable coarse-graining followed by a symplectic transformation
and column operations.
$1$ describes the ancilla qubits,
and $x'-1 = x'+1$ does the Ising model.
\end{proof}
Note that an almost identical treatment appears in \cite{GrasslRotteler2006Convolutional}.

\section{Two dimensions}
If $D=2$, the lattice is $\Lambda = \mathbb{Z}^2$, and our base ring is $R = \FF_2 [x,\bar x, y, \bar y]$.

The following asserts that the local relations --- 
a few terms in the Hamiltonian that multiply to identity in a nontrivial way
as in 2D Ising model, or the kernel of $\sigma$ ---
among the terms in a code Hamiltonian,
can be completely removed for exact Hamiltonians in two dimensions~\cite{Bombin2011Structure}.
We prove a more general version.
\begin{lem}
If $G \xrightarrow{\sigma} P \xrightarrow{\epsilon} E$ is exact
over $R = F_2[ x_1, \bar x_1, \ldots, x_D, \bar x_D ]$,
There exists $\sigma' : G' \to P$ such that $\im \sigma' = \im \sigma$ and
\[
 0 \to G_{D-2} \to \cdots \to G_1 \to G' \xrightarrow{\sigma'} P \xrightarrow{\epsilon} E
\]
is an exact sequence of free $R$-modules. If $D=2$, one can choose $\sigma'$ to be injective.
\label{lem:coker-epsilon-resolution-length-D}
\end{lem}
\noindent
The lemma is almost the same as the Hilbert syzygy theorem~\cite[Corollary~15.11]{Eisenbud}
applied to $\coker \epsilon$,
which states that any finitely generated module over a polynomial ring with $n$ variables
has a finite free resolution of length $\le n$, by finitely generated free modules.
A difference is that our two maps on the far right in the resolution 
has to be related as $\epsilon = \sigma^\dagger \lambda$.
To this end, we make use of a constructive version of Hilbert syzygy theorem via Gr\"obner basis.
\begin{prop}\cite[Theorem~15.10, Corollary~15.11]{Eisenbud}
Let $\{ g_1, \ldots, g_n \}$ be a Gr\"obner basis of a submodule of
a free module $M_0$ over a polynomial ring.
Then, the S-polynomials $\tau_{ij}$ of $\{ g_i \}$
in the free module $M_1 = \bigoplus_{i=1}^n S e_i$
generate the syzygies for $\{g_i\}$.
If the variable $x_1,\ldots, x_s$ are absent from the initial terms of $g_i$,
one can define a monomial order on $M_1$
such that $x_1,\ldots,x_{s+1}$ is absent from
the initial terms of $\tau_{ij}$.
If all variables are absent from the initial terms of $g_i$,
then $M_0/(g_1,\ldots,g_n)$ is free.
\label{prop:constructive-syzygy}
\end{prop}
\begin{proof}[Proof of \ref{lem:coker-epsilon-resolution-length-D}]
Without loss of generality
assume that the $t \times 2q$ matrix $\epsilon$ have entries with nonnegative exponents,
so $\epsilon$ has entries in $S=\FF_2[x_1,\ldots,x_D]$.
Below, every module is over the polynomial ring $S$ unless otherwise noted.
Let $E_+$ be the free $S$-module of rank equal to $\mathrm{rank}_R~ E$.

If $g_1,\cdots, g_{2q}$ are the columns of $\epsilon$,
apply Buchberger's algorithm to obtain a Gr\"obner basis
$g_1,\cdots, g_{2q}, \ldots, g_n$ of $\im \epsilon$.
Let $\epsilon'$ be the matrix whose columns are $g_1,\ldots, g_n$.
We regard $\epsilon'$ as a map $M_0 \to E_+$.
By Proposition~\ref{prop:constructive-syzygy},
the initial terms of the syzygy generators (S-polynomials) 
$\tau_{ij}$ for $\{ g_i \}$ lacks the variable $x_1$.
Writing each $\tau_{ij}$ in a column of a matrix $\tau_1$, we have a map $\tau_1 : M_1 \to M_0$.

By induction on $D$, we have an exact sequence
\[
 M_{D} \xrightarrow{\tau_D} M_{D-1} \xrightarrow{\tau_{D-1}}
 \cdots \xrightarrow{\tau_1} M_0 \xrightarrow{\epsilon'} E_+ 
\]
of free $S$-modules, where the initial terms of columns of $\tau_D$ lack all the variables.
By Proposition~\ref{prop:constructive-syzygy} again, $M'_{D-1} = M_{D-1} / \im \tau_{\tau_D}$ is free.
Since $\ker \tau_{D-1} = \im \tau_D$, we have
\[
 0 \to M'_{D-1} \xrightarrow{\tilde \tau_{D-1}}
 \cdots \xrightarrow{\tau_1} M_0 \xrightarrow{\epsilon'} E_+ 
\]
Since $g_{2q+1},\ldots,g_n$ are $S$-linear combinations of $g_1,\ldots, g_{2q}$,
there is a basis change of $M_0$ so that the matrix representation of $\epsilon'$ becomes
\[
 \epsilon' \cong
\begin{pmatrix}
 \epsilon & 0
\end{pmatrix} .
\]
With respect to this basis of $M_0$,
the matrix of $\tau_1$ is
\[
 \tau_1 \cong
\begin{pmatrix}
 \tau_{1u} \\
 \tau_{1d}
\end{pmatrix}
\]
where $\tau_{1u}$ is the upper $2q \times t'$ submatrix.
Since $\ker \epsilon' = \im \tau_1$, 
The first row $r$ of $\tau_{1d}$ should generate $1 \in S$.
(This property is called unimodularity.)
Quillen-Suslin theorem~\cite[Chapter~XXI Theorem~3.5]{Lang}
states that there exists a basis change of $M_1$
such that $r$ becomes $\begin{pmatrix}1 & 0 & \cdots & 0 \end{pmatrix}$.
Then, by some basis change of $M_0$, one can make 
\[
\epsilon' \cong
\begin{pmatrix}
 \epsilon & 0
\end{pmatrix} ,
\quad
 \tau_{1d} \cong
\begin{pmatrix}
 1 & 0 \\
 0 & \tau'_{1d}
\end{pmatrix} .
\]
where $\tau'_{1d}$ is a submatrix.
By induction on the number of rows in $\tau_{1d}$,
we deduce that the matrix of $\tau_1$ can be brought to
\[
\epsilon' \cong
\begin{pmatrix}
 \epsilon & 0
\end{pmatrix} ,
\quad
 \tau_1 \cong
\begin{pmatrix}
 \sigma'' & \sigma' \\
 I & 0
\end{pmatrix}
\]
Note that $\epsilon \sigma'' = 0$ and $\epsilon \sigma' = 0$.
The basis change of $M_0$ by $\begin{pmatrix} I & -\sigma'' \\ 0 & I \end{pmatrix}$ gives
\[
\epsilon' \cong
\begin{pmatrix}
 \epsilon & 0
\end{pmatrix} ,
\quad
 \tau_1 \cong
\begin{pmatrix}
 0 & \sigma' \\
 I & 0
\end{pmatrix} .
\]
The kernel of $\begin{pmatrix} \sigma' \\ 0 \end{pmatrix}$ determines $\ker \tau_1 = \im \tau_2$.
Let $M'_1$ denote the projection of $M_1$ such that the sequence
\[
 0 \to M'_{D-1} \xrightarrow{\tilde \tau_{D-1}}
 \cdots \to M_2 \to M'_1 \xrightarrow{\sigma'} M'_0 \xrightarrow{\epsilon} E_+ 
\]
of free $S$-modules is exact.

Taking the ring of fractions with respect to the multiplicatively closed set 
\[
U = \{x_1^{i_1} \cdots x_D^{i_D} | i_1,\ldots,i_D \ge 0\},
\]
we finally obtain the desired exact sequence over $U^{-1}S = R$ 
with $P = U^{-1}M'_0$ and $E = U^{-1}E_+$.
Since $\im \sigma = \ker \epsilon$, we have $\im \sigma' = \im \sigma$.
\end{proof}

\begin{lem}
Let $R$ be a Laurent polynomial ring in $D$ variables over a finite field $\FF$,
and $N$ be a module over $R$.
Suppose $J = \ann_R N$ is a proper ideal such that $\dim R/J = 0$.
Then, there exists an integer $L \ge 1$ such that
\[
 \ann_{R'} N = (x_1^L -1, \ldots, x_D^L -1) \subseteq R'
\]
where $R' = \FF[x_1^{\pm L},\ldots,x_D^{\pm L}]$ is a subring of $R$.
\label{lem:zero-dim-ideal-over-finite-field}
\end{lem}

This is a variant of Lemma~\ref{lem:coarse-graining-single-polynomial-in-1D}.
The annihilator $J = \ann_R N$ is the set of all elements $r \in R$ such that
$r n = 0$ for any $n \in N$.
It is an ideal;
if $r_1, r_2 \in \ann_R N$, then
$r_1+r_2$ is an annihilator since $(r_1 + r_2)n = r_1 n + r_2 n = 0$,
and $a r_1 \in \ann_R N$ for any $a \in R$ since $(a r_1) n = a(r_1 n)=0$.
If $R' \subseteq R$ is a subring and $N$ is an $R$-module,
$N$ is an $R'$-module naturally.
Clearly, $J' = \ann_{R'} N$ is by definition equal to $(\ann_{R} N) \cap R'$.
Note that $J'$ is the kernel of the composite map $R' \hookrightarrow R \to R/J$.
Hence, we have an algebra homomorphism $\varphi' : R'/J' \to R/J$.
Although $R'$ is a subring, it is isomorphic to $R$ via the correspondence $x_i^L \leftrightarrow x_i$.
Therefore, we may view $\varphi'$ as a map $\varphi: R/I \to R/J$ for some ideal $I \subseteq R$.
It is a homomorphism such that $\varphi(x_i) = x_i^L$.
Considering the algebras as the set of all functions on the algebraic sets $V(I)$ and $V(J)$ defined by $I$ and $J$, respectively,
we obtain a map $\hat \varphi : V(J) \to V(I)$.
Intuitively, $\hat \varphi$ maps each point $(a_1,\ldots,a_D) \in \FF^D$ to $(a_1^L,\ldots,a_D^L) \in \FF^D$.
In a finite field, any nonzero element is a root of unity.
Since $\dim R/J = 0$, which means that $V(J)$ is a finite set, we can find a certain $L$
so $V(I)$ would consist of a single point. A formal proof is as follows.

\begin{proof}
Since $R$ is a finitely generated algebra over a field,
for any maximal ideal $\mm$ of $R$, the field $R/\mm$ is a finite extension of $\FF$
(Nullstellensatz~\cite[Theorem~4.19]{Eisenbud}).
Hence, $R/\mm$ is a finite field.
Since $x_i$ is a unit in $R$, the image $a_i \in R/\mm$ of $x_i$ is nonzero.
$a_i$ being an element of finite field, a power of $a_i$ is $1$.
Therefore, there is a positive integer $n$ such that
$\bb_n = (x_1^n-1,\ldots,x_D^n-1) \subseteq \mm$.
Since $x^n-1$ divides $x^{nn'}-1$,
we see that there exists $n \ge 1$ such that $\bb_n \subseteq \mm_1 \cap \mm_2$
for any two maximal ideals $\mm_1, \mm_2$.
One extends this by induction to any finite number of maximal ideals.

Since $\dim R/J = 0$, any prime ideal of $R/J$ is maximal and 
the Artinian ring $R/J$ has only finitely many maximal ideals.
$\mathrm{rad}~ J$ is then the intersection of the contractions (pull-backs)
of these finitely many maximal ideals.
Therefore, there is $n \ge 1$ such that
\[
 \bb_n \subseteq \mathrm{rad}~ J .
\]
Since $R$ is Noetherian, $(\mathrm{rad}~J)^{p^r} \subseteq J$ for some $r \ge 0$
where $p$ is the characteristic of $\FF$. Hence, we have 
\[
\bb_{np^r} \subseteq \bb_n^{p^r} \subseteq (\mathrm{rad}~J)^{p^r} \subseteq J.
\]

Let $L = np^r$.
If $R' = \FF[x_1^L,\bar x_1^L,\ldots,x_D^L,\bar x_D^L]$,
$\ann_{R'} N$ is nothing but $J \cap R'$.
We have just shown $\bb_L \cap R' \subseteq J \cap R'$.
Since $J$ is a proper ideal, we have $1 \notin J \cap R'$.
Thus, $\bb_L \cap R'= J \cap R'$ since $\bb_L \cap R'$ is maximal in $R'$.
\end{proof}

\begin{theorem}
For any two-dimensional degenerate exact code Hamiltonian,
there exists an equivalent Hamiltonian such that
\[
\ann \coker \epsilon = (x-1,y-1).
\]
Thus, $\coker \epsilon$ is a torsion module.
\label{thm:structure-2d-ann-coker-epsilon}
\end{theorem}
The content of Theorem~\ref{thm:structure-2d-ann-coker-epsilon}
is presented in \cite{Bombin2011Structure}.
We will comment on it after the proof.
\begin{proof}
By Lemma~\ref{lem:coker-epsilon-resolution-length-D},
we can find an equivalent Hamiltonian such that the generating map $\sigma$ for its stabilizer module
is injective:
\[
 0 \to G \xrightarrow{\sigma} P .
\]
Let $t$ be the rank of $G$.
The exactness condition says
\[
 0 \to G \xrightarrow{\sigma} P \xrightarrow{\epsilon} E
\]
is exact where $\epsilon = \sigma^\dagger \lambda_q$ and $E$ has rank $t$.
Applying Proposition~\ref{prop:exact-sequence},
since $\overline{ I(\sigma) } = I(\epsilon)$ 
and hence in particular $\codim I(\sigma) = \codim I(\epsilon)$,
we have that $q=t$ and $\codim I(\epsilon) \ge 2$ if $I(\epsilon) \ne R$.
But, $I(\epsilon) \ne R$ by Corollary~\ref{cor:unit-characteristic-ideal-means-nondegeneracy}.

Since $q=t$, $I(\epsilon)$ is equal to the initial Fitting ideal,
and therefore has the same radical as the annihilator of $\coker \epsilon = E/ \im \epsilon$.
(See \cite[Proposition~20.7]{Eisenbud} or \cite[Chapter~XIX Proposition~2.5]{Lang}.)
In particular, $\dim R / (\ann \coker \epsilon) = 0$.
Apply Lemma~\ref{lem:zero-dim-ideal-over-finite-field} to conclude the proof.
\end{proof}

An interpretation of the theorem is the following.
For systems of qubits, Theorem~\ref{thm:structure-2d-ann-coker-epsilon} says that
$x+1$ and $y+1$ are in $\ann \coker \epsilon$.
In other words, any element $v$ of $E$ is a charge,
and a pair of $v$'s of distance 1 apart
can be created by a local operator.
Equivalently, $v$ can be translated by distance 1 by the local operator.
Since translation by distance 1 generates all translations of the lattice,
we see that any excitation can be moved through the system by some sequence of local operators.
This is exactly what happens in the 2D toric code:
Any excited state is described by a configuration of magnetic and electric charge,
which can be moved to a different position by a string operator.

Moreover, since $(x-1,y-1) = \ann \coker \epsilon$,
the action of $x,y \in R$ on $\coker \epsilon$ is the same as the identity action.
Therefore, the $R$-module $\coker \epsilon$ is completely determined
up to isomorphism by its dimension $k$ as an $\FF_2$-vector space.
In particular, $\coker \epsilon$ is a finite set,
which means there are finitely many charges.
The module $K(L)$ of Pauli operators acting on the ground space (logical operators),
can be viewed as $K(L) = \Tor_1(\coker \epsilon, R/\bb_L)$.
Thus, $K(L)$ is determined by $k$ up to $R$-module isomorphisms.
This implies that the translations of a logical operator are all equivalent.
It is not too obvious at this moment whether the symplectic structure,
or the commutation relations among the logical operators,
of $K(L)$ is also completely determined.

Yoshida~\cite{Yoshida2011Classification} argued a similar result assuming
that the ground-state degeneracy should be independent of system size.
Bombin~\cite{Bombin2011Structure} later claimed without the constant degeneracy assumption
that one can choose locally independent stabilizer generators in a `translationally invariant way'
in two dimensions, for which Lemma~\ref{lem:coker-epsilon-resolution-length-D} is a generalization,
and that there are finitely many topological charges,
which is immediate from Theorem~\ref{thm:structure-2d-ann-coker-epsilon} since $\coker \epsilon$ is a finite set.
The claim is further strengthened assuming extra conditions by Bombin~\emph{et al.}~\cite{BombinDuclosCianciPoulin2012},
which can be summarized by saying
that $\sigma$ is a finite direct sum of $\sigma_\text{2D-toric}$ in Example~\ref{eg:2d-toric}.

\begin{rem}
\label{rem:long-string}
Although the strings are capable of moving charges on the lattice,
it could be very long compared to the interaction range.
Consider
\[
 \epsilon_p = 
\begin{pmatrix}
 p(x) & p(y) & 0         & 0          \\
 0    & 0    & p(\bar y) & -p(\bar x)
\end{pmatrix}
\]
where $p$ is any polynomial. It defines an exact code Hamiltonian.
For instance, the choice $p(t) = t-1$ reproduces the 2D toric code of Example~\ref{eg:2d-toric}.
Now let $p(t)$ be a primitive polynomial of the extension field $\FF_{2^w}$ over $\FF_2$.
$p(t)$ has coefficients in the base field $\FF_2$ and factorizes in $\FF_{2^w}$
as $p(t) = (t-\theta)(t-\theta^2)(t-\theta^{2^2})\cdots(t-\theta^{2^{w-1}})$.
(See \cite[Chapter~V Section~5]{Lang}.)
The multiplicative order of $\theta$ is $N = 2^w -1$.
The degree $w$ of $p(t)$ may be called the interaction range.
If the charge $e = \begin{pmatrix}1 & 0 \end{pmatrix}^T$ at $(0,0)\in \ZZ^2$
is transported to $(a,b) \in \ZZ^2 \setminus \{ (0,0) \}$ by some finitely supported operator,
we have $(x^a y^b -1)e \in \im \epsilon$.
That is, $x^a y^b -1 \in (p(x),p(y))$.
Substituting $x \mapsto \theta$ and $y \mapsto \theta^{2^m}$,
we see that $\theta^{a+2^m b} = 1$ or $a+2^m b \equiv 0 \pmod N$ for any $m \in \ZZ$.
In other words, $a \equiv -b \equiv -2b \pmod N$. It follows that $|a| + |b| \ge N$.%
\footnote{In case of qudits with prime dimensions $p$,
the lower bound will be $N/(p-1)$.}
Therefore, the length of the string segment transporting a charge is exponential in the interaction range $w$.
\end{rem}

\begin{example}[Wen plaquette~\cite{Wen2003Plaquette}]
\label{eg:wen-plaquette}
This model consists of a single type of interaction ($t=q=1$)
\[
\xymatrix@!0{
X \ar@{-}[r] & Y \ar@{-}[d] \\
Y \ar@{-}[u] & X \ar@{-}[l]
} \quad \quad
\sigma_\text{Wen} = 
\begin{pmatrix}
 1 + x + y + xy \\
\hline
 1 + xy
\end{pmatrix}
\]
where $X,Y$ are abbreviations of $\sigma_x,\sigma_y$.
It is known to be equivalent to the 2D toric code.
Take the coarse-graining given by $R' = \FF_2[x',y',\bar x', \bar y']$ where
\[
 x' = x \bar y, \quad \quad y' = y^2 .
\]
(The coarse-graining considered in this example
is intended to demonstrate a non-square blocking of the old lattice
to obtain a `tilted' new lattice, and is by no means special.)
As an $R'$-module, $R$ is free with basis $\{ 1, y \}$.
With the identification $R = (R' \cdot 1) \oplus (R' \cdot y)$,
we have $x \cdot 1 = x' \cdot y$, $x \cdot y = x'y' \cdot 1$,
and $y \cdot 1 = 1 \cdot y$, $y \cdot y = y' \cdot 1$.
Hence, $x$ and $y$ act on $R'$-modules
as the matrix-multiplications on the left:
\[
x \mapsto 
\begin{pmatrix}
0  & x'y'\\
x' &   0
\end{pmatrix},
\quad
y \mapsto
\begin{pmatrix}
0 & y'\\
1 & 0
\end{pmatrix}.
\]
Identifying
\[
 R^n = [ (R' \cdot 1) \oplus (R' \cdot y) ] \oplus \cdots \oplus [(R' \cdot 1) \oplus (R' \cdot y)] ,
\]
our new $\sigma$ on the coarse-grained lattice becomes
\[
 \sigma' =
\begin{pmatrix}
1+x'y' & y'+x'y'\\
1+x'   & 1+x'y'\\
\hline
1+x'y' & 0 \\
0      & 1+x'y'
\end{pmatrix}.
\]
By a sequence of elementary symplectic transformations, we have
\begin{align*}
\sigma'
&
\xrightarrow[E_{1,3}(1)]{E_{2,4}(1)}
\begin{pmatrix}
0      & y'+x'y' \\
1+x'   & 0       \\
1+x'y' & 0       \\
0      & 1+x'y'
\end{pmatrix} 
\xrightarrow[E_{3,2}(y')]{E_{4,1}(\bar y')}
\begin{pmatrix}
0      & y'+x'y' \\
1+x'   & 0       \\
1+y' & 0       \\:
0      & x'y'+x'
\end{pmatrix}
\\
& \xrightarrow[\times \bar x' \bar y']{\text{col.2}}
\begin{pmatrix}
0      & 1 + \bar x' \\
1+x'   & 0       \\
1+y' & 0       \\
0      & 1+\bar y'
\end{pmatrix}
\xrightarrow{1 \leftrightarrow 3}
\begin{pmatrix}
1+y' & 0     \\
1+x'   & 0   \\
0      & 1+ \bar x'\\
0      & 1 + \bar y'\\
\end{pmatrix},
\end{align*}
which is exactly the 2D toric code.
\hfill $\Diamond$
\end{example} 


\section{Three dimensions}
\label{sec:3d}

In the previous section, we derived a consequence of the exactness of code Hamiltonians.
The two-dimensional Hamiltonian was special 
so we were able to characterize the behavior of the charges more or less completely.
Here, we prove a weaker property of three dimensions that there must exist a nontrivial charge
for any exact code Hamiltonian.
It follows from Theorems~\ref{thm:charge-equals-torsion},%
\ref{thm:logarithmic-energy-barrier-for-fractal-operator}
that such a charge can spread through the system
by surmounting the logarithmic energy barrier.
\begin{lem}
Suppose $D=3$,
\[
  0 \to G_1 \xrightarrow{\sigma_1} G \xrightarrow{\sigma} P \xrightarrow{\epsilon=\sigma^\dagger \lambda_q} E
\]
is exact, and $I(\sigma) \subseteq \mm = (x-1,y-1,z-1)$.
Then, $\coker \epsilon$ is not torsion-free.
\label{lem:char-ideal-with-support-on-origin-implies-fractal}
\end{lem}
\begin{proof}
Suppose on the contrary $\coker \epsilon$ is torsion-free.
We have an exact sequence
\[
 0 \to G_1 \xrightarrow{\sigma_1} G \xrightarrow{\sigma} P \xrightarrow{\epsilon} E \to E'.
\]
If $G_1 = 0$, Proposition~\ref{prop:injective-sigma-implies-fractal} implies the conclusion.
So we assume $G_1 \ne 0$, and therefore we have $I(\sigma_1) = R$ by Proposition~\ref{prop:exact-sequence}.

Let us localize the sequence at $\mm$, so $I(\sigma_1)_\mm = R_\mm$.
Since $\rank (G_1)_\mm = \rank (\sigma_1)_\mm$, the matrix of $(\sigma_1)_\mm$ becomes
\[
 (\sigma_1)_\mm =
\begin{pmatrix}
 0 \\
 I
\end{pmatrix}
\]
for some basis of $(G_1)_\mm$ and $G_\mm$. See the proof of 
Lemma~\ref{lem:associated-maximal-localized-homology}.
In other words, there is an invertible matrix $B \in \mathrm{GL}_{ t \times t}( R_\mm )$ such that
\[
 \sigma_\mm B = 
\begin{pmatrix}
 \tilde \sigma & 0
\end{pmatrix}
\]
where $\tilde \sigma$ is the $2q \times t'$ submatrix.
Note that the antipode map is a well-defined automorphism of $R_\mm$
since $\overline \mm = \mm$.

Since $\epsilon = \sigma^\dagger \lambda_q$, we have
\begin{equation}
 B^\dagger \epsilon_\mm = 
\begin{pmatrix}
 \tilde \sigma^\dagger \\
 0
\end{pmatrix} \lambda_q
=
\begin{pmatrix}
 \tilde \sigma^\dagger \lambda_q \\
 0
\end{pmatrix} .
\label{eq:tilde-epsilon-sigma-dagger}
\end{equation}
Therefore, we get a new exact sequence
\[
   0 \to G' \xrightarrow{\tilde \sigma} P_\mm \xrightarrow{\tilde \epsilon = \tilde \sigma^\dagger \lambda_q} R_\mm^{t'}
\]
where $G' = G_\mm / \im (\sigma_1)_\mm$ is a free $R_\mm$-module and $t' = \rank G'$.
It is clear that 
$\rank \tilde \epsilon = \rank \tilde \sigma$.
Setting $S = R_\mm$ in Proposition~\ref{prop:injective-sigma-implies-fractal} implies that
$\coker \tilde \epsilon$ is not torsion-free.
But, since we are assuming $\coker \epsilon_m$ is torsion-free,
$\coker \tilde \sigma^\dagger$ is also torsion-free by Eq.~\eqref{eq:tilde-epsilon-sigma-dagger}.
This is a contradiction.
\end{proof}

\begin{theorem}
For any three-dimensional, degenerate and exact code Hamiltonian,
there exists a fractal generator.
\label{thm:fractal-exists-in-3D}
\end{theorem}
\begin{proof}
By Lemma~\ref{lem:coker-epsilon-resolution-length-D}, there exists an equivalent Hamiltonian
such that
\[
  0 \to G_1 \xrightarrow{\sigma_1} G \xrightarrow{\sigma} P \xrightarrow{\epsilon=\sigma^\dagger \lambda_q} E
\]
is exact. The existence of a fractal generator is a property of the equivalence class
by Proposition~\ref{prop:fractal-property-of-eq-class}.
If we can find a coarse-graining such that $I(\sigma') \subseteq (x'-1,y'-1,z'-1)$,
then Lemma~\ref{lem:char-ideal-with-support-on-origin-implies-fractal} shall imply the conclusion.

Recall that $\epsilon_L$ and $\sigma_L$ denote the induced maps
by factoring out $\bb_L = (x^L-1,y^L-1,z^L-1)$. See Sec.~\ref{sec:degeneracy}.
There exists $L$ such that $K(L) = \ker \epsilon_L / \im \sigma_L \ne 0$ by 
Corollary~\ref{cor:unit-characteristic-ideal-means-nondegeneracy}.
Consider the coarse-grain by $x'=x^L,~ y'=y^L,~ z'=z^L$.
Let $R' = F_2[x'^{\pm 1}, y'^{\pm 1}, z'^{\pm 1}]$ denote the coarse-grained base ring.
If $K'(L')$ denotes $\ker \epsilon'_{L'} / \im \sigma'_{L'}$ as $R'$-module,
we see that $K'(1) = K(L)$ as $\FF_2$-vector space.
In particular, $K'(1) \ne 0$.
Put $\mm = (x'-1,y'-1,z'-1) = \bb'_1 \subseteq R'$.
Then, $K'(1)_{\mm} = K'(1) \ne 0$.
By Lemma~\ref{lem:associated-maximal-localized-homology},
we have $I(\sigma') \subseteq \mm$.
\end{proof}

Yoshida argued that when the ground-state degeneracy is constant independent of system size
there exists a string operator~\cite{Yoshida2011feasibility}.
To prove it, we need an algebraic fact.
\begin{prop}
Let $M$ be a finitely presented $R$-module, and $T$ be its torsion submodule.
Let $I$ be the first non-vanishing Fitting ideal of $M$.
Then,
\[
 \rad I \subseteq \rad \ann T .
\]
\label{prop:support-torsion-fitting-ideal}
\end{prop}
\begin{proof}
Let $\pp$ be any prime ideal of $R$ such that $I \not\subseteq \pp$.
By the calculation of the proof of Lemma~\ref{lem:associated-maximal-localized-homology},
$M_\pp$ is a free $R_\pp$-module, and hence is torsion-free.
Since $T$ is embedded in $M$, it follows that $T_\pp = 0$, or equivalently, $\ann T \not\subseteq \pp$.
Since the radical of an ideal is the intersection of all primes containing it~\cite[Proposition 1.8]{AtiyahMacDonald},
the claim is proved.
\end{proof}

\begin{cor}
Let $T$ be the set of all point-like charges modulo locally created ones
of a degenerate and exact code Hamiltonian in three dimensions
of characteristic dimension zero.
Then, one can coarse-grain the lattice such that
\[
 \ann T = (x-1,y-1,z-1) .
\]
\label{cor:annT-3d-characteristic-dimension-0}
\end{cor}
\noindent
The corollary says that any point-like charge is attached to strings and is able to move freely through the lattice.
The condition is implied by Lemma~\ref{lem:k-growing-sequence-L} 
if the ground-state degeneracy is constant independent of the system size
when defined on a periodic lattice.
\begin{proof}
By Theorem~\ref{thm:charge-equals-torsion}, $T$ is the torsion submodule of $\coker \epsilon$.
By Theorem~\ref{thm:fractal-exists-in-3D}, $T$ is nonzero.
Setting $M = \coker \epsilon$ in Proposition~\ref{prop:support-torsion-fitting-ideal},
the associated ideal $I_q(\epsilon)$ is the first non-vanishing Fitting ideal of $M$.
Since $\dim R / I_q(\epsilon) = 0$ by assumption, we have $\dim R / \ann T = 0$.
Lemma~\ref{lem:zero-dim-ideal-over-finite-field} implies the claim.
\end{proof}

\begin{example}[Toric codes in higher dimensions]
\label{eg:highD-toric-codes}
Any higher-dimensional toric code can be treated similarly as for the two-dimensional case.
In three dimensions one associates each site with $q=3$ qubits.
It is easily checked that
\[
 \sigma_\text{3D-toric} = 
\begin{pmatrix}
1 + \bar x & 0   & 0      & 0    \\
1 + \bar y & 0   & 0      & 0    \\
1 + \bar z & 0   & 0      & 0    \\
\hline
0          & 0   & 1+z    & 1+y  \\
0          & 1+z & 0      & 1+x  \\
0          & 1+y & 1+x    & 0    \\
\end{pmatrix} .
\]

Both two- and three-dimensional toric codes have the property 
that $\coker \epsilon$ is not torsion-free.
However, in two dimensions any element of $E$ is a physical charge,
whereas in three dimensions $E$ contains physically irrelevant elements.
Note that in both cases, $1+x$ and $1+y$ are fractal generators.
Being consisted of two terms, they generate the `string operators.'

The 4D toric code~\cite{DennisKitaevLandahlEtAl2002Topological}
has $\sigma_x$-type interaction and $\sigma_z$-type interaction.
Originally the qubits are placed on every plaquette of 4D hypercubic lattice;
instead we place $q=6$ qubits on each site.
The generating map $\sigma$ for the stabilizer module is written as
a $12 \times 8$-matrix ($t=8$)
\[
\sigma_\text{4D-toric} =
\begin{pmatrix}
 \sigma_X & 0 \\
 0        & \sigma_Z
\end{pmatrix}
\]
where
\begin{align*}
 \sigma_X &= 
\begin{pmatrix}
1+y & 1+x & 0   & 0   \\
1+w & 0   & 0   & 1+x \\
1+z & 0   & 1+x & 0   \\
0   & 1+z & 1+y & 0   \\
0   & 1+w & 0   & 1+y \\
0   & 0   & 1+w & 1+z
\end{pmatrix}, \\
\bar \sigma_Z &=
\begin{pmatrix}
0   & 0   & 1+w & 1+z \\
0   & 1+z & 1+y & 0   \\
0   & 1+w & 0   & 1+y \\
1+w & 0   & 0   & 1+x \\
1+z & 0   & 1+x & 0   \\
1+y & 1+x & 0   & 0
\end{pmatrix}.
\end{align*}
Note the bar on $\sigma_Z$.

Theorem~\ref{thm:fractal-exists-in-3D} does not prevent
the absence of a fractal generator in four or higher dimensions.
Indeed, this 4D toric code lacks any fractal generator.
To see this, it is enough to consider $\sigma_Z$
since $\overline{\coker \sigma_X^\dagger} \cong \coker \sigma_Z^\dagger$ as $R_4$-modules,
where $R_4 = \FF_2[x^{\pm 1},y^{\pm 1},z^{\pm 1},w^{\pm 1}]$.
If
\[
\epsilon_1 =
 \begin{pmatrix}
1+x & 1+y & 1+z & 1+w
 \end{pmatrix} : R_4^4 \to R_4,
\]
then
\[
 R_4^6 \xrightarrow{\sigma_Z^\dagger} R_4^4 \xrightarrow{\epsilon_1} R_4
\]
is exact.
(A direct way to check it is to compute S-polynomials~\cite[Chapter~15]{Eisenbud}
of the entries of $\epsilon_1$,
and to verify that they all are in the rows of $\sigma_Z$.)
Hence, $\coker \sigma_Z^\dagger$ is torsion-free by Proposition~\ref{prop:nofractal-torsionless}.

For the toric codes in any dimensions, $\sigma$ has nonzero entries of form $x_i-1$.
The radical of the associated ideal $I(\sigma)$ is equal to $\mm = (x_1-1,\ldots, x_D-1)$.
So $\mm$ is the only maximal ideal of $R$ that contains $I(\sigma)$.
The characteristic dimension is zero.
If $2 \nmid L$, since $(\bb_L)_{\mm} = \mm_{\mm}$,
$(\sigma_L)_\mm$ is a zero matrix.
Any other localization of $\sigma_L$ does not contribute to $\dim_{\FF_2} K(L)$
by Lemma~\ref{lem:associated-maximal-localized-homology}.
Therefore, if $2 \nmid L$, $K(L)$ has constant vector space dimension independent of $L$.

There is a more direct way to compute the $R$-module $K(L)$.
For the three-dimensional case, consider a free resolution of $R_3/\mm$,
where $R_3 = \FF_2[x^{\pm 1},y^{\pm 1},z^{\pm 1}]$,
as
\[
 0 \to 
R_3^1 
\xrightarrow{
\partial_3 = 
\begin{pmatrix} a \\ b \\ c \end{pmatrix}
}
R_3^3
\xrightarrow{  
\partial_2 =
\begin{pmatrix}
 0 & -c & b  \\
 c & 0 & -a  \\
 -b & a & 0   
\end{pmatrix}
}
R_3^3
\xrightarrow{
\partial_1 =
\begin{pmatrix}
a & b  & c
\end{pmatrix}
}
R_3^1
\to
R_3/\mm
\to 0
\]
where $a=x-1$, $b=y-1$, and $c=z-1$.
We see that
\begin{equation}
 \sigma_\text{3D-toric} = \bar \partial_3 \oplus \partial_2 ,
 \quad \text{and} \quad
 \epsilon_\text{3D-toric} = \bar \partial_2 \oplus \partial_1  .
\label{eq:resolution-F}
\end{equation}
Therefore,
\begin{align*}
 K(L)_\text{3D-toric} \cong \Tor_1(\coker \epsilon_\text{3D-toric},R_3/\bb_L) 
&\cong \Tor_2( \overline{ R_3/\mm}, R_3/\bb_L ) \oplus \Tor_1( R_3/\mm, R_3/\bb_L ).
\end{align*}
Using $\Tor(M,N) \cong \Tor(N,M)$ and the fact that a resolution of $R_3/\bb_L$ is Eq.~\eqref{eq:resolution-F}
with $a,b,c$ replaced by $x^L-1,y^L-1,z^L-1$, respectively,
we have 
\[
 \Tor_i(R_3/\mm,R_3/\bb_L) \cong \Tor_i(R_3/\mm,R_3/\mm) \cong (\FF_2)^{_3 C _i}
\]
for each $0 \le i \le 3$.
Therefore, $K(L)_\text{3D-toric} \cong (\FF_2)^{_3 C _2} \oplus (\FF_2)^{_3 C _1} \cong (\FF_2)^6$.
The four-dimensional case is similar:
\[
 K(L)_\text{4D-toric} \cong \Tor_2 (R_4/\mm, R_4/\bb_L) \oplus \Tor_2 (\overline{R_4/\mm}, R_4/\bb_L) \cong \left( (\FF_2)^{_4 C _2} \right)^2 .
\]
The calculation here is closely related
to the cellular homology interpretation of toric codes.
\hfill $\Diamond$
\end{example} 

\begin{example}[Chamon model~\cite{Chamon2005Quantum,BravyiLeemhuisTerhal2011Topological}]
This three-dimensional model consists of single type of term in the Hamiltonian.
The generating map is
\[
\sigma_\text{Chamon} =
\begin{pmatrix}
 x+\bar x + y + \bar y \\
\hline
 z+\bar z + y + \bar y
\end{pmatrix}.
\]
Since
\[
 \sigma^\dagger \lambda_1
\begin{pmatrix}
0 \\ 
1
\end{pmatrix}
=
(1+ x \bar y)
\begin{pmatrix}
 0 \\
\bar x + y
\end{pmatrix},
\]
$1+x \bar y$ is a fractal generator.
Consisted of two terms, it generates a string operator.
The degeneracy can be calculated using Corollary~\ref{cor:k-formulas}.
Assume all the three linear dimensions of the system are even.
Put
\[
 S = R / (x + \bar x + y + \bar y, z + \bar z + y + \bar y, x^{2l} -1, y^{2m}-1, z^{2n} -1 ).
\]
Then, the $\log_2$ of the degeneracy is $k = \dim_{\FF_2} S$.
In $S$, we have $x + \bar x = y + \bar y = z + \bar z$.
Since $S$ has characteristic 2, it holds that
\[
 w^{p+1} + w^{-p-1} = (w + w^{-1})(w^p + w^{p-2} + \cdots + w^{-p})
\]
for $p \ge 1$ and $w = x,y,z$. By induction on $p$, we see that $w^{p} + w^{-p}$
is a polynomial in $w + w^{-1}$. Therefore,
\[
 x^p + \bar x ^p = y^p + \bar y^p = z^p + \bar z^p
\]
for all $p \ge 1$ in $S$.
Put $g = \gcd(l,m,n)$.
Since $x^l + x^{-l} = y^m + y^{-m} = z^n + z^{-n} = 0$ in $S$,
we have $x^g + x^{-g} = y^g + y^{-g} = z^g + z^{-g} = 0$.

Applying Buchberger's criterion with respect to the lexicographic order in which $x \prec y \prec z$,
we see that
\[
 S = \FF_2[x,y,z] / (z^2 + zx^{2l-1}+ zx+1, y^2 + yx^{2l-1}+yx+1, x^{2g}+1)
\]
is expressed with a Gr\"obner basis. Therefore,
\[
 k = \dim_{\FF_2} S = 8 \gcd(l,m,n).
\]
\hfill $\Diamond$
\label{eg:ChamonModel}
\end{example}

\begin{example}[Levin-Wen fermion model~\cite{LevinWen2003Fermions}]
\label{eg:Levin-Wen-fermion-model}
The 3-dimensional model is originally defined 
in terms of Hermitian bosonic operators $\{ \gamma^{ab} \}_{a,b=1,\ldots,6}$,
squaring to identity if nonzero,
such that $\gamma^{ab}=-\gamma^{ba}$,
$[\gamma^{ab},\gamma^{cd}]=0$ if $a,b,c,d$ are distinct,
and
$\gamma^{ab}\gamma^{bc}=i\gamma^{ac}$ if $a \neq c$.
An irreducible representation is given by Pauli matrices acting on $\mathbb{C}^2 \otimes \mathbb{C}^2$,
and their commuting Hamiltonian fits nicely into our formalism.
The model was proposed to demonstrate that the point-like excitations may actually be fermions.
\begin{align*}
 \sigma_\text{Levin-Wen} &=
\begin{pmatrix}
 1+z & 1+z & x+y \\
 y+y z & x+x z & x+y \\
 y+z & 1+x & 1+x \\
 y+z & z+x z & y+x y
\end{pmatrix} \\
 \epsilon_\text{Levin-Wen} &=
\begin{pmatrix}
 y+z & y+z & y+y z & 1+z \\
 z+x z & 1+x & x+x z & 1+z \\
 y+x y & 1+x & x+y & x+y
\end{pmatrix}
\end{align*}
Here we multiplied the rows of $\epsilon_\text{Levin-Wen}$ by suitable monomials to avoid negative exponents.
One readily verifies that $\ker \epsilon_\text{Levin-Wen} = \im \sigma_\text{Levin-Wen}$.
The model is symmetric under the spatial rotation by $\pi/3$ about $(1,1,1)$ axis.
Indeed, if one changes the variables as $x \mapsto y \mapsto z \mapsto x$ and apply a symplectic transformation
\begin{equation}
 \omega = 
 \begin{pmatrix}
 1 & 0 & 1 & 0 \\
 0 & 1 & 0 & 1 \\
 1 & 0 & 0 & 0 \\
 0 & 1 & 0 & 0
 \end{pmatrix} :
 \begin{cases}
   XI \mapsto YI \\
   IX \mapsto IY \\
   ZI \mapsto XI \\
   IZ \mapsto IX
\end{cases} ,
\label{eq:symplectic-Levin-Wen}
\end{equation}
then $\sigma_\text{Levin-Wen}$ remains the same up to permutations of columns.

The torsion submodule $T$ of $C = \coker \epsilon_\text{Levin-Wen}$,
which describes the point-like charges according to Theorem~\ref{thm:charge-equals-torsion},
is
\begin{equation}
 T = R \cdot \begin{pmatrix} 1+y \\ 1+x \\ 0 \end{pmatrix} .
\label{eq:Levin-Wen-torsion}
\end{equation}
In order to see this, first shift the variables $a = x+1, b=y+1, c=z+1$.
Then, $\epsilon_\text{Levin-Wen}$ becomes
\[
\epsilon_\text{Levin-Wen} =
\begin{pmatrix}
 b+c & b+c & c+b c & c \\
 a+a c & a & c+a c & c \\
 a+a b & a & a+b & a+b
\end{pmatrix}
 =: \phi
\]
We will verify that $N = C / T$ is torsion-free.
A presentation of $N = \coker \phi'$ is obtained by joining the generator of $T$ to the matrix $\phi$.
\[
 \phi' = 
\begin{pmatrix}
 b+c & b+c & c+b c & c & b \\
 a+a c & a & c+a c & c & a \\
 a+a b & a & a+b & a+b & 0
\end{pmatrix}
\]
Column operations of $\phi'$ give
\[
 \phi' \cong
\begin{pmatrix}
 0 & c & b & 0 & 0 \\
 c & 0 & a & 0 & 0 \\
 b & a & 0 & 0 & 0
\end{pmatrix} = 
\begin{pmatrix}
 \partial_2 & 0 & 0
\end{pmatrix}
\]
where $\partial_2$ is from Eq.~\eqref{eq:resolution-F}.
Therefore, $\phi'$ generates the kernel of $\partial_1$,
and by Proposition~\ref{prop:nofractal-torsionless}, $N = \coker \phi'= \coker \partial_2$ is torsion-free.

The torsion submodule $T$ of $C = \coker \phi$ is annihilated by $a$, $b$, or $c$
(See Corollary~\ref{cor:annT-3d-characteristic-dimension-0}):
\[
 a \begin{pmatrix} b \\ a \\ 0 \end{pmatrix} =\phi \begin{pmatrix} 1 \\ 1+a \\ 0 \\ a \end{pmatrix}, \quad
 b \begin{pmatrix} b \\ a \\ 0 \end{pmatrix} =\phi \begin{pmatrix} 1 \\ 1+b \\ 1 \\ 1 \end{pmatrix}, \quad
 c \begin{pmatrix} b \\ a \\ 0 \end{pmatrix} =\phi \begin{pmatrix} 0 \\ 0 \\ 1 \\ 1 \end{pmatrix}.
\]
Therefore, $T$ is isomorphic to $\coker \partial_1 \cong \FF_2$ of Eq.~\eqref{eq:resolution-F}.
The arguments $h_x,h_y,h_z$ of $\phi$ can be thought of as \emph{hopping} operators for the charge.
According to \cite{LevinWen2003Fermions}, one can check that the charge is actually a fermion
from the commutation values among, for example, $h_x,h_y,\bar y h_y$.

Consider a short exact sequence
\[
 0 \to T \to C \to N \to 0 .
\]
The corresponding sequence for 3D toric code splits,
i.e., $C \cong T \oplus N$, while this does not.
It implies that this model is not the same as the 3D toric code.

Now we can compute the ground-state degeneracy, or $\dim_{\FF_2} K(L)$.
Tensoring the boundary condition
\[
B = R/\bb_L = R/(x^L-1,y^L-1,z^L-1)
\]
to the short exact sequence,
we have a long exact sequence
\[
 \cdots \to
\Tor_1(T,B) \xrightarrow{\delta'} \Tor_1( C, B) \xrightarrow{\delta} \Tor_1(N,B) 
\to T \otimes B \to C \otimes B \to N \otimes B \to 0 .
\]
Hence, $K(L) \cong \Tor_1( C,B)$ has vector space dimension $\dim_{\FF_2} \im \delta + \dim_{\FF_2} \ker \delta$.
Since the sequence is exact, $\dim_{\FF_2} \ker \delta = \dim_{\FF_2} \im \delta'$.
As we have seen in Example~\ref{eg:highD-toric-codes},
\begin{align*}
 \Tor_1(T,B) & \cong \Tor_1(R/\mm,B) \cong (\FF_2)^3 , \quad \text{and} \\
 \Tor_1(N,B) & \cong \Tor_2(R/\mm,B) \cong (\FF_2)^3 .
\end{align*}
It follows that $\dim_{\FF_2} K(L) \le \dim_{\FF_2} \Tor_1(N,B) + \dim_{\FF_2} \Tor_1(T,B) = 6$.

It is routine to verify that $\bb_4 \subseteq I_2(\phi) \subseteq \mm := (x+1,y+1,z+1)$.
Recall the decomposition $K(L) = \bigoplus_\pp K(L)_\pp$ where $\pp$ runs over all maximal ideals of $R/\bb_L$.
Due to Lemma~\ref{lem:associated-maximal-localized-homology}, this decomposition consists of only one summand $K(L)_\mm$.
When $L$ is odd, since $(\bb_L)_\mm = \mm_\mm$, we know $K(L)_\mm = K(1)_\mm$.
Since $\phi \mapsto 0$ under $a=b=c=0$, we see $\dim_{\FF_2} K(1) = 4$.
The logical operators in this case are
\[
 \begin{pmatrix}  0 \\ 0 \\ \widehat z  \\ \widehat z  \end{pmatrix}
\smile
 \widehat x \cdot \widehat{y \bar z}
 \begin{pmatrix}  0 \\ 1 \\ 0 \\ 0 \end{pmatrix}
 \quad ; \quad
 \begin{pmatrix}  \widehat x \\ \widehat x \\ 0 \\ 0 \end{pmatrix}
\smile
 \widehat z \cdot \widehat{x y} 
 \begin{pmatrix}  1 \\ 1 \\ 1 \\ 0 \end{pmatrix}
\]
where $\widehat \mu = \sum_{n=0}^{L-1} \mu^n$ so $\mu \cdot \widehat \mu = \widehat \mu$,
and symplectic pairs are tied.
The left elements are string-like, and the right surface-like.

When $L$ is even, the following are $\FF_2$-independent elements of $K(L)$.
As there are 6 in total, the largest possible number, we conclude that $K(L)$ is 6-dimensional,
i.e., the number of encoded qubits is 3 when linear dimensions are even.
\[
 \begin{pmatrix}  0 \\ 0 \\ \widehat z \\ \widehat z \end{pmatrix}
\smile
 \widehat x ' \widehat y '
 \begin{pmatrix}  1+y \\ x+xy \\ 0 \\ 1+x+y+xy \end{pmatrix}
 ;
 \begin{pmatrix}  \widehat x \\ \widehat x \\ 0 \\ 0 \end{pmatrix}
\smile
 \widehat y ' \widehat z '
 \begin{pmatrix}  1+z \\ 1+z \\ 1+z \\ y+yz \end{pmatrix}
  ;
 \begin{pmatrix}  \widehat y \\ \widehat y \\ \widehat y \\ \widehat y \end{pmatrix}
\smile 
 \widehat z ' \widehat x '
 \begin{pmatrix}  0 \\ 1+x+z+xz \\ 1+x \\ 1+x \end{pmatrix}
\]
where $\widehat \mu ' = \sum_{i=0}^{L/2-1} \mu^{2i}$ so $(1+\mu)\widehat \mu ' = \widehat \mu$.
The pairs are transformed cyclically by $x \mapsto y \mapsto z \mapsto x$
together with the symplectic transformation $\omega$ of Eq.~\eqref{eq:symplectic-Levin-Wen}.
\hfill $\Diamond$
\end{example}

\section{Discussion}
\label{sec:lowD-discussion}

There are many natural questions left unanswered.
Perhaps, it would be the most interesting to answer how much 
the associated ideal $I(\sigma)$
determines about the Hamiltonian.
Note that the very algebraic set defined by the associated ideal
is not invariant under coarse-graining.
For instance, in the characteristic dimension zero case,
the algebraic set can be a several points in the affine space,
but becomes a single point under a suitable coarse-graining.
However, the geometry of the algebraic set seems to be crucial
to prove, for example, ``no-strings rule.''
See Chapter~\ref{chap:cubic-code}.

It is interesting on its own to prove or disprove
that the elementary symplectic transformations 
generate the whole symplectic transformation group.
In the zero-dimensional case where $\Lambda$ is the trivial group,
it is true as we have already seen in 
Proposition~\ref{prop:elem-sym-transformation-generates-whole-BlockCodeCase}.
The one-dimensional case is also true 
because it is implied by the computation in the proof of Theorem~\ref{thm:1D-classification}.
A classical problem answered affirmatively by Suslin~\cite{Suslin1977Stability}
is that any sufficiently large invertible matrix over a polynomial ring
is a finite product of elementary matrices such as row operations and scalar multiplications.
Later, an algorithmic proof is given by Park and Woodburn~\cite{ParkWoodburn1995Algorithmic}.
A similar problem under a confusingly similar name `symplectic group' over polynomial rings
is solved by Grunewald~{\em et al.}~\cite{GrunewaldMennickeVaserstein1991SymplecticGroup},
who defined the `symplectic group' as
$ \left\{ S \in \mathrm{Mat}\left(n, \FF[x_1,\ldots,x_n]\right) ~\middle| ~ S^T \lambda S = \lambda \right\}$
where $T$ is the transpose.
Kopeyko~\cite{Kopeyko1999Symplectic} generalized it to include Laurent polynomials,
but still the `symplectic group' is different from ours
since the antipode map is absent from the definition
\[
 \left\{ S \in \mathrm{Mat}\left(n, \FF[x_1^{\pm 1},\ldots,x_n^{\pm 1}]\right) ~\middle| ~ S^T \lambda S = \lambda \right\} .
\]

The characteristic dimension is not proven to be invariant under coarse-graining.
It suffices to have an upper bound on $\dim_{\FF_2} K(L)$ in Corollary~\ref{cor:upper-bound-k}.
without the condition that $2 \nmid L$.
A closely related object is the Hilbert function.
Given a graded module $M = \bigoplus_{s=0}^{\infty} M_s$ 
over a polynomial ring with coefficients in a field $\FF$,
the Hilbert function $f_M$ is a numerical function defined by $f_M(s) = \dim_\FF M_s$.
Since $K(L) = \dim_{\FF_2} \Tor_1(\coker \epsilon, R / \bb_L)$,
the Hilbert function might be useful if we could make $\coker \epsilon$ graded.
A technical difficulty would be that the ideal $\bb_L=(x_1^L-1,\ldots,x_D^L-1)$ 
is not a power of $\mm=(x_1 - 1,\ldots,x_D-1)$.
See \cite[Chapter~12]{Eisenbud} and \cite{IyengarPuthenpurakal2007HilbertSamuelFunctions}.

Lastly, an important problem is to give a criterion to a module $M$
that can be realized as $\coker \sigma^\dagger$
for an exact code Hamiltonian described by $\sigma$.
In the two-dimensional case, we know the answer ---
$\ann M = (x-1,y-1)$ by Theorem~\ref{thm:structure-2d-ann-coker-epsilon}.

\chapter{Cubic code}
\label{chap:cubic-code}

The toric code~\cite{Kitaev2003Fault-tolerant} 
is usually defined on a planar graph with qubits residing on edges.
The star operator and plaquette operators 
are defined according to the data of the graph.
The model is thus a priori lattice dependent.
However, a certain set of important properties of the model
turns out to be lattice independent.
Particularly the ground-state subspace 
conveys a structure that depends only on the topology of the underlying space, 
insensitive to the microscopic detail.
If the underlying space is a genus $g$ (oriented) surface,
the degeneracy is $4^g$.
Even if the Hamiltonian is perturbed,
the energy splitting of lowest $4^g$ states
is exponentially small in the system size
as long as the perturbation is small enough~%
\cite{WenNiu1990GroundState,
Kitaev2003Fault-tolerant,
BravyiHastingsMichalakis2010stability}.
It is a signature of \emph{topological order}.

Authors have used the term `topological quantum order' to mean
all or a part of the following properties:
Ground-state degeneracy as a function of topology,
no spontaneous symmetry breaking,
anyonic particle content~%
\cite{
WenWilczekZee1989Chiral,
Wen1989Degeneracy,
Einarsson1990Fractional,
WenNiu1990GroundState,
ReadSachdev1991LargeN,
Wen1991SpinLiquid},
robust edge modes against 
perturbations~\cite{KaneMele2005,FuKaneMele2007Topological,HasanKane2010TIReview},
locally indistinguishable ground states~\cite{
BravyiHastingsMichalakis2010stability,
BravyiHastings2011short,
MichalakisPytel2011stability},
and topological entanglement entropy~\cite{KitaevPreskill2006Topological,LevinWen2006Detecting}.
Here, we take the local indistinguishability of ground states
as a definition of topological quantum order.
The local indistinguishability is precisely the one used in Chapter~\ref{chap:alg-theory},
as well as in the proof of gap stability 
results~\cite{BravyiHastings2011short,
MichalakisPytel2011stability}.
See Lemma~\ref{lem:local-tqo=exact}.
We ask what quantum phases are possible in the class of code Hamiltonians.

An important example of topological quantum order 
presented in this thesis is \emph{cubic code}.
In this chapter we explain how the model is found,
and study its consequences.
The cubic code is an exact local additive code with translation symmetry 
on the simple cubic lattice, 
where the exactness is as in Chapter~\ref{chap:alg-theory}.
This model is a gapped unfrustrated spin Hamiltonian,
and is topologically ordered in the very sense we just defined,
but breaks many aspects of conventional models of topological order.
Most prominently, the excitations or charges of the cubic code
\emph{cannot} be interpreted as particles
since they are immobile;
the hopping term appears only after $L$-th
or higher order perturbation theory, where $L$ is the linear system size.
The immobility implies that the system spends
quite a long time to reach its thermal equilibrium
in response to environment's change.
(Later, we will give quantitative statements regarding this.)
Moreover, due to the topological quantum order,
the phase is stable with respect to arbitrary but small perturbations;
the immobility of excitations is also protected~\cite{BravyiHastingsMichalakis2010stability}.

The immobility of the charges translates 
into the theory of quantum error correcting codes
as the absence of string logical operators.
We call this property by \emph{no-strings rule}.
The notion of string logical operators might be intuitive
if one imagines the toric codes in two or three 
dimensions~\cite{DennisKitaevLandahlEtAl2002Topological}.
However, this intuition is too model specific;
the strings in discrete lattices are not well-defined objects
since a set of points in the lattice does not in general have
a well-defined dimensionality.
We overcome this issue by defining \emph{string segments}
that capture characteristics that are responsible for the mobility of charges.

We proceed by translating conditions for the absence of the logical string segments
into our algebraic framework,
in order to systematically search for codes without logical strings;
the cubic code is not an ad hoc model.
By Corollary~\ref{cor:annT-3d-characteristic-dimension-0}%
~\cite{Yoshida2011feasibility},
the characteristic dimension of the cubic code must be 1
and the degeneracy must generally grow with the system size.
An explicit formula for the degeneracy is given.
Additionally, a real-space renormalization computation is presented.
It seems that the cubic code is a fixed point of a certain unconventional kind.
Namely, the model is a direct sum of two daughter models $A$ and $B$ at a coarse-grained lattice
where one daughter model $A$ is the same as the original model, but the other $B$ is not.
The model $B$ produces two copies of itself at a further coarse-grained lattice.
The renormalization group flow continues to branch.
Next, we compute the thermal partition function,
and show that there is no finite temperature phase transition.

\section{String segments and no-strings rule}

The most important property of one-dimensional objects
is that a finite part of it has two disjoint boundary points.
An intuitive role of string operator 
is to move an excitation at its one boundary point
to another boundary point.
Essentially, it is a concatenation or juxtaposition of hopping operators.
The hopping operator is, as a whole, a finitely supported operator.
Therefore, if the hopping operator acts on the vacuum (ground state),
the overall effect is to create a trivial or neutral charge,
which consists of two spatially separated charges.
When either of the two charges is neutral by itself,
the hopping is meaningless 
because trivial charges can be annihilated locally
and created anywhere arbitrarily;
only the hopping of nontrivial charges is important.
Now we can define string segments and their absence.
\begin{defn}
\cite{Haah2011Local,BravyiHaah2011Energy}
A {\bf string segment} is a finitely supported Pauli operator
that creates excitations contained in the union of two finite boxes
of {\bf width} $w$.
The string segment is {\bf nontrivial} 
if the charge contained in one of the boxes is nontrivial.
The distance between the boxes is the {\bf length} of the string segment.
We say a model obeys {\bf no-strings rule}
if the length of any nontrivial string segment of width $w$
is bounded by $\alpha w$ for some constant $\alpha \ge 1$.
\label{defn:no-strings}
\end{defn}
\noindent
The no-strings rule may seem too strong than necessary;
why do we need an upper bound by a \emph{linear} function?
It is rather a technicality 
that is necessary to prove a logarithmic energy barrier theorem
in Section~\ref{sec:log-barrier}.
However, our definition seems sharp yet broad enough to derive further results.

An immediate consequence of the no-strings rule
in the case of a translation-invariant code Hamiltonian
is that there are infinitely many charges.
If there were only finitely many,
then in the sequence of all translations of a charge $c$
there would be an equivalent charge $c'$.
It means $c - c'$ is neutral 
and therefore $c-c'$ is created by a finitely supported operator,
which is a nontrivial string segment.
This string segment can be juxtaposed many times
to give arbitrarily long string segments with fixed width.
This is a contradiction to the no-strings rule.
Also, in view of Corollary~\ref{cor:annT-3d-characteristic-dimension-0}
and Lemma~\ref{lem:k-growing-sequence-L},
a translationally invariant three-dimensional exact code Hamiltonian that obeys the no-strings rule,
must have a growing ground-state degeneracy with respect to the system size
when defined on lattices with periodic boundary conditions.

\section{Search for models}

We wish to find an interesting class of models
that are relatively handy to deal with.
The following lemma shows such a family.
Recall that for an exact code Hamiltonian,
we have $\ker \sigma^\dagger \lambda = \im \sigma$.
\begin{lem}
The generating matrix of a code Hamiltonian
\[
\sigma =
\begin{pmatrix}
f & 0 \\
g & 0 \\
0 & g' \\
0 & f' \\
\end{pmatrix},
\]
where $f,g,f',g' \in R$, elements of group algebra of $\ZZ^D$,
is \emph{exact}
if and only if 
\[
\gcd(f,g)=1,
\quad \quad
g' = \bar g,
\quad
 f' = -\bar f 
\]
up to units of $R$.
\label{lem:two-terms-CSS-exactness-condition}
\end{lem}
\noindent
It is a generalization of $\sigma$ in Remark~\ref{rem:long-string}.
The 2D toric code of Example~\ref{eg:2d-toric} falls into this form.
\begin{proof}
Observe that there is no mixing between 
the $\sigma^x$ part (the upper two rows of $\sigma$)
and the $\sigma^z$ part (the lower two rows of $\sigma$).
So $\ker \sigma^\dagger \lambda = \im \sigma$ can be checked
separately.
That is, we need to verify whether the kernel of
\[
\epsilon_z = \begin{pmatrix} \bar f & \bar g \end{pmatrix}
\]
is generated by $\begin{pmatrix} g' \\ f' \end{pmatrix}$,
and whether the kernel of
\[
\epsilon_x = \begin{pmatrix} \bar g' & \bar f' \end{pmatrix}
\]
is generated by $\begin{pmatrix} f \\ g \end{pmatrix}$.
If $f = f_1 h$ and $g = g_1 h$ where $\gcd(f,g)=h$,
then $\begin{pmatrix} \bar g_1 \\ -\bar f_1 \end{pmatrix} \in \ker \epsilon_z$.
We have $\begin{pmatrix} \bar g_1 & -\bar f_1 \end{pmatrix}
 = r \begin{pmatrix} g' & f' \end{pmatrix}$ for some $r \in R$.
Since $\gcd(g_1, f_1)=1$, $r$ must be a unit.
Hence, we may set $g' = \bar g_1$ and $f' = -\bar f_1$.
Then, $\begin{pmatrix} f_1 \\ g_1 \end{pmatrix} \in \ker \epsilon_x$,
so $h$ is a unit. The converse is straightforward.
\end{proof}

We stay with the form of $\sigma$ as in
Lemma~\ref{lem:two-terms-CSS-exactness-condition}
for its simplicity.
Now we consider the no-strings rule.
The equivalence classes of topologically nontrivial charges
is $\coker \epsilon$ by Theorem~\ref{thm:charge-equals-torsion},
which decomposes as $\coker \epsilon_z \oplus \coker \epsilon_x$.
Since the two summands are related by the antipode map,
it suffices to consider $\coker \epsilon_x \cong R/(f,g)$ only.
A necessary condition for the no-strings rule is that
any fractal generator should \emph{not} consist of two terms.
For Theorem~\ref{thm:structure-2d-ann-coker-epsilon},
the minimal dimension we should be interested in is 3;
$ R = \FF_2[x^{\pm 1},y^{\pm 1},z^{\pm 1}]$.
Let $I = (f,g)$, the ideal of $R$ generated by $f$ and $g$.
If, for example, $f$ factorizes as $f=(1+x)f_1$ and $f_1 \notin I$,
then $(1+x)$ is a string generator.
So it is necessary that the Gr\"obner basis of the ideal $(f,g)$ does not have
any two-term factor such as $x+1$, $y+1$, or $z+1$.
In order to have a degenerate code Hamiltonian 
(Corollary~\ref{cor:unit-characteristic-ideal-means-nondegeneracy}),
we have to have a non-unit associated ideal.
For simplicity, we demand that $x=y=z=1$ is a root of $f$ and $g$
so that the associated ideal is contained in a maximal ideal $(x+1,y+1,z+1)$.

We further assume that $f$ and $g$ have exponents $0$ or $1$ in each variable.
Thus, $f$ and $g$ are linear combinations of $1,x,y,z,xy,yz,zx,xyz$ over $\FF_2$.
Naively there are $2^{16}$ possibilities, but they are not all different.
For example, three choices $(f(x,y,z),g(x,y,z))$, $(x f(\frac{1}{x}),y,z), x g(\frac{1}{x},y,z))$,
and $(f(y,z,x),g(y,z,x))$ define the same models 
because they are related by the reflection about $yz$-plane
or the $\pi/3$ rotation about $(1,1,1)$-axis.
Up to these symmetries of the unit cube, there are 392 pairs of $f,g$.
An exhaustive search gives
10 models in Table~\ref{tb:10-models-wo-string} that satisfy all our requirements.
The most symmetric model defined by
\[
f = 1+ x + y + z, \quad \quad g = 1+xy+yz+zx
\]
will be called {\bf cubic code}.

\begin{table}[htbp]
\centering
\begin{tabular}{c|c|c}
\hline
   & $f$ & $g$ \\
\hline
1 & $1+x+y+z$ & $1+x y+x z+y z$ \\
2 & $x+y+z+y z$ & $1+y+x y+z+x z+x y z$ \\
3 & $1+x+y+z$ & $1+x z+y z+x y z$ \\
4 & $1+x+z+y z$ & $1+y+x y+x z$ \\
5 & $1+x+z+y z$ & $y+z+x z+y z$ \\
6 & $1+x+y+z$ & $1+y+x z+y z$ \\
7 & $1+x+y+z$ & $1+z+y z+x y z$ \\
8 & $1+x+z+y z$ & $1+y+x y+z+x z+y z$ \\
9 & $1+x+y+z$ & $x y+z+x z+y z$ \\
10 & $1+x+z+y z$ & $1+y+x y+x z+y z+x y z$ \\
\hline
\end{tabular}
\caption{Complete list of cubic codes. 
Each pair of polynomials defines an exact degenerate code Hamiltonian
according to Lemma~\ref{lem:two-terms-CSS-exactness-condition}.
They potentially obey the no-strings rule.
Especially, code 1 indeed obeys the no-strings rule as proven in Section~\ref{sec:cubic-code}.
The table is exhaustive, up to symmetries of simple cubic lattice and local symplectic transformations,
under the following criteria:
(i) $f$ and $g$ have exponents $0$ or $1$ for each variable,
(ii) any member of the Gr\"obner basis of the ideal $(f,g)$ is not divided by any two-term factor, and
(iii) $f$ and $g$ become zero when $x=y=z=1$.
In the main text, the {\em cubic code} refers specifically to the code 1 in the table.
The numbering is consistent with a table in \cite{Haah2011Local}.}
\label{tb:10-models-wo-string}
\end{table}

\section{Cubic code}
\label{sec:cubic-code}

Written out explicitly, cubic code
is the translation-invariant negative sum of two types of interaction terms as in Figure~\ref{fig:CubicCode}.
\begin{figure}[b]
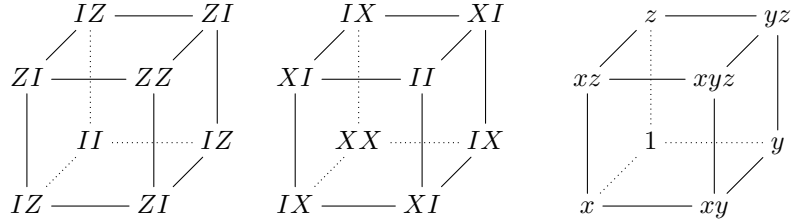

\[
\drawgenerator{ZI}{ZZ}{IZ}{ZI}{IZ}{II}{ZI}{IZ} 
\quad
\drawgenerator{XI}{II}{IX}{XI}{IX}{XX}{XI}{IX}
\quad \quad
\drawgenerator{xz}{xyz}{x}{xy}{y}{1}{yz}{z}
\]
\caption{Stabilizer generators of the 3D cubic code. Here
$X\equiv \sigma^x$ and $Z\equiv \sigma^x$ represent single-qubit Pauli operators,
while $I$ is the identity operator.
Double-letter indices represent two-qubit Pauli operators,
for example,
$IZ\equiv I\otimes Z$, $ZZ\equiv Z\otimes Z$, $II\equiv I\otimes I$, etc.}
\label{fig:CubicCode}
\end{figure}
There are two qubits per site,
and the two-letter notation stands for tensor product of Pauli matrices.
For example, $XI = \sigma^x \otimes I$, $ZZ = \sigma^z \otimes \sigma^z$, etc.
The third cube specifies the coordinate system of the simple cubic lattice.
The generating map for the stabilizer module is
\[
\sigma_\text{cubic-code} =
\begin{pmatrix}
1 + x + y + z    & 0 \\
1 + xy + yz + zx & 0 \\
0 & 1 + \bar x \bar y + \bar y \bar z + \bar z \bar x \\
0 & 1 + \bar x + \bar y + \bar z \\
\end{pmatrix}
\]
where one has to interchange the first and second qubit.
The associated ideal is contained in a prime ideal of codimension 2 in $\FF_2[x^{\pm 1},y^{\pm 1},z^{\pm 1}]$:
\[
 I(\sigma) \subseteq ( 1+x+y+z,~ 1+xy+yz+zx ) = \pp_{xyz}.
\]
Since $\codim I(\sigma) \ge 2$, the characteristic dimension is 1.
Since $\coker \epsilon_\text{cubic-code} = R / \pp_{xyz} \oplus R / \overline{ \pp_{xyz}}$,
any nonzero element of $\pp_{xyz}$ is a fractal generator.
Since the conditions used in the search for the model
were only necessary conditions for the no-strings rule.
A rigorous treatments is as follows.
We prove a purely algebraic statement,
of which the no-strings rule is an interpretation.
A more elementary method can be found in \cite{Haah2011Local}.
\begin{lem}
Let $S = \FF_2[x,y,z]$ be a polynomial ring, and $\pp = (1+x+y+z,~1+xy+yz+zx) \subseteq S$ an ideal.
If $m_1 e_1 + m_2 e_2 \in \pp$ for polynomials $e_1, e_2$
and monomials $m_1, m_2$ such that $\gcd( m_1, m_2 ) = 1$,
then only one of the following is true:
\begin{itemize}
\item $e_1 \in \pp$ and $e_2 \in \pp$.
\item $\max( \deg e_1, \deg e_2 ) \ge \max( \deg m_1, \deg m_2 )$.
\end{itemize}
\label{lem:alg-proof-no-string-cubic}
\end{lem}
\begin{proof}
Put $\FF_4 = \{ 0, 1, \omega, \omega^2 \}$,
i.e., $\omega$ is the primitive third root of unity over the binary field.
Consider a ring homomorphism $\phi : S \to U:=\FF_4[t]$ defined by
\[
\phi : x \mapsto 1+t, \quad y \mapsto 1+\omega t, \quad z \mapsto 1+\omega^2 t .
\]
It maps the two generators of $\pp$ to zero in $U$.
\begin{align*}
 \phi(1+x+y+z)    &= 4 + (1+\omega+\omega^2) t = 0 ,\\
 \phi(1+xy+yz+zx) &= 4 + 2(1+\omega+\omega^2)t + (1+\omega+\omega^2)t^2 = 0 .
\end{align*}
Moreover, $\ker \phi$ is precisely $\pp$.
(This can be verified by eliminating the variable $x$ 
and computing the Gr\"obner basis of $((y+1)+\omega(z+1),\omega^2+\omega+1)$ in an elimination monomial order.)
If $m_1 e_1 + m_2 e_2 \in \pp$, then $\phi(m_1)\phi(e_1) = \phi(m_2)\phi(e_2)$.
Since $m_1$ and $m_2$ are co-prime monomials and $\phi(x),\phi(y),\phi(z)$ are pairwise co-prime,
it follows that $\phi(m_1)$ and $\phi(m_2)$ are nonzero and co-prime.
Therefore, $\phi(e_1) = 0$ if and only if $\phi(e_2) = 0$ if and only if $e_1, e_2 \in \pp$,
which is the first case.
If $\phi(e_1) \neq 0$, then $\phi(m_1)$ must divide $\phi(e_2)$ and $\phi(m_2)$ must divide $\phi(e_1)$.
Since $\phi$ is degree-preserving, we have the second case.
\end{proof}
\begin{theorem}
The cubic code obeys the no-strings rule with the constant $\alpha = 1$ under $\ell_\infty$-metric.
\label{thm:CubicCode-no-strings}
\end{theorem}
\begin{proof}
Since the cubic code is translationally invariant,
we may use the formalism of Chapter~\ref{chap:alg-theory}.
Since the cubic code is of CSS type, where $\coker \epsilon$ is a direct sum of isomorphic summands,
we only have to consider $\sigma^x$-type charges.
The set of all virtual charges is $R^1 = R = \FF_2[x^{\pm 1},y^{\pm 1},z^{\pm 1}]$
and the set of all trivial charges are given by a submodule (ideal) $\pp_{xyz} = (1+x+y+z,1+xy+yz+zx) \subseteq R$.
Note that $R$ is the localization of $S = \FF_2[x,y,z]$ by a single element $xyz$,
and $\pp_{xyz}$ is the localization of $\pp = (1+x+y+z,~1+xy+yz+zx)$ by the same single element $xyz$.

Let $e_1$ and $e_2$ be the charges contained in two boxes of a string segment.
Since they are overall trivial, we have $e_1 + x^i y^j z^k e_2 \in \pp$ where $i,j,k \in \ZZ$.
Equivalently, we may write $m_1 e_1 + m_2 e_2 \in \pp$ where $e_1$ and $e_2$ have nonnegative exponents,
and $m_1, m_2$ are monomials such that each variable ($x$, $y$, or $z$) appears
only in one of $m_1$ or $m_2$.
Since $\pp_{xyz} \cap S = \pp$, which can be verified using Gr\"obner basis techniques,
we are in the situation of Lemma~\ref{lem:alg-proof-no-string-cubic}.
The width of the string segment is the maximum of the degrees of $e_1$ and $e_2$,
and the $\ell_1$-distance between the boxes enclosing $e_1$ and $e_2$ is $\le \max( \deg m_1, \deg m_2 )$.
If we use $\ell_\infty$-distance, the constant $\alpha$ in the no-strings rule is $1$.
\end{proof}

Let us explicitly calculate the ground-state degeneracy
when the Hamiltonian is defined on $L \times L \times L$ 
cubic lattice with periodic boundary conditions.
By Corollary~\ref{cor:k-formulas},
\[
 k = \dim_{\FF_2} R / (\pp_{xyz} + \bb_L) \oplus R / (\overline{\pp_{xyz}} + \bb_L) = 2 \dim_{\FF_2} R / (\pp_{xyz} + \bb_L),
\]
where $\bb_L = ( x^L -1,~ y^L -1,~ z^L -1 )$.
So the calculation of ground-state degeneracy comes down to the calculation of
\[
 d = \dim_{\FF_2} T' / \pp
\]
where $T' = \FF_2[x,y,z]/(x^{n_1}-1, y^{n_2}-1, z^{n_3}-1 )$.

We may extend the scalar field to any extension field
without changing $d$.
Let $\FF$ be the algebraic closure of $\FF_2$ and let
\[
 T=\FF[x,y,z]/(x^{n_1}-1, y^{n_2}-1, z^{n_3}-1 )
\]
be an Artinian ring.
By Proposition~\ref{prop:Artin-ring},
it suffices to calculate for each maximal ideal $\mm$ of $T$
the vector space dimension 
\[
d_\mm = \dim_{\FF} (T / \pp)_\mm
\]
of the localized rings, and sum them up.

Suppose $n_1,n_2,n_3 > 1$.
By Nullstellensatz, any maximal ideal of $T$ is of form
$\mm=(x-x_0,y-y_0,z-z_0)$ where $x_0^{n_1}=y_0^{n_2}=z_0^{n_3}=1$.
(If $n_1 = n_2 = n_3 = 1$, then $T$ becomes a field, and there is no maximal ideal other than zero.)
Put $n_i = 2^{l_i}n_i'$ where $n_i'$ is not divisible by $2$.
Since the polynomial $x^{n_1} -1$ contains the factor $x-x_0$ with multiplicity $2^{l_1}$,
it follows that
\[
 T_\mm = \FF [x,y,z]_\mm / ( x^{2^{l_1}}+a',~ y^{2^{l_2}}+b',~ z^{2^{l_3}}+c')
\]
where $a' = x_0^{2^{l_1}}, b' = y_0^{2^{l_2}}, c' = z_0^{2^{l_3}}$.
Hence, $(T/\pp)_\mm \cong \FF[x,y,z] / I'$ where
\[
 I' = (x+y+z+1, xy+xz+yz+1 ,~ x^{2^{l_1}}+a',~ y^{2^{l_2}}+b',~ z^{2^{l_3}}+c').
\]
If $I' = \FF[x,y,z]$, then $d_\mm = 0$.

Without loss of generality, we assume that $l_1 \le l_2 \le l_3$.
By powering the first two generators of $I'$,
we see that $(x_0,y_0,z_0)$ must be a solution of them in order for $I'$ not to be a unit ideal.
Eliminating $z$ and shifting $x \to x+1$, $y \to y+1$,
our objective is to calculate the Gr\"obner basis for the proper ideal
\[
 I = (x^2+xy+y^2, x^{2^{l_1}}+a,~ y^{2^{l_2}}+b )
\]
where $a=a'+1$ and $b = b'+1$. So
\[
 d_\mm = \dim_{\FF} \FF[x,y]/I.
\]
One can easily deduce by induction that $y^{2^m} + x^{2^m-1}(m x +y) \in I$ for any integer $m \ge 0$.
And $b = \omega a^{2^{l_2 - l_1}}$ for a primitive third root of unity $\omega$.
So we arrive at
\[
 I = ( y^2 + yx + x^2,~ yx^{2^{l_2} -1 } + b ( 1 + l_2 \omega^2 ),~ x^{2^{l_1}} + a )
\]
We apply the Buchberger criterion.
If $a \neq 0$, i.e., $x_0 \neq 1$,
then $b \neq 0$ and $I = (x+(\omega^2 + l_2)y, x^{2^{l_1}} + a)$,
so $d_\mm = 2^{l_1}$

If $a=b=0$,
then $I = (y^2 + yx + x^2, yx^{2^{l_2}-1}, x^{2^{l_1}} )$.
The three generators form Gr\"obner basis if $l_2 = l_1$.
Thus, in this case, $d_\mm = 2^{l_1+1}-1$.
If $ l_2 > l_1$, then $d_\mm = 2^{l_1+1}$.

To summarize, except for the special point $(1,1,1)\in \FF^3$ of the affine space,
each point in the algebraic set
\[
 V = \left\{ (x,y,z) \in \FF^3 ~\middle|~
\begin{matrix}
x+y+z+1 = xy+xz+yz+1 = 0 \\  
x^{n_1'}-1 = y^{n_2'}-1 = z^{n_3'}-1 = 0
\end{matrix}
\right\}
\]
contribute $2^{l_1}$ to $d$. The contribution of $(1,1,1)$ is either $2^{l_1 +1}$ or $2^{l_1 +1}-1$.
The latter occurs if and only if $l_1$ and $l_2$,
the two smallest numbers of factors of $2$ in $n_1,n_2,n_3$,
are equal. Let $d_0 = \#V$ be the number of points in $V$.
The desired answer is
\[
 d = 2^{l_1} (d_0 -1) + 
\begin{cases} 
2^{l_1 +1} -1 & \text{if $l_1 = l_2$ } \\
2^{l_1 +1}    & \text{otherwise}
\end{cases}
\]
where $l_1 \le l_2 \le l_3$ are the number of factors of $2$ in $n_i$.

The algebraic set defined by $(x+y+z+1,~xy+xz+yz+1)$ is the union of two isomorphic lines
intersecting only at $x=y=z=1$, one of which is parametrized by $x \in \FF$ as
\[
 (1+x, 1+\omega x, 1+\omega^2 x) \in \FF^3,
\]
and another is parametrized as
\[
 (1+x,1+\omega^2 x, 1+\omega x) \in \FF^3.
\]
where $\omega$ is a primitive third root of unity.
Therefore, the purely geometric number $d_0 = 2 d_1 -1$ can be calculated by
\[
 d_1 = \deg_x \gcd \left( (1+x)^{n_1'}+1, (1+\omega x)^{n_2'}+1, (1+\omega^2 x)^{n_3'}+1 \right) .
\]
Using $(\alpha+\beta)^{2^p} = \alpha^{2^p} + \beta^{2^p}$ and $\omega^2 + \omega + 1 = 0$,
one can easily compute some special cases as summarized in the following corollary.
Some values of $k$ for small $L$ are presented in Table~\ref{tb:k-numerics} and Figure~\ref{fig:k-numerics}.

\begin{figure}[htbp]
\centering
\begin{minipage}{.45\textwidth}
\includegraphics[width=\textwidth]{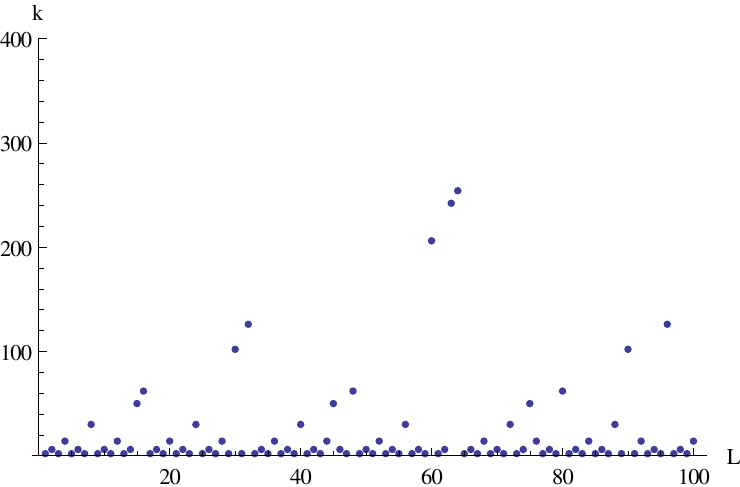}
\end{minipage}
\begin{minipage}{.45\textwidth}
\includegraphics[width=\textwidth]{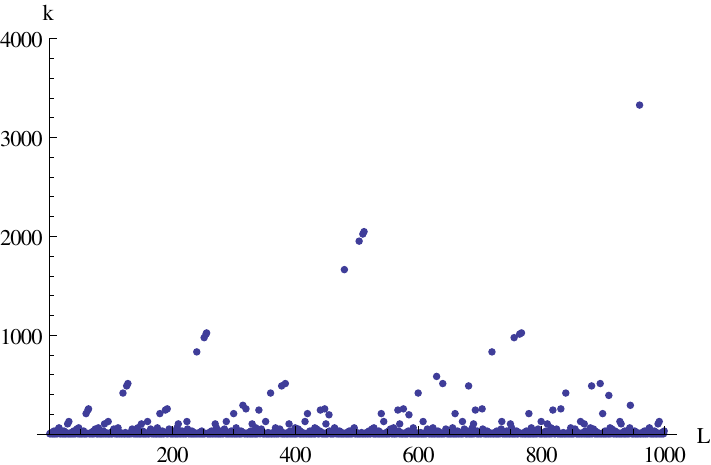}
\end{minipage}
\caption{Number of encoded qubits $k$ of the cubic code
defined on $L\times L \times L$ periodic lattice}
\label{fig:k-numerics}
\end{figure}

\begin{cor}
\label{cor:cubic-code-degeneracy-formula}
Let $2^k$ be the ground-state degeneracy of
the cubic code on the cubic lattice of size $L^3$ with periodic boundary conditions.
($k=k(L)$ is the number of encoded qubits.)
Then
\begin{align*}
\frac{k+2}{4} 
&= \deg_x \gcd \left( (1+x)^{L}+1,~ (1+\omega x)^{L}+1,~ (1+\omega^2 x)^{L}+1 \right)_{\FF_4} \\
&= \begin{cases}
     1    & \text{if $L = 2^p+1$}, \\
     L    & \text{if $L =2^p$}, \\
     L-2  & \text{if $L = 4^p -1$}, \\
     1    & \text{if $L = 2^{2p+1} -1$}.
   \end{cases}
\label{eq:cubic-code-formula-k}
\end{align*}
where $\omega^2 + \omega + 1 = 0$ and $p \ge 1$ is any integer.
If $L = 2^r L'$, then $k(L)+2 = 2^r( k(L') + 2 )$.
\end{cor}

\begin{table}[tp]
\centering
\begin{tabular}{r|rll}
\hline
\multicolumn{4}{c}{$(k+2)/4 = 1 + 12 \sum_n q_n(L)$} \\
\multicolumn{4}{c}{$q_n(L)$ is nonzero only if $n | L$.} \\
\hline
$q_n(L)$&  & $n$ & \\
\hline
1 & 15 & $= 2^4 -1$ & $=3 \cdot 5$  \\
5 & 63 & $= 2^6 -1$ & $=3^2 \cdot 7$  \\
20 & 255 & $= 2^8 -1$ & $=3 \cdot 5 \cdot 17$  \\
80 & 1023 & $=2^{10}-1$ & $=3 \cdot 11 \cdot 31$  \\
322 & 4095 & $= 2^{12} -1$ & $=3^2 \cdot 5 \cdot 7 \cdot 13$ \\
\hline
5 & 341 & $=(2^{10}-1)/(2^2 -1)$ & $= 11 \cdot 31$  \\
6 & 1365 & $=(2^{12}-1)/(2^2 -1)$ & $=3 \cdot 5 \cdot 7 \cdot 13$  \\
49 & 5461 & $=(2^{14}-1)/(2^2 -1)$ & $=43 \cdot 127$  \\
\hline
4 & 455 & $=(2^{12}-1)/(2^3 +1)$ & $= 5 \cdot 7 \cdot 13$  \\
3 & 585 & $=(2^{12}-1)/(2^3 -1)$ & $= 3^2 \cdot 5 \cdot 13$  \\
9 & 9709& $=(2^{18}-1)/(3(2^3+1))$& $= 7\cdot 19 \cdot 73$  \\
5 & 11275&$=11(2^{10}+1)$ & $=5^2 \cdot 11 \cdot 41$ \\
\hline
\end{tabular}
\caption{Numerical values of $k$ for \emph{odd} linear size $L$
computed from Corollary~\ref{cor:cubic-code-degeneracy-formula}.
The list is complete if $2 \le L \le 20000$.
For example, if $L=945=15 \cdot 63$, then $\frac{k+2}{4} = 1 + 12\cdot(1+5) = 73$.
}
\label{tb:k-numerics}
\end{table}

\section{Real-space renormalization of the cubic code}
\label{sec:real-space-RG-CubicCode}

Renormalization group refers to a machinery 
to extract essential properties of the system at long distances.
For lattice spin models, the so-called ``block spin'' method amounts to
considering a sequence of coarse-grained lattices
and finding effective Hamiltonians pertaining to the coarse-grained lattices~\cite{Wilson1975RG,Kadanoff1977RG}.
In the sequence of coarse-graining, short-ranged correlations will disappear,
and long-ranged essential correlations will remain.
Thus, if a generic Hamiltonian retains its form under the renormalization group flow,
we expect that the Hamiltonian should represent a proper phase of matter,
and can wonder about universal aspects of the phase.

For topologically ordered systems,
it is interesting to look at entanglement structure instead of correlation functions,
since, as they are gapped, the correlation functions of local operators would decay 
exponentially with the distance between the regions the operators act on~\cite{HastingsKoma2006CorDecay},
and hence, if they are renormalization group fixed points, they will have zero correlation length.
As we wish to ignore local deformations of the system,
it is legitimate to apply a finite depth quantum circuit (local unitary) to the system
so as to remove some local
entanglement~\cite{VerstraeteCiracLatorreEtAl2005Renormalization,
Vidal2007ER,AguadoVidal2007Entanglement,Vidal2008class,ChenGuWen2010transformation}.
More precisely, one applies a finite composition of unitary operators on a ground state,
each of which can be written as a product of local unitary operators of disjoint supports,
and then try to identify spins in a product state.

Code Hamiltonians admit an even simpler renormalization scheme.
Under the local unitaries the Hamiltonian is conjugated,
and spins in product states is readily identified by single-spin operators
such as $\sigma^z$ acting on unentangled spins.
If we restrict ourselves to translationally invariant case,
we may also assume that the local unitaries obey the translation invariance.
For example, it is well-known that the 2D toric code model
is a renormalization group fixed point under this scheme%
~\cite{AguadoVidal2007Entanglement}.

In Chapter~\ref{chap:alg-theory}, we have developed enough tools for the simplified renormalization group flow computation.
In fact, we have done a computation in Example~\ref{eg:wen-plaquette}.
To warm up, let us compute the renormalization of the 2D toric code.
The generating matrix is
\[
 \sigma =
\begin{pmatrix}
 1+x & 0 \\
 1+y & 0 \\
\hline
\hline
  0  & 1+\frac{1}{y} \\
  0  & 1+\frac{1}{x} \\
\end{pmatrix}
\]
If we coarse-grain the lattice by blocking two sites in $x$-direction,
then every module is now viewed as a module over $\FF_2[x^{\pm 2}, y^{\pm 1}]$,
and each entry of $\sigma$ is replaced by a $2 \times 2$ matrix,
since $\sigma$ is a map between free modules $G$ and $P$.
Concretely,
\[
 \sigma' = 
\begin{pmatrix}
 1 & 1 & 0 & 0 \\
 x & 1 & 0 & 0 \\
 1+y & 0 & 0 & 0 \\
 0 & 1+y & 0 & 0 \\
\hline
\hline
 0 & 0 & 1+\frac{1}{y} & 0 \\
 0 & 0 & 0 & 1+\frac{1}{y} \\
 0 & 0 & 1 & \frac{1}{x} \\
 0 & 0 & 1 & 1 
\end{pmatrix}
\]
where the double line distinguishes $\sigma^x$-part and $\sigma^z$-part.
Here, the variable $x$ really means the translation along $x$-direction by \emph{two} units of the original lattice.
Applying local unitaries, we see
\[
 \sigma' \to
\left(
\begin{array}{cccc}
 1 & 1 & 0 & 0 \\
 0 & 1+x & 0 & 0 \\
 0 & 1+y & 0 & 0 \\
 0 & 1+y & 0 & 0 \\
\hline
\hline
 0 & 0 & 0 & 0 \\
 0 & 0 & 0 & 1+\frac{1}{y} \\
 0 & 0 & 1 & \frac{1}{x} \\
 0 & 0 & 1 & 1
\end{array}
\right)
\to
\left(
\begin{array}{cccc}
 1 & 1 & 0 & 0 \\
 0 & 1+x & 0 & 0 \\
 0 & 1+y & 0 & 0 \\
 0 & 0 & 0 & 0 \\
\hline
\hline
 0 & 0 & 0 & 0 \\
 0 & 0 & 0 & 1+\frac{1}{y} \\
 0 & 0 & 0 & 1+\frac{1}{x} \\
 0 & 0 & 1 & 1
\end{array}
\right)
\to
\left(
\begin{array}{c|c|c|c}
 1 & 0 & 0 & 0 \\
\hline
 0 & 1+x & 0 & 0 \\
 0 & 1+y & 0 & 0 \\
\hline
 0 & 0 & 0 & 0 \\
\hline
\hline
 0 & 0 & 0 & 0 \\
\hline
 0 & 0 & 0 & 1+\frac{1}{y} \\
 0 & 0 & 0 & 1+\frac{1}{x} \\
\hline
 0 & 0 & 1 & 0
\end{array}
\right).
\]
In the last matrix, it is evident that the first qubit (all the first qubits on every site)
and the fourth qubit are disentangled.
Ignoring those, we see that the generating matrix and also the Hamiltonian retain the original form.
A drawback of this computation is that it is hard to understand why this should happen.
Fortunately, there is a better understanding for two-dimensional quantum double models~\cite{Kitaev2003Fault-tolerant},
which includes the 2D toric code model,
using so-called $G$-injective PEPS by Schuch, Cirac, and P\'erez-Garc\'ia~\cite{SchuchCiracPerez-Garcia2010G-injective}.
Their conclusion is that
the quantum double model is constructed with the regular representation of a finite symmetry group $G$,
and it is a renormalization group fixed point
because of the plethysm $\mathbb{C}G \otimes \mathbb{C}G \cong \mathbb{C}G \otimes (\mathbb{C}1)^{\oplus |G|}$ 
of the regular representation.

Now we return to the cubic code. We perform a similar computation as above,
and find that under blocking of $2 \times 2 \times 2$ sites (16 spins in total)
the cubic code model $A$ decomposes into two non-interacting Hamiltonians $A$ and $B$
living in the coarse-grained lattice,
one of which is the same as the original $A$,
but the other $B$ looks different.
Detailed calculation will be given below.
The generating matrices are as follows.
\[
 \sigma_A = 
\begin{pmatrix}
 1+x+y+z & 0 \\
 1+x y+y z+z x & 0 \\
\hline
\hline
 0 & 1+\frac{1}{xy}+\frac{1}{yz}+\frac{1}{zx} \\
 0 & 1+\frac{1}{x}+\frac{1}{y}+\frac{1}{z} \\
\end{pmatrix}, \quad
 \sigma_B = 
\left(
 \begin{array}{cccc}
 x+z & 1+x &   &   \\
 1+x & 1+z &   &   \\
 x+y & 1+y &   &   \\
 1+y & 1+x &   &   \\
\hline
\hline
   &   & 1+\frac{1}{y} & \frac{1}{x}+\frac{1}{y} \\
   &   & 1+\frac{1}{x} & 1+\frac{1}{y} \\
   &   & 1+\frac{1}{x} & \frac{1}{x}+\frac{1}{z} \\
   &   & 1+\frac{1}{z} & 1+\frac{1}{x}
\end{array}
\right).
\]
We can repeat the renormalization group flow computation for the model $B$ only.
We find that after blocking of $2 \times 2 \times 2$ sites,
$B$ is renormalized to the identical two copies of $B$ itself.
(Calculation will be given below.)
\[
 A \xrightarrow{2\times 2 \times 2} A \oplus B, \quad \text{ and } \quad
B \xrightarrow{2 \times 2 \times 2} B \oplus B.
\]
The renormalization group yields an explicit method to produce the ground states of the cubic code,
Start with a state one wish to encode, put auxiliary qubits in the trivial state,
apply the inverse local unitaries, and iterate the overall process on the \emph{refined} lattice with more auxiliary qubits.
It is slightly different from the MERA (Multiscale Entanglement Renormalization Ansatz)
prescription of 2D toric code state developed 
in~\cite{AguadoVidal2007Entanglement}.
Rather a so-called branching MERA~\cite{EvenblyVidal2012class} is more appropriate.

We do not have a deeper understanding why this should happen.
However, there are some consistency checks
from the degeneracy formula and the annihilator of the module of topological charges.
Corollary~\ref{cor:cubic-code-degeneracy-formula} says that
the number of encoded qubits $k$ is
\begin{align*}
k(L) = 4L -2 
&= \left( 4 \frac{L}{2} - 2\right) & &+ \left( 4 \frac{L}{2} \right) \\
&= \left( 4 \frac{L}{4} - 2\right) + \left( 4 \frac{L}{4} \right) & &+ \left( 4 \frac{L}{4} \right) + \left( 4 \frac{L}{4} \right)
\end{align*}
when the linear system size $L$ is a power of 2.
The expression is decomposed to display contributions from the model $A$ and the model $B$ explicitly.
It suggests that the model $A$, the original cubic code, 
cannot give rise to an identical pair of models at coarse-grained lattice,
whereas $B$ can.
It is not too clear whether the two models $A$ and $B$ are really non-isomorphic.

On the other hand, the annihilator of the topological charge module 
tells us there is something special about $2 \times 2 \times 2$ blocking.
The topological charge module is the torsion part of the virtual excitation module factored by trivial charge module,
i.e., the torsion submodule of $\coker \epsilon$.
As for the cubic code, $\coker \epsilon$ is a torsion module, which is a direct sum of two isomorphic summands.
\begin{align*}
& \coker \epsilon = R/\pp_{xyz} \oplus R / \overline{ \pp_{xyz} } , \\
& R = \FF_2[x^{\pm 1},y^{\pm 1},z^{\pm 1}],\\
& \pp_{xyz} = (1+x+y+z,1+xy+yz+zx) \subset R.\\
\end{align*}
where the bar denotes antipode map.
We focus on the first summand $R/\pp_{xyz}$.
If we replace $\FF_2$ with its algebraic closure $\FF$ for convenience,
$R/\pp_{xyz} = (\FF[x,y,z]/\pp)_{xyz}$ is the coordinate ring of an affine variety defined by $\pp$
localized on the complement of the union of three planes $xyz=0$.
The variety is a union of two isomorphic lines, 
as we have seen in the degeneracy calculation of Section~\ref{sec:cubic-code}.
One line is parameterized as
\[
x = 1+ t, \quad y = 1+ \omega t, \quad z = 1 + \omega^2 t
\]
and the other is
\[
x = 1+t, \quad y = 1+ \omega^2 t, \quad z = 1+ \omega t
\]
where $\omega$ is the primitive third root of unity.
The ideal $\pp_{xyz}$ is the annihilator of a module $M = R/\pp_{xyz}$.
Under coarse-graining $M$ is promoted to a module over the coarse-grained
translation group algebra, which is $R ' = \FF_2[x^{\pm 2}, y^{\pm 2}, z^{\pm 2}]$
in our $2 \times 2 \times 2$ blocking.
Then, the annihilator of $M$ in the new ring $R'$ is just $R' \cap I$.
It is easy to verify that $R' \cap \pp_{xyz} = (1+x^2 + y^2 + z^2, 1+ x^2 y^2 + y^2 z^2 + z^2 x^2 ) = \pp'$
using Gr\"obner basis.
As rings, $R'/\pp'$ and $R/\pp_{xyz}$ are isomorphic.
Thus, it is consistent that $A$ renormalizes to something similar to itself.
Further, we see that the charge annihilator of $B$ must be the same as that of $A$.

This observation tells us that $2^3$ blocking is special.
If we had blocked $3^3$ sites, we would not see the self-reproducing behavior,
since the charge annihilator would be different:
\begin{align*}
 \FF_2[x^{\pm 3}, y^{\pm 3}, z^{\pm 3}] \cap \pp_{xyz} = 
& \left( 1+{y'}+{y'}^2+{y'}^3+{z'}+{y'} {z'}+{y'}^2 {z'}+{z'}^2+{y'} {z'}^2+{z'}^3, \right. \\
& \left. 1+{x'}+{x'}^2+{y'}+{x'} {y'}+{y'}^2+{z'}+{x'} {z'}+{y'} {z'}+{z'}^2 \right)
\end{align*}
where $x'=x^3,y'=y^3,z'=z^3$.

\subsection*{Calculation}

The generating matrix $\sigma_A$
for the stabilizer module of the cubic code
transforms under the coarse-graining by blocking two sites along $x$-direction as
\begin{align*}
 \sigma_A \xrightarrow{\text{coarse-grain }x^2 \to x} \quad
& \sigma_1 = \begin{pmatrix} \sigma_{1X} & 0 \\ 0 & \sigma_{1Z} \end{pmatrix} \\
& \sigma_{1X} =
\begin{pmatrix}
 1+y+z   & 1     \\
 x       & 1+y+z \\
 1+y z   & y+z   \\
 x y+x z & 1+y z \\
\end{pmatrix}, \quad
\sigma_{1Z} =
\begin{pmatrix}
 1+\frac{1}{y z} & \frac{1}{x y}+\frac{1}{x z} \\
 \frac{1}{y}+\frac{1}{z} & 1+\frac{1}{y z} \\
 1+\frac{1}{y}+\frac{1}{z} & \frac{1}{x} \\
 1 & 1+\frac{1}{y}+\frac{1}{z} \\
\end{pmatrix}.
\end{align*}
We apply elementary symplectic transformations.
Recall $\dagger$ is the transpose followed by entry-wise antipode 
map $x \mapsto x^{-1}$, $y \mapsto y^{-1}$, $z \mapsto z^{-1}$.
\[
\sigma_2 = \begin{pmatrix} \sigma_{2X} & 0 \\ 0 & \sigma_{2Z} \end{pmatrix}
 = 
\begin{pmatrix} r_1 & 0 \\ 0 & r_1^\dagger \end{pmatrix} 
\begin{pmatrix} \sigma_{1X} & 0 \\ 0 & \sigma_{1Z} \end{pmatrix}
\begin{pmatrix} c_1 & 0 \\ 0 & c_1^\dagger \end{pmatrix}
 \]
where
\begin{align*}
\sigma_{2X} &=
\left(\begin{array}{c|c}
 0 & 1 \\
\hline
 1+x+y^2+z^2 & 0  \\
 1+y+y^2+z+y z+z^2 & 0  \\
\hline
 0 & 0  \\
\end{array}\right),
& &
\sigma_{2Z} =
\left(
\begin{array}{c|c}
 0 & 0 \\
\hline
0 & 1+\frac{1}{y^2}+\frac{1}{y}+\frac{1}{z^2}+\frac{1}{z}+\frac{1}{y z} \\
0 & 1+\frac{1}{x}+\frac{1}{y^2}+\frac{1}{z^2} \\
\hline
1 & 0 \\
\end{array}
\right) ,
\\
 r_1 &= 
\begin{pmatrix}
 1     & 0   & 0     & 0 \\
 1+y+z & 1   & 0     & 0 \\
 y+z   & 0   & 1     & 0 \\  
 1+y z & y+z & 1+y+z & 1 \\
\end{pmatrix},
& &
c_1 =
\begin{pmatrix}
 1     & 0 \\
 1+y+z & 1 \\
\end{pmatrix}.
\end{align*}
The first and fourth qubit may be factored out from $\sigma_2$.
A subsequent coarse-graining by blocking two sites in each $y$- and $z$-direction gives
\begin{align*}
\sigma_2 \xrightarrow{y^2 \to y,~~z^2 \to z} \quad
& \sigma_3 = \begin{pmatrix} \sigma_{3X} & 0 \\ 0 & \sigma_{3Z} \end{pmatrix} \\
\end{align*}
\begin{align*}
& \sigma_{3X} =
\begin{pmatrix}
 1+x+y+z & 0 & 0 & 0 \\
 0 & 1+x+y+z & 0 & 0 \\
 0 & 0 & 1+x+y+z & 0 \\
 0 & 0 & 0 & 1+x+y+z \\
 1+y+z & 1 & 1 & 1 \\
 z & 1+y+z & z & 1 \\
 y & y & 1+y+z & 1 \\
 y z & y & z & 1+y+z
\end{pmatrix},
\end{align*}
\begin{align*}
 & \sigma_{3Z} =
\begin{pmatrix}
 1+\frac{1}{y}+\frac{1}{z} & \frac{1}{z} & \frac{1}{y} & \frac{1}{y z} \\
 1 & 1+\frac{1}{y}+\frac{1}{z} & \frac{1}{y} & \frac{1}{y} \\
 1 & \frac{1}{z} & 1+\frac{1}{y}+\frac{1}{z} & \frac{1}{z} \\
 1 & 1 & 1 & 1+\frac{1}{y}+\frac{1}{z} \\
 1+\frac{1}{x}+\frac{1}{y}+\frac{1}{z} & 0 & 0 & 0 \\
 0 & 1+\frac{1}{x}+\frac{1}{y}+\frac{1}{z} & 0 & 0 \\
 0 & 0 & 1+\frac{1}{x}+\frac{1}{y}+\frac{1}{z} & 0 \\
 0 & 0 & 0 & 1+\frac{1}{x}+\frac{1}{y}+\frac{1}{z} \\
\end{pmatrix}
\end{align*}
We keep applying elementary symplectic transformations.
\[
 \sigma_4 = \begin{pmatrix} \sigma_{4X} & 0 \\ 0 & \sigma_{4Z} \end{pmatrix}
=
\begin{pmatrix} r_{2X} & 0 \\ 0 & r_{2Z} \end{pmatrix}
\begin{pmatrix} \sigma_{3X} & 0 \\ 0 & \sigma_{3Z} \end{pmatrix}
\begin{pmatrix} c_{2X} & 0 \\ 0 & c_{2Z} \end{pmatrix}
\]
where
\begin{align*}
\sigma_{4X} = 
\left(\begin{array}{c|c|cc}
 0 & 1 & 0 & 0  \\
 0 & 0 & 0 & 0  \\
\hline
 0 & 0 & 1+x+y+z & 0  \\
 0 & 0 & 0 & 1+x+y+z  \\
 1+x+y+z & 0 & 0 & 0  \\
 x+x y+z+x z & 0 & 1+y & y+z \\
 y+x y & 0 & 1+z & 1+y       \\
 x y+y z & 0 & y+z & 1+z     \\
\end{array}\right),
\end{align*}
\begin{align*}
\sigma_{4Z} =
\left(\begin{array}{c|ccc}
 0 & 0 & 0 & 0 \\
 1 & 0 & 0 & 0 \\
\hline
 0 & 1+\frac{1}{y} & 1+\frac{1}{z} & \frac{1}{y}+\frac{1}{z} \\
 0 & \frac{1}{y}+\frac{1}{z} & 1+\frac{1}{y} & 1+\frac{1}{z} \\
 0 & \frac{1}{x}+\frac{1}{x y}+\frac{1}{z}+\frac{1}{x z} & \frac{1}{y}+\frac{1}{x y} & \frac{1}{x y}+\frac{1}{y z} \\
 0 & 1+\frac{1}{x}+\frac{1}{y}+\frac{1}{z} & 0 & 0 \\
 0 & 0 & 1+\frac{1}{x}+\frac{1}{y}+\frac{1}{z} & 0 \\
 0 & 0 & 0 & 1+\frac{1}{x}+\frac{1}{y}+\frac{1}{z}
\end{array}\right),
\end{align*}
\begin{align*}
r_{2X} = 
\begin{pmatrix}
 1 & 0 & 0 & 0 & 1 & 0 & 0 & 0   \\
 1+y+z & 1 & 1 & 1 & 1+x+y+z & 0 & 0 & 0   \\
 0 & 0 & 1 & 0 & 0 & 0 & 0 & 0   \\
 0 & 0 & 0 & 1 & 0 & 0 & 0 & 0   \\
 1 & 0 & 0 & 0 & 0 & 0 & 0 & 0   \\
 1+y+z & 0 & 0 & 0 & 1+y+z & 1 & 0 & 0   \\
 y & 0 & 0 & 0 & y & 0 & 1 & 0   \\
 y & 0 & 0 & 0 & y & 0 & 0 & 1   \\
\end{pmatrix},
\end{align*}
\begin{align*}
r_{2Z} =
\begin{pmatrix}
 0 & 1+\frac{1}{x}+\frac{1}{y}+\frac{1}{z} & 0 & 0 & 1 & 1+\frac{1}{y}+\frac{1}{z} & \frac{1}{y} & \frac{1}{y} \\
 0 & 1 & 0 & 0 & 0 & 0 & 0 & 0 \\
 0 & 1 & 1 & 0 & 0 & 0 & 0 & 0 \\
 0 & 1 & 0 & 1 & 0 & 0 & 0 & 0 \\
 1 & \frac{1}{x} & 0 & 0 & 1 & 0 & 0 & 0 \\
 0 & 0 & 0 & 0 & 0 & 1 & 0 & 0 \\
 0 & 0 & 0 & 0 & 0 & 0 & 1 & 0 \\
 0 & 0 & 0 & 0 & 0 & 0 & 0 & 1 \\
\end{pmatrix}.
\end{align*}
\begin{align*}
c_{2X} =
\begin{pmatrix}
 1 & 0 & 0 & 0 \\
 x & 1 & 1 & 1 \\
 0 & 0 & 1 & 0 \\
 0 & 0 & 0 & 1 \\
\end{pmatrix}, \quad
c_{2Z} =
\begin{pmatrix}
 1 & 1+\frac{1}{y}+\frac{1}{z} & \frac{1}{y} & \frac{1}{y} \\
 0 & 1 & 0 & 0 \\
 0 & 0 & 1 & 0 \\
 0 & 0 & 0 & 1
\end{pmatrix},
\end{align*}
Factoring out the first and second qubit from $\sigma_4$
and applying elementary symplectic transformations,
we arrive at
\[
 \sigma_5 = 
\begin{pmatrix} \sigma_{5X} & 0 \\ 0 & \sigma_{5Z} \end{pmatrix}
=
\begin{pmatrix} r_{3X} & 0 \\ 0 & r_{3Z} \end{pmatrix}
\left( \begin{pmatrix} \sigma_{4X} & 0 \\ 0 & \sigma_{4Z} \end{pmatrix} \rvert_{\{1,2\}^c} \right)
\begin{pmatrix} c_{3X} & 0 \\ 0 & c_{3Z} \end{pmatrix}
\]
where
\begin{align*}
\sigma_{5X} =
\begin{pmatrix}
 1+x+y+z & 0 & 0  \\
 1+x y+x z+y z & 0 & 0  \\
 0 & x+z & 1+x  \\
 0 & 1+x & 1+z  \\
 0 & x+y & 1+y  \\
 0 & 1+y & 1+x  \\
\end{pmatrix},
\quad
\sigma_{5Z} =
\begin{pmatrix}
 0 & 0 & 1+\frac{1}{x y}+\frac{1}{x z}+\frac{1}{y z} \\
 0 & 0 & 1+\frac{1}{x}+\frac{1}{y}+\frac{1}{z} \\
 1+\frac{1}{y} & \frac{1}{x}+\frac{1}{y} & 0 \\
 1+\frac{1}{x} & 1+\frac{1}{y} & 0 \\
 1+\frac{1}{x} & \frac{1}{x}+\frac{1}{z} & 0 \\
 1+\frac{1}{z} & 1+\frac{1}{x} & 0
\end{pmatrix},
\end{align*}
\begin{align*}
r_3 =
\begin{pmatrix}
 0 & 0 & 1 & 0 & 0 & 0 \\
 0 & 0 & 1 & 1 & 1 & 1 \\
 1 & 1 & x & 1 & 0 & 0 \\
 1 & 0 & 1 & 1 & 1 & 0 \\
 1 & 0 & 0 & 0 & 1 & 0 \\
 0 & 1 & 1+x & 1 & 0 & 0  \\
\end{pmatrix},
\quad
r_{3Z} =
\begin{pmatrix}
1 & \frac{1}{x} & 1 & 1 & 1 & 1 \\
0 & 0 & 0 & 0 & 0 & 1 \\
1 & 0 & 0 & 0 & 1 & 1 \\
0 & 1 & 0 & 1 & 0 & 1 \\
0 & 1 & 0 & 1 & 1 & 0 \\
1 & 1 & 0 & 0 & 1 & 1
\end{pmatrix},
\end{align*}
\begin{align*}
c_3 =
\begin{pmatrix}
 1 & 0 & 0 \\
 1 & 1 & 0 \\
 x & 0 & 1 \\
\end{pmatrix},
\quad
c_{3Z} =
\begin{pmatrix}
1 & 0 & 1 \\
0 & 1 & 1 \\
0 & 0 & 1
\end{pmatrix}.
\end{align*}
It is clear that $\sigma_A$, $\sigma_1$, $\sigma_2$, $\sigma_3$, $\sigma_4$, and $\sigma_5$ are all equivalent
because we applied only elementary symplectic transformations.
$\sigma_5$ shows a decomposition of $\sigma$ into two non-interacting two models,
one of which is $\sigma_A$ at $2 \times 2 \times 2$ coarse-grained lattice and another is $\sigma_B$.
We perform a similar renormalization for $\sigma_B$.
The coarse-graining by blocking two sites in $x$-direction gives
\begin{align*}
\sigma_B \xrightarrow{x^2 \to x} \quad
&\sigma_{B1} = \begin{pmatrix} \sigma_{B1X} & 0 \\ 0 & \sigma_{B1Z} \end{pmatrix} \\
&
\sigma_{B1X} =
\begin{pmatrix}
 z & 1 & 1 & 1   \\
 x & z & x & 1   \\
 1 & 1 & 1+z & 0 \\
 x & 1 & 0 & 1+z \\
 y & 1 & 1+y & 0 \\
 x & y & 0 & 1+y \\
 1+y & 0 & 1 & 1 \\
 0 & 1+y & x & 1 \\
\end{pmatrix},
\quad
\sigma_{B1Z} =
\begin{pmatrix}
 1+\frac{1}{y} & 0 & \frac{1}{y} & \frac{1}{x} \\
 0 & 1+\frac{1}{y} & 1 & \frac{1}{y} \\
 1 & \frac{1}{x} & 1+\frac{1}{y} & 0 \\
 1 & 1 & 0 & 1+\frac{1}{y} \\
 1 & \frac{1}{x} & \frac{1}{z} & \frac{1}{x} \\
 1 & 1 & 1 & \frac{1}{z} \\
 1+\frac{1}{z} & 0 & 1 & \frac{1}{x} \\
 0 & 1+\frac{1}{z} & 1 & 1
\end{pmatrix}.
\end{align*}
Apply elementary symplectic transformations to factor out trivial qubits.
\[
\sigma_{B2} =
 \begin{pmatrix} \sigma_{B2X} & 0 \\ 0 & \sigma_{B2Z} \end{pmatrix}
=
\begin{pmatrix} r_{4X} & 0 \\ 0 & r_{4Z} \end{pmatrix}
\begin{pmatrix} \sigma_{B1X} & 0 \\ 0 & \sigma_{B1Z} \end{pmatrix}
\begin{pmatrix} c_{4X} & 0 \\ 0 & c_{4Z} \end{pmatrix}
\]
where
{\small
\begin{align*}
&\sigma_{B2X} =                    & & \sigma_{B2Z} = \\
&\left(
\begin{array}{c|cc|c}
 0 & 0 & 0 & 1 \\
 0 & 0 & 0 & 0 \\
 1 & 0 & 0 & 0 \\
 0 & 0 & 0 & 0 \\
\hline
 0 & 1+y & y+z & 0       \\
 0 & 1+x+z+y z & y+z & 0 \\
 0 & 0 & x+y+z+y z & 0   \\
 0 & y+z & 1+x+y+y z & 0 \\
\end{array}
\right),
& &
\left(\begin{array}{c|c|c|c}
 0 & 0 & 0 & 0 \\
 0 & 0 & 1 & 0 \\
 0 & 0 & 0 & 0 \\
 1 & 0 & 0 & 0 \\
\hline
 0 & 1+\frac{1}{x}+\frac{1}{z}+\frac{1}{y z} & 0 & \frac{1}{y}+\frac{1}{z} \\
 0 & 1+\frac{1}{y} & 0 & \frac{1}{y}+\frac{1}{z} \\
 0 & \frac{1}{y}+\frac{1}{z} & 0 & 1+\frac{1}{x}+\frac{1}{y}+\frac{1}{y z} \\
 0 & 0 & 0 & \frac{1}{x}+\frac{1}{y}+\frac{1}{z}+\frac{1}{y z}
\end{array}\right),
\end{align*}
\begin{align*}
& r_{4X} = &  & r_{4Z} = \\
& \begin{pmatrix}
 1 & 0 & z & 0 & 0 & 0 & 0 & 0    \\
 y & 1 & 1+y & 0 & z & 1 & 1 & 1  \\
 0 & 0 & 1 & 0 & 0 & 0 & 0 & 0    \\
 1+y & 0 & 1 & 1 & 1 & 1 & 1+z & 0\\
 0 & 0 & y & 0 & 1 & 0 & 0 & 0         \\
 1+y & 0 & x+z+y z & 0 & 0 & 1 & 0 & 0 \\
 0 & 0 & 1+y & 0 & 0 & 0 & 1 & 1       \\
 1 & 0 & z & 0 & 0 & 0 & 0 & 1         \\
\end{pmatrix},
& &
\begin{pmatrix}
 1 & 1 & 0 & 1+\frac{1}{z} & 0 & 1+\frac{1}{y} & 1 & 1 \\
 0 & 1 & 0 & 0 & 0 & 0 & 0 & 0 \\
 \frac{1}{z} & \frac{1}{x} & 1 & \frac{1}{x} & \frac{1}{y} & \frac{1}{x} & 1+\frac{1}{y} & 0 \\
 0 & 0 & 0 & 1 & 0 & 0 & 0 & 0 \\
 0 & \frac{1}{z} & 0 & 1 & 1 & 0 & 0 & 0 \\
 0 & 1 & 0 & 1 & 0 & 1 & 0 & 0 \\
 0 & 1 & 0 & 1+\frac{1}{z} & 0 & 0 & 1 & 0 \\
 0 & 0 & 0 & 1+\frac{1}{z} & 0 & 0 & 1 & 1 \\
\end{pmatrix}.
\end{align*}
\[
c_{4X} =
\begin{pmatrix}
 1 & 1 & 0 & 0   \\
 0 & 1 & 1+z & 0 \\
 0 & 0 & 1 & 0   \\
 0 & 1+z & z & 1 \\
\end{pmatrix},
\quad
c_{4Z} =
\begin{pmatrix}
1 & 1 & 0 & \frac{1}{y} \\
0 & 1 & 0 & 1 \\
0 & 1+\frac{1}{y} & 1 & 1 \\
0 & 0 & 0 & 1
\end{pmatrix},
\]
}
Factoring out trivial qubits from $\sigma_{B2}$ and coarse-graining by blocking two sites along $z$-direction,
we have
\begin{align*}
& \sigma_{B2}|_{\{1,2,3,4\}^c} \xrightarrow{z^2 \to z} \quad 
\sigma_{B3} = \begin{pmatrix} \sigma_{B3X} & 0 \\ 0 & \sigma_{B3Z} \end{pmatrix}, \text{ where }
\end{align*}
\begin{align*}
& \sigma_{B3X} = & & \sigma_{B3Z} =\\
& \begin{pmatrix}
 1+y & 0 & y & 1    \\
 0 & 1+y & z & y    \\
 1+x & 1+y & y & 1  \\
 z+y z & 1+x & z & y\\
 0 & 0 & x+y & 1+y  \\
 0 & 0 & z+y z & x+y\\
 y & 1 & 1+x+y & y  \\
 z & y & y z & 1+x+y\\
\end{pmatrix},
& &
\begin{pmatrix}
 1+\frac{1}{x} & \frac{1}{z}+\frac{1}{y z} & \frac{1}{y} & \frac{1}{z} \\
 1+\frac{1}{y} & 1+\frac{1}{x} & 1 & \frac{1}{y} \\
 1+\frac{1}{y} & 0 & \frac{1}{y} & \frac{1}{z} \\
 0 & 1+\frac{1}{y} & 1 & \frac{1}{y} \\
 \frac{1}{y} & \frac{1}{z} & 1+\frac{1}{x}+\frac{1}{y} & \frac{1}{y z} \\
 1 & \frac{1}{y} & \frac{1}{y} & 1+\frac{1}{x}+\frac{1}{y} \\
 0 & 0 & \frac{1}{x}+\frac{1}{y} & \frac{1}{z}+\frac{1}{y z} \\
 0 & 0 & 1+\frac{1}{y} & \frac{1}{x}+\frac{1}{y}
\end{pmatrix}.
\end{align*}
Again apply elementary symplectic transformations.
\[
\sigma_{B4} =
 \begin{pmatrix} \sigma_{B4X} & 0 \\ 0 & \sigma_{B4Z} \end{pmatrix}
=
\begin{pmatrix} r_{5X} & 0 \\ 0 & r_{5Z} \end{pmatrix}
\begin{pmatrix} \sigma_{B3X} & 0 \\ 0 & \sigma_{B3Z} \end{pmatrix}
\begin{pmatrix} c_{5X} & 0 \\ 0 & c_{5Z} \end{pmatrix}
\]
where
\[
 \sigma_{B4X} =
\left(\begin{array}{c|c|c|c}
 0 & 0 & 1 & 0 \\
\hline
 1+x+y+x y^2+z+y z & 0 & 0 & x y+y^2+x y^2+y z \\
\hline
 0 & 0 & 0 & 0 \\
\hline
 1+x^2+x y+x^2 y+z+y^2 z & 0 & 0 & x^2 y+z+y z+y^2 z \\
 x+y+x y+y^2 & 0 & 0 & 1+y+x y+y^2 \\
 z+y^2 z & 0 & 0 & x+y+y z+y^2 z \\
\hline
 0 & 1 & 0 & 0 \\
\hline
 y+x y+x y^2+y^2 z & 0 & 0 & 1+x+y+y^2+x y^2+z+y z+y^2 z
\end{array}\right),
\]
\[
 \sigma_{B4Z} =
\left(\begin{array}{cc|c|c}
 0 & 0 & 0 & 0 \\
\hline
 1+\frac{1}{y^2} & 1+\frac{1}{x} & 0 & 1+\frac{1}{x}+\frac{1}{y}+\frac{1}{y z} \\
\hline
 0 & 0 & 1 & 0 \\
\hline
 1+\frac{1}{y} & 1+\frac{1}{y} & 0 & 1+\frac{1}{z} \\
 1+\frac{1}{x}+\frac{1}{x y} & \frac{1}{z} & 0 & \frac{1}{x z}+\frac{1}{y z} \\
 \frac{1}{y^2} & \frac{1}{y} & 0 & 1+\frac{1}{x}+\frac{1}{z}+\frac{1}{y z} \\
\hline
 0 & 0 & 0 & 0 \\
\hline
 1+\frac{1}{y^2} & 0 & 0 & \frac{1}{x}+\frac{1}{y}+\frac{1}{z}+\frac{1}{y z}
\end{array}\right),
\]
\[
 r_{5X} =
\left(\begin{array}{cccccccc}
 1 & 0 & 0 & 0 & 0 & 0 & 0 & 0 \\
 1+x+y+x y+z & 1 & 0 & 0 & 0 & 0 & 1+y & 0 \\
 1+x+y & y & 1 & 1 & 1+x & 1+y & x+y & 1+y \\
 1+x^2+y z & 0 & 0 & 1 & 0 & 0 & 1+x & 0 \\
 x+y & 0 & 0 & 0 & 1 & 0 & 0 & 0 \\
 z+y z & 0 & 0 & 0 & 0 & 1 & 0 & 0 \\
 0 & 0 & 0 & 0 & 0 & 0 & 1 & 0 \\
 y+x y+z+y z & 0 & 0 & 0 & 0 & 0 & y & 1
\end{array}\right),
\]
\[
 r_{5Z} =
\left(\begin{array}{cccccccc}
 1 & 1+\frac{1}{x}+\frac{1}{y}+\frac{1}{x y}+\frac{1}{z} & \frac{1}{x y} & 1+\frac{1}{x^2}+\frac{1}{y z} & \frac{1}{x}+\frac{1}{y} & \frac{1}{z}+\frac{1}{y z} & 0 & \frac{1}{y}+\frac{1}{x y}+\frac{1}{z}+\frac{1}{y z} \\
 0 & 1 & \frac{1}{y} & 0 & 0 & 0 & 0 & 0 \\
 0 & 0 & 1 & 0 & 0 & 0 & 0 & 0 \\
 0 & 0 & 1 & 1 & 0 & 0 & 0 & 0 \\
 0 & 0 & 1+\frac{1}{x} & 0 & 1 & 0 & 0 & 0 \\
 0 & 0 & 1+\frac{1}{y} & 0 & 0 & 1 & 0 & 0 \\
 0 & 1+\frac{1}{y} & 1+\frac{1}{y} & 1+\frac{1}{x} & 0 & 0 & 1 & \frac{1}{y} \\
 0 & 0 & 1+\frac{1}{y} & 0 & 0 & 0 & 0 & 1
\end{array}\right),
\]
\[
 c_{5X} = 
\left(\begin{array}{cccc}
 y & 0 & 1 & 1+y \\
 1+x+x y & 1 & 1+x & y+x y \\
 1+y & 0 & 1 & y \\
 0 & 0 & 0 & 1
\end{array}\right),
\quad
c_{5Z} =
\left(\begin{array}{cccc}
 \frac{1}{y} & 0 & 1 & \frac{1}{z} \\
 0 & 1 & 0 & 1 \\
 1+\frac{1}{y} & 0 & 1 & \frac{1}{z} \\
 0 & 0 & 0 & 1
\end{array}\right).
\]
Factor out the first, third, and seventh qubits from $\sigma_{B4}$,
and coarse-grain by blocking two sites along $y$-direction.
\[
\sigma_{B4} \xrightarrow{y^2 \to y} \quad 
\sigma_{B5} = \begin{pmatrix} \sigma_{B5X} & 0 \\ 0 & \sigma_{B5Z} \end{pmatrix}, \text{ where }
\]
\[
\sigma_{B5X} =
\left(\begin{array}{cccc}
 1+x+x y+z & 1+z & y+x y & x+z \\
 y+y z & 1+x+x y+z & x y+y z & y+x y \\
 1+x^2+z+y z & x+x^2 & z+y z & x^2+z \\
 x y+x^2 y & 1+x^2+z+y z & x^2 y+y z & z+y z \\
 x+y & 1+x & 1+y & 1+x \\
 y+x y & x+y & y+x y & 1+y \\
 z+y z & 0 & x+y z & 1+z \\
 0 & z+y z & y+y z & x+y z \\
 x y+y z & 1+x & (1+x+z)(1+y) & 1+z \\
 y+x y & x y+y z & y+y z & (1+x+z)(1+y)
\end{array}\right),
\]
\[
\sigma_{B5Z} =
\left(\begin{array}{cccccc}
 1+\frac{1}{y} & 0 & 1+\frac{1}{x} & 0 & 1+\frac{1}{x} & \frac{1}{y}+\frac{1}{y z} \\
 0 & 1+\frac{1}{y} & 0 & 1+\frac{1}{x} & 1+\frac{1}{z} & 1+\frac{1}{x} \\
 1 & \frac{1}{y} & 1 & \frac{1}{y} & 1+\frac{1}{z} & 0 \\
 1 & 1 & 1 & 1 & 0 & 1+\frac{1}{z} \\
 1+\frac{1}{x} & \frac{1}{x y} & \frac{1}{z} & 0 & \frac{1}{x z} & \frac{1}{y z} \\
 \frac{1}{x} & 1+\frac{1}{x} & 0 & \frac{1}{z} & \frac{1}{z} & \frac{1}{x z} \\
 \frac{1}{y} & 0 & 0 & \frac{1}{y} & 1+\frac{1}{x}+\frac{1}{z} & \frac{1}{y z} \\
 0 & \frac{1}{y} & 1 & 0 & \frac{1}{z} & 1+\frac{1}{x}+\frac{1}{z} \\
 1+\frac{1}{y} & 0 & 0 & 0 & \frac{1}{x}+\frac{1}{z} & \frac{1}{y}+\frac{1}{y z} \\
 0 & 1+\frac{1}{y} & 0 & 0 & 1+\frac{1}{z} & \frac{1}{x}+\frac{1}{z}
\end{array}\right).
\]
Apply elementary symplectic transformations to factor out two more qubits.
\[
\sigma_{B6} =
 \begin{pmatrix} \sigma_{B6X} & 0 \\ 0 & \sigma_{B6Z} \end{pmatrix}
=
\begin{pmatrix} r_{6X} & 0 \\ 0 & r_{6Z} \end{pmatrix}
\begin{pmatrix} \sigma_{B5X} & 0 \\ 0 & \sigma_{B5Z} \end{pmatrix}
\begin{pmatrix} \id_{4 \times 4} & 0 \\ 0 & c_{6Z} \end{pmatrix}
\]
where
\[
 \sigma_{B6X} =
\left(\begin{array}{cccc}
 1+x+x y+z & 1+z & y+x y & x+z \\
 y+y z & 1+x+x y+z & x y+y z & y+x y \\
\hline
 0 & 0 & 0 & 0 \\
\hline
 x y+x^2 y & 1+x^2+z+y z & x^2 y+y z & z+y z \\
 x+y & 1+x & 1+y & 1+x \\
 y+x y & x+y & y+x y & 1+y \\
 z+y z & 0 & x+y z & 1+z \\
\hline
 0 & 0 & 0 & 0 \\
\hline
 x y+y z & 1+x & (1+x+z)(1+y) & 1+z \\
 y+x y & x y+y z & y+y z & (1+x+z)(1+y) 
\end{array}\right),
\]
\[
 \sigma_{B6Z} =
\left(\begin{array}{c|c|c|ccc}
 0 & \frac{1}{y}+\frac{1}{x y} & 0 & \frac{1}{y^2}+\frac{1}{y} & \frac{1}{x}+\frac{1}{y}+\frac{1}{z}+\frac{1}{x z} & \frac{1}{x^2}+\frac{1}{x}+\frac{1}{x y}+\frac{1}{x z} \\
 0 & 1+\frac{1}{y} & 0 & 1+\frac{1}{x} & 1+\frac{1}{z} & 1+\frac{1}{x} \\
\hline
 1 & 0 & 0 & 0 & 0 & 0 \\
\hline
 0 & 1+\frac{1}{y} & 0 & 1+\frac{1}{y} & 1+\frac{1}{z} & 1+\frac{1}{z} \\
 0 & \frac{1}{x y}+\frac{1}{y z} & 0 & \frac{1}{y}+\frac{1}{x y} & 1+\frac{1}{x}+\frac{1}{z^2}+\frac{1}{x z} & 1+\frac{1}{x^2}+\frac{1}{z^2}+\frac{1}{y z} \\
 0 & 1+\frac{1}{x} & 0 & \frac{1}{x y}+\frac{1}{z} & \frac{1}{x}+\frac{1}{z} & \frac{1}{x^2}+\frac{1}{x} \\
 0 & 0 & 0 & \frac{1}{y^2}+\frac{1}{y} & 1+\frac{1}{x}+\frac{1}{y}+\frac{1}{z} & \frac{1}{y}+\frac{1}{x y} \\
\hline
 0 & 0 & 1 & 0 & 0 & 0 \\
\hline
 0 & 0 & 0 & \frac{1}{y^2}+\frac{1}{y} & 1+\frac{1}{x}+\frac{1}{y}+\frac{1}{z} & 1+\frac{1}{x}+\frac{1}{x y}+\frac{1}{z} \\
 0 & 1+\frac{1}{y} & 0 & 0 & 1+\frac{1}{z} & \frac{1}{x}+\frac{1}{z}
\end{array}\right),
\]
\[
 r_{6X} =
\left(\begin{array}{cccccccccc}
 1 & 0 & 0 & 0 & 0 & 0 & 0 & 0 & 0 & 0 \\
 0 & 1 & 0 & 0 & 0 & 0 & 0 & 0 & 0 & 0 \\
 1+y & 0 & 1 & 1 & 1+x & x & y & 0 & 1+y & 0 \\
 0 & 0 & 0 & 1 & 0 & 0 & 0 & 0 & 0 & 0 \\
 0 & 0 & 0 & 0 & 1 & 0 & 0 & 0 & 0 & 0 \\
 0 & 0 & 0 & 0 & 0 & 1 & 0 & 0 & 0 & 0 \\
 0 & 0 & 0 & 0 & 0 & 0 & 1 & 0 & 0 & 0 \\
 x+y & 0 & 0 & 0 & 1+x+z & x & y & 1 & 1+y & 0 \\
 0 & 0 & 0 & 0 & 0 & 0 & 0 & 0 & 1 & 0 \\
 0 & 0 & 0 & 0 & 0 & 0 & 0 & 0 & 0 & 1
\end{array}\right),
\]
\[
 r_{6Z} =
\left(\begin{array}{cccccccccc}
 1 & 0 & 1+\frac{1}{y} & 0 & 0 & 0 & 0 & \frac{1}{x}+\frac{1}{y} & 0 & 0 \\
 0 & 1 & 0 & 0 & 0 & 0 & 0 & 0 & 0 & 0 \\
 0 & 0 & 1 & 0 & 0 & 0 & 0 & 0 & 0 & 0 \\
 0 & 0 & 1 & 1 & 0 & 0 & 0 & 0 & 0 & 0 \\
 0 & 0 & 1+\frac{1}{x} & 0 & 1 & 0 & 0 & 1+\frac{1}{x}+\frac{1}{z} & 0 & 0 \\
 0 & 0 & \frac{1}{x} & 0 & 0 & 1 & 0 & \frac{1}{x} & 0 & 0 \\
 0 & 0 & \frac{1}{y} & 0 & 0 & 0 & 1 & \frac{1}{y} & 0 & 0 \\
 0 & 0 & 0 & 0 & 0 & 0 & 0 & 1 & 0 & 0 \\
 0 & 0 & 1+\frac{1}{y} & 0 & 0 & 0 & 0 & 1+\frac{1}{y} & 1 & 0 \\
 0 & 0 & 0 & 0 & 0 & 0 & 0 & 0 & 0 & 1
\end{array}\right),
\]
\[
 c_{6Z} =
\left(\begin{array}{cccccc}
 1 & 0 & 1 & \frac{1}{y} & 1 & 1+\frac{1}{x}+\frac{1}{z} \\
 0 & 1 & 0 & 0 & 0 & 0 \\
 0 & \frac{1}{y} & 1 & 0 & \frac{1}{z} & 1+\frac{1}{x}+\frac{1}{z} \\
 0 & 0 & 0 & 1 & 0 & 0 \\
 0 & 0 & 0 & 0 & 1 & 0 \\
 0 & 0 & 0 & 0 & 0 & 1
\end{array}\right).
\]
The third and eighth qubits in $\sigma_{B6}$ are trivial.
We finally finish our lengthy transformation.
\[
 \sigma_{B7} =
 \begin{pmatrix} \sigma_{B7X} & 0 \\ 0 & \sigma_{B7Z} \end{pmatrix}
=
\begin{pmatrix} r_{7X} & 0 \\ 0 & r_{7Z} \end{pmatrix}
\left( \begin{pmatrix} \sigma_{B6X} & 0 \\ 0 & \sigma_{B6Z} \end{pmatrix}\rvert_{\{3,8\}^c} \right)
\begin{pmatrix} c_{7X} & 0 \\ 0 & c_{7Z} \end{pmatrix}
\]
where
\[
 \sigma_{B7X} =
\left(\begin{array}{cccc}
 x+z & 1+x & 0 & 0 \\
 1+x & 1+z & 0 & 0 \\
 x+y & 1+y & 0 & 0 \\
 1+y & 1+x & 0 & 0 \\
 0 & 0 & x+z & 1+x \\
 0 & 0 & 1+x & 1+z \\
 0 & 0 & x+y & 1+y \\
 0 & 0 & 1+y & 1+x
\end{array}\right), \quad
\sigma_{B7Z} =
\left(\begin{array}{cccc}
 1+\frac{1}{y} & \frac{1}{x}+\frac{1}{y} & 0 & 0 \\
 1+\frac{1}{x} & 1+\frac{1}{y} & 0 & 0 \\
 1+\frac{1}{x} & \frac{1}{x}+\frac{1}{z} & 0 & 0 \\
 1+\frac{1}{z} & 1+\frac{1}{x} & 0 & 0 \\
 0 & 0 & 1+\frac{1}{y} & \frac{1}{x}+\frac{1}{y} \\
 0 & 0 & 1+\frac{1}{x} & 1+\frac{1}{y} \\
 0 & 0 & 1+\frac{1}{x} & \frac{1}{x}+\frac{1}{z} \\
 0 & 0 & 1+\frac{1}{z} & 1+\frac{1}{x}
\end{array}\right),
\]
\[
 r_{7X} =
\left(\begin{array}{cccccccc}
 1+z & 1 & 0 & 0 & z & 0 & 0 & 1 \\
 1+x & 1 & 1 & x & 0 & 1 & 0 & 1 \\
 1+x+y & 0 & 1 & x & 1 & 1+y & y & 1 \\
 1+y & 1 & 1 & 0 & 1+x & y & y & 1 \\
 1+x & 1 & 1 & x+z & 0 & 0 & 0 & 1 \\
 x & 1 & 1 & x+z & 0 & 1 & 1 & 1 \\
 1+x & 1 & 1 & 1+x+z & 1 & 1 & 1 & 1 \\
 0 & 0 & 0 & 1 & 0 & 0 & 0 & 0
\end{array}\right),
\]
\[
r_{7Z} =
\left(\begin{array}{cccccccc}
 0 & 0 & 1 & 0 & 0 & 0 & 0 & 1 \\
 0 & 1 & 0 & 0 & 0 & 1 & 1 & 1 \\
 0 & 1 & 0 & 0 & 0 & 0 & 0 & 1 \\
 1 & 1 & \frac{1}{x} & 0 & 1 & 0 & 1 & 0 \\
 1+\frac{1}{y} & 1 & 1+\frac{1}{x}+\frac{1}{x y} & 0 & 1+\frac{1}{y} & 1 & 1+\frac{1}{y} & \frac{1}{y} \\
 1+\frac{1}{x}+\frac{1}{y} & \frac{1}{x} & \frac{1}{x^2}+\frac{1}{x}+\frac{1}{x y}+\frac{1}{z} & 0 & \frac{1}{x}+\frac{1}{y} & 0 & \frac{1}{x}+\frac{1}{y} & 1+\frac{1}{y}+\frac{1}{z} \\
 1+\frac{1}{x} & \frac{1}{x} & \frac{1}{x^2}+\frac{1}{x}+\frac{1}{z} & 0 & \frac{1}{x} & 0 & 1+\frac{1}{x} & 1+\frac{1}{z} \\
 1+\frac{1}{z} & \frac{1}{z} & \frac{1}{x z} & 1 & \frac{1}{z} & \frac{1}{z} & 1 & 1+\frac{1}{z}
\end{array}\right),
\]
\[
 c_{7X} =
\left(\begin{array}{cccc}
 1 & 0 & 0 & 1+z \\
 0 & 1 & 0 & 1 \\
 1 & 0 & 1 & 1+z \\
 1 & 1 & 0 & 1+z
\end{array}\right),\quad
c_{7Z} =
\left(\begin{array}{cccc}
 \frac{1}{z} & 1+\frac{1}{z} & 1 & 0 \\
 1 & 1 & 0 & 0 \\
 \frac{1}{y} & 1+\frac{1}{y} & 0 & 1 \\
 0 & 1 & 0 & 0
\end{array}\right).
\]
We see that $\sigma_{B7}$ is a direct sum of two copies of $\sigma_B$.
This verifies the bifurcation:
\[
 A \xrightarrow{2\times 2 \times 2} A \oplus B, \text{ and } B \xrightarrow{2 \times 2 \times 2} B \oplus B.
\]

\section{Thermal partition function of the cubic code}
\label{sec:thermal-partition-function}

The only relevant property of the cubic code in relation to a thermal partition function
is that the generating map $\sigma$ for the stabilizer module is injective.
\[
 0 \to G \xrightarrow{\sigma} P
\]
This property is shared with 1D Ising model (Example~\ref{eg:1d-ising}),
2D toric code (Example~\ref{eg:2d-toric}), and 3D Chamon model (Example~\ref{eg:ChamonModel}).
Slightly more generally, Lemma~\ref{lem:coker-epsilon-resolution-length-D} says that
for two-dimensional exact code Hamiltonians one can choose an injective generating map 
that gives rise to an equivalent Hamiltonian in the sense of Definition~\ref{defn:equiv-H}.

Let $H = -\sum_i P_i$ be the Hamiltonian of the cubic code on $L \times L \times L$ periodic lattice.
We set the coupling constant to be $1$ so that each term literally squares to the identity $P_i^2 = 1$.
The thermal partition function is
\[
 \calZ = \calZ(\beta) = \trace \exp ( - \beta H )
\]
where $\beta$ is an inverse temperature.
Since $H$ consists of commuting terms,
the exponential function can be written as products.
\[
 \calZ = \trace \prod_i (\cosh \beta + P_i \sinh \beta) 
 = \cosh^M \beta \trace \prod_i (1 + x P_i)
\]
where $M=tL^3$ is the number of terms in $H$ and $x = \tanh \beta$.
Since the trace of a Pauli operator that is not identity is zero,
only the term in the expansion of the product 
that is proportional to the identity contributes to $\calZ$.
If there are $N=qL^3$ qubits in the system, 
the trace of the identity operator is $2^N$.
\[
 \frac{\calZ}{2^N \cosh^M \beta} = \sum_{\gamma \in \Gamma} x^{|\gamma|}
\]
where $\gamma$ runs over all possible collections of terms $P_i$ of $H$
that multiply to the identity. $\Gamma$ contains the empty collection.
Thus, $\Gamma$ is in one-to-one correspondence with $\ker \sigma_L$.
(It may be nonzero due to the periodic boundary conditions
although $\sigma$ was injective. See Section~\ref{sec:degeneracy}.)
Therefore, the cardinality of $\Gamma$ is precisely
\[
 |\Gamma| = 2^{\dim_{\FF_2} \ker \sigma_L} = 2^{k} ,
\]
the ground-state degeneracy.
The second equality is implied by Corollary~\ref{cor:k-formulas}.

We know from Corollary~\ref{cor:cubic-code-degeneracy-formula} that $k$ is bounded by
a linear function of $L$ for the cubic code; $k = O(L)$.
Therefore, the partition function is sandwiched as
\[
 | \log \calZ(\beta) - L^3\log (2^q \cosh^t \beta) | \le k \log 2 = O(L).
\]
Thus, the free energy per unit volume in the thermodynamic limit
is just a smooth function $\log( 2^q \cosh^t \beta )$ for all nonzero temperatures.
This should contrast with three- or higher-dimensional toric code model 
where the free energy density has a singularity at a nonzero temperature~\cite{NussinovOrtiz2008Autocorrelations}.

The analyticity of the free energy does not directly invalidate
a possibility of self-correction of encoded quantum information.
The latter is rather a dynamical process,
which may have little to do with the thermal equilibrium.
In fact, it is quite subtle to analyze self-correcting power
or to give a criterion on it, based on the thermal partition function.
The existence or possibility of a good decoding process
to extract the encoded quantum information
would be a more appropriate way to address the question of self-correction.
For instance, in the four-dimensional toric code model
the thermal expectation value for any bare logical operator is zero.
However, a special dressed logical operator defined with respect to a decoding algorithm
can have a non-vanishing expectation value in the thermodynamic
limit~\cite{ChesiLossBravyiEtAl2010Thermodynamic}.
For the cubic code, the analyticity of the partition function suggests that
it does not allow such a nonzero thermal expectation value for logical operators.
Nevertheless, we can show that the characteristic time scale
of the expectation values of logical operators dressed by a decoding algorithm,
is very large at low temperatures.
This is the topic of the next chapters.

\chapter{Consequences of no-strings rule}
\label{chap:consq-no-strings}

The no-strings rule prohibits a nontrivial charge 
to travel a distance longer than a constant times its size.
We repeat the rule.

\begin{center}
\parbox{.8\textwidth}{
Definition~\ref{defn:no-strings}:
{\em
A string segment is a finitely supported Pauli operator
that creates excitations contained in the union of two finite boxes (anchor regions)
of width $w$.
The string segment is nontrivial 
if the charge contained in one of the boxes is nontrivial.
The distance between the boxes is the length of the string segment.
We say that a model obeys no-strings rule
if the length of any nontrivial string segment of width $w$
is bounded by $\alpha w$ for some constant $\alpha \ge 1$.
}}
\end{center}

\noindent
The cubic code is an extreme case where a point charge cannot move distance $1$
(Theorem~\ref{thm:CubicCode-no-strings}).
It does not mean that isolation or diffusion of charges is impossible;
Theorem~\ref{thm:logarithmic-energy-barrier-for-fractal-operator}
actually gives a local process for the isolation to happen in the translationally invariant case.
The isolation process for a charge has to pay an energy penalty 
upper bounded by the logarithm of the distance from others.

In this chapter, we illustrate the energy landscape of the cubic code Hamiltonian
in a view towards self-correcting quantum memory.
We show that the separation or isolation of charges
requires an energy barrier logarithmically high in the separation distance.
This can be visualized as follows.
Imagine all the energy eigenstates as points on an imaginary ``land,''
and introduce a metric on the land by the minimum number of local operations
one has to apply in order to map one state to another.
Consider a function on the land given by the energy of the energy eigenstate.
Our logarithmic energy barrier means that the energy function, or ``energy landscape,''
has a macroscopic number of local minimums separated by macroscopic energy barriers.
These minimums correspond to low-energy excited states 
in which the separation between defects is approximately the system size.

The energy landscape with the large number of the local minimums 
suggests a possibility of a spin glass phase at a sufficiently low temperature.
Note that there is no quenched disorder in the Hamiltonian~\cite{BinderYoung1986SpinGlass,Weissman1993SpinGlass},
and the glassy feature, if present, would be protected ``topologically''~\cite{BravyiHastingsMichalakis2010stability}.
A spin glass phase can indeed be realized for some classical spin Hamiltonians 
with logarithmic energy barriers such as the model discovered by Newman and Moore~\cite{NewmanMoore1999Glassy}.
Interplay between the topological order and the spin glassiness has been studied recently by
several authors~\cite{Chamon2005Quantum,TsomokosOsborneCastelnovo2010Interplay}.

The logarithmic energy barrier is an optimal bound up to constants
because of Theorem~\ref{thm:logarithmic-energy-barrier-for-fractal-operator}
and \ref{thm:fractal-exists-in-3D}.
An obvious difficulty in proving the lower bound on the energy barrier
is that there are a huge number of paths to isolate a charge.
Relying on scale-invariant nature of the no-strings rule,
we introduce a technique which may be regarded
as a renormalization group in the space of error paths.

Also, we attempt to estimate the length of the separation process:
How many local operations do we need in order to isolate a charge?
Put differently,
how large is a quantum mechanical tunneling probability between two ground states
in presence of local perturbations?
Since the tunneling amplitude would be non-vanishing only for
virtual process that is mediated by an operator acting on the ground space,
the question translates into the theory of error correcting codes
as to find the weight distribution of logical operators.
We do not answer this question at a satisfactory level,
but we find a superlinear lower bound on the code distance assuming the no-strings rule.

\section{Logarithmic energy barrier}
\label{sec:log-barrier}

We consider a regular $D$-dimensional cubic lattice $\Lambda$
with periodic boundary conditions and linear size $L$, that is, $\Lambda=\ZZ_L^D$.
Each site $u\in \Lambda$ is populated by a finite number of qubits.
The class of models we are going to discuss is that of code Hamiltonians
\begin{equation}
H=-\sum_{a=1}^{M} G_a,
\label{eq:stabH}
\end{equation}
where each term $G_a$  is a multi-qubit Pauli operator 
(a tensor product of $I,X,Y,Z$ with an overall $\pm 1$ sign)
and different terms commute with each other;
$G_a G_b =G_b G_a$ and $G_a^2=I$.
We assume that each generator $G_a$ acts nontrivially (by $X,Y$ or $Z$)
only on a set of qubits located at vertices of an elementary cube.
It is allowed to have more than one generator per cube.
Any short-range stabilizer Hamiltonian can be written in this form by performing
a coarse-graining of the lattice.
We continue to assume that $H$ is frustration-free%
\footnote{This is always the case for independent generators $G_a$. 
Since our goal is to obtain a lower bound on the
energy barrier, we can assume that the generators are independent, although it does not
play any role in our analysis.},
The Hamiltonian may or may not be translation-invariant.

Consider any multi-qubit Pauli operator $E$.
A state $\psi = E\,  \psi_0 $
is an excited eigenstate of $H$. Obviously, $G_a \, \psi=\pm \psi$
where the sign depends on whether $G_a$ commutes (plus)
or anticommutes (minus) with $E$.
Any flipped generator ($G_a\, \psi=-\psi$) will be referred to as a {\bf defect}.
Different defects may occupy the same elementary cube.
It should be emphasized that 
a {\em configuration of  defects}, called a {\bf syndrome}, in $\psi$ is the same for all ground states $\psi_0$.
An eigenstate with $m$ defects has energy $2m$ above the ground state.
For brevity, we shall use the term {\bf vacuum} for a  ground state of $H$
whenever its choice  is not important.
A Pauli operator $E$ whose action on the vacuum creates no defects is either a stabilizer ($E\in \mathcal G$),
or a {\bf logical operator} ($E\notin \mathcal G$, but $E$ commutes with $\mathcal G$).
In the former case any ground state of $H$ is invariant under $E$.
In the latter case $E$ maps some ground state of $H$ to an orthogonal ground state.

A Hamiltonian is  said to have topological order if it
has a degenerate ground state and different ground states are locally indistinguishable.
We shall need a slightly stronger version of this condition
that involves properties of both ground and excited states.
These properties depend on a length scale $L_{tqo}$ that must
be bounded as $L_{tqo}\ge  L^\gamma$ for some constant $\gamma>0$.
(For code Hamiltonians, $L_{tqo}$ corresponds to the code distance.)
Most of stabilizer code Hamiltonians with topological order satisfy our conditions
with $L_{tqo} \sim L$.
Our first topological quantum order condition concerns ground states,
which is a rephrasing of Definition~\ref{defn:tqo}:
\begin{defn}
\label{defn:TQO1}
{\bf TQO1} refers to the following condition: 
If a Pauli operator $E$ creates no defects when applied to the vacuum
and its support can be enclosed by a cube of linear size $L_{tqo}$ then $E$
is a stabilizer,  $E\in \mathcal G$.
\end{defn}
Our second TQO condition concerns excited states.
A cluster of defects $S$ will be called {\bf neutral} if it can be created
from the vacuum by a Pauli operator $E$ whose support can be enclosed
by a cube of linear size $L_{tqo}$ without creating any other defects.
Otherwise we say that $S$ is a {\bf charged} cluster.
Given a region $A \subseteq \Lambda$ we shall use a notation
$\mathcal B _r(A)$ for the $r$-neighborhood of $A$, that is, a set of all points that have distance at most $r$ from $A$.
Here and below we use $\ell_{\infty}$-distance on $\ZZ_L^D$.
We shall need the following condition saying that neutral clusters of defects can be
created from the vacuum locally.%
\footnote{If a lattice has a boundary, charged defects
might be created locally on the boundary, as it is the case for the planar version of the toric
code. This is the reason why we restrict ourselves to periodic boundary conditions.}
\begin{defn}
\label{defn:TQO2}
{\bf TQO2} refers to the following condition:
Let $S$ be a neutral cluster of defects 
and $C_{min}(S)$ be the smallest cube that encloses $S$.
Then $S$ can be created from the vacuum 
by a Pauli operator supported on $\mathcal B_1(C_{min}(S))$.
\end{defn}
\begin{rem}
\label{rem:TransInv-strongTQO}
A translationally invariant exact code Hamiltonian satisfies
both of our topological order conditions.
In particular, the cubic code satisfies the present TQO conditions.
[~$\because$
TQO1 is immediate from Lemma~\ref{lem:local-tqo=exact},
and TQO2 follows from an argument similar to the proof of Lemma~\ref{lem:local-tqo=exact}.
Indeed, given a neutral cluster $e$ of defects,
we have an equation $e = \epsilon p$ for some $p$ of the Pauli module.
$p$ is computed by the standard division algorithm applied to the columns of
the excitation map $\epsilon$.
If $e$ is centered at origin, then the degree of $p$ does not exceed
that of $e$, which implies the second TQO condition.%
]
\end{rem}

Let us consider a process of building a cluster of defects $S$ (syndrome) from the vacuum.
It can be described by an {\bf error path} 
--- a finite sequence of local Pauli errors  $E_1,\ldots,E_T$ such that
$E = E_T \cdots E_2 E_1$ creates $S$ from the vacuum.
For simplicity, we assume that each local error  $E_t$ 
is a single-qubit Pauli operator $\sigma^x$, $\sigma^y$, or $\sigma^z$.
Applying this sequence of errors to a ground state $\psi_0$
generates a sequence of states $\{ \psi(t)\}_{t=0,\ldots,T}$,
where and $\psi(T)=E \, \psi_0$ is the excited state of defect configuration $S$.
We will say that $S$ has {\bf energy barrier $\omega$}
if for any Pauli operator $E$ that creates $S$ from the vacuum 
and for any error path implementing $E$,
at least one of the intermediate states has more than $\omega$ defects.
Note that we do not impose any restriction on the length of the path $T$ as long as it is finite.
In particular, one and the same error may be repeated in the error path several times
at different time steps.
Similarly, we consider the energy barrier for a logical operator $\overline{P}$;
we say that a logical operator $\overline{P}$ has {\bf energy barrier $\omega$}
if for any error path implementing $\psi_0 \mapsto \overline P \, \psi_0$
at least one of the intermediate states $\psi(t)$ has more than $\omega$ defects.

For any integer $p\ge 0$, define a {\bf level-$p$ unit of length $\xi(p)$} as%
\footnote{The choice of the constant $10$ is somewhat arbitrary.
We do not try to optimize constants in our proof.}
\begin{equation}
\xi(p)=(10\alpha)^p, \quad p=0,1,\ldots.
\end{equation}
Let $S$ be any non-empty syndrome.
Recall that each defect in $S$ can be associated with some elementary cube of the lattice.
\begin{defn}
\label{defn:psparse}
A syndrome $S$ is said to be {\bf sparse at level $p$} 
if the set of elementary cubes occupied by the defects in $S$
can be partitioned into a disjoint union of clusters
such that each cluster has diameter at most $\xi(p)$
and any pair of distinct clusters combined together has diameter larger
than $\xi(p+1)$. Otherwise, $S$ is {\bf non-sparse at level $p$}.
\end{defn}
\noindent
For example, suppose all defects in $S$ occupy the same elementary cube.
Since an elementary cube has diameter $1$, such a syndrome $S(t)$ is sparse at any level $p \ge 0$.
If $S$ occupies a pair of adjacent cubes, $S(t)$ is sparse at any level $p\ge 1$,
but is non-sparse at level $p=0$.
Note that the partition of $S$ into clusters required for level-$p$ sparsity
is unique whenever it exists.
The non-sparsity provides a lower bound on the number of defects in a cluster
as follows.
\begin{lem}
\label{lemma:counting}
A non-empty syndrome $S$ that is non-sparse at all levels $q=0,\ldots,p$
contains at least $p+2$ defects.
\end{lem}
\begin{proof}
Let $C^{(0)}_1,\ldots, C^{(0)}_g$ be elementary cubes occupied by $S$.
Obviously, $S$ contains at least $g$ defects.
Since $S$ is non-empty and non-sparse at level $0$,
we have $g \ge 2$ and there exists
a pair of cubes $C_a^{(0)}, C_b^{(0)}$ such that the union
$C_a^{(0)} \cup C_b^{(0)}$ has diameter at most $\xi(1)$.
Combining the pair $C_a^{(0)}, C_b^{(0)}$  into a single cluster,
we obtain a partition $S = C_1^{(1)} \cup \cdots \cup C_{g-1}^{(1)}$
where each cluster $C_a^{(1)}$ has diameter at most $\xi(1)$.
Since $S$ is non-sparse at level $1$, we have $g-1 \ge 2$,
and there exists a pair of clusters $C_a^{(1)}, C_b^{(1)}$ such that the union
$C_a^{(1)}\cup C_b^{(1)}$ has diameter at most $\xi(2)$.
Combining the pair $C_a^{(1)}, C_b^{(1)}$ into a single cluster and
proceeding in the same way we arrive at $g \ge p+2$.
\end{proof}

A configuration of defects created by applying a Pauli operator $E$ to the vacuum
will be called a {\em syndrome caused by $E$}.
The process of building up a logical operator $\overline{P}$ 
by a sequence of local errors $E_1,\ldots,E_T$ can be described by
a {\bf syndrome history} $\{S(t)\}_{t=0,\ldots,T}$.
Here $S(t)$ is the syndrome caused by the product $E_t \cdots E_2 E_1$,
a partial implementation of $\overline{P}$ up to a step $t$.
The following concerns the size of local errors
\begin{lem}
\label{lem:large-support-for-nontrivial-transition}
Let $Q_j$ be Pauli operators causing  a chain of transitions
\[
\mathrm{vac} \xrightarrow{Q_1} S_1 \xrightarrow{Q_2}  S_2 \xrightarrow{Q_3} \ldots \xrightarrow{Q_r} S_r
\xrightarrow{Q_{r+1}} \mathrm{vac}.
\]
Let $P_j$ be some Pauli operator creating the syndrome $S_j$ from the vacuum.
Suppose the support of any operator $P_j$ and any operator $Q_j$
can be enclosed by $n$ or less cubes of linear size $R$ such that $4nR < \ltqo$.
Then, the product $\overline{Q} = Q_1\cdots Q_r Q_{r+1}$ is a stabilizer.
\end{lem}
\begin{proof}
Let $\psi_0$ be any ground state.
Define a sequence of states
\begin{align*}
\psi(1) &= P_1 Q_1\cdot \psi_0, \\
\psi(j+1) &=  (P_j P_{j+1}) Q_{j+1} \cdot \psi(j) \quad \text{ for $j=1,\ldots,r-1$},\\
\psi(r+1) &= Q_{r+1} P_{r} \cdot \psi(r).
\end{align*}
Obviously,
\begin{align*}
\psi(j) &= \pm P_j \cdot (Q_1 \cdots Q_j) \cdot  \psi_0 \quad \text{for $j=1,\ldots,r$} \\
\psi(r+1) &= \pm \overline{Q}\cdot \psi_0.
\end{align*}
It follows that all states $\psi(j)$ are ground states,
and the transition from $\psi(j)$ to $\psi(j+1)$ 
can be caused by a Pauli operator
\[
O_j = P_j P_{j+1} Q_{j+1}.
\]
Let $M_j$ be the support of $O_j$. 
By assumption, $M_j$ can be enclosed by at most $3n$ cubes of linear size $R$.
If $M_j$ is a connected set,
i.e., one can connect any pair of qubits from $M_j$ by a path $(u_1,\ldots,u_l)$
such that the distance between $u_a$ and $u_{a+1}$ is $1$,
then $M_j$ can be enclosed by a single cube of linear size at most $3nR$.
That $3nR < \ltqo$ implies $O_j$ is a stabilizer by the topological order condition.
Generally, $M_j$ consists of several disconnected components $M_j^\alpha$,
such that the distance between any pair of distinct components is at least $2$.
The restriction of $O_j$ on a connected component 
commutes with any stabilizer generator,
and is supported in a box of linear size $3nR$,
which is smaller than $\ltqo$ by assumption.
Therefore, the restrictions are stabilizers,
and their product $O_j$ is also a stabilizer.
In other words, $\psi(j+1)=\pm \psi(j)$ for all $j$.
It means that $\overline{Q}\, \psi_0 = \pm \psi_0$ for any ground state $\psi_0$.
We conclude that $\overline{Q}$ is a stabilizer.
\end{proof}

The no-strings rule says that an isolated charged defect belonging to some sparse syndrome
cannot be moved further than distance $\alpha$ away by a sequence of local errors.
Since the no-strings rule is scale invariant,
it may be applied to a coarse-grained lattice to show that isolated
charged clusters cannot be moved further than distance $\alpha \xi(p)$ away.
In order to exploit the scale invariance,
we define a {\bf level-$p$ syndrome history}
as a subsequence of the original syndrome history $\{S(t)\}_{t=0,\ldots,T}$
that includes only those syndromes $S(t)$ that are non-sparse
at all levels $q=0,\ldots,p-1$.
The level-$0$ syndrome history includes all syndromes $S(t)$,
the level-$1$ syndrome history omits $S(t)$ that is sparse at level 0,
and so on.
When $S(t')$ and $S(t'')$ are a consecutive pair of level-$p$ syndromes,
we define a {\bf level-$p$ error $E$} connecting $S(t')$ and $S(t'')$
as the product of all single-qubit errors $E_j$ that occurred between $S(t')$ and $S(t'')$.
Level-$p$ errors are represented by horizontal arrows on Fig.~\ref{fig:RG}.
Note that we do not have any bound on the number
of single-qubit errors $E_j$ in the interval between $S(t')$ and $S(t'')$.
In the worst case, $E$ could act nontrivially on every qubit in the system.
A main technical lemma of this chapter below asserts, loosely speaking, that
if a syndrome history does not have a deep enough non-sparse hierarchy,
any error path is equivalent to one that is localized around the defects.

\begin{lem}
\label{lemma:RG2}
Let $S'= S(t')$ and $S''= S(t'')$ be a consecutive pair of syndromes in the level-$p$ syndrome history
of a Hamiltonian obeying the no-strings rule.
Let $E$ be the product of all errors $E_j$ that occurred between $S'$ and $S''$.
Suppose that any $S(t)$ contains at most $m$ defects.
If 
\[
16 m \xi(p) < L_{tqo},
\]
then there exists an error $\tilde{E}$  supported on $\mathcal B _{\xi(p)}(S'\cup S'')$
such that $E\tilde{E}$ is a stabilizer.
\end{lem}

\begin{figure}[tb]
\centerline{\includegraphics[height=4cm]{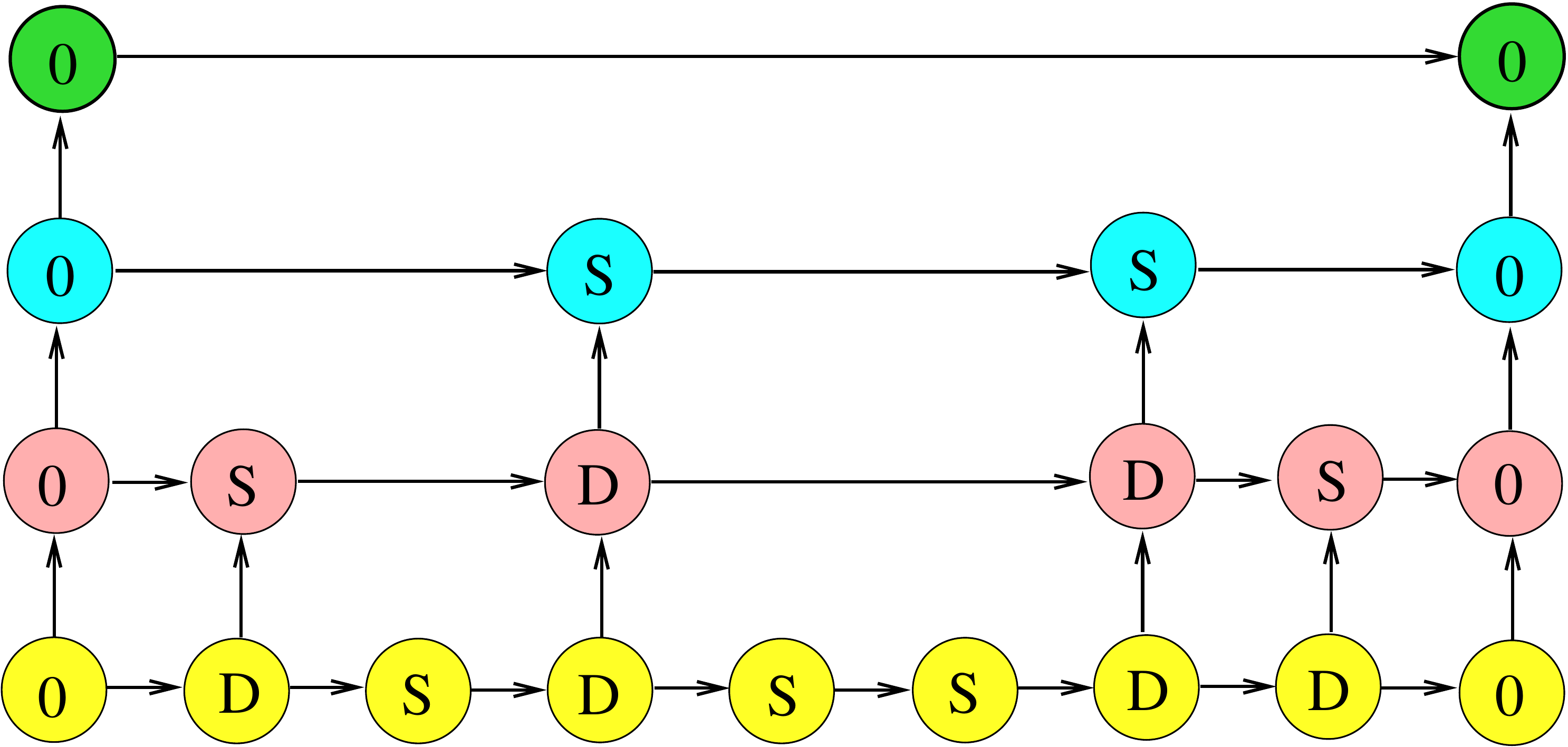}}
\caption{Renormalization group technique 
used to prove a logarithmic  lower bound on the energy barrier for logical operators.
Horizontal axis represents time. 
Vertical axis represents RG level $p=0,1,\ldots,p_{max}$.
A sequence of level-$0$ errors (single-qubit Pauli operators) 
implementing a logical operator $\overline{P}$
defines a level-$0$ syndrome history (yellow circles)
that consists of sparse (S) and non-sparse (D) syndromes.
The history begins and ends with the vacuum ($0$).
For any level $p\ge 1$ we define a level-$p$ syndrome history
by retaining only non-sparse syndromes at the lower level.
A syndrome is called non-sparse at level $p$ if it
cannot be partitioned into clusters of size $\le (10\alpha)^p$
separated by distance $\ge (10\alpha)^{p+1}$,
where $\alpha$ is a constant coefficient from the no-strings rule.
Each  level-$p$ error (horizontal arrows) connecting syndromes $S',S''$
is equivalent to the product of all level-$(p-1)$ errors between $S',S''$
modulo a stabilizer.
We prove that these stabilizers can be chosen such that
level-$p$ errors act on $2^{O(p)}$ qubits.
Since at the highest  level $p=p_{max}$
a single level-$p$ error is a logical operator,
one must have $p_{max}=\Omega(\log{L})$.
We prove that level-$p$ non-sparse syndromes contain $\Omega(p)$ defects
which implies that at least one syndrome at level $p=p_{max}-2$
consists of $\Omega(\log{L})$ defects.}
\label{fig:RG}
\end{figure}

\begin{proof}
The proof is by induction in $p$.
When $p=0$, $E=E_j$ is a single-qubit error.
If the qubit acted on by $E$ does
not belong to $\mathcal B_{1}(S'\cup S'')$, one must have $S'=S''$.
It means that $E$ is a single-qubit error with a trivial syndrome.
The topological order condition implies that $E$ is a stabilizer.
Choosing $\tilde{E}=I$ proves the lemma for $p=0$.

Suppose the assertion is true for some level $p \ge 0$.
Let $S'=S(t')$ and $S''=S(t'')$ be consecutive syndromes in the level-$(p+1)$
history. Consider first the trivial case when
$S'=S(t')$ and $S''=S(t'')$ are also consecutive syndromes in the level-$p$ history.
Then $S'$ and $S''$ are connected by a single%
\footnote{%
The word `single' does not necessarily mean a single blob of errors supported in a small ball; 
it just means a single arrow in the level-$p$ history.
For instance, suppose that a model permits neutral point defects.
Two neutral defects separated by distance 1 is a non-sparse at level 0 and 1.
If we annihilate one of them, put sparsely many neutral defects,
and finally put a neutral defect adjacent to one of the defects,
then the whole process is described by a consecutive pair in the level 1 history.
The `single' level-1 error in this case consists of many components.%
}
level-$p$ error $E$ which, by induction hypothesis, has support on $\mathcal B _{\xi(p)}(S'\cup S'')$
modulo stabilizers.
The latter is contained in $\mathcal B_{\xi(p+1)}(S'\cup S'')$ which proves the induction step.

The nontrivial case is when there is at least one level-$p$ syndrome between $S'$ and $S''$.
The interval of the level-$p$ syndrome history between $S'$ and $S''$ can be represented
(after properly redefining the time variable $t$) as
\[
S'\xrightarrow{E_\text{lead}}
S(1) \xrightarrow{E_1}
S(2) \xrightarrow{E_2}
\cdots
 \xrightarrow{E_{\tau-1}}
S(\tau) \xrightarrow{E_\text{tail}}
S''.
\]
Here all syndromes $S(1),\ldots,S(\tau)$ are sparse at the level $p$ and
all transitions are caused by level-$p$ errors.
The sparsity implies that
the set of elementary cubes occupied by  $S(t)$ has a unique
partition into a disjoint union of clusters $C_a(t)$ such that
each cluster has diameter at most $\xi(p)$ and the distance
between any pair, if any, of clusters is at least
\begin{align*}
\mathrm{dist}(C_a(t),C_b(t)) & \ge   \xi(p+1)-2\xi(p) \ge (10\alpha -2)\xi(p)  \ge   8\alpha \xi(p).
\end{align*}
Represent any intermediate syndrome as a disjoint union
\begin{equation}
\label{S1+}
S(t)=S^c(t) \cup S^n(t), \quad t=1,\ldots,\tau,
\end{equation}
where $S^c(t)$ and $S^n(t)$ include all charged and all neutral clusters $C_a(t)$, respectively.
Let $g$ be  the number of clusters in $S^c(t)$. We claim  that $g$  does not depend on $t$.
Indeed, since a level-$p$ error $E_t$ acts on $\xi(p)$-neighborhood of $S(t) \cup S(t+1)$
by the induction hypothesis,
the sparsity condition implies that $E_t$
cannot create or annihilate a charged cluster $C_a(t)$ from the vacuum, 
or map a charged cluster  to a neutral cluster and vice versa.
The same argument shows that each  cluster $C_a(t)\subseteq  S^c(t)$ can `move' 
at most by $\xi(p)$ per time step,
that is, we can parameterize
\[
S^c(t)=C_1(t) \cup \ldots \cup C_g(t)
\]
such that
a `world-line' of the $a$-th charged cluster
obeys the continuity condition
\begin{equation}
\label{cont+}
\mathrm{dist}(C_a(t+1),C_a(t)) \le \xi(p).
\end{equation}
We can now use the no-strings rule to show that
all charged clusters are `locked' near their
initial positions, so that their world-lines
are essentially parallel to the time axis.  More precisely, we claim that
\begin{equation}
\label{locking+}
\mathrm{dist}(C_a(t), C_a(1)) \le \alpha  \xi(p)  \quad \text{for all $1\le t \le \tau$}.
\end{equation}
Indeed, suppose Eq.~\eqref{locking+} is false for some $a$.
Using the continuity Eq.~\eqref{cont+},
one can find a time step $t_1$
such that $\mathrm{dist}(C_a(t_1), C_a(1))>\alpha \xi(p)$ and
$\mathrm{dist}(C_a(t), C_a(1)) \le \alpha \xi(p)$ for all $1\le t<t_1$.
Let $E_{close}$ be the product of all level-$p$ errors $E_j$
that occurred between $S(1)$ and $S(t_1)$ within distance $(2+\alpha)\xi(p)$
from $C_a(1)$. Since all intermediate syndromes are sparse at level $p$,
the net effect of $E_{close}$ is to annihilate the  charged cluster  $C_a(1)$
and create the charged cluster $C_a(t_1)$.
Equivalently, applying $E_{close}$ to the vacuum creates a pair of charged clusters
$C_a(1)$ and $C_a(t_1)$.
However, this contradicts to the no-strings rule since
$C_a(1)$ and $C_a(t_1)$ have linear size at most $\xi(p)$ while
the distance between them is greater than $\alpha \xi(p)$.
Thus we have proved Eq.~\eqref{locking+}.

We say that $\vec{x}\in \Lambda$ is {\em close to $S'$} if $\vec{x}\in \mathcal B_{\xi(p+1)}(S')$,
and $\vec{x}\in \Lambda$ is {\em close to $S''$} if $\vec{x}\in \mathcal B_{\xi(p+1)}(S'')$.
Let $E_t$ be the level-$p$ error causing the transition from  $S(t)$ to $S(t+1)$,
where $t=1,\ldots,\tau-1$.
By induction hypothesis, we may assume that the support of $E_t$ is in $\mathcal B_{\xi(p)}(S(t) \cup S(t+1))$.
Let $E^c_t$ be the restriction of $E_t$ onto $\mathcal B_{\xi(p)}(S^c(t)\cup S^c(t+1))$,
and $E^n_t$ be the restriction of $E_t$ onto $\mathcal B_{\xi(p)}(S^n(t)\cup S^n(t+1))$.
The sparsity of $S(t)$ implies that
\begin{equation}
\label{EcEn+}
E_t = E^c_t \cdot E^n_t.
\end{equation}
We claim that {\em any error $E^c_t$ is close to $S'$}.
Indeed, each cluster in $S^c(1)$ is within distance $2\xi(p)$ from $S'$ 
since, otherwise, a single level-$p$ error $E_\text{lead}$,
supported in $\mathcal B_{\xi(p)}(S' \cup S(1) )$ by induction hypothesis (modulo stabilizers),
would be able to create a charged cluster from the vacuum.
Using Eq.~\eqref{locking+}, we infer that 
$C_a(t) \subseteq \mathcal B_{(2+\alpha)\xi(p)}(S')$ for all $a=1,\ldots,g$.
Therefore, $E^c_t$ is close to $S'$.

We wish to find a ``localized'' leading error $\tilde{E}_\text{lead}$
that maps the syndrome $S'$ to $S^c(1)$ such that the support of
$\tilde{E}_\text{lead}$ is close to $S'$.
For each neutral cluster $C \in S^n(1)$ of diameter at most $\xi(p)$,
let $O'(C)$ be a Pauli operator creating $C$ from the vacuum.
Because of our TQO2, we can choose $O'(C)$ to be supported in $\mathcal B_1(C)$.
Set
\[
\tilde{E}_\text{lead} = E_\text{lead} \prod_{ C \in S^n(1) } O'(C),
\]
We have seen that $S^c(1)$ is within $(2+\alpha)\xi(p)$-neighborhood of $S'$.
If $E_\text{lead}$ has a disconnected component $E(C)$,
isolated by distance $\xi(p)$ and centered at a neutral cluster $C$ of $S^n(1)$,
then $E(C) O'(C)$ is a stabilizer.
Hence, $\tilde E_\text{lead}$ is close to $S'$ modulo stabilizers.
If $E_\text{lead}$ does not have such a component, $\tilde E_\text{lead}$ is already close to $S'$.
Certainly, $\tilde E_\text{lead}$ maps $S'$ to $S^c(1)$.
We can apply the same rules to the error $E_\text{tail}$ and the syndrome $S^n(\tau)$.
We find a localized error
\[
\tilde{E}_\text{tail} = E_\text{tail} \cdot \prod_{C \in S^n(\tau) } O''(C)
\]
modulo stabilizers
such that $\tilde{E}_\text{tail}$ maps $S^c(\tau)$ to $S''$ 
and the support of $\tilde{E}_\text{tail}$ is close to $S''$.
The operator $O''(C)$ above creates a neutral cluster $C \in S^n(\tau)$ from the vacuum.

We can now define a localized level-$(p+1)$ error $\tilde{E}$
whose support is close to $S'\cup S''$ as
\begin{equation}
\label{tildeE_cluster}
\tilde{E}=\tilde{E}_\text{tail} \cdot E^c_{\tau-1} \cdots E^c_{1} \cdot \tilde{E}_\text{lead}.
\end{equation}
By construction, it describes an error path
\[
S'\xrightarrow{\tilde{E}_\text{lead}}
S^c(1) \xrightarrow{E^c_1}
S^c(2) \xrightarrow{E^c_2}
\cdots
 \xrightarrow{E^c_{\tau-1}}
S^c(\tau) \xrightarrow{\tilde{E}_\text{tail}}
S''.
\]
It remains to check that $E \cdot \tilde{E}$ is a stabilizer.
Combining Eq.~\eqref{EcEn+} and
Eq.~\eqref{tildeE_cluster} we
conclude that
\[
E\cdot \tilde{E} = \pm (\tilde{E}_\text{lead} E_\text{lead}) \cdot E^n_1 \cdots E^n_{\tau-1} \cdot (\tilde{E}_\text{tail} E_\text{tail}).
\]
Applying  $E\cdot \tilde{E}$ to the vacuum generates the following chain of transitions:
\begin{align}
\label{chain1_cluster}
\mathrm{vac} \xrightarrow{\tilde{E}_\text{lead} E_\text{lead}}
S^n(1) \xrightarrow{E^n_1}
S^n(2) \xrightarrow{E^n_2} \cdots 
 \xrightarrow{E^n_{\tau-1}}
S^n(\tau) \xrightarrow{\tilde{E}_\text{tail} E_\text{tail}}
\mathrm{vac}
\end{align}
Each syndrome $S^n(t)$ consists of at most $m$ neutral clusters of diameter $\xi(p)$, 
i.e., it can be created from the vacuum by an error
whose support can be enclosed by at most $m$ cubes of linear size $2+\xi(p)$,
due to our TQO2.
The first transition $\tilde{E}_\text{lead} E_\text{lead}$
or the last transition $\tilde{E}_\text{tail} E_\text{tail}$
is caused by errors whose support can be enclosed by at most $m$ cubes of linear size $2+\xi(p)$.
All intermediate ones are supported on $\mathcal B_{\xi(p)}(S^n(t)\cup S^n(t+1))$.
Hence their support can be enclosed by at most $2m$ cubes of linear size $2+\xi(p)$.
Now the statement that $E\cdot \tilde{E}$ is a stabilizer follows from Lemma~\ref{lem:large-support-for-nontrivial-transition}.
\end{proof}

Now we state two theorems that apply to any stabilizer Hamiltonian Eq.~\eqref{eq:stabH}
on a $D$-dimensional lattice that obeys the topological order condition (TQO1,2)
and the no-strings rule.
The numerical constants for the Hamiltonian are
$\alpha \ge 1$ in the no-strings rule
and $1 \ge \gamma>0$ in the bound $\ltqo \ge L^\gamma$
where $L$ is the linear system size.

\begin{theorem}
\label{thm:lop-log-energy-barrier}
The energy barrier for any logical operator is at least $c \log{L}$ for some constant $c=c(\alpha,\gamma)$.
\end{theorem}
\begin{proof}
Consider a syndrome history of an implementation of a logical operator $E$.
We keep the initial and the final syndromes (the empty syndromes) at all levels;
the syndrome history starts and ends with the vacuum at any level $p$.
It suffices to treat the case where
all intermediate syndromes $S(t)$ are non-empty.
Let $p_{max}$ be the highest RG level, that is, the smallest integer $p \ge 0$
such that a single level-$p$ error $E$ maps the vacuum to itself, see Fig.~\ref{fig:RG}.
Let $m$ be the maximum number of defects in the syndrome history at any given moment.

Suppose that $16m \xi(p_{max}) < \ltqo$.
Then Lemma~\ref{lemma:RG2} applied to the level-$p_{max}$ syndrome history
with $S'=S''=\emptyset$ (vacuum),
would imply $\tilde{E}=I$.
Since $E$ is not a stabilizer by assumption, we must have $16 m \xi(p_{max}) \ge \ltqo$.
Since the syndrome history must contain a syndrome $S(t)$ non-sparse at all levels $0,\ldots,p_{max}-2$,
Lemma~\ref{lemma:counting} implies $m \ge p_{max}$.
And, TQO1 requires $\ltqo \ge L^\gamma$.
Therefore, $ 16(10\alpha)^{2m} \ge 16m(10\alpha)^m \ge 16 m (10\alpha)^{p_{max}} \ge L^\gamma $,
and $m = \Omega(L)$.
\end{proof}

\begin{theorem}
\label{thm:cluster-log-energy-barrier}
Let $S$ be a neutral cluster of defects
containing a charged cluster $S' \subseteq S$ of diameter $r$
such that there are no other defects within distance $R$ from $S'$.
If $r+R < L_{tqo}$,
then the energy barrier for creating $S$ from the vacuum is at least $c \log R$
for some constant $c=c(\alpha)$.
\end{theorem}
\begin{proof}
Let $S$ be a neutral cluster of defects and $E$ be a Pauli operator creating $S$ from the vacuum,
with $S' \subset S$ of diameter $r$ being charged.
Consider a hierarchy of syndrome histories similar to the one shown on Fig.~\ref{fig:RG},
where we now maintain the initial syndrome $\emptyset$ and the final syndrome $S$ for all levels.
Let $p_{max}$ be the highest RG level.
Then a single level-$p_{max}$ error $E$ creates $S$ from the vacuum.
Since there must be a syndrome that is non-sparse for all levels $0,1,\ldots,p_{max}-2$,
Lemma~\ref{lemma:counting} implies $m \ge p_{max}$
where $m$ is the maximum number of defects in the syndrome history.

Suppose $16 m \xi(p_{max}) < \ltqo$ 
Lemma~\ref{lemma:RG2} implies that 
$E$ is the equivalent modulo stabilizers to $\tilde{E}$ 
supported on $\mathcal B _{\xi(p_{max})}(S)$.
If $\xi(p_{max}) < R/4$,
then $\tilde{E}$ must act on two separated regions, one near $S'$ and another far from $S'$.
This means that $S'$ alone can be created by a Pauli operator whose support is enclosed by a cube of linear size $r+R/2$.
Since $r+R < L_{tqo}$, it is contradictory to the assumption that $S'$ is charged.
Therefore, $\xi(p_{max}) \ge R/4$, and $m = \Omega(R)$.
If $16 m \xi(p_{max}) \ge \ltqo$, then, since $\ltqo > R$, we also have $m = \Omega(R)$.
\end{proof}

\section{Superlinear code distance}
\label{sec:superlinear-distance}

We have shown that in a process of isolating a charged cluster, there is a logarithmic energy barrier.
The following theorem quantifies how long the process must be.
The proof again makes use of renormalization group,
and shows that there is a subset of `fractal dimension' $\gamma > 1$ in the support of $E$,
the operator that makes two separated clusters from the vacuum of which one is charged.
We assume that $E$ has weight minimum possible.

Let $w$ be an odd positive number.
We say a set of sites $C \subset \Lambda$ is a {\bf level-$p$ chunk} if $\mathrm{diam}(C) < w^p$.
A \emph{path} in the lattice is a finite sequence of sites $(u_1,u_2,\ldots,u_n)$ such that $d(u_i,u_{i+1})=1$.
(Recall that we use the $l_\infty$ metric $d$.)
Using paths, we can say whether a set is connected.

\begin{defn}
A connected level-$p$ chunk $C \subseteq S$ is \emph{maximal} with respect to a set of sites $S$
if there exist a connected subset $C^\circ \subseteq C$
and a path $\zeta = (u_1,\ldots,m,\ldots,u_n) \subseteq C^\circ$ satisfying
\begin{enumerate}
 \item[(i)]   $d(u_1,u_n)=w^p-w^{p-1}$,
 \item[(ii)]  $d(u_1,m),d(u_n,m)\ge \frac{w^p-w^{p-1}}{2}$,
 \item[(iii)]  $C^\circ$ contains the connected component of $m$ in $\mathcal{B}_{\frac{w^{p}-w^{p-1}}{2}}(m) \cap S$, and
 \item[(iv)] $C$ contains the connected component of $C^\circ$ in $\mathcal{B}_{\frac{w^{p-1}}{2}}(C^\circ)\cap S$.
\end{enumerate}
\end{defn}

The last two conditions restricts the position of $\zeta$ in $C$ such that $\zeta$ lies sufficiently far from the boundary of $C$.
The site $m$ will be referred to as a {\bf midpoint} of $C$.
Let $S$ be the support of the Pauli operator $E$,
any restriction of which obeys the no-strings rule.

\begin{lem}
Given a path $\zeta$ in $S$ joining $u_1$ and $u_n$ such that $d(u_1,u_n)=lw^p-1$,
there are $l$ disjoint maximal chunks of level $p$ whose midpoints are on $\zeta$.
\label{lem:max_chunks_on_path}
\end{lem}
\begin{proof}
For convenience, we assume that the $z$-coordinates of $u_1$ and $u_n$ are
$0$ and $lw^p-1$, respectively.
Consider $l+1$ planes $P_i$ perpendicular to the $z$-axis, whose $z$-coordinates are $iw^p$ for $i=0,1,\ldots,l$.
In each region between the two consecutive planes $P_{i-1}$ and $P_i$,
there is a subpath $\zeta_i=(u_{j_{i-1}},\ldots,u_{j_{i}})$ such that $d(u_{j_{i-1}},u_{j_{i}})=w^p-1$.
Choose $m_i \in \zeta_i$ such that $d(u_{j_{i-1}},m_i),d(m_i,u_{j_{i}}) \ge \frac{w^p -1}{2}$.
Let $C_i$ be the connected component of $m_i$ within $S \cap \mathcal B_{\frac{w^p}{2}}(m_i)$.
Add, if necessary, some points of $\zeta_i$ to $C_i$ to get a maximally connected $C'_i$.
This $C'_i$ is a maximal chunk of sites with midpoint being $m_i$.
Any two $C'_i$'s are disjoint since each of them lies in a unique region enclosed by $P_{i-1}$ and $P_i$.
\end{proof}

\begin{lem}
For sufficiently large $w$,
a maximal level $(p+1)$ chunk $C$ with respect to $S$
admits a decomposition into $w+1$ or more maximal chunks of level $p$ with respect to $S$.
\label{lem:many_max_in_max}
\end{lem}
\begin{proof}
Recall that $S$ is the support of the Pauli operator $E$, any restriction of which obeys no-strings rule.
Define the \emph{boundary} of a subset $U$ of $S$ to be $\partial U = \mathcal B_1(U) \cap U^c \cap S$.
Then, any subset $U$ of sites with boundary enclosed in a two disjoint regions can be regarded as a string segment.

By the definition of the maximal chunk,
there exists a path $(u_1,\ldots,m,\ldots,u_n)$ in $C^\circ \subseteq C$ such that $d(u_1,u_n)=w^{p+1}-w^p$.
We assume that the $z$-coordinates of $u_1, u_n$ differ by $w^{p+1}-w^p$.
We will show that there are sufficiently long and separated paths in $C^\circ$,
to which we apply Lemma~\ref{lem:max_chunks_on_path} to find $w+1$ maximal chunks of level $p$.
They will lie in $\mathcal B_\frac{w^p}{2}(C^\circ)$, and hence in $C$.

Let $M$ ($N$) be the subset of $S$ consisted of sites
whose $z$-coordinates differ from that of $u_1$ ($u_n$) by at most $\eta w^p$.
First, suppose $\partial C^\circ$ is not contained in $M \cup N$.
Since $u_1 \in M$ and $u_n \in N$, there is a site $s \in C^\circ$
adjacent (of distance 1) to $\partial C^\circ$ such that $d(s,u_1),d(s,u_n) > \eta w^p$.
Furthermore, $d(s,m) \ge \frac{w^{p+1}-w^p}{2}-1$;
otherwise, $C^\circ$ contains a site in the boundary, which is a contradiction.

Consider the shortest network $\mathcal{N}$ of paths in $C^\circ$ connecting four sites $u_1,m,u_n,s$.
(The length of a network of paths is the number of sites in the union of the paths.)
Let $\zeta$ be the shortest path in $\mathcal{N}$ from $u_1$ to $u_n$.
If $s$ is not contained in $\mathcal B_{3w^p}(\zeta)$,
then $\zeta' \subseteq \mathcal{N}$ joining $s$ to a site on $\zeta$
has a subpath $\zeta'' \subseteq \zeta'$ of diameter at least $2w^p$
such that $\zeta''$ is separated from $\zeta$ by $w^p$.
Applying Lemma~\ref{lem:max_chunks_on_path} to $\zeta$ and $\zeta''$,
we find at least $w+1$ maximal chunk of level $p$.
If $m$ is not contained in $\mathcal B_{3w^p}(\zeta)$,
a similar argument reveals at least $w+1$ maximal chunk of level $p$.

Suppose both $s$ and $m$ are contained in $\mathcal B_{3w^p}(\zeta)$.
Observe that $\zeta \setminus ( \mathcal B_{4w^p}(s) \cup \mathcal B_{4w^p}(m) )$
consists of three connected components $\zeta_1,\zeta_2,\zeta_3$,
two of which have diameter $\ge \frac{w^{p+1}-w^p}{2}-8w^p$ and the other has diameter $\ge (\eta - 4)w^p$.
Two distinct $\mathcal B_{\frac{w^p}{2}}(\zeta_i)$ and $\mathcal B_{\frac{w^p}{2}}(\zeta_j)$ ($i,j=1,2,3$)
do not overlap because of the minimality of $\zeta$. Applying Lemma~\ref{lem:max_chunks_on_path},
we find $w + \eta - 21$ maximal chunks of level $p$.
Choosing $\eta > 21$, we get the desired result.

Next, suppose $\partial C^\circ$ is contained in $M \cup N$.
Let $s_M, s_N \in C^\circ \setminus (M\cup N)$ be sites adjacent to $M$ and $N$, respectively.
The separation between $M$ and $N$ is $(w-1-2\eta)w^p$.
If it is greater than $\eta' \alpha w^p$, there must be a $z$-plane $P$ that contains $s_M$ or $s_N$
such that $P \cap C^\circ$ has diameter $> \eta' w^p$; Otherwise, the no-strings rule is violated.
Let $v_1, v_2 \in P \cap C^\circ$ be sites separated by $\eta' w^p$.
The four sites, $u_1,u_n,v_1,v_2$ are sufficiently separated and connected by some paths in $C^\circ$.
Arguing as before, we find $(w -1) + \eta'-16$ maximal chunks of level $p$.
The choice of $\eta' > 17$ and $w > 1+2\eta+\eta'\alpha$ proves the lemma.
\end{proof}

\begin{theorem}
Let $E$ be a Pauli operator creating $S$, a neutral cluster of defects
containing a charged cluster $S' \subseteq S$ of diameter $r$
such that there are no other defects within distance $R$ from $S'$.
If $r+2R < L_{tqo}$, then the weight of $E$ must be $\ge c R^\gamma$ for some constant $\gamma > 1$ and $c$.
\label{thm:separation_charged_defects}
\end{theorem}
\begin{proof}
The support of the minimal Pauli operator $E$ in Theorem~\ref{thm:separation_charged_defects}
must admit a path connecting $S'$ and $S \setminus S'$.
Otherwise, $S'$ can be regarded as being created locally,
and our topological order condition demands the cluster be neutral.
Since the path has length $\ge R$,
Lemma~\ref{lem:max_chunks_on_path} says we have a maximal chunk of level $p$
where $p$ is such that $w^p \le R < w^{p+1}$.
Lemma~\ref{lem:many_max_in_max} implies any maximal chunk of level $p$ must contain
at least $(w+1)^p$ sites.
This concludes the proof with $\gamma=\frac{\log(w+1)}{\log w}>1$.
\end{proof}
A similar argument proves the lower bound $d=\Omega(L^\gamma)$ on the code distance $d$ of the cubic code
since the minimal logical operator must contain a path of length $L$.

\chapter{Renormalization group decoder and error threshold theorem}
\label{chap:RG-decoder}

Any error correcting scheme would be comprised of a chosen code space,
an encoding procedure, and a decoding procedure.
We have studied a way to choose a code space via additive/stabilizer code formalism.
Our focus has been the situation where the code space is realized as a ground space of a local Hamiltonian.
The encoding is a process in which one prepares a state
that is to be transferred or stored,
and then one embeds the state into the designed code space.
For a concatenated code the encoding would be hierarchical
resembling the very way the code is constructed.
Interestingly, a ground state of toric code model
can also be prepared in a similarly hierarchical way~\cite{AguadoVidal2007Entanglement}.
In case of the cubic code, for example, the encoding can be done
by inverting the real-space renormalization group procedure
presented in Section~\ref{sec:real-space-RG-CubicCode}.

The decoder of a quantum code
restores a damaged state into the code space.
In contrast to its name, the decoder
should not reveal any information that is encoded.
Rather, it prepares the state appropriate for next information processing step
which assumes that the state is in the code space;
it detects errors and suggests an operator that would undo the errors.
The performance of the decoder is measured
by how closely the damaged state is restored to the original encoded state.
In this chapter, we explain a decoding algorithm,
called \emph{renormalization group decoder},
that is applicable for a family of topological codes
including the cubic code.
A very similar idea appears in Harrington's thesis~\cite{Harrington2004thesis}.
A decoder for 2D toric code with a similar name was proposed by 
Duclos-Cianci and Poulin~\cite{Duclos-CianciPoulin2009Fast}.
The two decoders are conceptually similar, and the running times are the same up to a multiplicative constant.
Our decoder is however advantageous 
for its simplicity and applicability.
In particular, our decoder is the only decoder so far
that has a positive error threshold under stochastic error
when used with the cubic code.
In fact, our decoder provides a universal positive error threshold 
for all topological codes in a given number of spatial dimensions,
as we prove in Section~\ref{sec:threshold}.

Formally, if we restrict ourselves to local additive codes,
the decoder is an association of a Pauli operator $P$ to any possible syndrome $S$
such that the Pauli operator $P$ transforms $S$ to the empty syndrome.
Recall that the syndrome measurement reveals locations of defects (flipped stabilizer generators)
created by an unknown error. The renormalization group (RG) decoder attempts to annihilate the defects comprising
the syndrome $S$ by dividing them into disjoint connected clusters $S=C_1\cup \ldots \cup C_m$ 
and then trying to annihilate each cluster $C_a$  individually.
More specifically, the decoder checks whether $C_a$ can be annihilated by a Pauli operator  $P_a$
supported  on a sufficiently small spatial region  $b(C_a)$ enclosing $C_a$.
If such a local annihilation operator $P_a$ exists, the decoder  updates the syndrome by erasing
all the defects comprising $C_a$, records the operator $P_a$, and moves on to the next cluster.
If $C_a$ cannot be annihilated, the decoder skips it.
The annihilation operator $P_a$ is not unique.
However, if the enclosing region $b(C_a)$ is small enough to ensure that no logical operator
can be supported on $b(C_a)$, all annihilation operators $P_a$ must be equivalent modulo stabilizers
and the choice of $P_a$ does not matter.

After all clusters $C_a$ have been examined, the decoder is left with a new configuration of defects
$S'$, which is typically smaller than the original one.
If no defects are left, i.e., $S'=\emptyset$, the decoder stops and returns the product of all recorded Pauli operators $P_a$.
If $S'\ne \emptyset$,  the decoder applies a scale transformation
increasing the unit of length by some constant factor and repeats all the above steps starting from the syndrome $S'$.
The scale transformation potentially merges several unerased clusters $C_a$ into a single connected cluster
whereby giving the decoder one more attempt to annihilate them.

The full decoding algorithm is the iteration of partitioning the defects into the connected clusters and calculating the annihilation operators. It declares failure and aborts if the recorded operator cannot annihilate all the defects before the rescaled unit length is comparable to the lattice size.

A detailed  implementation of the RG decoder must  be tailored to a specific lattice geometry
and  a stabilizer code under consideration.
It must include a precise definition of  the connected clusters of defects $C_a$
and the enclosing regions $b(C_a)$. It must also include an algorithm for choosing the annihilation
operators $P_a$, a schedule for increasing the unit of length,
and clearly stated conditions under which the decoder aborts.
In the rest of this chapter we describe an efficient implementation of the RG decoder
for arbitrary stabilizer codes satisfying topological order conditions defined in the previous chapter.
The only part of this implementation specialized for the 3D cubic code 
is the ``broom algorithm'' of Section~\ref{subs:broom}.
As we have noted in Remark~\ref{rem:TransInv-strongTQO},
the our topological order conditions are satisfied by every translationally invariant exact code.
It turns out that the broom algorithm is also applicable for every translationally invariant code.

\section{Assumptions and conventions}
\label{sec:rg-decoder-assumptions}

Let $\Lambda$ be the regular 3D cubic lattice of linear size $L$ with periodic boundary conditions
along all coordinates $x,y,z$.
We shall label sites of $\Lambda$ by triples of integers $(i,j,k)$ defined modulo $L$
and measure the distance between sites using the $\ell_\infty$-metric. In other words,
the distance $d(u,v)$ between a pair of sites $u$ and $v$ is the smallest
integer $r$ such that $u$ and $v$ can be enclosed by a cubic box with dimensions
$r\times r\times r$. For example, $d(u,v)=1$ whenever $u$ and $v$ belong to the same
edge, plaquette, or elementary cube of the lattice.
Each site of $\Lambda$ represents one or several physical qubits (two qubits for the 3D Cubic Code).
Each elementary cube $c$ represents a spatial location of one or several  stabilizer generators
For example, there are two generators for the 3D Cubic Code.
A generator located at cube $c$ may act only on
qubits located at vertices of $c$. We shall label each elementary cube by
coordinates of its center, the triple of half-integers $(i,j,k)$ defined modulo $L$.
The distance $d(c,c')$ between a pair of cubes $c$ and $c'$ is the
distance between their centers.
For example, $d(c,c')=1$ whenever $c$ and $c'$ share a vertex, an edge, or a plaquette.

A {\bf defect} is a stabilizer generator whose eigenvalue has been flipped as a result of the error.
We shall use a term {\bf cluster of defects}, or simply {\bf cluster}
for any set of defects.  Define the {\em diameter} of a cluster $d(C)$ as
the maximum distance $d(c,c')$ where $c,c'\in C$.
Here and below  the distance between defects is defined as the distance between the cubes
occupied by these defects.
Given two non-empty clusters $C$ and $C'$,
define a distance $d(C,C')$ as the minimum distance $d(c,c')$ where $c\in C$ and $c'\in C'$.
Given an integer $r$, we shall say that a cluster $C$ is {\bf connected at scale $r$},
or simply {\bf $r$-connected}, if $C$
cannot be partitioned into two proper subsets $C=C'\cup C''$ such that $d(C',C'')>r$.
A maximal $r$-connected subset of a cluster $C$  is called a {\bf $r$-connected component} of $C$.
The {\bf minimal enclosing box} $b(C)$ of a cluster $C$ is the smallest
rectangular box $B$ enclosing all defects of $C$ 
such that all vertices of $B$ are dual sites of $\Lambda$.  
Note that the minimal enclosing box $b(C)$ is unique as long as $d(C)<L/2$;
if a cluster $C$ has diameter $L/2$, 
one may have two boxes with the same dimensions
enclosing $C$ that `wrap' around the lattice in two different ways.

Let $\cal G$ be the abelian group
generated by the stabilizer generators. Elements of $\cal G$ are called stabilizers.
Let $S(P)$ be the {\bf syndrome} of a Pauli operator $P$, that is,
the set of all stabilizer generators anticommuting with $P$.
The syndrome can be viewed as a cluster of defects.
\vspace{3mm}
\begin{center}
\parbox{.8\textwidth}{\em
We assume that our topological code obeys TQO1 and TQO2 of 
Definitions~\ref{defn:TQO1},\ref{defn:TQO2} throughout the chapter,
but not the no-strings rule of Definition~\ref{defn:no-strings}.
}
\end{center}
\vspace{3mm}
The 3D cubic code satisfies both of TQO1 and TQO2 with
$\ltqo = \half L$, since it is exact.
In order to avoid unnecessary complications due to boundaries,
we always assume that $\ltqo \le \half L$.
Below we consider only topological stabilizer codes.
Continued from the previous chapter,
a cluster of defects $C$ is called {\bf neutral} 
if it can be created from the vacuum by a Pauli operator $P$ 
supported on a cube of linear size $\ltqo$. 
Otherwise, the cluster is said to be {\bf charged}.
For example, the 2D toric code~\cite{Kitaev2003Fault-tolerant} has two types of defects:
magnetic charges (flipped plaquette operators) and electric charges (flipped star operators).
In this case,
a cluster of defects $C$ is neutral if and only if
$C$ contains even number of magnetic charges and even number of electric charges.
It follows from TQO2
that any neutral cluster of defects $C$ can be annihilated 
by a Pauli operator supported on the $1$-neighborhood
of the minimum enclosing box $b(C)$.

\section{Renormalization group decoder}
\label{subs:rgdecoder}

We are now ready to define our RG decoder precisely.
Recall that $d(C)$ is the diameter of a cluster $C$, and $\ltqo \le \half L$ by convention.
\begin{center}
\fbox{
\parbox{0.95\textwidth}{{\bf TestNeutral} \\
{\bf Input} $S$ : a set of defects, {\bf Output} $P$ : a Pauli operator.\\
1. Compute the minimal enclosing box $B$ of $S$.\\
2. {\bf if} $d(B) > \ltqo$, {\bf then} {\bf return} $I$.\\
3. Try to compute a Pauli $P$ supported on the 1-neighborhood of $B$ such that $S(P) = S$.\\
4. {\bf if} a consistent $P$ is found {\bf then} {\bf return} $P$ {\bf else} {\bf return} $I$.
}}
\end{center}
TQO2 implies that
TestNeutral successfully computes the correcting Pauli operator for any neutral cluster.
Step~1 is easy as we discuss in the end of Section~\ref{subs:clusters}.
The specification of Step~3 depends on the code,
but it always has an efficient implementation
using the standard stabilizer formalism~\cite{Gottesman1998Theory}.
In general, the condition $S(P)=S$ can be described by
a system of $O(V)$ linear equations over $O(V)$ binary variables parameterizing $P$,
where $V$ is the volume of $B$.
The running time is then $O(V^3)$ by the Gauss elimination.
In the special case of the 3D cubic code or more generally all translationally invariant codes,
there is a much more efficient algorithm running in time $O(V)$
which we describe in Section~\ref{subs:broom}.

\begin{figure*}[p]
\begin{minipage}{.48\textwidth}
\centering
\includegraphics[width=\textwidth,height=\textwidth]{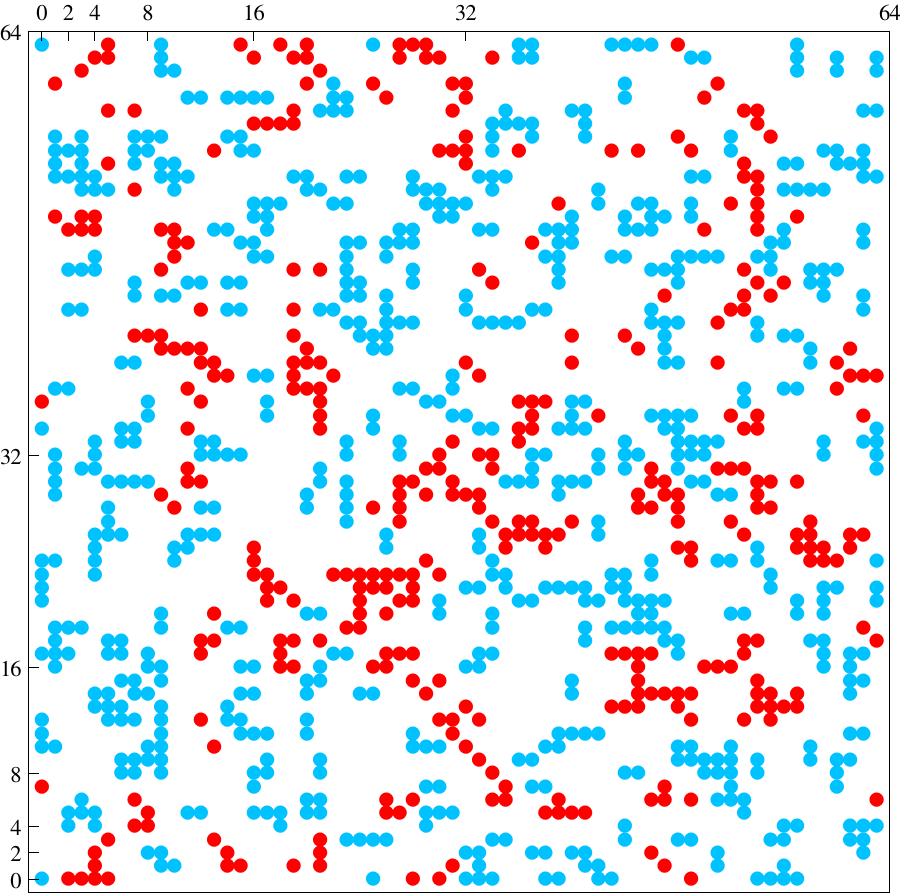}
Level 0: Unit length 1
\includegraphics[width=\textwidth,height=\textwidth]{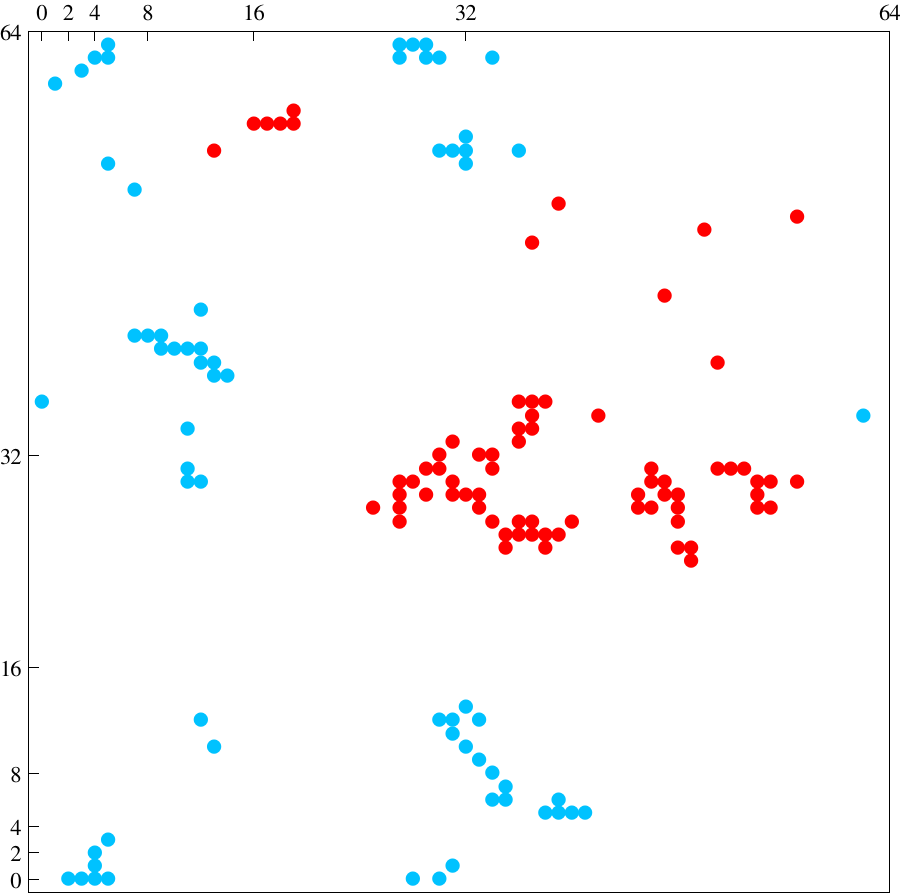}
Level 2: Unit length 4
\end{minipage}
\begin{minipage}{.48\textwidth}
\centering
\includegraphics[width=\textwidth,height=\textwidth]{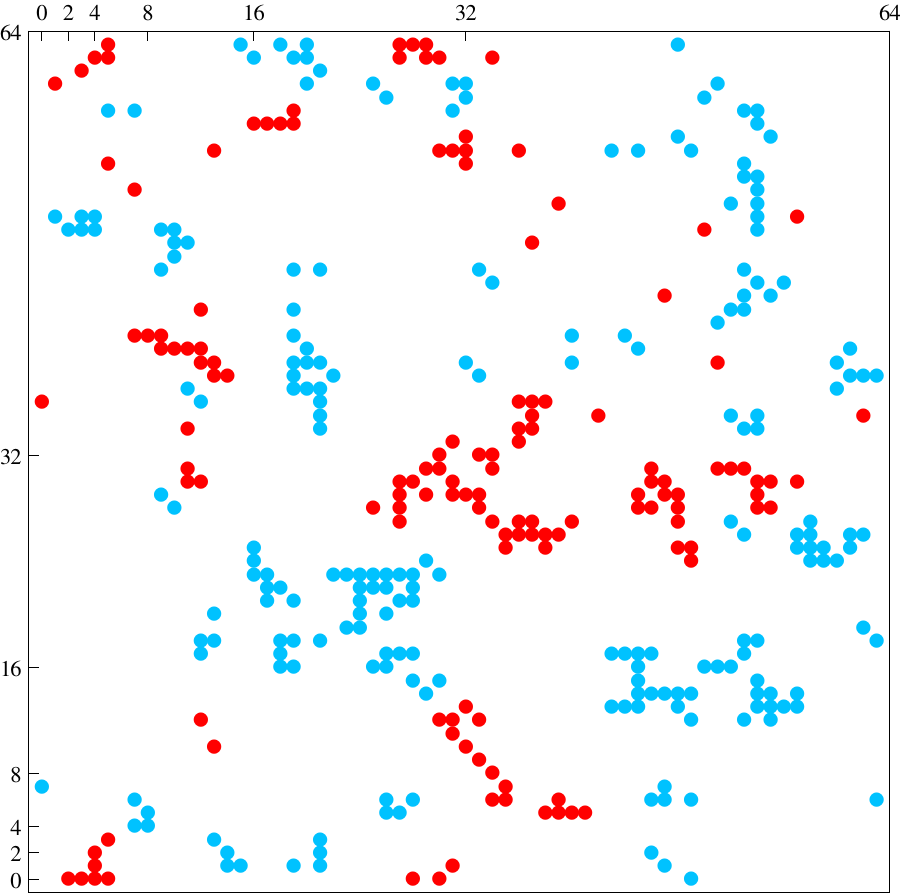}
Level 1: Unit length 2
\includegraphics[width=\textwidth,height=\textwidth]{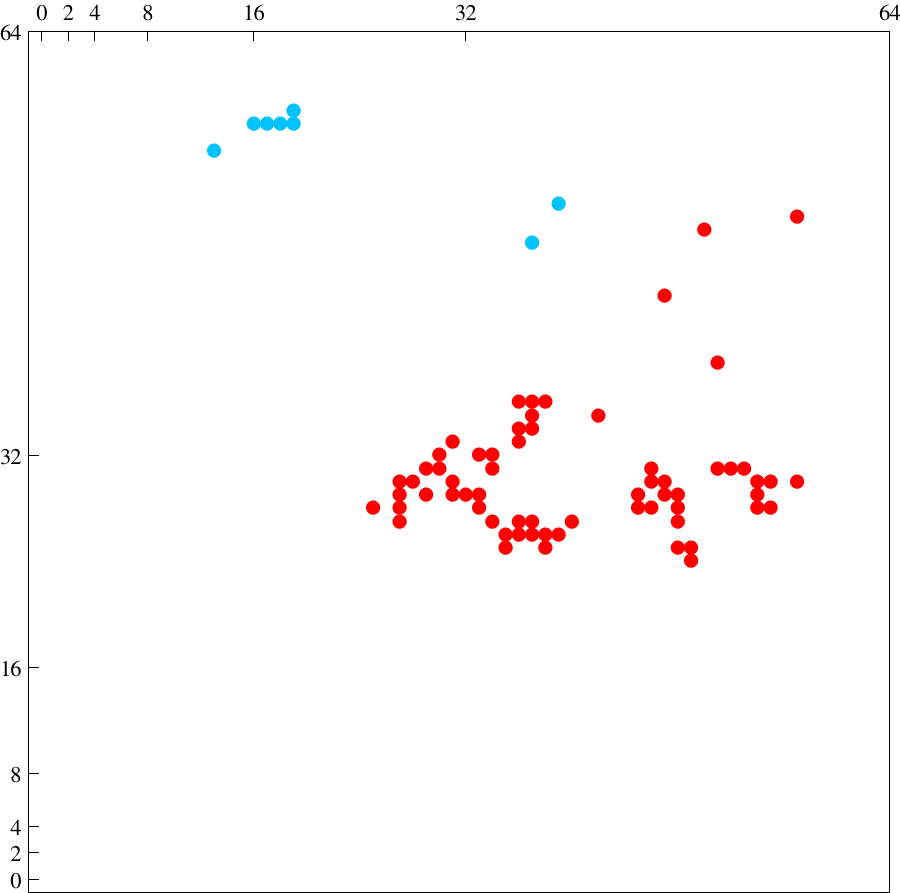}
Level 3: Unit length 8
\end{minipage}
\caption{%
Configuration of defects at consecutive levels of renormalization group decoder.
Bit-flip errors are generated stochastically for an illustrative purpose,
over $64 \times 64$ lattice with a uniform error rate $8\%$ per qubit.
The dots are flipped stabilizer generators (defects) of the 2D toric code.
At each level $p$, RG decoder decomposes the set of defects into connected clusters ---
connected components in a graph with vertices representing defects and edges connecting
pairs of defects separated by distance $2^p$ or less (not shown). 
Each neutral cluster containing even number of defects (blue)
is annihilated by applying a local Pauli operator supported 
in the smallest rectangular box enclosing the cluster.
Charged clusters containing odd number of defects (red) cannot be annihilated locally. 
All defects in the red clusters are  passed to the next RG level $p+1$.
In this example the red clusters from level 3 are annihilated at level 4 (not shown)
and RG decoder successfully annihilated all defects returning the
corrupted state to the originally encoded state. 
}
\label{fig:snapshot-RGdecoder}
\end{figure*}

Let $p_{M}$ be the largest integer such that $2^{p_M} < \ltqo$.
For any integer $0\le p\le p_M$, we define
the {\bf level-$p$ error correction}:
\begin{center}
\fbox{
\parbox{0.95\textwidth}{{\bf EC($p$)} \\
{\bf Input} $S$ : a syndrome, {\bf Output} $P$ : a Pauli operator.\\
1. Partition $S$ into $2^p$-connected components: $S=C_1\cup \ldots \cup C_m$.\\
2. For each component, compute $P_a = $ TestNeutral$(C_a)$.\\
3. {\bf return} the product $P_1 P_2 \cdots P_m$.
}}
\end{center}
The overall running time of EC($p$)  is polynomial in the number of qubits $N$.
Step~1 can be done, for example, by examining the distance between all pairs of defects,
forming a graph whose edges connect pairs of defects separated by distance $\le 2^p$, and finding
connected components of this graph.
A more efficient algorithm with the running time $O(N)$
is described in Section~\ref{subs:clusters}.
Since $V=O(N)$ and $m=O(N)$,
we see that the worst-case running time of EC($p$) is $O(N^2)$.
For instance, one can consider nested boxes,
near the faces of which many defects lie.
However, clusters are created from the vacuum
with some probability which we expect to be smaller than, say, $\half$.
So the number of clusters that have overlapping minimal enclosing box
appears with exponentially small probability.
On average, the running time of EC($p$) will be $O(N)$.
The full RG decoding algorithm is as follows.
\begin{center}
\fbox{
\parbox{0.95\textwidth}{{\bf RG Decoder} \\
{\bf Input} $S$ : the syndrome, {\bf Output} $P_{ec}$ : a Pauli operator.\\
1. Set $P_{ec}=I$.\\
2. {\bf for} $p=0$ {\bf to}  $p_M$ {\bf do}\\
\makebox[2em]{}  Let $Q=$EC($p$)($S$).\\
\makebox[2em]{}  Update $P_{ec} \leftarrow P_{ec} \cdot Q$ and $S\leftarrow S\oplus S(Q)$.\\
\makebox[1em]{} {\bf end for} \\
3. {\bf if} $S=0$ {\bf then} return $P_{ec}$ {\bf else} declare failure.
}}
\end{center}
Here the notation $S\oplus S(Q)$ stands for the symmetric difference of the sets
$S$ and $S(Q)$ or addition modulo two, if the syndromes are represented by binary strings.
The discussion above implies that the RG decoder has running time $O(N^2\log{N})$ in the worst case,
and $O(N\log{N})$ in the case of sparse syndromes.
An example of RG decoder in action is illustrated in Figure~\ref{fig:snapshot-RGdecoder}.

Being a physically realizable operation, 
any decoder should be written as a trace preserving completely positive map on the set of density matrices.
By measuring a syndrome $S$, the decoder projects the state onto a subspace $\Pi_S$ of the syndrome $S$,
and then applies a correcting operator $P_{ec}(S)$:
\[
\Phi_{ec}(\rho)=\sum_S P_{ec}(S)\Pi_S \, \rho \Pi_S P_{ec}(S)^\dag
\]
where the sum is over all possible syndromes.
Thus, to conform with this equation 
our decoder should return some operator that is consistent with the measured syndrome,
rather than declaring a ``failure.'' 
It is however no better than initializing the memory with an arbitrary state.

\section{Cluster decomposition}
\label{subs:clusters}
\newcommand{\ls}{r} 

Given a length scale $\ls$, the cluster decomposition of defects
is to partition the defects into maximally connected subsets
(connected components) of the syndrome at scale $\ls$.
Naively, the task to compute the decomposition of all the defects into clusters
will take time $O(m^2)$ where $m$ is the total number of defects.
The density of defects will typically be constant irrespective of the system size,
and the computation time for decomposition will be $O(N^2)$,
where $N$ is the volume of the system.
However, by exploiting the geometry of simple cubic lattice,
we can do it in time $O(N)$.
This is the optimal scaling
since we have to sweep through the whole system anyway
to identify the position of defects.

If $\ls = 1$, the problem is to label the connected components of binary array \cite{AsanoTanaka2010,KiranEtAl2011Labeling}.
Given a defect $u_0$, we can compute the connected component
containing $u_0$ in time $O(m)$ where $m$ is the number of defects in the component.
One prepares an empty queue (first-in-first-out data structure),
and puts $u_0$ into it.
The subsequent computing is as follows:
(i) Pop out the first element $u$ from the queue,
and of the neighborhood put the \emph{unlabeled} defects into the queue and label them.
(ii) Repeat until the queue becomes empty.
Every defect in a connected component $j$ is stored in the queue only once.
Hence, this process computes the component $j$ of a given defect
in time proportional to the number $m_j$ of defects in $j$.
One examines the whole system in some order
and finds the connected component whenever there is an unlabeled defect.
The total computation time is proportional to $N + O(1)\sum_j m_j = O(N)$,
since the connected components are disjoint.

For $\ls > 1$, the algorithm begins by dividing the whole lattice
into boxes of linear size $\ls$ or smaller.
The defects in a box certainly belong to a single connected component
(recall that we use the $\ell_\infty$ metric).
The defects in the boxes $B, B'$ belong to the same component
if and only if there is a pair $u \in B$, $v \in B'$ of defects
such that $d(u,v) \le \ls$.
In other words, we evaluate the binary function
\[
 \delta(B, B') =
\begin{cases}
1 & \text{if there are $u \in B, v \in B'$ such that } d(u,v) \le \ls ,\\
0 & \text{otherwise}
\end{cases}
\]
for each neighbor $B'$ of $B$;
if $B$ does not meet $B'$,
we know that $\delta(B,B') = 0$.

Given the table of $\delta$,
we can finish computing the decomposition
in time $O((L/r)^D)$ as in the $\ls = 1$ case.
We show that the computation of $\delta(B,B')$ can be done
in time $O(r^D)$ where $D$ is the dimension of the lattice.
Then, the total time to compute the table of $\delta$ for all adjacent boxes
will be $O(r^D (L/r)^D )=O(N)$.
Let $B$ and $B'$ be adjacent.
For clarity of presentation, we restrict to $D=2$.
Suppose $B$ and $B'$ meet along an edge parallel to $x$-axis.
Since any difference $|x-x'|$ of $x$-coordinates
of the defects in $B\cup B'$ is at most $\ls$,
we only need to compare $y$-coordinates.
That is, the problem is reduced to one dimension.
It suffices to pick two defects from $B$ and $B'$, respectively,
that are the closest to the $x$-axis.
If the $y$-coordinates differ at most by $\ls$,
then $\delta(B,B') = 1$;
otherwise, $\delta(B,B')=0$.

Suppose $B$ and $B'$ meet at a vertex.
Without loss of generality, we assume $B$ is in the third quadrant,
and $B'$ is in the first quadrant.
Define a binary function $\delta'$ on $B'$ as
\[
 \delta'(i,j)=
\begin{cases}
 1 & \text{if there is a defect $(x,y) \in B'$ where $ x \le i$ and $ y \le j$}, \\
 0 & \text{otherwise}.
\end{cases}
\]
The function table of $\delta'$ is computed in time $O(\ls^2)$.
It is important to note that $\delta'(i,j) = 1$ implies $\delta'(i+1,j) = \delta'(i,j+1)=1$.
One starts from the origin and sets $\delta'(\half,\half)=1$
if there is a defect at $(\half,\half)$; otherwise $\delta'(\half,\half)=0$.
Here, $(\half,\half)$ means the elementary square in the first quadrant
that is the closest to the origin.
Proceeding by a lexicographic order of the coordinates,
one sets $\delta'(i,j) = 1$
if $\delta'(i-1,j)=1$, or $\delta'(i,j-1)=1$, or there is a defect at $(i,j)$;
otherwise $\delta'(i,j)=0$.
It is readily checked that this procedure correctly computes $\delta'$.
Equipped with this $\delta'$ table,
we can immediately test for each defect in $B$
whether there is a defect in $B'$ within distance $\ls$.
Thus, we have computed $\delta(B,B')$ in time $O(r^2)+O(m)$
where $m$ is the number of defects in $B$,
which is at most $O(r^2)$.
The computation of $\delta$ in higher dimensions is similar.

The computation of the minimal enclosing box for each cluster is also efficient.
Given the coordinates of the $m$ points in the cluster,
we read out, say, $x$-coordinates $x_1,\ldots,x_m$.
The minimal enclosing interval $B_x$ of $x_1,\ldots,x_m$ under periodic boundary conditions,
is the complement of the longest interval between consecutive points $x_i$ and $x_{i+1}$
which can be computed in time $O(m)$.
$B_x$ is unambiguous if the diameter of the cluster is smaller than $L/2$.
The minimal enclosing box is the product set $B_x \times B_y \times B_z$,
whose vertices are computed in time $O(m)$.

\section{Gr\"obner basis and broom algorithm}
\label{subs:broom}

Now we describe an efficient algorithm for the 3D cubic code
that tests whether a cluster is neutral. If the test is positive,
the algorithm also returns a
Pauli operator $E$ that annihilates the cluster.

A crucial property of the 3D cubic code is
that it is translationally invariant;
it is described by a few Laurent polynomials 
over the variables that represent translations.
See Chapter~\ref{chap:alg-theory}.
The polynomials form a matrix $\sigma$ satisfying $\sigma^\dagger \lambda \sigma = 0$,
where $\dagger$ is transposition followed by entry-wise antipode map,
$x \mapsto x^{-1}$, etc., and $\lambda$ is an alternating full rank matrix.
Since we are working with qubits, the alternating matrix is actually symmetric.
More important than $\sigma$ is the excitation map $\epsilon = \sigma^\dagger \lambda$.
A Pauli operator described by a column matrix $p$ produces a syndrome described by $\epsilon p$.
Thus, the neutrality of a cluster $c$ is equivalent to the existence of $p$ of finitely many terms
such that $c = \epsilon p$.
That is, $c$ is neutral if and only if $c \in \im \epsilon$;
the neutrality test is really a submodule membership problem.
Gr\"obner basis provides an efficient algorithmic answer to the membership problem:
Compute a Gr\"obner basis $B$ for the module $\im \epsilon$.
It can be done by, for example, Buchberger algorithm applied to columns of $\epsilon$~\cite{PauerUnterkircher1999}.
The Gr\"obner basis is computed only once for a given code.
Then, the neutrality test is straightforward:
\begin{itemize}
\item[(1)] Express a cluster as a column matrix $e$ of Laurent polynomials.
This step takes running time $O(V)$ where $V$ is the volume of the cluster.
\item[(2)] Run a standard division algorithm with respect to $B$.
It generates an explicit expression
\[
 e = \sum_i c_i b_i + r
\]
where $b_i \in B$, and $r$ is a unique remainder that cannot be further reduced by $B$.
During the division the degree does not increase.
Therefore, the running time of the division is $O(V)$.
\item[(3)] If $r = 0$, then the cluster is neutral, and $c_i$ give the annihilating operator for the cluster.
If $r \neq 0$, then the cluster is charged.
\end{itemize}

\begin{figure}
\centering
\begin{minipage}{.4\textwidth}
\includegraphics[width=\textwidth]{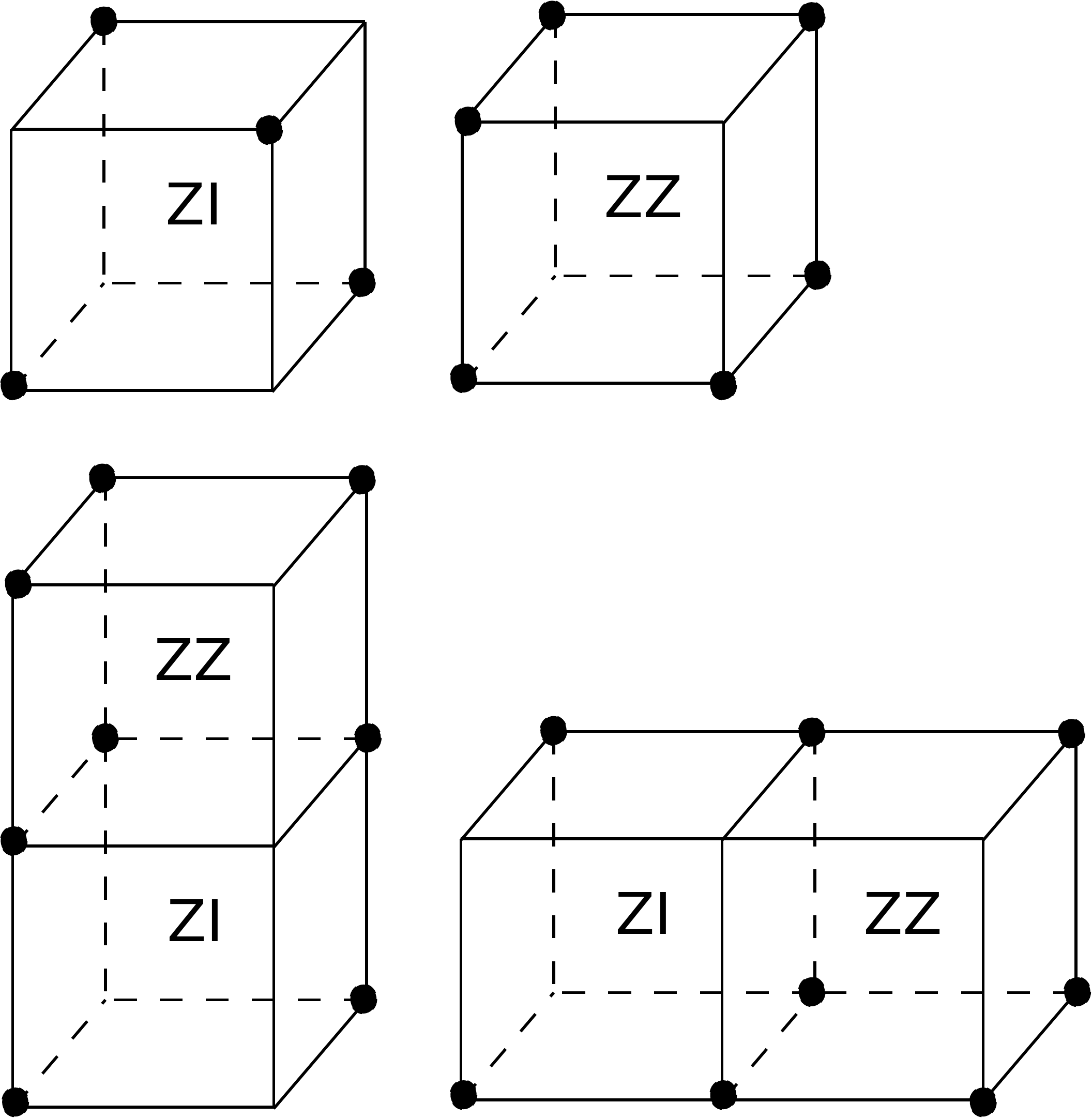}
\end{minipage}
\begin{minipage}{.2\textwidth}
{\drawgenerator{xz}{xyz}{x}{xy}{y}{1}{yz}{z}}
\end{minipage}
\caption{Elementary syndromes created by $Z$ errors. The vertices which are on the dual lattice, represent the defects created by the error at the center. The elementary syndrome by $ZI$ is used to push the defects to the bottom and to the left. The syndromes by errors of weight three is used to push the defects to the bottom-left corner.
The cube on the right specifies the coordinate system.
}
\label{fig:elementary-syndrome}
\end{figure}

The 3D cubic code is simple because it is Calderbank-Shor-Steane type code;
$X$- and $Z$-type errors can be treated separately.
Since there is one $X$-type stabilizer generator,
the syndrome caused by $Z$ errors is expressed by one Laurent polynomial,
which we call a \emph{syndrome polynomial}.
And neutral syndromes are described by an ideal (submodule) $I=(xyz+xy+yz+zx,~  xyz+x+y+z)$.%
\footnote{
The generators of $I$ are different from those presented in Chapter~\ref{chap:cubic-code},
but they are related by redefinition of lattice coordinate system.
The antipode map applied to $I$ yields $(1+x+y+z,~1+xy+yz+zx)$.}
For simplicity, suppose that the syndrome polynomial is of nonnegative exponents.
A Gr\"obner basis%
\footnote{%
The basis presented here is not the \emph{reduced} Gr\"obner basis,
by which we mean a basis where no term is divisible by a leading term of other elements in the basis.
The presented basis is actually what is used in the numerical simulation of the cubic code
in Chapter~\ref{chap:q-mem}.
It also matches with Figure~\ref{fig:elementary-syndrome}.
}
of $I$ is
\begin{align}
x+y+z+\mathbf{xyz}, \nonumber \\
x + y + \mathbf{x y} + z + x z + y z, \nonumber \\
x + y + z + x z + y z + z^2 + \mathbf{x z^2} + y z^2, \label{eq:gb-cubic-code} \\
x + y + x y + y^2 + \mathbf{x y^2} + z + y z + y^2 z,\nonumber \\
y^2 + y z + y^2 z + z^2 + y z^2 + \mathbf{y^2 z^2} \nonumber
\end{align}
where leading terms are marked as bold.
The following is an graphical explanation for the division algorithm.
Figure~\ref{fig:elementary-syndrome} shows a subset of a Gr\"obner basis for the cubic code.
One can directly see that the first four polynomials in Eq.~\eqref{eq:gb-cubic-code}
matches the diagrams.
A step-by-step explanation is as follows.
We fix a box $B$ that encloses all the defects in the neutral cluster.
We will sweep the defects to bottom-left-back corner.
Since each defect is a $Z_2$ charge, they will disappear in the end.
The algorithm begins with the top-right foremost vertex of $B$ on the dual lattice.
If there is a defect at $(x+\half,y+\half,z+\half)$,
we apply $ZI$ at $(x,y,z)$ to eliminate it.
This might create another defects as there are four defects in the elementary syndrome.
Important is that the potentially new defects
are all contained in the box $B$ we started with.
Continuing with $ZI$ we push all the defects in the foremost plane of $B$
to the left vertical line and the bottom horizontal line.
During this process, we record where $ZI$ has been applied.

For the defects on the vertical line at the left or on the horizontal line at the bottom,
we use the operator of weight 3 to further move the defects to the bottom-left corner.
See Fig.~\ref{fig:elementary-syndrome}.
This will in general create more defects behind, all of which are still contained in $B$.
Thus, we have moved all the defects on the foremost plane to the bottom-left corner
except for the three sites $t,u,v$:
\begin{equation}
\xymatrix@!0{ t \ar@{-}[d] & o \\ u \ar@{-}[r] & v }
\label{eq-diagram:last-three-defects}
\end{equation}
That is, if $E'$ is the recorded operator during the sweeping process,
the syndrome $S(EE') \subseteq B$ has potential defects only at $t,u,v$ on the foremost plane.

Let $o$ be at $(x_o+\half,y_o+\half,z_o+\half)$.
By considering the multiplication by suitable stabilizer generators $Q^Z$,
we can assume that $EE'$ is the identity on the plane $x=x_o$, except $(x_o,y_o,z_o)$.
Since there is no defect at $o$,
the operator at $(x_o,y_o,z_o)$ has to commute with $XX$;
it is either $II$ or $ZZ$.
Applying $ZZ$ if necessary, the operator at $(x_o,y_o,z_o)$ will become $II$,
and the defects at $t,u,v$ will disappear.
In this way, we have successfully pushed all the defects to the next-to-foremost plane.
We emphasize that the box $B$ still envelops all the defects,
and further $B$ can be shrunk in one direction.

Due to the threefold symmetry of the cubic code,
one can carry out this broom algorithm along any of three directions.
We will have, at last, a box $B$ of volume 1 that encloses all defects.
The defects in the cluster must be from one of the three elementary syndrome cubes
created either by $ZI$, $IZ$, or $ZZ$, which are easily eliminated.
It is clear that in time $O(V)$ the error operator has been computed up to stabilizer,
where $V$ is the initial volume of the minimal enclosing box of the cluster.

\section{Threshold theorem for topological stabilizer codes}
\label{sec:threshold}

In this section we prove that any topological stabilizer code can tolerate stochastic local errors with
a small constant rate assuming that the error correction is performed using the RG decoder.
We assume without loss of generality that each stabilizer generator is supported on a unit cube.
Each site of the lattice may contain finitely many qubits. A generator at a cube $c$ may act only on qubits of $c$.
We shall assume that errors at different sites are independent and identically distributed.
More precisely, let $E(P)$ be the set of sites at which a Pauli error $P$ acts nontrivially. We shall assume that
\begin{equation}
\mathrm{Pr}[E(P)=E] = (1- \epsilon)^{V-|E|} \epsilon^{|E|}
\label{eq:random-error-distribution}
\end{equation}
where $0 \le \epsilon \le 1$ is the error rate and $V=L^D$ is the total number of sites (the volume of the lattice).
For example, the depolarizing noise in which every qubit experiences $X,Y,Z$ errors
with the probability $p/3$ each, satisfies
Eq.~(\ref{eq:random-error-distribution}) with the error rate $\epsilon=1-(1-p)^q$, where $q$ is the number of qubits per site.
\begin{theorem}
Suppose a family of stabilizer codes has topological order satisfying TQO1,2.
(In particular, every translationally invariant exact codes do.)
Then, there exists a constant threshold $\epsilon_0>0$
such that for any $\epsilon<\epsilon_0$ the RG decoder corrects random independent errors with rate $\epsilon$
with the failure probability at most $e^{-\Omega(L^\eta)}$ for some constant $\eta>0$.
\label{thm:universal-threshold}
\end{theorem}
\noindent
In the rest of this section we prove the theorem.
Our proof borrows some techniques from \cite{Gray2001Guide,Gacs1986Reliable,Harrington2004thesis},
specifically Section~5.1 of Gray's review~\cite{Gray2001Guide} on G\'acs' 1D cellular automata~\cite{Gacs1986Reliable}.

Recall that we use $\ell_\infty$-metric, so a cube of linear size $r$ thus has diameter $r$.
We keep the terminologies and conventions from Section~\ref{sec:rg-decoder-assumptions},
and our decoder is what we have explained in the previous sections:
%
The \emph{level-$p$ error correction} {\em EC($p$)} on a syndrome $S$ is the following subroutine.
(i) find all neutral $2^p$-connected components $M$ of $S$,
(ii) for each $M$ found at step 1, calculate and apply a Pauli operator $P$ supported on the 1-neighborhood of $b(M)$ that annihilates $M$, and update the syndrome accordingly.
Calling the full RG decoder on a syndrome $S$ involves the following steps:
(i) run EC(0), EC(1), ..., EC($\lfloor \log_2 \ltqo \rfloor$),
(ii) if the resulting syndrome $S$ is empty, return the accumulated Pauli operator applied by the subroutines EC($p$). Otherwise, declare a failure.

Below we shall use the term `error' 
both for the error operator $P$ and for the subset of sites $E$ acted on by $P$, 
whenever the meaning is clear from the context.
Let us choose an integer $Q \gg 1$ and find a class of errors which are properly corrected by the RG decoder,
see Lemma~\ref{lem:correctability-criterion} below.
We will see later that this class of errors actually includes all errors which are likely to appear for small enough error rate.
\begin{defn}
Let $E$ be a fixed error.
A site $u\in E$ is called a {\bf level-$0$ chunk}.
A non-empty subset of $E$ is called a {\bf level-$n$ chunk} ($n\ge 1$) 
if it is a disjoint union of two level-$(n-1)$ chunks and its diameter is at most $Q^n /2$.
\end{defn}
The term `chunk,' not to be confused with the usage in Section~\ref{sec:superlinear-distance},
is chosen in order to avoid confusion with `cluster', which is used for a set of defects.
Note that a level-$n$ chunk contains exactly  $2^n$ sites.
Given an error $E$, let $E_n$ be the union of all level-$n$ chunks of $E$.
If $u \in E_{n+1}$, then by definition $u$ is an element of a level-$(n+1)$ chunk.
Since a level-$(n+1)$ chunk is a union of two level-$n$ chunks, $u$ is contained in a level-$n$ chunk. Hence, $u \in E_n$, and the sequence $E_n$ form a descending chain
\[
 E = E_0 \supseteq E_1 \supseteq \cdots \supseteq E_m ,
\]
where $m$ is the smallest integer such that $E_{m+1}=\emptyset$.
Let $F_i = E_i \setminus E_{i+1}$, so $E = F_0 \cup F_1 \cup \cdots \cup F_m$ is expressed as a disjoint union, which we call the \emph{chunk decomposition} of $E$.
\begin{prop}
Let $Q \ge 6$ and $M$ be any $Q^n$-connected component of $F_n$. Then
$M$ has diameter $\le Q^n$ and is separated from $E_n \setminus M$ by distance $> \frac{1}{3}Q^{n+1}$.
\label{prop:structure-Fn}
\end{prop}
\begin{proof}
We claim that for any pair of sites $u \in F_n = E_n \setminus E_{n+1}$ and $v \in E_n$ we have $ d( u,v) \le Q^n$ or $d( u,v ) > \frac{1}{3} Q^{n+1}$.
Suppose on the contrary to the claim, that there is a pair $u \in F_n$ and $v \in E_n$ such that $ Q^n < d(u,v) \le Q^{n+1} /3$. Let $C_u \ni u$ and $C_v \ni v$ be level-$n$ chunks that contains $u$ and $v$, respectively. Since the diameters of $C_{u,v}$ are $\le Q^n /2$ and $d(u,v) > Q^n$, we deduce that $C_u$ and $C_v$ are disjoint. On the other hand,
\[
 d( C_u \cup C_v ) \le d( u,v ) + d(C_u) + d(C_v)  \le Q^{n+1}/2
\]
since $Q \ge 6$. Thus, $C_u \cup C_v$ is a level-$(n+1)$ chunk that contains $u$
which shows that $u \in E_{n+1}$. It contradicts to our assumption that $u\in F_n=E_n \setminus E_{n+1}$.
\end{proof}
Note that in the chunk decomposition a $Q^n$-connected component $P$ of $E_n$ may not be separated from the rest $E \setminus P$ by distance $> Q^n$.
\begin{lem}
Let $Q \ge 10$. If the length $m$ of the chunk decomposition of an error $E$ satisfies $Q^{m+1} < \ltqo$, then $E$ is corrected by the RG decoder.
\label{lem:correctability-criterion}
\end{lem}
\begin{proof}
Consider any fixed error $P$ supported on a set of sites $E$.
Let $E=F_0 \cup F_1 \cup \cdots \cup F_m$ be the chunk decomposition of $E$, and
let $F_{j,\alpha}$ be the $Q^j$-connected components of $F_j$.
Also, let $B_{j,\alpha}$ be the $1$-neighborhood of the smallest box enclosing the syndrome
created by the restriction of $P$ onto $F_{j,\alpha}$.
Proposition~\ref{prop:structure-Fn} implies that
\begin{equation}
\label{subchunks2}
d(B_{j,\alpha})\le Q^j+2 \quad \mbox{and} \quad d(B_{j,\alpha}, B_{k,\beta})>\frac13 Q^{1+\min{(j,k)}}-2.
\end{equation}
Let $P_{ec}^{(p)}$ be the accumulated correcting operator returned by the levels $0,\ldots,p$  of the RG decoder.
Let us use induction in $p$ to prove the following statement.
\begin{enumerate}
\item The operator $P_{ec}^{(p)}$ has support on the union of the boxes $B_{j,\alpha}$.
\item  The operators $P_{ec}^{(p)}$ and $P$ have the same restriction on $B_{j,\alpha}$ modulo stabilizers
for any $j$ such that $2^p\ge Q^j+2$.
\end{enumerate}
The base of induction is $p=0$.  Using Eq.~(\ref{subchunks2}) we conclude that any
$1$-connected component of the syndrome $S(P)$ is fully contained
inside some box $B_{j,\alpha}$. It proves that $P_{ec}^{(0)}$ has support on the union of the boxes $B_{j,\alpha}$.
The second statement is trivial for $p=0$.

Suppose we have proved the above statement for some $p$.
Then the operator $P\cdot P_{ec}^{(p)}$ has support only inside
boxes $B_{j,\alpha}$ such that $2^p<Q^j+1$ (modulo stabilizers). It follows that any
$2^{p+1}$-connected component of the syndrome caused by $P\cdot P_{ec}^{(p)}$ is contained in some box
$B_{j,\alpha}$ with $2^p<Q^j+1$. Note that the RG decoder never adds new defects;
we just need to check that $2^{p+1}$-connected components do not cross the boundaries
between the boxes $B_{j,\alpha}$ with $2^p<Q^j+1$. This follows from Eq.~(\ref{subchunks2}).
Hence $P_{ec}^{(p+1)}$
has support in the union of $B_{j,\alpha}$. Furthermore, if
$2^p<Q^j+1\le 2^{p+1}$, the cluster of defects created by $P\cdot P_{ec}^{(p)}$
inside $B_{j,\alpha}$ forms a single  $2^{p+1}$-connected component
of the syndrome examined by  EC$(p+1)$.
This cluster is neutral since we assumed $Q^{m+1}<L_{tqo}$.
Hence $P_{ec}^{(p+1)}$ will annihilate this cluster.
The annihilation operator is equivalent to the restriction of $P\cdot P_{ec}^{(p)}$
onto $B_{j,\alpha}$ modulo stabilizers, since the linear size
of $B_{j,\alpha}$ is smaller than $\ltqo$. It proves the induction hypothesis
for the level $p+1$. 
\end{proof}
\noindent
The preceding lemma says that errors by which the RG decoder could be confused are those from very high level chunks. What is the probability of the occurrence of such a high level chunk if the error is random according to Eq.\eqref{eq:random-error-distribution}? Since our probability distribution of errors depend only on the number of sites in $E$, this question is completely percolation-theoretic.

Let us review some terminology from the percolation theory\cite{Grimmett1999Percolation}. An event is a collection of configurations. In our setting, a configuration is a subset of the lattice. Hence, we have a partial order in the configuration space by the set-theoretic inclusion. An event $\mathcal E$ is said to be \emph{increasing} if $E \in \mathcal{E}, E \subseteq E'$ implies $E' \in \mathcal{E}$. For example, the event defined by the criterion that there exists an error at $(0,0)$, is increasing.
The \emph{disjoint occurrence} $\mathcal A \circ \mathcal B$ of the events $\mathcal A$ and $\mathcal B$ is defined as the collection of configurations $E$ such that $E = E_a \cup E_b$ is a disjoint union of $E_a \in \mathcal A$ and $E_b \in \mathcal B$. To illustrate the distinction between $\mathcal A \circ \mathcal B$ and $\mathcal A \cap \mathcal B$, consider two events defined as $\mathcal A = $ ``there are errors at $(0,0)$ and at $(1,0)$'', and $\mathcal B = $ ``there are errors at $(0,0)$ and at $(0,1)$''. The intersection $\mathcal A \cap \mathcal B$ contains a configuration $\{ (0,0),(1,0),(0,1) \}$, but the disjoint occurrence $\mathcal A \circ \mathcal B$ does not. A useful inequality by van den Berg and Kesten (BK) reads \cite{BergKesten1985, Grimmett1999Percolation}
\begin{equation}
 \mathrm{Pr}[ \mathcal A \circ \mathcal B ] \le \mathrm{Pr}[ \mathcal A ] \cdot \mathrm{Pr}[ \mathcal B]
\label{eq:BK-inequality}
\end{equation}
provided the events $\mathcal A$ and $\mathcal B$ are increasing.

\begin{proof}[Proof of Theorem~\ref{thm:universal-threshold}]
Consider a $D$-dimensional lattice
and a random error $E$ defined by Eq.~(\ref{eq:random-error-distribution}).
Let $B_n$ be a fixed cubic box of linear size $Q^n$ and
$B_n^+$ be the  box of linear
size $3Q^n$ centered at $B_n$.
Define the following probabilities:
\begin{align*}
p_n &= \mathrm{Pr}\left[ \mbox{$B_n$ has a nonzero overlap with a level-$n$ chunk of $E$}\right] \\
\tilde{p}_n &= \mathrm{Pr}\left[\mbox{$B_n^+$ contains a level-$n$ chunk of $E$}\right] \\
q_n &= \mathrm{Pr}\left[\mbox{$B_n^+$ contains $2$ disjoint  level-$(n-1)$ chunks of $E$}\right] \\
r_n &= \mathrm{Pr}\left[\mbox{$B_n^+$  contains a level-$(n-1)$ chunk of $E$}\right]
\end{align*}
Note that all these probabilities do not depend on the choice of the box $B_n$
due to translation invariance. Since a level-$0$ chunk is just a single site
of $E$, we have $p_0=\epsilon$.
We begin by noting that
\[
p_n\le \tilde{p}_n\le q_n.
\]
Here we used  the fact that
any level-$n$ chunk has diameter at most $Q^n/2$ and that any level-$n$ chunk
consists of a disjoint pair of level-$(n-1)$ chunks.
Let us fix the box $B_n^+$ and let ${\cal Q}_n$ be the event that
$B_n^+$ contains a disjoint pair of  level-$(n-1)$ chunks of $E$.
Let ${\cal R}_n$ be the event that $B_n^+$ contains a level-$(n-1)$ chunk of $E$.
Then ${\cal Q}_n={\cal R}_n\circ {\cal R}_n$.
It is clear that  ${\cal Q}_n$ and ${\cal R}_n$  are increasing events.
Applying the van den Berg and Kesten inequality we arrive at
\[
q_n\le r_n^2.
\]
Finally, since $B_n^+$ is a disjoint union of $(3Q)^D$ boxes of linear size $Q^{n-1}$,
the union bound yields
\[
r_n\le (3Q)^D p_{n-1}.
\]
Combining the above inequalities we get
$p_n \le (3Q)^{2D} p_{n-1}^2$, and hence
\[
p_n \le (3Q)^{-2D}((3Q)^{2D} \epsilon)^{2^n}.
\]
The probability $p_n$ is  doubly exponentially small in $n$
whenever $\epsilon < (3Q)^{-2D}$.
 If there exists at least one level-$n$ chunk, there is always a box of linear size $Q^n$ that overlaps with it. Hence, on the finite system of linear size $L$, the probability of the occurrence of a level-$m$ chunk is bounded above by $L^D p_m$.
 Employing
 Lemma~\ref{lem:correctability-criterion}, we conclude that the RG decoder fails with probability at most $p_{fail}=L^D p_m$
  for any $m$ such that $Q^{m+1} < \ltqo$. Since we assumed that $\ltqo\ge L^\delta$,
 one can choose $m\approx \delta \log{L}/\log{Q}$.
 In this case $p_{fail}=\exp{(-\Omega(L^\eta))}$ for $\eta \approx \delta / \log{Q}$.
 We have proved our theorem with $\epsilon_0 = (3Q)^{-2D}$ where $Q = 10$.
\end{proof}

\section{Benchmark of the decoder}
\label{sec:benchmark}

Given a decoder, a family of quantum codes indexed by code length (system size)
is said to have an \emph{error threshold} $p_c$
if the probability for decoder to fail approaches zero in the limit of large code length
provided the random error rate $p$ is less than $p_c$.
We tested our decoder with respect to random uncorrelated bit-flip errors on the well-studied 2D toric code.
The error threshold is measured to be $8.4(1)\%$ using $\ell_1$-metric.
See Fig.~\ref{fig:threshold}.
It is reasonably close to the best known value $10.3\%$
based on the perfect matching algorithm~\cite{DennisKitaevLandahlEtAl2002Topological, Harrington2004thesis},
or $9\%$ based on a renormalization group decoder of similar nature to ours~\cite{Duclos-CianciPoulin2009Fast}.
This is remarkable for our decoder's simplicity and applicability.
The 3D cubic code has threshold $\gtrsim 1.1\%$ under independent bit-flip errors 
using $\ell_\infty$-metric.
Note that in these simulations we do not use TestNeutral$'$ of Remark~\ref{rem:TestNeutralprime}.

\begin{figure}[ht]
\centering
\begin{minipage}{.48\textwidth}
\includegraphics[width=\textwidth]{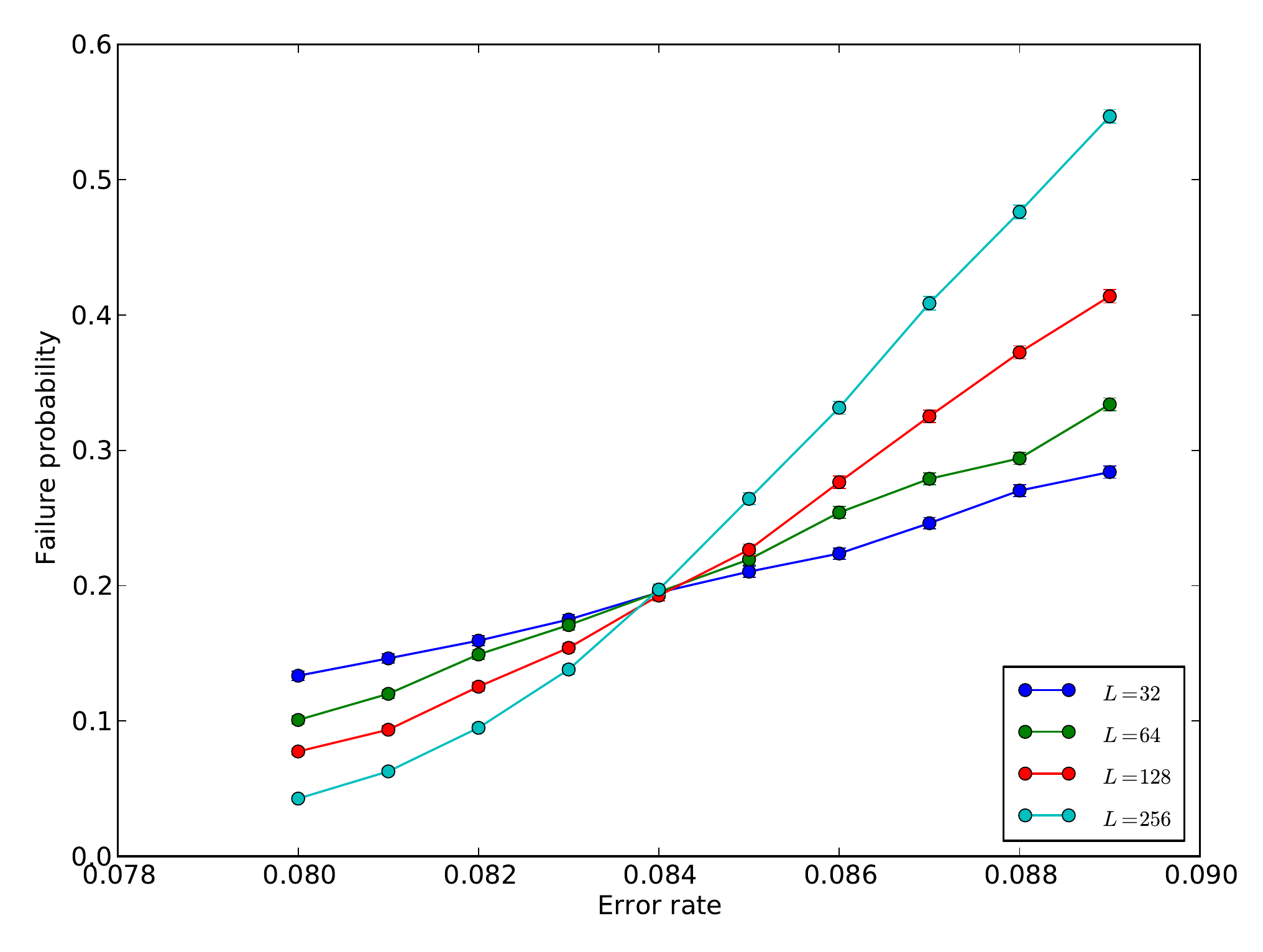}
\end{minipage}
\begin{minipage}{.48\textwidth}
\includegraphics[width=\textwidth]{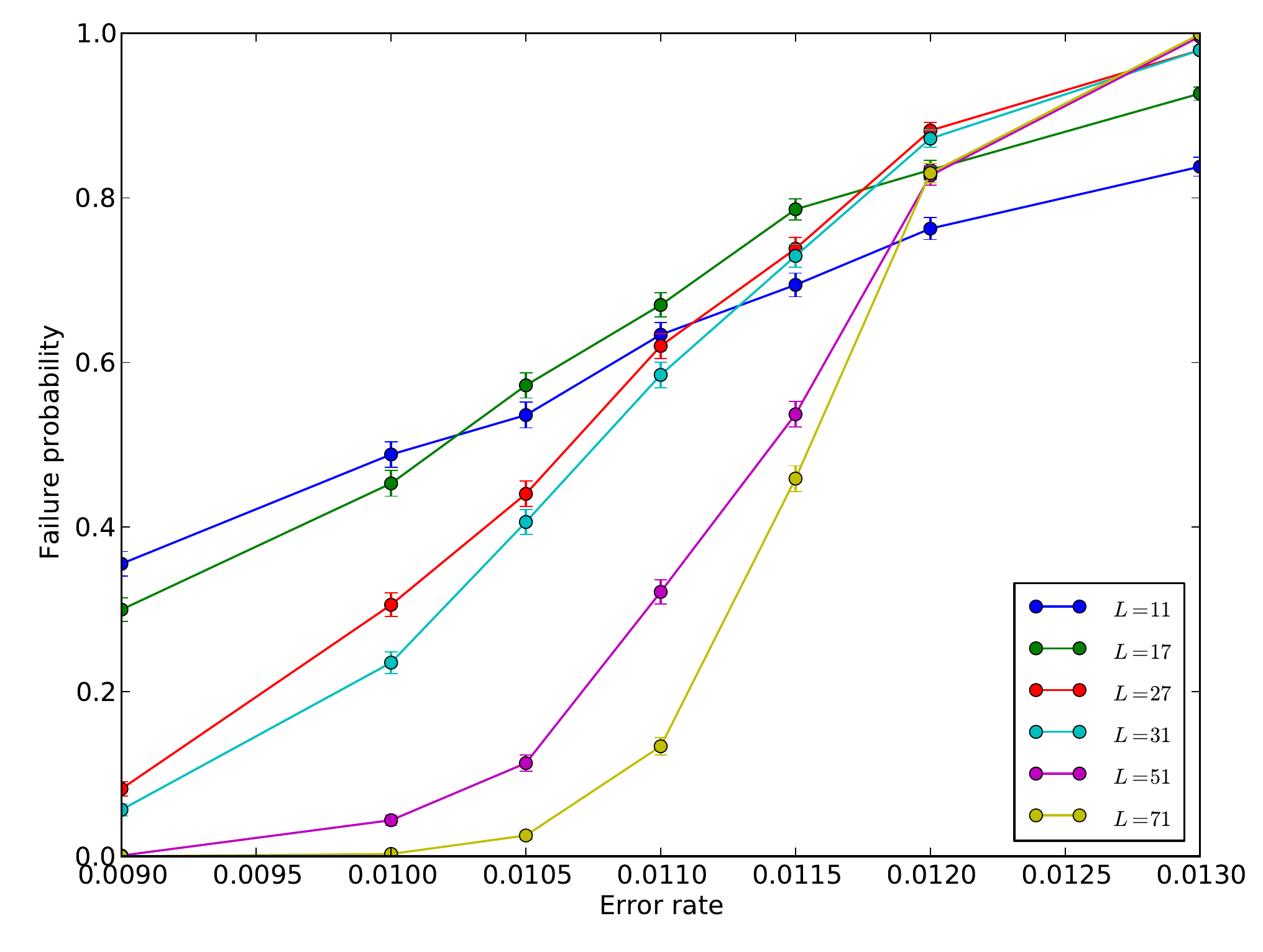}
  \end{minipage}
\caption{
The thresholds of 2D toric code (left) and 3D cubic code (right)
under independent random bit-flip errors using our RG decoder.
The left shows simulation data for 2D toric code under $\ell_1$-metric,
The right shows the data for 3D cubic code under $\ell_\infty$-metric.
The thresholds are measured to be $p_c(\text{2D toric})=8.4(1)\%$
and $p_c(\text{3D cubic}) \gtrsim 1.1 \%$.
}
\label{fig:threshold}
\end{figure}

\chapter{Self-correcting quantum memory}
\label{chap:q-mem}

We now directly assess the cubic code as a self-correcting quantum memory.
The results in this chapter depend on previous chapters.
In Chapter~\ref{chap:consq-no-strings} 
we have found an energy barrier to isolate a nontrivial charge.
It implies that nontrivial charges are localized in their original position where they are created.
It suggests further that geometrically localized defects would form a neutral cluster,
and the annihilation of them would not likely cause any undetected logical error.
The renormalization group decoder and the threshold theorem 
in Chapter~\ref{chap:RG-decoder} have been motivated by this intuition.
In this chapter, we show that 
the performance of the RG decoder against thermal errors
matches our expectation, too.
We prove that if the no-strings rule is satisfied,
then the memory time grows as a power law 
$L^{c\beta}$ where $\beta$ is the inverse temperature of the heat bath.
The bound is valid when the system size is small enough $L \le e^{c' \beta}$.
Note that our analysis does not tell anything conclusive for the system in thermodynamic limit.

A few remarks can be made to the validity regime $L \le e^{c' \beta}$.
It is reasonable to expect that the no-strings rule and a good decoder
would be sufficient to guarantee that the memory time increases with the system size.
However, when the entropy is considered, the situation is more complicated.
From the degeneracy formula of the cubic code in Chapter~\ref{chap:cubic-code},
we know that there are system sizes where the number $k$ of encoded qubits is 2.
Since the cubic code has exactly one $X$-type stabilizer generator $G_X$ in each elementary cube,
$k/2$ is equal to the number of ways that $G_X$'s multiply to the identity.
Therefore, when $k = 2$, any configuration of defects is allowed
as long as they are in an even number,
and the number of excited states at a particular energy 
is given by a combinatorial factor.
In contrast, the two-dimensional Ising model has
only exponentially many configurations at a particular energy
(the number of self-avoiding walks).
The energy barrier of the cubic code is lower than that of the 2D Ising model,
while the entropic contribution is stronger in the cubic code
than in the 2D Ising model.
The inequality $L \le e^{c' \beta}$ can be understood as a requirement 
that entropic contribution should not be too large.

Perhaps, this is already hinted from the smooth thermal partition function of the cubic code
presented in Section~\ref{sec:thermal-partition-function}.
It appears that the strong entropic contribution is unavoidable at the presence of point-like defects.
We have seen from Chapter~\ref{chap:lowD-codes} 
that the point-like defects always exists in three-dimensional translationally invariant topological codes.
Thus, it would be impossible to have a truly self-correcting quantum memory
based on local translationally invariant quantum codes in three dimensions
whose memory time increases unbounded with the system size,
similar to the 4D toric code~\cite{AlickiHorodeckiHorodeckiEtAl2010thermal, ChesiLossBravyiEtAl2010Thermodynamic}.

We model the thermal interaction by taking Davies weak coupling limit~\cite{Davies1974}.
In order to make use of the results from previous chapters,
we continue to assume topological quantum order (TQO) condition 1 and 2 
defined in Chapter~\ref{chap:consq-no-strings},
and the no-strings rule of Chapter~\ref{chap:cubic-code}.
The two TQO conditions demand that ground states must be locally indistinguishable,
and that any locally created cluster of defects must be created from a ground state
by an operator supported on the immediate neighborhood of the cluster.
Note that any translationally invariant exact code Hamiltonian, such as the cubic code,
automatically satisfies both of the TQO conditions.
Of course, the cubic code satisfies the no-strings rule.

One more technical requirement to show the long memory time is that the number $k$ of encoded qubits,
or the ground-state degeneracy must be small.
It is an ironic requirement at least for the cubic code
because the small $k$ implies a large number of excited states and large entropic contribution.
In the three-dimensional case, we know from Corollary~\ref{cor:annT-3d-characteristic-dimension-0}
that the characteristic dimension must be 1 in order for the no-strings rule to be obeyed.
The nonzero characteristic dimension generally implies a growing $k(L)$ as a function of $L$.
It is not so clear whether it is always possible to find a family of lattice sizes $\{ L_i \}$
such that $k(L_i)$ is small.
Although we do not know how to resolve this, the cubic code causes no problem since we know
there is an infinite family $\{L_i\}$ such that $k(L_i) = 2$ 
by Corollary~\ref{cor:cubic-code-degeneracy-formula}.

\section{Previous work}

Alternative routes towards quantum self-correction in topological memories proposed in the literature,
focus on  finding new mechanisms for suppressing diffusion of topological defects
(here and below we only consider zero-dimensional defects).
Arguably, the simplest of such mechanisms would be to have no topological defects in the first place.
Unfortunately, this seems to require four spatial dimensions.
The 4D toric code~\cite{DennisKitaevLandahlEtAl2002Topological} provides the only known example of a truly self-correcting quantum memory.
As was shown by Alicki and Horodecki's~\cite{AlickiHorodeckiHorodeckiEtAl2010thermal},
the memory time of the 4D toric code
grows exponentially with the lattice size for small enough bath temperature.
The first  3D topological memory in which diffusion of defects is constrained by superselection
rules was proposed by Chamon~\cite{Chamon2005Quantum}, see also~\cite{BravyiLeemhuisTerhal2011Topological}.
Topological defects in this model have a limited mobility restricted to certain
subspaces of $\mathbb{R}^3$ or have no mobility at all.
However, the model has no macroscopic energy barrier that could suppress the diffusion.
2D topological memories
in which diffusion of anyons is suppressed by effective long-range interactions
were studied by Chesi et al.~\cite{ChesiRoethlisbergerLoss2010Self-Correcting}
and Hamma et al~\cite{HammaCastelnovoChamon2009ToricBoson}.
A quenched disorder and Anderson localization were proposed as
a means of suppressing propagation of defects at the zero temperature
by Wootton and Pachos~\cite{WoottonPachos2011disorder} and,
independently, by Stark et al.~\cite{StarkImamogluRenner2011Localization},
see also~\cite{BravyiKoenig2011}.
A no-go theorem for quantum self-correction based on 3D stabilizer Hamiltonians
in which ground-state degeneracy does not depend on lattice dimensions
was proved by Yoshida~\cite{Yoshida2011feasibility}.
A different line of research initiated
by Pastawski et al.~\cite{PastawskiClementeCirac2011dissipation}
focuses  on quantum memories in which active error correction
is imitated by engineered dissipation driving the memory system
towards the ground state (as opposed to the Gibbs state).
Finally, let us emphasize that quantum  self-correction is technically different from
the thermal stability of topological phases,
see, for instance,~\cite{CastelnovoChamon2007Entanglement,
IblisdirEtAl2010thermal,
NussinovOrtiz2008Autocorrelations,
Hastings2011warmTQO}.
While the latter  attempts to establish the presence (or absence) of topological order
in the equilibrium thermal state, quantum self-correction is mostly concerned with the relaxation time
towards the equilibrium state.

\section{Storage scheme}

In order to use any memory, either classical or quantum,
a user must be able to write, store, and read information.
In this section we describe these steps
formally for a topological quantum memory based on the 3D cubic code.
The Hamiltonian of the memory is
\begin{equation*}
H=-J \sum_{c} G^X_c + G^Z_c,
\end{equation*}
where the sum runs over all $L^3$ elementary cubes $c$ and the operators $G^X_c$, $G^Z_c$
act on the qubits of $c$ as shown on Fig.~\ref{fig:CubicCode}.
The positive coupling constant $J$ is set to $J=\half$ for simplicity.

Suppose at time $t=0$ the memory system is initialized in some ground state $\rho(0)$
encoding a quantum state to be stored; the ground space is the code space.
We model interaction between the memory system and the thermal bath
using the Davies weak coupling limit~\cite{Davies1974}.
It provides a Markovian master equation of the following form:
\begin{equation}
\dot{\rho}(t)=-i[H,\rho(t)] + \mathcal L(\rho(t)), \quad t\ge 0.
\label{eq:Markovian-master}
\end{equation}
Here $\rho(t)$ is the state of the memory system at time $t$ 
and $\mathcal L$ is the Lindblad generator describing dissipation of energy.
To define $\mathcal L$, let us choose some set of self-adjoint operators 
$\{A_\alpha\}$ through which the memory can couple to the bath.
We assume that each $A_\alpha$ acts nontrivially on a constant number of qubits.
For example, $\{A_\alpha\}$ could be the set of all single-qubit Pauli operators.
Let 
\begin{equation}
A_{\alpha}(t)= e^{iHt} A_\alpha e^{-iHt} = \sum_\omega e^{-i\omega t} A_{\alpha,\omega} .
\label{eq:jump-ops}
\end{equation}
We will see later in Eq.~\eqref{Acom} that 
$A_{\alpha,\omega}$ maps eigenvectors of $H$ with energy $E$ 
to eigenvectors with energy $E-\omega$.
Then
\begin{equation}
\mathcal L( \rho) = \sum_\alpha \sum_\omega h(\alpha,\omega) 
\left(
A_{\alpha,\omega} \rho A_{\alpha,\omega}^\dag - \half \{  \rho,A_{\alpha,\omega}^\dag A_{\alpha,\omega} \} 
\right).
\label{eq:lindbladian}
\end{equation}
The coefficient $h(\alpha,\omega)$ is the rate of quantum jumps caused by $A_\alpha$ 
transferring energy $\omega$ from the memory to the bath.
It must obey the detailed balance condition~\cite{Davies1974}
\begin{equation}
\label{DB}
h(\alpha,-\omega)=e^{-\beta \omega} h(\alpha,\omega),
\end{equation}
where $\beta$ is the inverse bath temperature.
The detailed balance condition Eq.~\eqref{DB} is the only part
of our model that depends on  the bath temperature. 
It guarantees that the Gibbs state $\rho_\beta \sim e^{-\beta H}$ 
is the fixed point of the dynamics as $\mathcal L(\rho_\beta)=0$.
Identities in the proof of Proposition~\ref{prop:L-Liouville-hermitian}
are useful to see $\mathcal L(\rho_\beta) = 0$.
This is a unique fixed point under certain natural ergodicity conditions~\cite{Spohn1977}.
Furthermore, we shall assume that $\|A_\alpha\|\le 1$ and
\begin{equation}
\label{eq:weak}
\max_{\alpha,\omega} h(\alpha,\omega) = O(1).
\end{equation}
Let us remark that the Davies weak coupling limit was adopted
as a  model of the thermal dynamics
in most of the previous works with a rigorous analysis of quantum self-correction;
see, for instance,~\cite{AlickiHorodeckiHorodeckiEtAl2010thermal,
AlickiFannesHorodecki2009thermalization,
ChesiLossBravyiEtAl2010Thermodynamic,
ChesiRoethlisbergerLoss2010Self-Correcting}.

The final state $\rho(t)$ generated by the Davies dynamics 
can be regarded as a corrupted version of the initial encoded state $\rho(0)$.
A decoder retrieves the encoded information from $\rho(t)$
by performing a syndrome measurement and an error correction.
A syndrome measurement involves a non-destructive eigenvalue measurement 
of all stabilizer generators $G^X_c$, $G^Z_c$.
The measured syndrome $S$ can be regarded 
as a classical bit string that assigns an eigenvalue $\pm 1$ to each generator. 
The error correction step is specified by an algorithm 
that takes as input the measured syndrome $S$ 
and returns a correcting Pauli operator $P_{ec}(S)$.
Let $\Pi_S$ be the projector onto the subspace with syndrome $S$.
The net action of the decoder on states can be described
by  a trace preserving completely positive (TPCP) linear map
\begin{equation}
\label{decoder}
\Phi_{ec}(\rho)=\sum_S P_{ec}(S)\Pi_S \, \rho \Pi_S P_{ec}(S)^\dag,
\end{equation}
where the sum runs over all possible syndromes.
We choose the renormalization group decoder
presented in Chapter~\ref{chap:RG-decoder}.

\section{Properties of the Lindbladian}

\begin{prop}[\cite{ChesiLossBravyiEtAl2010Thermodynamic}]
\label{prop:Lindblad}
Each operator $A_{\alpha,\omega}$ acts non-trivially only on $O(1)$ qubits.
Furthermore,
\[
 \| \mathcal L_+ \|_1 \equiv \sup_X \frac{ \| \mathcal L_+(X) \|_1 }{ \| X \|_1 }  = O(N).
\]
\end{prop}
\begin{proof}
Recall that each operator $A_\alpha$ acts on $O(1)$ qubits.
Since $H$ is a sum of pairwise commuting terms, we can represent
$A_\alpha(t)$ as  $A_\alpha(t)=e^{iH_\alpha t} A_\alpha e^{-iH_\alpha t}$
where $H_\alpha$ is obtained from $H$ by retaining only those stabilizer generators that do not
commute with $A_\alpha$.
(Compare it with Eq.~\eqref{eq:jump-ops}.)
All such generators must share at least one qubit with $A_\alpha$.
Therefore  $A_\alpha(t)$ and $A_{\alpha,\omega}$ act non-trivially only $O(1)$ qubits.
Furthermore, since $H_\alpha$ has $O(1)$ distinct integer eigenvalues,
$A_\alpha(t)$ has only $O(1)$ distinct Bohr frequencies;
the summation over $\omega$ in Eq.~\eqref{eq:jump-ops} is finite.
The bound $\| A_{\alpha,\omega} \|\le 1$ follows trivially from our assumption $\| A_\alpha\|\le 1$.
The norm of $\mathcal{L}_+$ is then
bounded from Eq.~\eqref{eq:lindbladian} using triangle inequality, $\| X Y \|_1 \le \| X \| \cdot \| Y \|_1$, and Eq.~\eqref{eq:weak}.
\end{proof}
\begin{prop}
\begin{equation}
\label{Acom}
H A_{\alpha,\omega} - A_{\alpha,\omega}H  = - \omega A_{\alpha,\omega} \quad \text{and} \quad
H A_{\alpha,\omega}^\dag  - A_{\alpha,\omega}^\dag H =  + \omega A_{\alpha,\omega}^\dag
\end{equation}
\end{prop}
\begin{proof}
Let us drop the subscript $\alpha$ momentarily for notational convenience.
The spectral component $A_\omega$ is given by a Fourier transformation
\[
 A_\omega = \int \frac{\dd t}{T} \ e^{i \omega t} \left( e^{i H t} A e^{-i H t} \right)
\]
where the integral is over a period $T$ such that $\frac{f T}{2\pi}$ is an integer for every Bohr frequency $f$ of $A(t)$.
Since there are only finitely many Bohr frequencies, such $T$ exists.
Let $\ket \psi$ be any energy eigenstate of energy $E$.
Then,
\begin{align*}
 H A_\omega \ket \psi 
&= \int \frac{\dd t}{T} \ e^{i \omega t} H e^{i H t} A e^{-i H t} \ket \psi \\
&= \int \frac{\dd t}{T} \ e^{i (\omega-E) t} H e^{i H t} A \ket \psi \\
&= \int \frac{\dd t}{T} \ e^{i (\omega-E) t} (-i)\frac{\partial}{\partial t} e^{i H t} A \ket \psi \\
&= \int \frac{\dd t}{T} \ (-i)\frac{\partial}{\partial t} e^{i (\omega-E) t}  e^{i H t} A \ket \psi 
      -(\omega-E) \int \frac{\dd t}{T} \ e^{i (\omega-E) t}  e^{i H t} A \ket \psi  \\
&= \frac{-i}{T} \left[ e^{i \omega t} A(t) \right]_0^T 
   + (E-\omega) \int \frac{\dd t}{T} \ e^{i \omega t} e^{i H t} A e^{-i H t} \ket \psi \\
&= (E-\omega) A_\omega \ket \psi \\
&= -\omega A_\omega \ket \psi + A_\omega H \ket \psi.
\end{align*}
Since $\ket \psi$ was an arbitrary energy eigenstate, we have Eq.~\eqref{Acom}.
\end{proof}
\begin{prop}
\label{prop:com}
Let $\hat{H}$ be the linear map defined by $\hat{H}(X)=HX-XH$.
The map $\hat{H}$ commutes with the Lindblad generator $\mathcal L$.
Therefore, $\hat H (\rho(t)) = 0$, and $\rho(t)=e^{\mathcal L t}(\rho(0))$.
\end{prop}
\begin{proof}
It is elementary that $\hat H(XY) = \hat H(X) Y + X \hat H (Y)$.
Eq.~\eqref{Acom} reads $\hat H (A_{\alpha,\omega}) = -\omega A_{\alpha,\omega}$
and $\hat H (A_{\alpha,\omega}^\dagger) = +\omega A_{\alpha,\omega}^\dagger$.
Hence, $\hat H (A_{\alpha,\omega} A_{\alpha,\omega}^\dagger) = 0$.
Clearly,
\begin{align*}
\hat H \mathcal L (\rho) 
&= \hat H h \left( a \rho a^\dagger - \half \rho a^\dagger a - \half a^\dagger a \rho \right) \\
&= -h\omega a \rho a^\dagger + ha \hat H (\rho) a^\dagger + h\omega a \rho a^\dagger 
   - \half h \hat H (\rho) a^\dagger a 
   - \half h a^\dagger a \hat H(\rho) \\
&= \mathcal L \hat H (\rho)
\end{align*}
where $a = A_{\alpha,\omega}$, $h=h(\alpha,\omega)$ with the sum over $\alpha$ and $\omega$ understood,
and $\rho$ is arbitrary.
It follows that
\[
\rho(t)=e^{-i\hat{H} t + \mathcal L t}(\rho(0)) =e^{\mathcal L t} \circ  e^{-i\hat{H} t} (\rho(0))= e^{\mathcal L t}(\rho(0)),
\]
since $\hat H(\rho(0)) = [H,\rho(0)]=0$. It implies that $\hat H(\rho(t)) = 0$.
\end{proof}
\begin{prop}[\cite{AlickiHorodeckiHorodeckiEtAl2010thermal}]
Let $\mathcal L^*$ be the adjoint linear map of $\mathcal L$
(the one describing time evolution in the Heisenberg picture)
with respect to Hilbert-Schmidt inner product.
Define Liouville inner product by $\langle X,Y \rangle_\beta \equiv \trace{\rho_\beta X^\dag Y}$.
Then, $\mathcal L^*$ is self-adjoint with respect to Liouville inner product; 
$\langle X, \mathcal L^*(Y) \rangle_\beta = \langle \mathcal L^*(X), Y \rangle_\beta$.
\label{prop:L-Liouville-hermitian}
\end{prop}
\begin{proof}
We suppress the index $\alpha$ for notational convenience.
Using $e^H X e^{-H} = e^{\hat H} (X)$ and $\hat H (A_\omega) = - \omega A_\omega$, we have 
$e^{\beta H} A_\omega e^{-\beta H} = e^{\beta \hat H}(A_\omega) = e^{-\beta \omega} A_\omega$,
or $A_\omega \rho_\beta = e^{-\beta \omega} \rho_\beta A_\omega$. 
Hence, $\rho_\beta (A_\omega^\dagger A_\omega) = (A_\omega^\dagger A_\omega) \rho_\beta$.
Now,
\begin{align*}
\langle \mathcal L ^* X, Y \rangle - \langle X, \mathcal L^* Y \rangle 
&= \sum_\omega h(\omega) \trace\left( 
A_\omega^\dagger X^\dagger A_\omega Y \rho_\beta 
- \half X^\dagger A_\omega^\dagger A_\omega  Y \rho_\beta 
- \half A_\omega^\dagger A_\omega X^\dagger Y \rho_\beta
\right. \\
& \left. \quad \quad \quad \quad \quad \quad
- X^\dagger A_\omega^\dagger Y A_\omega \rho_\beta 
+ \half X^\dagger Y A_\omega^\dagger A_\omega  \rho_\beta 
+ \half X^\dagger A_\omega^\dagger A_\omega Y \rho_\beta
\right)\\
&= \sum_\omega h(\omega) \trace \left( A_\omega^\dagger X^\dagger A_\omega Y \rho_\beta 
- X^\dagger A_\omega^\dagger Y A_\omega \rho_\beta  \right) \\
&= \sum_\omega h(\omega) \trace \left( A_\omega^\dagger X^\dagger A_\omega Y \rho_\beta 
- X^\dagger A_\omega^\dagger Y  \rho_\beta A_\omega e^{-\beta \omega} \right)
\end{align*}
Applying the detailed balance condition $e^{-\beta \omega} h(\omega) = h(-\omega)$ of Eq.~\eqref{DB},
we see that
\begin{align*}
&=\sum_\omega h(\omega) \trace \left( A_\omega^\dagger X^\dagger A_\omega Y \rho_\beta \right)
- \sum_\omega h(-\omega) \trace \left( X^\dagger A_\omega^\dagger Y  \rho_\beta A_\omega \right) .
\end{align*}
Since $A_\omega$ is a Fourier transform of a Hermitian operator $A(t)$ in Eq.~\eqref{eq:jump-ops},
we apply $A_\omega^\dagger = A_{-\omega}$ to conclude
\begin{align*}
&=\sum_\omega h(\omega) \trace \left( A_\omega^\dagger X^\dagger A_\omega Y \rho_\beta \right)
- \sum_\omega h(-\omega) \trace \left( X^\dagger A_{-\omega} Y  \rho_\beta A_{-\omega}^\dagger \right)\\
& = 0 .
\end{align*}
The last equality is because $\omega$ is a dummy variable.
\end{proof}

\section{Analysis of thermal errors}
\label{sec:thermal}

We now analyze the relation between $\mathcal L$ and $\Phi_{ec}$.
Recall the following terminology and notations.
A ground state of the memory Hamiltonian will be referred to as a {\em vacuum}.
It will be convenient to perform an overall energy shift such that the vacuum has \emph{zero energy}.
A {\em Pauli operator} is an arbitrary tensor product of single-qubit Pauli operators
$X,Y,Z$ and the identity operators $I$. We will say that a Pauli operator creates
$m$ {\em defects} iff $P$ anticommutes with exactly $m$ stabilizer generators $G^X_c, G^Z_c$.
Equivalently, applying $P$ to the vacuum one obtains an eigenvector of $H$ with energy $m$.
For example, using the explicit form of the generators, see Figure~\ref{fig:CubicCode},
one can check that single-qubit $X$ or $Z$ errors create $4$ defects, while $Y$ errors create $8$ defects.
See Figure~\ref{fig:elementary-syndrome}.
We will say that a Pauli error $P$ is corrected by the decoder if $P_{ec}(S(P))=\pm PG$,
where $S(P)$ is the syndrome of $P$ and $G$ is a product of stabilizer generators.
Let $N=2L^3$ be the total number of qubits.
Note that $N$ is also the number of stabilizer generators for the cubic code.

Let $\Gamma=(P_0,P_1,\ldots,P_t)$ be a finite sequence of Pauli operators
such that  the operators $P_i$ and $P_{i+1}$ differ on at most one qubit
for all $0\le i <t$. 
We say that $\Gamma$ is an {\bf error path}
implementing a Pauli operator $P$ when $P_0=I$ and $P_t=P$.
Let $m_i$ be the number of defects created by $P_i$.
The maximum number of defects
\[
m(\Gamma)=\max_{0\le i\le t} \; m_i
\]
will be called an {\bf energy cost} of the error path $\Gamma$.
Given a Pauli operator $P$, we define its {\bf energy barrier} $\Delta(P)$
as the minimum energy cost of all error paths implementing $P$,
\[
\Delta(P)=\min_\Gamma \; m(\Gamma).
\]
Although the set of error paths is infinite,
the minimum always exists because the energy cost is a nonnegative integer.
In fact, it suffices to consider paths in which $P_i$ are all distinct.
The number of such paths is finite
since there are only finitely many Pauli operators for a given system size.

It is worth emphasizing that an operator $P$ may have a very large energy barrier even though
$P$ itself creates only a few defects or no defects at all.
Consider as an example the 2D Ising model, 
$H=- \half \sum_{\langle u v \rangle} Z_u Z_v$, where the sum runs
over all pairs of nearest neighbor sites on the square lattice 
of size $L\times L$ with open boundary conditions.
Then the logical-$X$ operator $P=\bigotimes_u X_u$ has an energy barrier $\Delta(P)=L$ 
since any sequence of bit-flips implementing $P$ must create a domain wall across the lattice.
It is clear that a Pauli operator $P$ acting nontrivially on $n$ qubits has energy barrier at most $O(n)$.

A naive intuition suggests that a stabilizer code Hamiltonian 
is a good candidate for being a self-correcting memory
if there exists an error correction algorithm, or a decoder,
that corrects all errors with a sufficiently small energy barrier.
Errors with a high energy barrier can confuse the decoder
and cause it to make wrong decisions,
but we expect that such errors are unlikely to be created by the thermal noise.
We make this intuition more rigorous.

Let $f$ be the maximum energy barrier of Pauli operators that appear in the expansion of
the quantum jump operators $A_{\alpha,\omega}$ or $A_{\alpha,\omega}^\dag A_{\alpha,\omega}$
of Eq.~\eqref{eq:lindbladian}.
Since $A_{\alpha,\omega}$ act on a constant number of qubits 
by Proposition~\ref{prop:Lindblad}, we have $f=O(1)$.
Below $m$ is an arbitrary energy cutoff.
Let $\mathcal D = \ker \hat H$ be the set of all operators that are commuting with the Hamiltonian $H$.
Every operator in $\mathcal D$ is block-diagonal in the energy eigenstate basis.
Since $\rho(0)$ is supported on the ground subspace of $H$, we have $\rho(0) \in \mathcal D$.
Below we only consider states from $\mathcal D$ and linear maps preserving $\mathcal D$.
Define an orthogonal  identity decomposition
\[
I=\Pi_- + \Pi_+
\]
where $\Pi_-$ projects onto the subspace with energy $<m+f$
and $\Pi_+$ projects onto the subspace with energy $\ge m+f$.
Introduce auxiliary Lindblad generators
\begin{equation}
\mathcal L_-( \rho)=\sum_\alpha \sum_\omega h(\alpha,\omega) \left(
B_{\alpha,\omega} \rho B_{\alpha,\omega}^\dag - \frac12 \{  \rho,B_{\alpha,\omega}^\dag B_{\alpha,\omega} \} \right),
\quad \mbox{where} \quad B_{\alpha,\omega} = \Pi_- A_{\alpha,\omega}
\label{Lp}
\end{equation}
and
\begin{equation}
\mathcal L_+( \rho)=\sum_\alpha \sum_\omega h(\alpha,\omega) \left(
C_{\alpha,\omega} \rho\,  C_{\alpha,\omega}^\dag - \frac12 \{  \rho,C_{\alpha,\omega}^\dag C_{\alpha,\omega} \} \right),
\quad \mbox{where} \quad C_{\alpha,\omega} = \Pi_+ A_{\alpha,\omega}.
\label{Lq}
\end{equation}
Simple algebra shows that $\mathcal L_-$ and $\mathcal L_+$ preserve $\mathcal D$ and
their \emph{restrictions on $\mathcal D$} satisfy
\begin{equation}
\label{LPQ}
\mathcal L=\mathcal L_- + \mathcal L_+.
\end{equation}
It is useful to note that any $X \in \mathcal D$ commutes with $\Pi_\pm$.
By abuse of notations, we shall apply Eq.\eqref{LPQ} as though it holds for all operators.

\begin{lem}
\label{lem:Lp}
Suppose that an error correction algorithm $s \mapsto P_{ec}(s)$ corrects
any Pauli error $P$ with the energy barrier smaller than $m+2f$.
Let $\Phi_{ec}$ be the corresponding decoder defined by Eq.\eqref{decoder}.
For any time $t\ge 0$ and for any state $\rho_0$ supported on the ground subspace of $H$ one has
\begin{equation*}
\Phi_{ec}(e^{\mathcal L_- t}(\rho_0))=\rho_0.
\end{equation*}
\end{lem}
\begin{proof}
Since all maps are linear, we may assume $\rho_0 =\ket{g}\bra{g}$ is a pure state.
Then, $e^{\mathcal L_- t}(\rho_0)$ is in the span of states of form $\ket{\psi}=\Pi_- E_n \cdots \Pi_- E_2 \Pi_- E_1 \ket{g}$,
where $E_i$ are Pauli operators that appears in the expansion of $A_{\alpha,\omega}$
or $A_{\alpha,\omega}^\dag A_{\alpha,\omega}$.
This follows from the Taylor expansion of $e^{\mathcal{L}_- t}$.
Since Pauli errors map eigenvectors of $H$ to eigenvectors of $H$, we conclude that
either $\ket{\psi}=0$, or $\ket{\psi}=E_n \cdots E_2 E_1 \ket g$.
Furthermore, the latter case is possible only if the Pauli operator $E\equiv E_n \cdots E_2 E_1$
has energy barrier smaller than $m+2f$. Indeed, definition of $\Pi_-$ implies that $E_j E_{j-1} \cdots E_1$
creates at most $m+f-1$ defects for all $j=1,\ldots,n$. By assumption, each operator $E_j$ can be implemented by an error
path with energy cost at most $f$. Taking the composition of all such error paths
one obtains an error path for $E$ with energy cost at most $m+2f-1$
and thus  $\Phi_{ec}$ will correct $E$.
Since $\Phi_{ec} \  \circ \  e^{\mathcal L_- t}$ is a TPCP map, we must have $\Phi_{ec}(e^{\mathcal L_- t}(\rho_0))=\rho_0$.
\end{proof}

For any decomposition $\mathcal L=\mathcal L_- + \mathcal L_+$ one has the following identity:
\begin{equation}
e^{\mathcal L t} = e^{\mathcal L_- t} + \int_0^t ds \, e^{\mathcal L_- (t-s)} \mathcal L_+ \, e^{\mathcal L s},
\label{eq:exp}
\end{equation}
which follows from the identity
\[
\frac{d}{ds} e^{\mathcal L_-(t-s)} e^{\mathcal L s} = e^{\mathcal L_- (t-s)} (-\mathcal L_- + \mathcal L) e^{\mathcal L s} =  e^{\mathcal L_- (t-s)} \mathcal L_+ \, e^{\mathcal L s}.
\]

\begin{lem}
\label{lem:memory-time}
Assume the supposition of Lemma~\ref{lem:Lp}.
Let $Q_m$ be the projector onto the (high energy) subspace with at least $m$ defects. Then
\begin{equation}
\label{eq:infidelity}
\epsilon(t)\equiv \| \Phi_{ec}(\rho(t))-\rho(0) \|_1 \le
O(tN) \mathrm{Tr} \, Q_m e^{-\beta H}
\end{equation}
for any initial state $\rho(0)$ supported on the ground subspace of $H$.
The time evolution of $\rho(t)$ is governed by the Lindblad equation, Eq.~\eqref{eq:Markovian-master},
with the inverse temperature $\beta$ of the bath.
\end{lem}
\begin{proof}
Write $\rho_0 \equiv \rho(0)$.
First, Proposition~\ref{prop:com} says $\rho(t) = e^{\mathcal L t} (\rho_0)$.
Applying Lemma~\ref{lem:Lp} and Eq.~\eqref{eq:exp} one arrives at
\begin{equation}
\epsilon(t)\le \int_0^t ds\, \| e^{\mathcal L_-(t-s)}  \mathcal L_+ \, e^{\mathcal L s} (\rho_0) \|_1
\le t \cdot \max_{0\le s\le t} \; \|  \mathcal L_+ \, e^{\mathcal L s} (\rho_0) \|_1.
\end{equation}
We shall use an identity
\begin{equation}
\mathcal L_+(X)=\mathcal L_+( Q_m X Q_m)
\end{equation}
valid for any $X \in \mathcal D$. Indeed,
any quantum jump operator $A_{\alpha,\omega}$ changes the energy at most by $f$,
so that $\mathcal L_+(Q_m^\perp X)=\mathcal L_+(XQ_m^\perp)=0$  for any $ X \in \mathcal D$ (note that any $ X \in \mathcal D$
commutes with $Q_m$).
We arrive at
\begin{equation}
\epsilon(t)
\le t \cdot \| \mathcal L_+(Q_m e^{\mathcal L s} (\rho_0) Q_m) \|_1
\le t \cdot \| \mathcal L_+ \|_1 \cdot \| Q_m e^{\mathcal L s} (\rho_0)  Q_m \|_1
\le O(tN) \trace{ Q_m  e^{\mathcal L s} (\rho_0)},
\label{epsilon_bound}
\end{equation}
where the maximization over $s$ is implicit.
We used Proposition~\ref{prop:Lindblad} and the positivity of $\mathcal L(\rho_0)$.
in the last inequality.
Since the ground-state energy of $H$ is zero, one has
\begin{equation}
\rho_0 = \mathcal Z_\beta \rho_\beta \rho_0,
\end{equation}
where $\mathcal Z_\beta$ is the partition function. It yields
\begin{equation}
\trace{ Q_m  e^{\mathcal L s} (\rho_0)}
= \trace{  \rho_0 \, e^{\mathcal L^* s}(Q_m)}
= \mathcal Z_\beta \trace{ \rho_\beta \rho_0 e^{\mathcal L^* s} (Q_m)}
= \mathcal Z_\beta \langle \rho_0, e^{\mathcal L^* s} (Q_m) \rangle_\beta,
\end{equation}
Proposition~\ref{prop:L-Liouville-hermitian} implies 
that the map $e^{\mathcal L^* s}$ is also self-adjoint with respect to the Liouville inner product.
Hence, we have
\begin{equation}
\trace{ Q_m  e^{\mathcal L s} (\rho_0)}
=   \mathcal Z_\beta \langle \rho_0, e^{\mathcal L^* s} (Q_m) \rangle_\beta
=   \mathcal Z_\beta \langle e^{\mathcal L^* s} (\rho_0), Q_m \rangle_\beta
=   \trace{e^{\mathcal L^* s}(\rho_0) Q_m e^{-\beta H}}
\le \trace{Q_m e^{-\beta H}},
\label{trace_final}
\end{equation}
where the last inequality is because
$e^{\mathcal L^* s}$ is a unital completely positive map
and $\rho_0 \le I$.
\end{proof}

\section{Correctability of errors with an energy barrier}
\label{sec:Correctability}

Let $P$ be an unknown Pauli error.
Suppose we are promised that $P$ has a sufficiently small
energy barrier, namely, $\Delta(P) \le c \log{L}$, for some constant $c$ that will be chosen later.
In this section we prove that any such error $P$ will be corrected by the RG decoder.

Assume throughout this section that a family of topological stabilizer codes $\{ {\cal C}_L\}_L$
defined in Section~\ref{sec:log-barrier} with $\ltqo \ge L^\gamma$
obey the no-strings rule with some constant $\alpha$ as in Definition~\ref{defn:no-strings}.
The $\ell_\infty$-metric $d$ is used; $d((x,y,z),(x',y',z'))= \max \{ |x-x'|,|y-y'|,|z-z'| \}$.
We use 
\[ 
\xi(p)= (10\alpha)^p 
\]
for notational convenience as in the previous chapter,
as well as the notion of \emph{level-$p$ sparseness} of Definition~\ref{defn:psparse}
and \emph{level-$p$ syndrome history}.
When $P_0,P_1,\ldots,P_t$ form an error path,
the corresponding sequence of syndromes $S(j)$ caused by $P_j$ acting on a vacuum
 is called a level-$0$ syndrome history.
A level-$p$ syndrome history, inductively defined, 
discards all level-$(p-1)$ sparse syndromes from level-$(p-1)$ syndrome history,
but keeps the initial and final syndromes.

Let $\Gamma=(P_0,P_1,\ldots,P_t)$ be an error path implementing $P$
with the energy cost $m(\Gamma)=m$.
Here $P_0=I$, $P_t=P$, while $E_j\equiv P_j P_{j-1}$ are single-qubit Pauli operators
for all $j$.
Lemma~\ref{lemma:counting} implies that
there is a level $p_{max} < m$ such that
in the level-$p_{max}$ syndrome history
only the initial empty syndrome $S(0)=0$ and the final syndrome $S(t)=S$
are possibly non-sparse.

\begin{lem}
\label{lemma:neutral-components}
Let $P$ be any Pauli error,
$S=S(P)$ be its syndrome,
and $m=\Delta(P)$ be its energy barrier.
Suppose
\begin{equation}
\label{msmall}
16m(10\alpha)^m < \ltqo.
\end{equation}
Then, there exists a stabilizer $G \in {\cal G}$ such that $P \cdot G$
has support on the $\xi(m)$-neighborhood of $S$.
Any $R$-connected component of $S$ is neutral for $2\xi(m) < R \le 4 \xi(m)$.
\end{lem}
\begin{proof}
Let us apply Lemma~\ref{lemma:RG2} to the level $p_{max}$,
the smallest integer such that
the level-$p_{max}$ syndrome history has the initial and final syndrome,
and level-$p_{max}$-sparse syndromes.
Since $p_{max} < m$, the condition $16 m \xi(p_{max}) < \ltqo$ 
in Lemma~\ref{lemma:RG2} is satisfied.
We have $S'=0$, $S''=S$, and $E=P$.
Hence, there exists a stabilizer $G \in {\cal G}$ such that
$P\cdot G$ is supported on the $\xi(p_{max})$-neighborhood of $S$.
It proves the first statement of the lemma.

Let $r = \xi(m) = (10\alpha)^m$.
Choose any $R$ such that $2r < R \le 4r$
and let $C_a$ be any $R$-connected component of $S$.
Since $C_a$ contains at most $m$ defects, the diameter of $C_a$ is at most $mR$.
Restricting $P \cdot G$ on the $r$-neighborhood of $C_a$,
we obtain a Pauli operator $P_a$
supported on a cube of linear size at most $mR+2r \le 4rm+2r <\ltqo$ by assumption.
Furthermore, the support of $P_a$ is separated from $(P_a)^{-1} (P \cdot G)$
by distance at least $R - 2r > 0$.
Hence, $P_a$ creates the cluster $C_a$ from the vacuum.
Therefore, $C_a$ is neutral.
\end{proof}

We wish to have a well-separated cluster decomposition.
\begin{lem}
\label{lemma:wellseparated-decomposition}
Let $S$ be any cluster of $m>0$ defects.
For any integer $\mu \ge 1$, there exists a nonnegative integer $ p < m$ 
and a decomposition 
\[
S = C_1 \cup \cdots \cup C_n \ \text{ such that }\ 
d(C_a) \le 4^p \mu \ \text{ and }\  d(C_a,C_b) > \half \cdot 4^{p+1} \mu \ \text{ for } a \neq b.
\quad \quad \quad \quad \mathrm{(*)}
\]
\end{lem}
\begin{proof}
The only nontrivial part is that $p$ can be chosen as $p < m$.
Let us say that a partition of $S$ into clusters is \emph{$p$-good} if it satisfies $\mathrm{(*)}$.
By grouping all defects occupying the same elementary cube into a cluster,
one obtains a partition $S=C_1\cup \ldots \cup C_g$. Obviously, $g\le m$, and $d(C_a) \le \mu$.
If this partition is not $0$-good,
then $g\ge 2$ and there is a pair, say, $C_1,C_2$ such that $d(C_1, C_2) \le 2\mu$.
Merging $C_1$ and $C_2$ into a single cluster $C'_2$,
one obtains a partition $S=C'_2 \cup C_3 \cup \ldots \cup C_g$ where $d(C'_2) \le 4\mu$.
If this partition is not $1$-good, then $g\ge 3$ and one can repeat the merging again.
After at most $g-1$ iterations, one arrives at a good partition.
\end{proof}
Note that the minimal enclosing boxes of distinct cluster do not overlap, since
\[
 d(b(C_a),b(C_b)) > 2 \cdot 4^p \mu - 4^p \mu - 4^p \mu = 0 .
\]
The following is the desired property of the RG decoder.
\begin{lem}
\label{lemma:RGdecoder-capability}
Let $P$ be any Pauli error with energy barrier $m=\Delta(P)$. Suppose
\[
 (160\alpha)^m < \ltqo .
\]
Then calling the RG decoder on the syndrome $S(P)$ returns a correcting operator
$P_{ec}$ such that $P P_{ec}$ is a stabilizer.
Thus, the RG decoder corrects $P$ if $\Delta(P) < \frac{\gamma}{\log (160 \alpha)} \log L$.
\end{lem}
\begin{proof}
Let $S=S(P)$ be the syndrome.
Let $p$ be the integer such that $2 \xi(m) < 2^p \le 4 \xi(m)$.
Setting $\mu=2^p$ in Lemma~\ref{lemma:wellseparated-decomposition},
$S$ is decomposed into $S=C_1 \cup \ldots \cup C_n$
such that $d(C_a) \le 2^{p'}$ and $ d(C_a,C_b) > 2^{p'+1}$ for all $a \ne b$,
where $p'$ is an integer such that $p \le p' < 2m+p$.
Since $16m(10\alpha)^m \le (160\alpha)^m$,
Lemma~\ref{lemma:neutral-components} implies that
each $C_a$ is neutral for being a disjoint union of neutral $2^p$-connected components.
The RG subroutines EC($s$) with $s=0,1,\ldots,p-1$,
can only annihilate some neutral $2^s$-connected components of $C_a$,
which does not alter the neutrality of $C_a$.
Therefore, the RG decoder from level-$0$ to $p$
will annihilate each cluster $C_a$, and hence $S$ at last.

We need to show that $P \cdot P_{ec}$ is a stabilizer,
where $P_{ec}$ is the returned correcting operator.
Let $B_a$ be the $(10\alpha)^m$-neighborhood of $b(C_a)$.
Our assumptions imply that $B_a$ has diameter smaller than $\ltqo$
and distinct $B_a$'s do not intersect.
By construction, the operators $P_{ec}$ and $P \cdot G$
have support in the union $B_1 \cup \cdots \cup B_n$.
Therefore, $P \cdot G \cdot P_{ec} = Q_1 \cdots Q_n$,
where $Q_a$ has support on $B_a$ and has trivial syndrome.
Topological order condition implies that $Q_a$ are stabilizers, so is the product.
\end{proof}

The full hierarchy of the RG decoder
is not necessary to correct the error with the low energy barrier.
A single level-$p$ error correction with $p$ proportional to $\log \ltqo$, will be sufficient.
We nevertheless include the hierarchy
since in practice it corrects errors with slightly higher
(although only by a constant factor) energy barrier
at a marginal slowdown of the decoder.
If one wishes to apply the decoder against random errors,
the hierarchy becomes necessary,
as we have discussed in Section~\ref{sec:threshold}.

\begin{rem}
\label{rem:TestNeutralprime}
A closer analysis reveals a simplification of TestNeutral defined in Section~\ref{subs:rgdecoder} for the 3D Cubic Code.
We defined TestNeutral to return the identity operator if a cluster turns out to be charged.
The modified TestNeutral$'$ just applies the broom algorithm and returns recorded operator in any case.
It gives the {\em same} characteristic as stated in the Lemma~\ref{lemma:RGdecoder-capability}.
EC$'(p)$ using TestNeutral$'$ will transform a charged cluster $C_a$ to a {\em different} cluster $C'_a$,
but $C'_a$ is still contained in $b(C_a)$.
Due to Lemma~\ref{lemma:wellseparated-decomposition}, $b(C_a)$ do not overlap
at a high level $p$, and EC$'(p)$ will eliminate neutral clusters at last.
This specialized version of RG decoder is used in our numerical simulation 
in Section~\ref{sec:numerics}.
\end{rem}

\section{A lower bound on memory time}
\label{sec:memory-time}
  
\begin{theorem}
\label{thm:memorytime}
For a code Hamiltonian in $D$ spatial dimensions
that satisfies TQO1, TQO2, and no-strings rule,
there exists a decoder $\Phi_{ec}$ and a constant $c,c'>0$
such that for any inverse temperature $\beta>0$,
any state $\rho(0)$ supported on the ground subspace of $H$,
and any evolution time $t \ge 0$ one has
\begin{equation}
\label{error-bound}
\epsilon(t)\equiv \| \rho(0) - \Phi_{ec}(\rho(t)) \|_1 \le O(t) \cdot 2^{k(L)} \cdot L^{D-c\beta}
\end{equation}
as long as $L\le e^{c' \beta}$.
The error correction algorithm used by the decoder
has running time $poly(L)$.
\end{theorem}
\begin{proof}
We use the renormalization group decoder of Chapter~\ref{chap:RG-decoder}.
Its running time is $poly(L)$.
Lemma~\ref{lemma:RGdecoder-capability} guarantees
that the supposition of Lemma~\ref{lem:memory-time} is satisfied
with $m = \Omega(\log L)$.
It remains to bound $\trace Q_m e^{-\beta H}$ from above, 
where $H$ has ground energy equal to $0$ by convention
and $Q_m$ is the projector onto the space of energy $m$ or higher.
The number of states of exactly $n$ defects,
is $2^{k(L)}$ times the number of configurations of $n$ defects,
the latter of which is upper bounded by the binomial coefficient $\binom{rL^D}{n}$
where $r$ is the maximum number of terms in the Hamiltonian that acts on an elementary cube.
(In case of the cubic code, $r=2$.)
Therefore,
\begin{align}
\epsilon(t)
&\le O(t L^D) 2^k \sum_{n \ge m} \binom{rL^D}{n} e^{-\beta n} \nonumber \\
&= O(t L^D) 2^k e^{-\beta m/2} \sum_{n \ge m} \binom{rL^D}{n} e^{-\beta n + \beta m /2 } \nonumber \\
&\le O(t L^D) 2^k e^{-\beta m/2} \sum_{n \ge m} \binom{rL^D}{n} e^{- \beta n / 2 } \nonumber \\
&\le O(t L^D) 2^k e^{-\beta m/2} \left( 1 + e^{-\beta/2} \right)^{rL^D}
\label{eq:tradeoff}
\end{align}
As long as $rL^D \le e^{\beta/2}$ the entropy contribution is bounded by a constant.
Since $m = \Omega(L)$, the proof is complete.
\end{proof}

Clearly, our bound is most useful when  $k(L)$ is small.
In the following we shall mostly be interested in the
smallest ground-state degeneracy, $k(L)=2$. 
This happens for any odd $3\le L\le 200$ such that $L$ is not a multiple of $15$ or $63$.
In fact, there exists an infinite sequence of lattice sizes
such that $k(L)=2$, for example, $k(2^p+1)=2$ for all $p\ge 1$
by Corollary~\ref{cor:cubic-code-degeneracy-formula}.

On the other hand, this requirement seems unnecessary; we believe that the $2^k$ factor should be removable.
A very rough argument is as follows.
Once we fix a ground state at $t=0$, the interaction with the thermal bath turns the system into an ensemble, 
which is likely to be supported on the states above the fixed ground state.
Eq.~\ref{error-bound} can be viewed as a probability estimation
that the system has visited a high energy state above the logarithmic cutoff,
below which topological sectors are unambiguously told.
Since quantum tunneling to another sector is highly suppressed by the macroscopic code distance,
the system would not be aware of the other topological sectors.

The upper bound on the storage error can be easily translated to a lower bound on the memory time.
Indeed, if one is willing to tolerate a fixed storage error $\epsilon$, say $\epsilon=0.01$,
the memory time $T_{mem}$ can be defined as the smallest $t \ge 0$ such that $\epsilon(t)\ge \epsilon$.
Assuming that the lattice size is chosen such that $k(L)=2$, Theorem~\ref{thm:memorytime} implies that
\begin{equation}
\label{Tmem-bound}
T_{mem} \ge L^{c\beta-3} \quad \text{ if } L \le e^{c'\beta}.
\end{equation}
Here we neglected the overall constant coefficient.
It shows that for low temperatures, $\beta \gg 1$,  and sufficiently small system size,
$L\ll e^{c' \beta}$,  the memory time grows with $L$ according to a power law
whose exponent is proportional to $\beta$.
To the best of our knowledge, this provides the first realistic example 
of a topological memory with a self-correcting behavior. 
Unfortunately, the bound is not conclusive in the thermodynamic limit.
At the optimally chosen lattice size,
the maximum memory time $T_{mem}(\beta)$ achievable at a given temperature $\beta$
is easily found from Eq.~\eqref{Tmem-bound}:
\begin{equation}
\label{Tmem-bound1}
T_{mem}(\beta) \ge e^{cc' \beta^2}
\end{equation}
for $\beta \gg 1$. For comparison, the memory time of the 2D toric code model grows
only exponentially with $\beta$~\cite{AlickiFannesHorodecki2009thermalization, ChesiRoethlisbergerLoss2010Self-Correcting}.
Depending on the value of the constant $cc'$ and the temperatures realizable in experiments,
the scaling Eq.~\eqref{Tmem-bound1} may be favorable enough to achieve macroscopic memory times.

The restriction $L \ll e^{c' \beta}$ implies that the average number of defects (flipped stabilizers)
in the equilibrium Gibbs state $\rho_\beta \sim e^{-\beta H}$
is small, that is, the Gibbs state has most of its weight on the ground subspace of $H$.
This might suggest that the thermal noise is irrelevant in the studied regime.
However, this is not the case.
If the evolution time is large enough, so that $\rho(t) \approx \rho_\beta$, 
the encoded information cannot be retrieved from $\rho(t)$,
since $\rho_\beta$ does not depend on the initial state.
If a time $t\sim e^{\Omega(\beta^2)}$ has elapsed, 
the system would have accommodated approximately $tL^3 e^{-\beta}\sim  e^{\Omega(\beta^2)}$ defects during the evolution.
This implies in particular that the system has endured $e^{\Omega(\beta^2)}$ errors which becomes significant for low temperatures.

\section{Numerical simulation}
\label{sec:numerics}

\begin{figure}[p]
\centering
 \includegraphics[width=.81\textwidth]{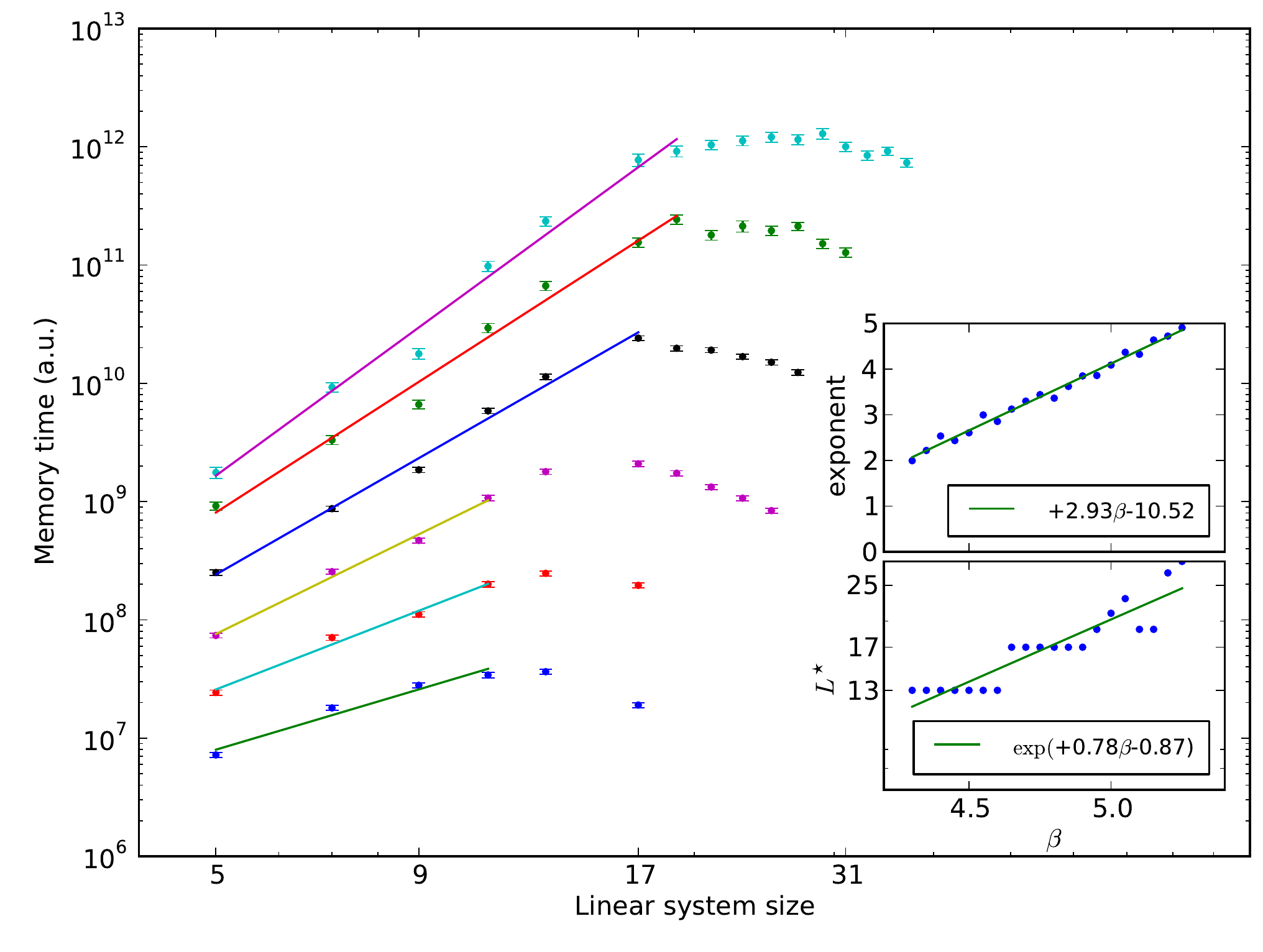}
\caption{The memory time $T_{mem}$ vs. the system size $L$. In the upper inset is shown the exponent of the power law fit of $T_{mem}$ for the first a few system sizes. It is clear that $T_{mem} \propto L^{2.93 \beta -10.5}$ when $L < L^\star$, where $L^\star$ is the optimal system size where $T_{mem}$ reaches maximum. The data for $\beta = 4.3, 4.5, 4.7, 4.9, 5.1, 5.25$ are shown.}
\label{fig:TvsL}
\end{figure}

\begin{figure}[p]
\centering
 \includegraphics[width=.81\textwidth]{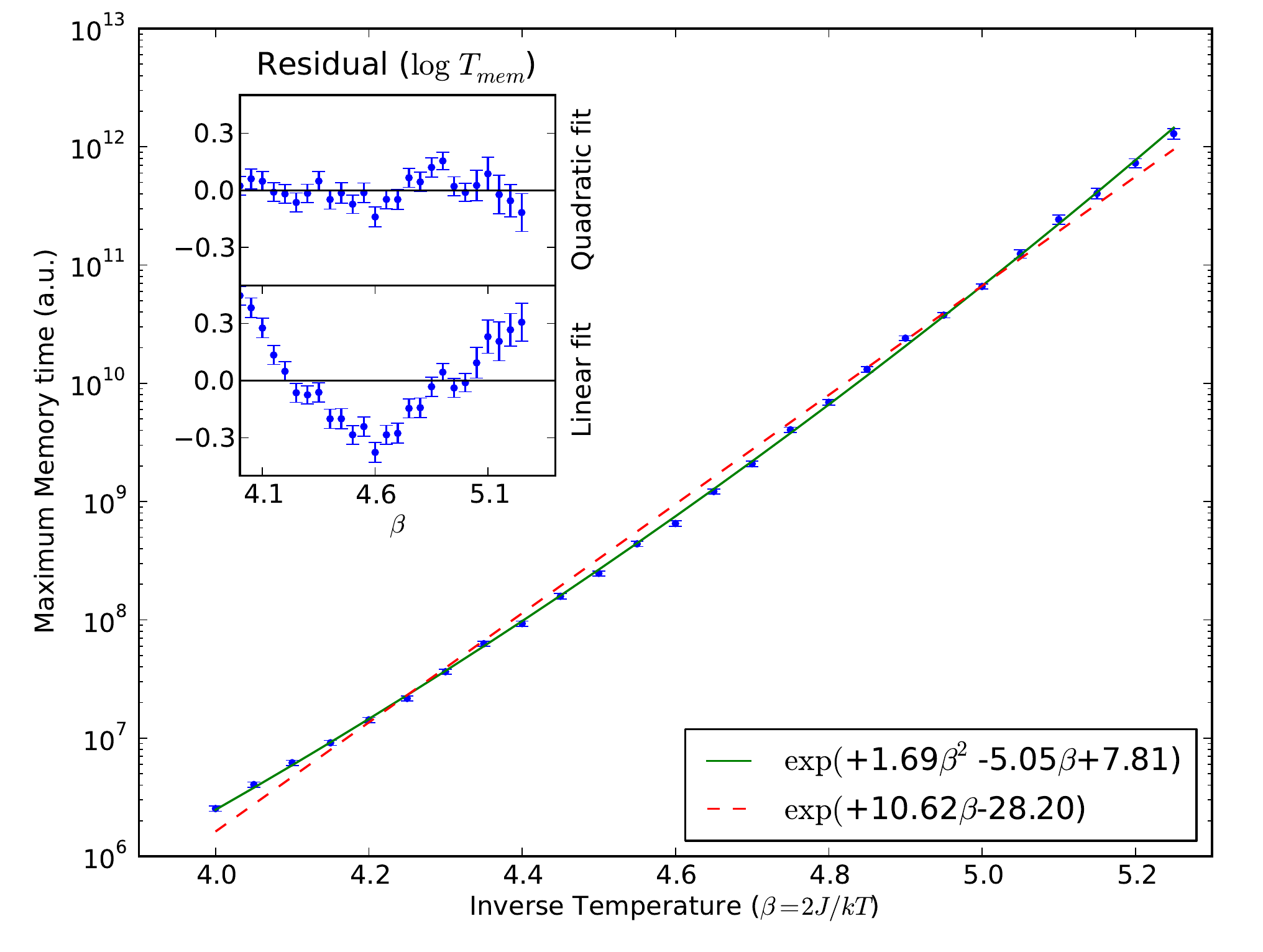}
\caption{The maximum memory time $T_{mem}$ vs. the inverse temperature $\beta$. The memory time is maximized with respect to the system size. The logarithm of $T_{mem}$ clearly follows a quadratic relation with $\beta$ as opposed to a linear one.}
\label{fig:TvsBeta}
\end{figure}

Since Theorem~\ref{thm:memorytime} only provides a lower bound on the memory time, 
a natural question is whether this bound is tight and, 
if so, what is the exact value of the constant coefficient $c$?
To answer this question, the memory time of the 3D cubic code has been computed numerically
for a range of $\beta$'s and $L$'s.
It should be emphasized that both Theorem~\ref{thm:memorytime} and our numerical simulation
use the \emph{same} decoder at the read-out step.
The numerical results strongly suggest that
our analytical bound is tight up to constant coefficients for our renormalization group decoder.
See Figure~\ref{fig:TvsL},\ref{fig:TvsBeta}.
It suggests that $T_{mem}\approx L^{2.93 \beta-10.52}$
as long as $L\le L^*\approx e^{0.78\beta - 0.87}$. 
The number $0.78$ in the estimate of $L^*$ should not be taken too seriously,
as one sees in the inset that the dependence of $L^*$ on $\beta$ is hard to tell quantitatively.
It is clear, however, that $L^*$ is increasing with the inverse temperature.

The interaction of the memory system with a thermal bath is simulated by Metropolis evolution.
As we wish to observe low temperature behavior we adopt continuous time algorithm by Bortz, Kalos, and Lebowitz (BKL)~\cite{BortzKalosLebowitz1975}.
A pseudo-random number generation package \verb|RngStream| by L'Ecuyer~\cite{LEcuyerEtAl2002RngStream} was used.
As before, the coupling constant in the Hamiltonian is set to $J = \half$ so a single defect has energy 1.
Although the cubic code is inherently quantum,
it is relevant to consider only $X$-type errors (bit flip) in the simulation,
thanks to the duality of the $X$- and $Z$-type stabilizer generators of the cubic code.
The simulation thus is purely classical.
The errors are represented by a binary array of length $2L^3$,
and the corresponding syndrome by a binary array of length $L^3$.

The memory time is measured to be the first time when the memory becomes unreliable.
There are two cases the memory is unreliable: 
either the broom algorithm fails to remove all the defects so we have to reinitialize the memory, 
or a nontrivial logical error is occurred. 
It is thus necessary in our simulation to keep track of the error operator during the time evolution.
In fact, most of the time, it was the broom algorithm's failure that made the memory unreliable.
Nontrivial logical errors occurred only for very small system sizes $L=5,7$.

It is too costly to decode the system every time it is updated. 
Alternatively, we have performed a trial decoding  every fixed time interval
\[
 T_{ec} = \frac{e^{4\beta}}{100}
\]
where $\beta$ is the inverse temperature.
Although the time evolution of the BKL algorithm is stochastic, 
a single BKL update typically advances time much smaller than $T_{ec}$. 
So it makes sense to decode the system every $T_{ec}$. 
The exponential factor appears naturally 
because BKL algorithm advances time exponentially faster as $\beta$ increases. 
It is to be emphasized that we do not alter the system by the
trial decodings (a copy of the actual syndrome has been created for each trial decoding). 

The system sizes $L^3$ for the simulation 
are chosen such that the code space dimension is exactly $2$,
for which the complete list of logical operators is known.
If the linear size $L$ is $\le 200$, 
this is the case when $L$ is not a multiple of 2, 15, or 63 by Corollary~\ref{cor:cubic-code-degeneracy-formula}.
For these system sizes, 
to check whether a logical operator is nontrivial 
is to compute the commutation relation with the known nontrivial logical operators.

The measured memory time for a given $L$ and $\beta$ is observed to follow an exponential distribution;
a memory system is corrupted with a certain probability given time interval.
Specifically, the probability that the measured memory time is $t$ is proportional to $e^{-t/\tau}$.
Thus the memory time should be presented as the characteristic time of the exponential distribution.
We choose the estimator for the characteristic time
to be the sample average $\bar T = \frac{1}{n}\sum_i^n T_i$.
The deviation of the estimator will follow a normal distribution for large number $n$ of samples.
We calculated the confidence interval to be the standard deviation of the samples divided by $\sqrt{n}$.
For each $L$, $400$ samples when $\beta \le 5.0$ and $100$ samples when $\beta > 5.0$ were simulated.
The computation was performed on IBM Blue Gene/P using 512 cores
located in IBM T.~J. Watson Research Center, Yorktown Heights, New York.
The result is summarized in Figure~\ref{fig:TvsL},\ref{fig:TvsBeta}.

Figure~\ref{fig:TvsBeta} clearly supports $\log T_{mem} = c \beta^2 + \cdots$.
Figure~\ref{fig:TvsL} demonstrates the power law for small system size:
\[
 T_{mem} \propto L^{2.93 \beta -10.5}
\]
We wish to relate some details of the model with the numerical coefficients.
The rigorous analysis of the previous section, 
gives a relatively small coefficient $c$ of the energy barrier for correctable errors by our RG decoder.
However, we expect that the coefficient of $\beta$ in the exponent is the same as the constant 
$c$ that appear in the energy barrier
\[
 E = c \log_2 R
\]
to create an isolated defect separated from the other by a distance $R$.
This is based on an intuition that 
the output $P'$ of the decoder would have roughly the same support 
as the real error $P$ for the most of the time, 
provided that the error has energy barrier less than $\Delta=c \log_2 \ltqo$.
Thus, an error of energy barrier less than $\Delta$ would be corrected by the decoder.
Our empirical formula supports this intuition.
It suggests that $c = 2.93 \log 2 = 2.03 \sim 2$.

\begin{figure}
\centering
\includegraphics[width=.5\textwidth]{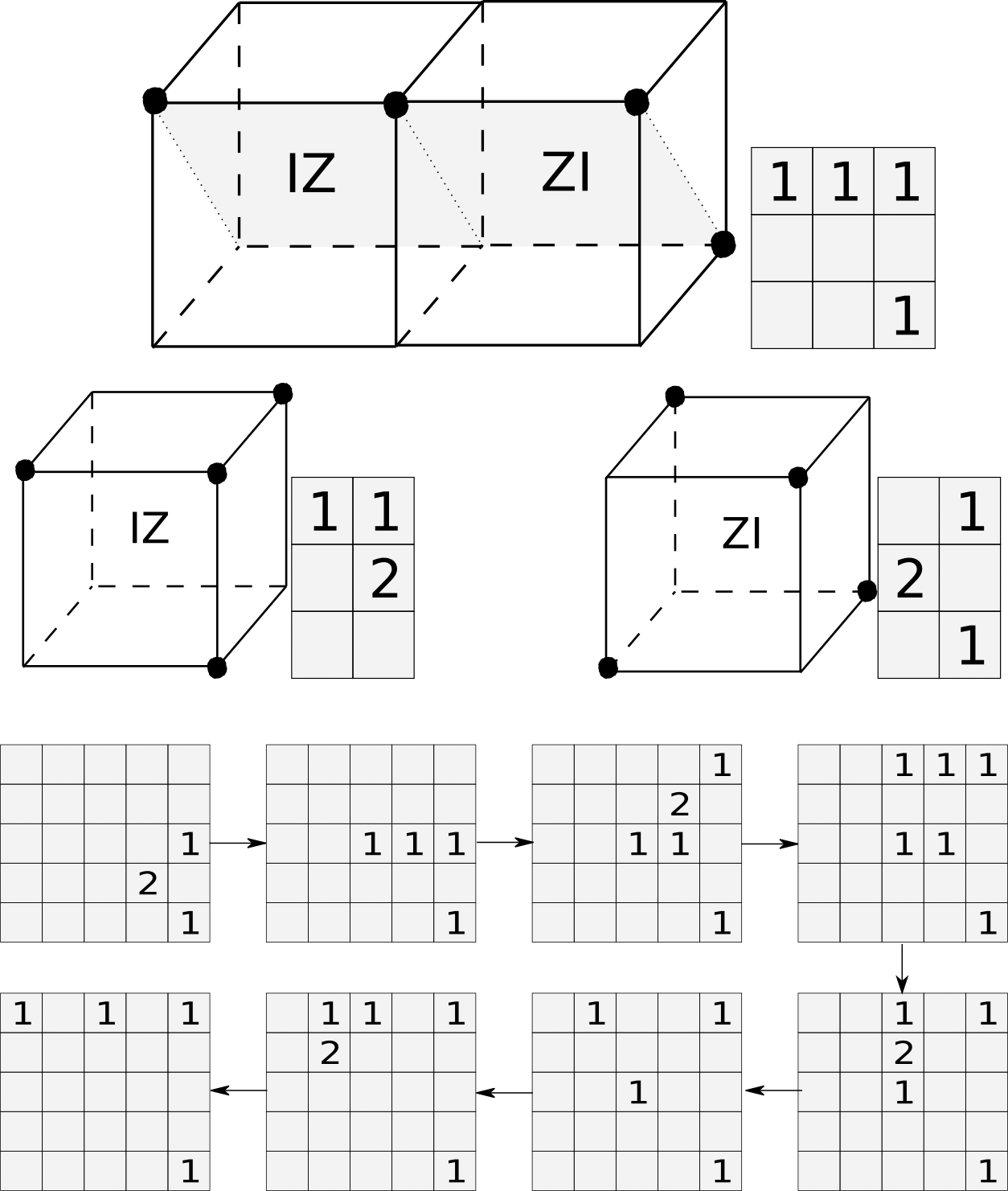}
\caption{Construction of a hook of level 2 from the vacuum.
The grid diagram represents the position and the number of defects in the ($x=z$)-plane.
For each transition, an operator of weight 1 is applied.
The total number of defects never exceeds 6.
From a level-$0$ hook (the second diagram in the sequence),
a level-$1$ hook (the last in the sequence) is constructed using extra 2 defects.}
\label{fig:hook}
\end{figure}

Indeed, we can illustrate explicitly an error path that separates 
a single defect from the rest by distance $2^p$ during which only $2p+4$ defects are needed.
Consider an error of weight 2 that creates 4 defects as shown in the top of Fig.~\ref{fig:hook}.
We call it the {\em level-0 hook}.
The bottom sequence depicts a process to create a configuration shown at the bottom-left,
which we call {\em level-1 hook}. 
One sees that level-$1$ hook is similar with ratio 2 to level-$0$, 
and is obtained from level-0 with extra 2 defects.
One defines level-$p$ hooks hierarchically.
We claim that a level-$p$ hook can be constructed from the vacuum using $2p+4$ defects.
The proof is by induction.
The case $p=1$ is treated in the diagrams.
Suppose we can construct level-$p$ hook using $2p+4$ defects.
Consider the $2^{nd}$, $4^{th}$, $6^{th}$, and $8^{th}$ steps in Fig.~\ref{fig:hook}.
They can be viewed as a minuscule version of level-$p$ steps 
that construct a level-$(p+1)$ hook from the level-$p$ hooks. 
It requires at most $2p+4+2$ defects to perform the level-$p$ step; 
this completes the induction.

It may not be obvious whether a high level hook corresponds to a nontrivial logical operator,
but such a large hook is bad enough to make our decoder to fail.

\appendix

\chapter{Commutative algebra}
\label{app:algebra}

We briefly review algebraic concepts and tools used in this thesis, mainly in Chapter~\ref{chap:alg-theory}~and~\ref{chap:lowD-codes}.
There are many nice textbooks including those by Lang~\cite{Lang}, Atiyah and MacDonald~\cite{AtiyahMacDonald}, and Eisenbud~\cite{Eisenbud}.
The book by Lang is a comprehensive textbook covering a wide range of topics in abstract algebra.
The book by Atiyah and MacDonald explains commutative algebra that may look too concise, but precisely for this reason it is very useful as a reference.
Examples are rare but essential.
The book by Eisenbud is also on commutative algebra and is extensive.
It covers more material than Atiyah-MacDonald.
In particular, our summary of Gr\"obner basis follows Eisenbud.
The chapter on Gr\"obner basis appears in the middle of the book,
but is relatively self-contained and elementary.
Here, we will omit many proofs and not try to be fully rigorous.
We explain theorems to the point where intuition can be developed.
Rigorous proofs can be found in one of the three books.

We start by recalling definitions for abelian groups.
An abelian group $G$ with the identity element denoted by $0$
is a set with an operation $+ : G \times G \to G$ such that
$g + g' = g' + g$ and $0 + g = g$.
It is required for $G$ to have inverses of $g$ denoted by $-g$ such that $g + (-g) = 0$.
$n$-fold sum of $g$ is simply denoted as $ng$, where $n \in \ZZ$.
Given two abelian groups $G$ and $H$, we can form a direct sum $G \oplus H$.
It is the set of all tuples $(g,h)$, where $ g \in G,~ h \in H$,
and the group operation $+$ is defined as $(g,h) + (g',h') = (g+g',h+h')$.
We can form a {\bf direct sum} $\bigoplus_\alpha G_\alpha$ of arbitrary family $\{ G_\alpha \}$ of groups.
It is the set of all indexed collections of group elements $(g_\alpha)$ 
where only {\em finitely many} $g_\alpha$ are nonzero.
The group operation is again defined component-wise.
Thus, any element in the direct sum is a sum of finitely many $g_\alpha \in G_\alpha$.
A sum of two abelian groups can be defined if they are subgroups of a parent group.
If $A,B \le C$ are subgroups, the {\bf sum} $A + B$ is the group of all elements of $C$ of form $a+b$ where $a \in A$ and $b \in B$.
Note that $A+B \cong A \oplus B$ if and only if $A \cap B = 0$.
Given a subgroup $N \le G$, we can form a quotient group $G/N$,
the set of all equivalent classes under the equivalence relation $[g]=[g']$ iff $g-g' \in N$.
In commutative algebra, almost everything is an abelian group.
On top of the abelian (additive) group structure, a new ``multiplication'' is added.

\section{Rings and homomorphisms}

The set of integers $\ldots, -2,-1,0,1,\ldots$ admits two operations, addition and multiplication.
There is 0 that has no effect under addition, and 1 that has no effect under multiplication.
One can always undo the addition because one can subtract a number.
However, the multiplication is not invertible within the set of integers
because fractions are not integers.
One convenient thing is that the multiplication does not care about the order.
A commutative ring is an abstraction of this structure.
It is a set, in which one can add and subtract. A multiplication exists but is not in general invertible.
An additive identity 0 exists, and a multiplicative identity 1 exists.
The distribution law $a(b+c) = ab + ac$ is assumed,
and the multiplication is commutative $ab = ba$.
A ring $R$ can consists of a single element, in which case $R$ is called a zero ring, if and only if $0=1$.
Indeed, if $a \in R$ and $1 = 0 \in R$, then $a = a \cdot 1 = a \cdot 0 = a \cdot ( 0 + 0 ) = a \cdot (1+1) = a + a = 0$.
Examples of rings are abundant: The set of all integers, the set of all complex numbers, the set of all square diagonal matrices of  a fixed size,
the set of polynomials, the set of all differentiable functions on a real line, the set of all continuous real-valued functions on a manifold, etc.
Is the set of all even integers a ring? 
No. Some authors define rings to include this case where the multiplicative identity $1$ is not provided,
but we avoid this case. Any ring is with 1.
Note that $1$ is unique; if $1'$ is also a multiplicative identity, then $1 = 1 \cdot 1' = 1'$.
The same is true for $0$.

A ring is always understood in terms of relations with other rings.
Given two rings $A$ and $B$ we consider a restricted class of maps between them.
That is, we require that the map obeys the ring structure of the rings.
$f : A \to B$ is a {\bf homomorphism} if $f(a+b) = f(a) + f(b)$ and $f(ab) = f(a)f(b)$ for any $a,b \in R$.
In addition, we assume $f(0) = 0$ and $f(1) = 1$.
(``morph'' means ``shape.'') 
The $+$ or the omitted $\cdot$ between $a$ and $b$ in $ab$ on the left-hand side
are the operations defined in $A$, whereas those in the right-hand side are in $B$.
The {\bf image} of a homomorphism $f$ is the subset of $B$ written as $f(A)$ defined by $\{ f(a) ~|~ a \in A \}$.
Is the image of a homomorphism a ring? Yes.

There is no point to speak of a map between two rings $A,B$ that is not a homomorphism.
If we are going to ignore the ring structure, we would rather say the map between the ``sets'' $A,B$.
We will simply say a {\bf map between rings} to mean a homomorphism.
We note more terminologies: An {\bf endomorphism} is a map from a ring into itself.
An {\bf isomorphism} is a map between two rings with a unique inverse.
An {\bf automorphism} is an isomorphism that is an endomorphism.

The ring of integers is so primitive in the following sense.
Let $A$ be an arbitrary ring. Consider a map $f : \ZZ \to A$.
$f(n) = \sum_{i=1}^{n} f(1)$ and $f(-n) = \sum_{i=1}^n (-f(1))$ where $n > 0$.
But, $f(1)=1$ is the unique multiplicative identity.
Therefore, $f$ is completely determined, though we just required $f$ be a homomorphism;
there is a unique nonzero map from $\ZZ$ into any ring.
How many endomorphisms are there for $\ZZ$?

The {\bf kernel} of a map (homomorphism!) $f$ is the subset of $A$ written as $\ker f$
defined by $\{ a \in A ~|~ f(a) = 0 \}$.
It is easy to see that the kernel is closed under the addition and multiplication.
Here, the closeness means that the result of the operation using two elements in a subset lies in the subset.
(It is pointless to speak of the closedness of an operation without reference to a subset.)
There is one more important property as we discuss below.

\section{Ideals and modules}

The kernel $I$ of a map $f$ between rings $R \to S$ has the following property:
\begin{equation}
 \forall r \in R,~ \forall x \in I~:~~ r x \in I
\label{cond-ideal}
\end{equation}
This is easily verified since $f(r x) = f(r)f(x) = f(r) \cdot 0 = 0$.
When a subset $I$ of a ring $R$ is closed under the multiplication and addition, and satisfies the above property,
we call $I$ to be an {\bf ideal} of $R$.
For example, in $\ZZ$, the set of all even numbers is an ideal denoted as $(2)$.
The property \eqref{cond-ideal} is trivial, because it reads a multiple of an even number is even.
This ideal is the kernel of the map $\ZZ \to \ZZ/(2)$,
where the latter is the ring of integers modulo 2.
Remember that there is a unique nonzero map from $\ZZ$ to any ring.
In fact, any ideal arises in this way: An ideal is the kernel of a ring homomorphism.

To understand this, we need to formalize how to construct {\bf quotient rings} or {\bf factor rings}.
Let $R$ be a ring and $I$ be an ideal.
They are both abelian groups; they are closed under the addition with identity $0$,
and contains additive inverses, the minus elements.
The quotient ring $R/I$ as a set is the same as the quotient group as an abelian group.
$R/I$ is the family of equivalence classes under the equivalence relation that
$[a] = [b] \in R/I$ iff $a-b \in I$.
The multiplication in $R/I$ is as expected: $[a] \cdot [b] = [ab]$.
Using the property $(*)$ one can verify that it is well-defined.
Consider a map $R \to R/I$ defined by $a \mapsto [a]$.
It should be straightforward that it is a ring homomorphism.
What is the kernel? Precisely, $I$.
Now it is an almost tautology that an ideal is a kernel of a ring homomorphism.
Often, we just write $a$ in place of the equivalence class notation $[a]$.
In this lazy notation, ``the map $R \to R/I$ is defined by $a \mapsto a$.''
This map is so important that it has its own name: {\bf canonical map} or {\bf quotient map}.

In leisurely words, an ideal extends what we treat as zeros;
any element of $I$ is zero in $R/I$.
Zero plus zero or zero times zero must be zero, so $I$ is closed under the addition and multiplication.
Zero times any element must be zero, so the property \eqref{cond-ideal} should hold.
How do we specify an ideal in a real calculation?
Consider a polynomial ring $\QQ[x]$, the set of all polynomials (of finitely many terms) in $x$ with coefficients in the rational numbers.
Suppose we decided that $x$ must be equal to $\half$.
In other words, we decided that $x - \half$ must be ``zero'' in $\QQ[x]/I$.
Then, as zero times anything must be zero, any element $(x-\half)\cdot f(x)$ must be zero in $\QQ[x]/I$, too;
it must be an element of the ideal we are defining.
We require no more elements are identified as zeros.
That is, we define $I = \{ f(x)(x-\half) ~|~ f(x) \in \QQ[x] \}$ as our ideal,
and carry out any computation in the quotient ring $\QQ[x]/I$.
We write $I$ as $(x-\half)$ by putting the {\bf generator} inside the parenthesis.
Computing in $\QQ[x]/I$ is the same as computing in $\QQ[x]$ and evaluate the polynomial at $x=\half$.

In general, if we write and ideal $I$ of $R$ as $(a,b,c)$,
then we mean
\[
 I = \{ ax + by + cz ~|~ x,y,z \in R \} .
\]
$I$ is said to be {\bf generated by $a,b,c$}.
A quick exercise: What is $\QQ[x]/(x-\half, x-1)$~? It is a zero ring.
Since $x-\half$ and $x-1$ are ``zeros,'' their difference $\half$ is zero.
Zero times 2 is zero, so $1$ is zero. Therefore, everything is zero.
We proved an ideal equality $(x-\half, x-1) = (1)$.
The whole ring $R$ viewed as an ideal is called the {\bf unit ideal} denoted by $(1)$.
In fact, when there is any invertible element in an ideal, it is the unit ideal.
For this reason, an invertible element in a ring is called a {\bf unit}.
Note that computing the minimal set of generators is in general a hard problem.
For instance, I do not know any algorithmic answer to questions like ``can this ideal be generated by three elements?''
However, the Gr\"obner basis gives a \emph{canonical} set of generators for an ideal of a polynomial ring,
and we can algorithmically compare two ideals.

\subsection*{Prime, maximal}

There is a very important class of ideals that generalizes prime numbers in $\ZZ$.
A prime number $p$ has a property that if $p$ divides a product of two integers $ab$
then $p$ divides either $a$ or $b$. A {\bf prime}%
\footnote{Do not use ``primary'' in place of ``prime.'' The adjective ``primary'' has a slightly different technical meaning.}
ideal $\pp$ of $R$ is an ideal not equal to $(1)$ such that
\begin{equation}
ab \in \pp \text{ implies } a \in \pp \text{ or } b \in \pp \text{ for any } a,b \in R.
\label{cond-prime}
\end{equation}
The condition is rephrased as $a \notin \pp$ and $b \notin \pp$ imply $ab \notin \pp$.
Is the ideal $(0)=\{0\}$ prime? It depends. In $\ZZ$ the zero ideal is prime because the product of two nonzero integers is nonzero.
In the polynomial ring $\QQ[x]$, $(0)$ is prime because the degree of nonzero polynomial does not decrease under any nonzero multiplication. 
If any nonzero elements of a nonzero ring $R$ has a multiplicative inverse, in which case $R$ is called a {\bf field},
then $(0)$ is prime because $ab = 0$ means $a = 0$ or $b = 0$.
However, in the ring of diagonal $2 \times 2$ matrices, 
$(0)$ is not prime because nonzero matrices $\begin{pmatrix} 1 & 0 \\ 0 & 0 \end{pmatrix}$
and $\begin{pmatrix} 0 & 0 \\ 0 & 1 \end{pmatrix}$ multiply to zero.
The ring in which $(0)$ is prime has a special name, {\bf (integral) domain}.
It is easy to verify that $R/\pp$ is an integral domain if and only if $\pp$ is prime,
applying the picture that $\pp$ defines zeros in $R/\pp$.

There is one more important property of prime ideals.
Let $f : A \to B$ be a map between two rings.
{\em If $\pp \subseteq B$ is a prime ideal, then $I = f^{-1}(\pp) \subseteq A$ is prime.}
Proof: $aa' \in I \Rightarrow f(a)f(a') \in \pp \Rightarrow f(a) \in \pp \vee f(a') \in \pp \Rightarrow a \in I \vee a' \in I$.

A subclass of prime ideals consists of maximal ideals.
A {\bf maximal ideal $\mm \neq (1)$} is defined by the maximal property:
\begin{equation}
 \mm \subsetneq \mm' \text{ implies } \mm' = (1) \text{ for any ideal } \mm'.
\end{equation}
It is a priori not vivid why maximal ideals are prime.
But it is easy to see. Consider $R/\mm$. If $a \in R/\mm$ is nonzero, that is $a \notin \mm$, then $I = (\mm, a) \supsetneq \mm$ and therefore $I = (1)$, which means there is an element $b \in R$ such that $b a + m = 1$ for some $m \in \mm$. By the canonical map, $b$ maps to a multiplicative inverse of $a$ in $R/\mm$. 
(The converse is also true. {\em $\mm \neq (1)$ is a maximal ideal if and only if $R/\mm$ is a field}.)
In other words, $R/\mm$ is a field, and therefore an integral domain. In particular, $\mm$ is prime.

\subsection*{Modules}

Ideals admit another viewpoint.
Let us forget the multiplication within an ideal $I$,
and treat $I$ as a separate set from the mother ring $R$.
It is still an abelian group, and satisfies \eqref{cond-ideal}.
The condition $(*)$ looks very similar to the scalar multiplication on vector spaces.
Indeed, consider a direct sum $R \oplus R$ of abelian groups,
the set of all tuples $(r,r')$ where $r,r' \in R$ are any elements.
There is an operation on $R \oplus R$ similar to \eqref{cond-ideal};
we can define $r \cdot (s,s')$ to be $(rs, rs')$,
similar to the scalar multiplication for a vector space.
Indeed, a two-dimensional vector space is precisely obtained in this way by setting $R = \mathbb{Q}, \mathbb{R}, \mathbb{C}$, etc.

Let us define an abstract notion.
Let $M$ be an abelian group with a {\em bilinear} operation $\cdot$ such that
\begin{equation}
 \forall r \in R,~ \forall m \in M~:~~ r \cdot m \in M .
\label{cond-module}
\end{equation}
The ``multiplication'' $\cdot$ here is not the same thing as the multiplication within the ring.
We assume expected formulas $(rs) \cdot m = r \cdot (s \cdot m)$, and simply write $(rs) \cdot m = rsm$ where $r,s \in R$ and $m \in M$.
We call $M$ an {\bf $R$-module} or a {\bf module over $R$}.
It is general than the notion of vector spaces.
A vector space is a module over a {\bf field}, a commutative ring in which multiplication by a nonzero element has an inverse.
An ideal is a subset of $R$ such that it is a module over $R$.%
\footnote{
Although an ideal is a valid module, often it is treated differently than a module.
One should sometimes be careful to apply things defined for modules,
especially when one reads the dimension theory of Eisenbud~\cite{Eisenbud}.}
Note that $R$ itself is an $R$-module via the multiplication within $R$.
The above example $R \oplus R$ is an $R$-module.

For an $R$-module $M$ if there exist finitely many elements $m_1,\ldots,m_n \in M$ such that
\begin{equation}
 M = \{ r_1 m_1 + \cdots + r_n m_n ~|~ r_1,\ldots, r_n \in R \},
\end{equation}
then we say $M$ is {\bf finitely generated}.
All of our modules in the thesis are finitely generated.
An ideal is finitely generated if it is finitely generated as a module.

There is a confusingly similar terminology that one must distinguish.
Suppose $A$ is a subring of $B$. That is, $A$ is a ring by itself and contained in a bigger ring $B$.
Or slightly more generally, suppose we are given a ring map $A \to B$.
The subring case is precisely when the map is an inclusion.
We say $B$ is a {\bf finitely generated $A$-algebra}
if there exists finitely many elements $b_1, \ldots, b_n \in B$ such that
any element of $b$ can be written as a polynomial in $b_1,\ldots, b_n$ with coefficients in the image of $A$.
In this case, $B$ is sometimes written as $B = A[b_1,\ldots, b_n]$.
A typical situation is when $A$ is a field such as $\QQ, \mathbb{C}$ and $B$ is a polynomial ring over $A$.
For example, $B = \QQ[x,y]$.
The reason it is a confusing terminology is because $B$ is not necessarily finitely generated $A$-module.
$\QQ[x]$, a finitely generated $\QQ$-algebra, has infinitely many generators $\{ 1, x, x^2, \ldots \}$ as a $\QQ$-module.
Is $\QQ$ a finitely generated $\ZZ$-algebra? 
No, because multiplying by integers cannot produce large denominators.

As the rings are understood in relation to others,
the modules should be understood via maps.
We required the ring homomorphisms to preserve the defining operations of the rings.
The same is true for the module maps. Given two modules $M$ and $N$ over $R$,
we define an {\bf $R$-linear map} or {\bf $R$-module homomorphism} $f : M \to N$ to satisfy
\begin{equation}
f(m+m') = f(m)+f(m')\quad \text{ and } \quad f(r \cdot m) = r \cdot f(m) \text{ for any } r \in R,~ m, m' \in M.
\end{equation}
Note that the $\cdot$ on the left-hand side is the operation \eqref{cond-module} for $M$
whereas that on the right-hand side is the operation for $N$.%
\footnote{
In group representation theory,
the $R$-linear maps are called equivariant maps,
where $R$ is the group algebra which may not be commutative.
The representation space is a module, 
the subrepresentation space is a submodule, and
the irreducible representation space is a simple submodule.}
The notions of kernel, image, endomorphism, automorphism, and isomorphism apply to module maps, too.
A simple exercise: Let $A,B$ be two $R$-algebras; there are ring maps $R \to A$ and $R \to B$.
The two algebras are naturally $R$-modules, when a map $A \to B$ becomes $R$-linear?
Answer: It becomes $R$-linear when the following diagram commutes.
\[
 \xymatrix@!0{
A \ar[rr]  & & B \\
           & R\ar[ul]\ar[ur] &
}
\]

\subsection*{Free}

An $R$-module isomorphic to $\bigoplus_\alpha R$ is called a {\bf free module}.
When there are finitely many summands, it is called {\bf finitely generated free module}.
It is the most convenient type of modules.
In particular, a module map between finitely generated free $R$-modules
is simply given by a matrix with entries in $R$, just as linear map between finite dimensional vector space
can be described by a matrix with entries in a field.
Note that elements of a finitely generated free $R$-module $M = R^{\oplus n}$ can be expressed by column matrices.
Let $e_i$ ($i=1,\ldots,n$) denote the canonical basis column matrices.
(We could say ``column vectors'' instead of column matrices. However, an element of a module is not a vector in general.)
A map $f$ from $M$ to any module is specified if we specify the image $f(e_i)$
because the image of other elements $a_1 e_1 + \cdots a_n e_n$ is must be $a_1 f(e_1) + \cdots + a_n f(e_n)$
by $R$-linearity.
Writing $f(e_i)$ in columns, we have a matrix representation of $f$.
An un-redundant set of generators of a free module is called a {\bf basis}.
The cardinality of a basis is called {\bf rank}.
(One can show that rank is independent of the choice of a basis.)
Only for free modules can we speak of bases.
The crucial difference between general modules and vector spaces is that
a module is in general not free, while a vector space, a module over a field, is always free.
That nonzero elements are invertible makes such a huge difference.

Note that any module can be described by free modules.
Take a {\bf generating} set $\{ m_\alpha \}$ of a module $M$; 
any element of $M$ is a {\em finite} $R$-linear combination $\sum_i r_{\alpha_i} m_{\alpha_i}$.
A trivial and useless choice would be to take whole $M$ as a generating set.
Let $F$ be a free $R$-module whose rank is the same as the cardinality of the generating set,
i.e., there is a surjective module map $F \to M$.
The kernel $N$ is a submodule of $F$, not necessarily free, and $M \cong F / N$.
One can carry the same process for the module $N$.
That is, one finds a free module $F'$ such that $\phi : F' \to N$ is a surjection.
Now $M$ is expressed as $M \cong F / \im \phi$.
This is conceptually important observation, but not too useful
because we do not have any control over the ranks of $F$ and $F'$.
We need some finiteness conditions.
An $R$-module $M$ is said to be {\bf finitely presented} 
if there is a map $\phi : F' \to F$ between finitely generated free modules
such that $M \cong F / \im \phi$. The latter expression $F/\im \phi$ is often abbreviated as $\coker \phi$.
The map $\phi$ is called a {\bf finite presentation} of $M$.
As we have seen above, $\phi$ is a matrix with entries in $R$,	
and $M$ is expressed by a single matrix $\phi$.
In case of a finite dimensional vector space, $\phi$ can always be brought to a diagonal matrix
with entries of $0$ or $1$, after basis change of $F$ and $F'$.
We will discuss more on a finite presentation in Section~\ref{sec:det-ideal}.

\subsection*{Noetherian}

A technically very important and convenient adjective is Noetherian.
It is as important as vector spaces having finite dimensions.
A {\bf module $M$ is Noetherian} if every increasing sequence of submodules is stationary,
i.e., if $M_1 \le M_2 \le \cdots \le M_n \le \cdots$ is a sequence of submodules of $M$,
then for all sufficiently large $n$ one has $M_n = M_{n+1}$.
Often this condition is referred to as the {\bf ascending chain condition} or {\bf a.c.c.}
{\em A module $M$ is Noetherian if and only if any submodule is finitely generated.}
If a submodule cannot be generated by finitely many elements, 
one can construct a strictly increasing infinite sequence of submodules
using an infinite subset of generators.
Conversely, if any submodule of $M$ is finitely generated,
then the union $\cup_{i=1}^\infty M_i$ of any increasing sequence of submodules $M_i$ of $M$
is also finitely generated, say, by $m_1,\ldots,m_r$.
Since $m_j$ is contained in some $M_{j'}$, there must be some $k$ such that $m_1,\ldots,m_r \in M_k$.
It follows that $M_k = M_{k'}$ for all $k' \ge k$.
A Noetherian module over a field is just a finite dimensional vector space.
A {\bf Noetherian ring} $R$ is a ring that is Noetherian as an $R$-module,
i.e., the a.c.c is satisfied with respect to the ideals of $R$.
In a Noetherian ring, any ideal is finitely generated.

Being Noetherian is preserved in many cases.
Loosely speaking, it says that a module carries a finite amount of data.
A module with a finite amount data manipulated finitely many times would still have a finite amount of data.
The following are facts:
\begin{itemize}
\item {\em A homomorphic image of Noetherian ring (module) is a Noetherian ring (module).}
\item {\em A submodule of Noetherian module is Noetherian.}
\item {\em A finitely generated algebra over a Noetherian ring is a Noetherian ring.}
\item In particular, {\em a polynomial ring over a field with finitely many variables is a Noetherian ring.}
\item {\em A finitely generated module over a Noetherian ring is Noetherian.}
\end{itemize}
There are more to mention about being Noetherian using tensor products and localization.
See Section~\ref{sec:localization}.

Remark that {\em a finitely generated module $M$ over a Noetherian ring $R$ always admits a finite presentation $\phi : F' \to F$.}
Since $M$ is finitely generated, the module $F$ is of finite rank.
$F$ is Noetherian, and therefore, $\im \phi$ is finitely generated, which means $F'$ can be taken to be of finite rank.
The ``matrix'' $\phi$ is of finite size.

Is the ring of all differentiable functions $\mathbb{R} \to \mathbb{R}$ Noetherian? No.
Is the ring of all trigonometric functions $\mathbb{R} \to \mathbb{R}$ of period $1/n$, where $n$ are positive integers, Noetherian? Yes.

\section{Gr\"obner basis}

Gr\"obner basis is a special generating set for ideals and modules.
It provides canonical presentations of modules and ideals, from which any concrete computational commutative algebra is built.
In all notions in the previous section,
I cannot imagine any systematic way to compute things concretely without Gr\"obner basis.
It is theoretically important too because it tells what is actually constructible.
A nice application of the Gr\"obner basis is a sharp version of Hilbert syzygy theorem~\cite[Corollary~15.11]{Eisenbud}.
In this section we {\em assume the (base) ring is a polynomial ring over a field $\FF$ with finitely many variables};
$R = \FF[x_1,\ldots,x_n]$.
Since $R$ is Noetherian, for any ideal there is a finite set of generators.

How do ideals look like in $\FF$? There are only $(0)$ and $(1)$ because any nonzero element is a unit.
How about $\FF[x]$? It is only slightly more complicated.
By the Euclidean algorithm, the $\gcd$ of two polynomials $f(x),~g(x)$ can be expressed as
\[
 \gcd( f(x), g(x) ) = a(x) f(x) + b(x) g(x) \in (f(x),g(x))
\]
An ideal generated by $f(x)$ and $g(x)$ is thus the same as an ideal generated by a single element $\gcd(f(x),g(x))$, 
in which case the ideal is called {\bf principal}.
By induction, one can always reduce the number of generators of an ideal in $\FF[x]$ if it is greater than $1$.
Since $\FF[x]$ is Noetherian, this is enough to imply that {\em any ideal of $\FF[x]$ is generated by a single element,}
i.e., {\em any ideal of $\FF[x]$ is principal}.
The principal generator $p(x)$ is important, not only because it is simple, but also because it gives a criterion
whether an arbitrary polynomial $h(x)$ is contained in $(p(x))$:
Divide the $h(x)$ by $p(x)$. The remainder is zero if and only if $h(x) \in (p(x))$.
Put differently, we can find a canonical representative in the quotient ring $\FF[x]/(p(x))$
to be the remainder $r(x)$  where $\deg r < \deg p$.
The $p(x)$ is the Gr\"obner basis for the ideal $(p(x))$ in $\FF[x]$.

In case of two or more variables, it is no more true that any ideal is principal.
Let us examine the {\bf division algorithm}.
Let $f(x)$ be a dividend and $p(x)$ be a divisor.
We first compare degrees of them.
If $\deg f < \deg p$, then the division is completed, and $f$ is the remainder.
Otherwise, we match the leading coefficients and then subtract a multiple $a(x)p(x)$ from $f(x)$.
$a(x)$ is a monomial such as $3x^2$.
The purpose is of course to have $\deg (f-ap) < \deg f$.
And then we iterate.
Essential is a total ordering among terms such that the ordering is preserved under multiplication by a monomial.
Now we consider a problem of computing the remainder of $f$ modulo $G = \{ g_1, \ldots, g_k \}$.
Equipped with the ordering, given a dividend $f$ and a set $G$ of divisors,
(Step-1) one should be able to match the leading term and kill it,
thereby ``reduce'' the leading term of the dividend.
(Step-2) One stops when the reduction becomes impossible.
A question arises.
During Step-1, there would be many possible choices among the divisors from $G$.
How do we guarantee that the remainder is independent of the choices?
This is a sound question, and the answer is no in general.
Even in the one variable case, the answer is no in general.
Consider $f = x^2 -1,~ g = x^3 -1$ in $\QQ[x]$.
The division by the set $\{f,g\}$ of a dividend $x^3$ gives $1$
if we kill the leading term by $x^3-1$,
or $x$ by $x^2-1$.

We formulate the problem as finding a generating set of an ideal $I$ 
that gives rise to a unique remainder under the division.
A solution is that we include more and more elements of $I$ into a generating set $G$
so that any leading term $l$ of a dividend can be removed 
where $l$ appears as a leading term of some elements of $I$.
In the one variable case, it is enough to include the $\gcd$ of all (finitely many) generators of $I$.
We have given enough motivation and technical problems to smooth out.

\subsection*{Definitions}

A monomial is a product of variables.
A {\bf monomial order} on the set of all monomials of $R = \FF[x_1,\ldots,x_n]$
is a total ordering $\succ$, under which any two monomials are comparable, such that
\begin{equation}
 x_i m \succ x_i m' \succ m' \quad \text{whenever} \quad m \succ m'
 \label{cond-TermOrder}
\end{equation}
for any monomial $m,m'$ and any variable $x_i$.
Since a monomial $x_1^{a_1} \cdots x_n^{a_n}$ is uniquely given by a $n$-tuple of nonnegative integers $(a_i) = (a_1,\ldots,a_n)$,
a monomial order is a total order on the hyper-octant of $\ZZ^n$ of nonnegative coordinates.
Note that under any monomial order, $1$ is the least monomial.

Two examples at least are important. The first one is the {\bf lexicographic order},
under which $(a_i) \succ (b_i)$ if and only if $a_i > b_i$ for the least $i$ such that $a_i \neq b_i$.
For example,
\[
 x_1 \succ_\text{lex} x_2^5 \succ_\text{lex} x_2 x_3^{100} .
\]
The second example is the {\bf degree reverse lexicographic order},
under which $(a_i) \succ (b_i)$ if and only if `$\sum a_i > \sum b_i$' 
or `$\sum a_i = \sum b_i$ and $a_i < b_i$ for the largest $i$ such that $a_i \neq b_i$.'
For example,
\begin{align*}
 x_1 \succ_\text{revlex} x_2 \succ_\text{revlex} \cdots \succ_\text{revlex} x_n \\
 x_3^{100} \succ_\text{revlex} x_2^{50} \succ_\text{revlex} x_1^{49} x_3 .
\end{align*}
The degree reverse lexicographic order is the most commonly used in computer software.

{\em Any monomial order is a well-order, i.e., every set of monomials has a minimal element.}
This statement is at the core of any finiteness proof regarding Gr\"obner basis.
The proof is very simple. Let $X$ be any set of monomials of $R$.
The submodule (ideal) generated by $X$ is finitely generated by $G$ because $R$ is Noetherian.
It means $X$ consists of multiples of a finitely many monomials of $G$.
Therefore, the minimal element of $G$, a finite set, is the minimal element of $X$.
The statement can be rephrased as {\em any decreasing sequence of monomials is stationary}
or as {\em any strictly decreasing sequence of monomials is finite.}

With a monomial order, we can define the {\bf leading term} or {\bf initial term} of a polynomial $f$
which is the greatest term of $f$ with respect to the monomial order.
(We are distinguishing ``term'' and ``monomial.'' A monomial is a product of variables only, and a term is $\FF$-multiple of a monomial.)
Let us denote the leading term of $f$ by $\lt(f)$.
Given any set $S$ of $R$, let the ideal generated by all leading terms of element of $S$ be denoted by $\lt(S)$.
Now, a {\bf Gr\"obner basis} $G$ of $I$ is a generating set of $I$ such that $\lt(I) = \lt(G)$.
Since $R$ is Noetherian, $\lt(I)$ is finitely generated.
Therefore, $G$ can be chosen to be a finite set.

\subsection*{Buchberger criterion}

{\em The division by a Gr\"obner basis results in a unique remainder.}
In particular, a polynomial is in $I$ if and only if the division by a Gr\"obner basis yields the zero remainder.
To see this, let $f$ be an arbitrary polynomial.
Let $r$ be a remainder obtained by killing large terms of $f$ by elements of $G$.
Let $r'$ be another remainder. We must show $r = r'$.
It is clear that $r-r'$ belongs to $I$. If nonzero, the leading term of $r-r'$ belongs to $\lt(I) = \lt(G)$.
However, $r$ or $r'$ has by construction no leading term that belongs to $\lt(G)$.
Therefore, $r-r'$ must be zero.

It is not clear yet how $G$ can be computed from given generators of $I$,
though we know a solution in the case where the generators of $I$ are in a single variable.
A strategy is hinted from one-variable case:
\begin{itemize}
\item[(1)] Start with any generating set $S$ of $I$. 
\item[(2)] Try to produce new polynomials whose initial terms are not contained in $\lt(S)$.
\item[(3)] Update $S$ by adjoining the new polynomials.
\item[(4)] Iterate.
\end{itemize}
It will end after finitely many iterations because the new initial terms form a descending sequence of monomials.
It remains to find an effective method to produce new initial terms.
Let us work with an example first.
Consider $I = (x^2,xy+y^2) \subset \QQ[x,y]$.
We will compute Gr\"obner basis with respect to two monomial orders.
The first one is the lexicographic order under which $x \succ y$.
We have $\lt(x^2) = x^2$ and $\lt(xy+y^2) = xy$.
Neither of them is a multiple of the other.
However, $y(x^2) - x(xy+y^2) = -xy^2$ has the leading term divisible by $xy$.
This combination is called ``S-polynomial'' of $x^2$ and $xy+y^2$.
So, $y(x^2) - x(xy+y^2) + y(xy+y^2) = y^3$.
We see that $y^3$ is a new initial term of $I$.
The generating set is expanded as $G = \{ x^2, xy + y^2, y^3 \}$.
We then again form S-polynomials using any two among $G$,
but fail to produce any new initial term. Indeed, $G$ is a Gr\"obner basis of $I$.
The second monomial order is the lexicographic order under which $x \prec y$.
Then, $\lt(x^2) = x^2$ and $\lt(xy+y^2) = y^2$.
We form an S-polynomial, $y^2(x^2)-x^2(xy+y^2) = -x^3 y$.
It is a multiple of $x^2$, and we do not get anything new.
Indeed, $\{x^2, xy+y^2\}$ is a Gr\"obner basis.

Given two polynomials $f,g$ let us fix the leading coefficient of $\gcd$ to be $1$.
The {\bf S-polynomial} of two polynomials $f$ and $g$ is the polynomial
\[
 \sigma(f,g) = \frac{\lt(g)}{\gcd(\lt(f),\lt(g))} f - \frac{\lt(f)}{\gcd(\lt(f),\lt(g))} g .
\]
Dividing an S-polynomial by a set of generators will give a non-unique remainder,
but we can hope that the calculation would find a new initial term.
Buchberger showed that {\em the process is sufficient to find all initial terms of an ideal}.
The {\bf Buchberger algorithm} to find a Gr\"obner basis is given by the above ``algorithm,''
where the step (2) is now well-defined by (2$'$):
\begin{itemize}
 \item [(2$'$)] For each pair $f,g$ of polynomials of $S$, compute $\sigma(f,g)$ and its remainder $r$ after division by $S$.
\[
\frac{\lt(g)}{\gcd(\lt(f),\lt(g))} f - \frac{\lt(f)}{\gcd(\lt(f),\lt(g))} g = \sum_k h_k s_k  + r
\]
where $s_k \in S$, $h_k \in R$, and $\lt(r) \notin \lt(S)$ or $r=0$.
\end{itemize}
In other words, the algorithm terminates with a correct answer, a Gr\"obner basis $G$,
if the remainder of any S-polynomial computed from a pair of polynomials in $G$ is zero.
This is called the {\bf Buchberger criterion}.
It follows that if $G$ is a Gr\"obner basis for the ideal $(G)$,
then any subset $G' \subseteq G$ is a Gr\"obner basis for the ideal $(G')$.

A Gr\"obner basis $G={f_1,\ldots,f_t}$ is {\bf reduced}
if any term of $f_i$ is not divisible by any $\lt(f_j)$ where $i \neq j$,
and the leading coefficients of all $f_i$ are $1$.
The reduced Gr\"obner basis can be obtained from any Gr\"obner basis.
{\em Given a monomial order, a reduced Gr\"obner basis is unique for an ideal.}
Therefore, {\em two ideals are the same if and only if their reduced Gr\"obner bases are the same.}

\subsection*{Applications}

An immediate application is to decide whether 
a set of polynomial equations $\{f_i(x_1,\ldots,x_n) = 0 \}$ admits a common solution.
If one obtains $1 = 0$ while manipulating $f_i$, then certainly there is no solution,
and this is the only case. That is, the ideal $I$ generated by $f_i$ contains $1$,
i.e., $I$ is the unit ideal,
if and only if there is no common solution to $f_i = 0$.
Since the ideal membership can be tested by the division algorithm using a Gr\"obner basis,
we can algorithmically answer this question by computing a Gr\"obner basis
in any monomial order and looking for $1$.

Note that known algorithms for Gr\"obner basis is not efficient in a computation complexity sense.
Any significant improvement seems impossible because if one can compute Gr\"obner basis efficiently,
then one can also solve, for example, Boolean satisfiability problem (SAT) efficiently.
A SAT is a decision problem that asks whether there exists an assignment to Boolean variables $x_1,\ldots,x_n$
such that a given formula using AND, OR, and parenthesis evaluates to $1$.
Any SAT can be formulated as a finding a solution to a system of polynomial equations.
Being Boolean can be expressed as $x_i(x_i-1) = 0$.
The OR of two binary variables $x,y$ can be expressed as $1-(1-x)(1-y)$.
The AND of two binary variables $x,y$ is the product $xy$.
In particular, the 3-SAT, the SAT where all the expressions are given by conjunctive normal form with 3 variables per clause,
is equivalent to finding a solution to systems of polynomial equations, each of which has degree at most 3.
The 3-SAT is known to be NP-complete.
However, finding Gr\"obner basis for a fixed number of variables is efficient,
i.e., {\em the number of calculation steps increases as a polynomial in the degree of the generating polynomials}.

We can also compute the intersection of an ideal with a subring.
The foundational case is when the subring is given by $T = \FF[x_1,\ldots,x_n]$
contained in $R = \FF[x_1,\ldots,x_n,y_1,\ldots,y_m]$.
If an ideal $I$ of $R$ is given, we wish to compute the intersection $T \cap I$,
which is an ideal of $T$.
We should use an {\bf elimination monomial order} under which $\lt(f) \in T$ implies $f \in T$.
The lexicographic order in which $x_a \prec y_b$ is an elimination order.
The following statement is true.
{\em If a Gr\"obner basis $G=\{f_1,\ldots,f_u,g_1,\ldots,g_v\}$ of $I$ under elimination order
is such that $f_i$ do not involve any variable $y_b~(b=1,\ldots,m)$ but $g_j$ do,
then $G' = \{f_1,\ldots,f_u\}$ is a Gr\"obner basis of $I \cap T$.}
The ideal $J$ generated by $G'$ is certainly contained in $I \cap T$.
It is clear that $G'$ is a Gr\"obner basis for $J$ because it satisfies the Buchberger criterion.
If $J \subsetneq I \cap T$, then there would be $f \in (I \cap T) \setminus J$ with the least leading term.
In particular, $\lt(f) \notin \lt(J)$.
Since $G$ is a Gr\"obner basis, $\lt(f) \in \lt(G)$.
Since $\lt(g_j)$ involve $y_b$ due to the elimination order, but $\lt(f) \in T$,
we must have $\lt(f) \in \lt(G') = \lt(J)$, a contradiction.

More generally, suppose we are given with a subring $T \subseteq \FF[y_1,\ldots,y_m]$
generated by some polynomials $h_i (y_1,\ldots,y_m)$ ($i=1,\ldots,n$) over $\FF$.
We wish to find $T \cap I$ for an ideal $I$ of $\FF[y_1,\ldots,y_m]$.
Introduce new variables $x_1,\ldots,x_n$, and consider 
\[
T' = \FF[x_1, \ldots, x_n] \xrightarrow{\phi} \FF[y_1,\ldots,y_m] \xrightarrow{\pi} \FF[y_1,\ldots,y_m] / I,
\]
where $\phi$ is defined by $x_i \mapsto h_i(y_1,\ldots,y_m)$ and $\pi$ is the canonical map.
$\phi$ maps onto $T$.
If we knew $K = \ker \pi \circ \phi$, the image $\phi(K)$ would be precisely $T \cap I$.
Let $R = \FF[x_1,\ldots,x_n,y_1,\ldots,y_m]$.
Note that $R$ contains $T'$ and $\FF[y_1,\ldots,y_m]$ as subrings.
Thus $\phi$ can be extended to $\tilde \phi$ as
\[
 \tilde \phi: R = \FF[x_1,\ldots,x_n,y_1,\ldots,y_m] \to \FF[y_1,\ldots,y_m] \quad \text{ defined by } \quad 
\begin{cases} x_i \mapsto h_i(y), \\ y_j \mapsto y_j . \end{cases}
\]
The kernel of $\pi \circ \tilde \phi$ is the ideal $J = (x_1 - h_1(y), \ldots, x_n - h_n(y) ) + I$.
We have $K = J \cap T'$, which can be computed by the method above.
For example, let $I = (1+x+y+z,~1+xy+yz+zx) \subseteq \FF_2[x,y,z]$. 
If a subring is $T = \FF_2[x^3,y^3,z^3]$, then an auxiliary ring is $T' = \FF_2[x',y',z']$.
The intersection $I \cap T$ can be found by computing
\[
 (1+x+y+z,~1+xy+yz+zx,~x'-x^3,~y'-y^3,~z'-z^3) \cap \FF[x',y',z'].
\]
In the end, one has to replace $x'$ with $x^3$, $y'$ with $y^3$, and $z'$ with $z^3$.

The Gr\"obner basis is also useful to compute the vector space dimension of $R/I$, where $R=\FF[x_1,\ldots,x_n]$.
The remainder $r$ after division by a Gr\"obner basis of $I$ is unique
and has a property that $\lt(r)$ is not divisible by any leading term of the Gr\"obner basis elements.
The remainders are unique representatives for elements of $R/I$.
It follows that there is a one-to-one correspondence as sets between $R/I$ and $R/\lt(I)$.
In particular, they are isomorphic as vector spaces.
The generators of $\lt(I)$ are monomials $m_1,\ldots,m_t$ computed by the Gr\"obner basis.
Recall that monomials are in one-to-one correspondence with the hyper-octant $H$ of nonnegative coordinates of $\ZZ^n$.
Therefore, the vector space basis of $R/\lt(I)$ is labeled by points of $H$
that are not contained in any cone whose vertex, the least element with respect to the monomial order, is $m_i$ ($1\le i \le t$).
For example, let $I = (x^2,xy+y^2) \subset \FF[x,y]=R$.
The reduced Gr\"obner basis under a lexicographic order in which $x \succ y$ is $\{ y^3, xy+y^2,x^2 \}$.
Therefore, $R/\lt(I)$ has a vector space basis $\{1,x,y,y^2\}$.
In a different lexicographic order in which $x \prec y$, the reduced Gr\"obner basis of $I$ is
$\{x^2,y^2+xy\}$. A vector space basis of $R/\lt(I)$ is $\{1, x, y, xy \}$.

\subsection*{Syzygies and generalization to free modules}

So far we have discussed only ideals and generators.
It is straightforward to generalize the notions to submodules 
of free modules over polynomial rings $R = \FF[x_1,\ldots,x_n]$.
We first need to generalize the monomial order.
Let $e_1,\ldots,e_r$ be a basis of free module $F = R^r$.
A {\bf monomial} of $F$ is a product of variables times one of the basis,
e.g., $x_1 x_2^5 e_3$ and $x_1^3 x_3 e_1$ are monomials.
A {\bf monomial order} on $F$ is defined in the same way as \eqref{cond-TermOrder}.
The well-ordering property holds with the same proof.
The leading terms are defined the same way, as well as the ideal generated by leading terms.
The division algorithm does not need to be restated.
The Buchberger criterion and algorithm make sense with a minor change that
{\em if two leading terms involve different basis elements, then the S-polynomial is defined to be zero.}
The elimination order are still applicable, and the calculation techniques are valid with respect to submodules.

Nontrivial relations between polynomials are called syzygies.
(Here, polynomials include basis elements, so they are really elements of a free module.)
More formally, if a submodule is defined by the image of a module map $\phi : F' \to F$
between finitely generated {\em free} modules, the kernel of $\phi$ is called {\bf syzygy}. 
For instance, the syzygies of $x$ and $y$ can be described by a submodule $M$ of $\FF[x,y]^2$ 
generated by $\begin{pmatrix} y \\ -x \end{pmatrix}$,
since the ideal $(x,y)$ with generators $x$ and $y$ is the image of the map
$\phi : R^2 \xrightarrow{\begin{pmatrix} x & y \end{pmatrix}} R$ and $M = \ker \phi$.
The relation between $x$ and $y$ is thus $y(x) -x(y) = 0$.
We have implicitly seen how to obtain nontrivial relations from the Buchberger algorithm.
We formed an S-polynomial of two polynomials and ran the division algorithm with respect to a set $G=\{g_i\}$ of polynomials:
\[
 \sigma(g_i,g_j) = \frac{\lt(g_j)}{\gcd(\lt(g_i),\lt(g_j))} g_i - \frac{\lt(g_i)}{\gcd(\lt(g_i),\lt(g_j))} g_j = \sum_k h^{ij}_k g_k + r ,
\]
where $h^{ij}_k \in R$.
If $G$ were a Gr\"obner basis, then by Buchberger criterion we would have $r=0$, 
and the expression would be a nontrivial relation among $g_i$.
More formally, let the module be generated by a Gr\"obner basis $\{g_i\}$.
That is, module is the image of the map $\phi : F' \to F$
where the basis $\{\epsilon_i\}$ of $F'$ is mapped as $\epsilon_i \mapsto g_i$.
Then,
\[
\tau_{ij} = \frac{\lt(g_j)}{\gcd(\lt(g_i),\lt(g_j))} \epsilon_i - \frac{\lt(g_i)}{\gcd(\lt(g_i),\lt(g_j))} \epsilon_j - \sum_k h^{ij}_k \epsilon_k 
\]
is mapped to zero under $\phi$.
It is a theorem that {\em $\ker \phi$ is generated by $\tau_{ij}$.}
$\{\tau_{ij}\}$ is often a redundant generating set for $\ker \phi$.
Be warned that $\tau_{ij}$ is defined to be identically zero if $\lt(g_i)$ and $\lt(g_j)$ involve different basis elements of $F$.

Given a finite presentation of a module $F_0/\im \phi_1$ where $\phi_1 : F_1 \to F_0$ is a module map 
between finitely generated free modules,
we can algorithmically find $\ker \phi_1$. This kernel is a finitely generated submodule of $F_1$, 
so it can be identified with the image of a map $\phi_2 : F_2 \to F_1$ between finitely generated free modules.
We can continue this as many times as we want, and obtain a sequence of maps
\[
\cdots \to F_n \xrightarrow{\phi_n} F_{n-1} \to \cdots \to F_2 \xrightarrow{\phi_2} F_1 \xrightarrow{\phi_1} F_0 .
\]
This sequence is called a {\bf free resolution} of $F_0/\im \phi_1$.
A particularly interesting case is when the resolution is finite, i.e., $F_n =0$ for some $n$.

\section{Localization}
\label{sec:localization}

\subsection*{Tensor product}
Operations of modules include direct sum and quotient,
both of which are automatic from the fact that modules are abelian groups.
Here we introduce another operation on modules called tensor product.
It is perhaps best defined by a categorical characterization,
which is the most useful to prove things.
Here we define it in a colloquial way.

Let us first review a familiar case of vector spaces over a field $\FF$.
The tensor product $V \otimes W$ of two vector spaces $V,W$ is
a collection of formal finite linear combinations $\sum_{i,j} a_{ij} v_i \otimes w_j$ 
of ``products'' $v_i \otimes w_i$ where $v_i \in V,~ w_j \in W$ and $a_{ij} \in \FF$.
The ``product'' $\otimes$ is multilinear since
\[
 (v_1 + v_2) \otimes w = v_1 \otimes w + v_2 \otimes w \quad \text{and} \quad
 v \otimes (w_1 + w_2) = v \otimes w_1 + v \otimes w_2.
\]
The scalar multiplication floats around $\otimes$ as
\[
 (a v) \otimes w = v \otimes (a w)
\]
for any $a \in \FF$.
In physics literature, it is sometimes said that a tensor is a multi-indexed object
that transforms in a certain way under transformations for ``indices.''
This phrasing focuses on the coefficients $a_{ij}$ and captures the multilinearity, as
\[
 \sum_{ij} a_{ij} v_i \otimes w_j = \sum_{ij} a_{ij} \left( \sum_a A_{ia}v_a' \right) \otimes \left( \sum_b B_{jb} w_b' \right)
 = \sum_{ab} \left(\sum_{ij}a_{ij} A_{ia}B_{jb} \right) v_a' \otimes w_b'.
\]

The {\bf tensor product of two modules} is defined in a similar way.
For an $R$-module $M$ and $N$,
we build a formal abelian group out of all possible expressions $x \otimes_R y$ for $x \in M$ and $y \in N$,
with the following equalities {\em imposed}:
\begin{align}
 (x + x') \otimes_R y &= x \otimes_R y + x' \otimes_R y , \nonumber \\
 x \otimes_R (y + y')  &= x \otimes_R y + x \otimes_R y' ,\nonumber \\
 r x \otimes_R y &= x \otimes_R r y .\label{R-linear-tensor}
\end{align}
The resulting abelian group is denoted by $M \otimes_R N$.
It is an $R$-module by defining an action $r \cdot (x \otimes_R y) = rx \otimes_R y$.
This is the tensor product of $M$ and $N$.
When there is no confusion about the base ring $R$, we write $\otimes$ instead of $\otimes_R$.

Note that any abelian group $M$ is a module over $\ZZ$ by action $n \cdot m = \sum_{i=1}^{|n|} \mathrm{sgn}(n)m$
for any $n \in \ZZ$ and any $m \in M$.
Hence, we can take the tensor product of any abelian groups as $\ZZ$-modules.
Since any $R$-module is an abelian group,
one realizes that there are at least two ways to form tensor products of two $R$-modules $M$ and $N$;
$M \otimes_\ZZ N$ and $M \otimes_R N$.
All but one equations of \eqref{R-linear-tensor} remain unchanged.
The last equation declaring $R$-linearity of $\otimes_R$ depends on $R$.
It makes a huge difference.
For instance, $\mathbb{C} \otimes_\mathbb{C} \mathbb{C} \cong \mathbb{C}$ is one dimensional,
but $\mathbb{C} \otimes_\QQ \mathbb{C}$ is infinite dimensional
since there are infinitely many irrational numbers.
$\mathbb{C} \otimes_\ZZ \mathbb{C}$ is even larger.
When the subscript is omitted, {\em the tensor product of two $R$-modules is taken over $R$ by convention}.

Since the tensor product of two modules is a module,
it makes sense to take the tensor product of three or more modules.
Fortunately, the order of the tensor product does not matter.
\[
 (L \otimes_R M) \otimes_R N \cong L \otimes_R (M \otimes_R N) \quad  \quad \quad M \otimes_R N \cong N \otimes_R M
\]

There are (unwelcome) phenomena for general modules that never occur for vector spaces.
The tensor product of two nonzero modules may be zero.
For example, the tensor product of two $\ZZ$-modules $\ZZ/(2)$ and $\QQ$ is zero
because $[a] \otimes b = [a] \otimes 2\frac{a}{2} = 2[a] \otimes \frac{a}{2} = 0 \otimes \frac{a}{2} = 0$.
An expression $x \otimes y$ may be zero in $M \otimes N$ but may not be zero in $M' \otimes N'$
where $M' \le M$ and $N' \le N$ are submodules.
For example, let $M=\ZZ$ and $N = \ZZ/(2)$ be $\ZZ$-modules. Choose $M' = 2\ZZ$ and $N' = N$.
Now $2 \otimes [1]$ is nonzero in $M' \otimes N'$, 
but in $M \otimes N$, it is equal to $2 \cdot 1 \otimes [1] = 1 \otimes [2] = 0$.

Below we explain a ``safe'' and powerful tensor product.

\subsection*{Ring of fractions}

The ring of rational numbers $\QQ$ is constructed from $\ZZ$ by inverting nonzero elements.
We wish to do a similar thing for general rings.
Let $U$ be a {\bf multiplicatively closed} subset of a ring $R$,
i.e., $1 \in U$ and any product of two elements of $U$ lies in $U$.
$U$ needs not be closed under addition.
For example, the set of all nonzero numbers in $\ZZ$ is a multiplicatively closed set.
More important is the complement of a prime ideal $\pp$.
$1$ is included in $R \setminus \pp$ because $\pp \neq (1)$.
If $a \notin \pp$ and $b \notin \pp$, then $ab \notin \pp$.
This is the defining property of the prime ideal.

We construct ``fractions'' by putting elements of $U \subseteq R$ in denominators
and elements of $R$ in numerators. The set of all fractions becomes a ring with the ``usual'' additions and multiplications
\begin{align}
 \frac{a}{p} + \frac{b}{q}    &= \frac{aq + bp}{pq}, \label{fraction-add} \\
\frac{a}{p} \cdot \frac{b}{q} &= \frac{ab}{pq}. \label{fraction-multiply}
\end{align}
Moreover, the elements of $U$ are invertible!
\[
 \frac{1}{p} \cdot p = \frac{p}{p} = 1, \quad \quad \frac{q a}{qp} = \frac{a}{p}.
\]
The new {\bf ring of fractions}%
\footnote{
Rigorously, the ring of fractions is a collection of equivalence classes of $R \times U$;
$(a,p) = (b,q)$ if and only if there exists $s \in U$ such that $s(aq-bp) = 0$.
The reason we do not define the equivalence using ``$aq=bp$'' is
that there could be zero-divisors in $R$.
Interested readers might want to check that 
$U^{-1}R$ indeed becomes a ring using \eqref{fraction-add} and \eqref{fraction-multiply} as definitions.
}
is denoted by $U^{-1}R$ or $R[U^{-1}]$.
Two special cases are so important that they deserve separate notations.
\begin{itemize}
 \item If $U = R \setminus \pp$, the ring of fractions is called the {\bf localization of $R$ at $\pp$} and denoted by $R_\pp$.
 \item If $U = \{1, f, f^2, f^3, \ldots\}$ for some $f \in R$, the ring of fractions is denoted by $R_f$.
\end{itemize}
The original ring $R$ lives in $U^{-1}R$ via a canonical map $r \mapsto \frac{r}{1}$.
With this canonical map, we omit $1$ in the denominators.
Note that $U^{-1}R$ is an $R$-module.
For example, $\ZZ_{2} = \ZZ[\half]$ is a ring of fractions with denominators are powers of $2$.
(Yes, this is a confusing notation since $\ZZ_n$ is used to mean $\ZZ/(n)$.)
$\ZZ_{(2)}$ is a ring of fractions with odd denominators.
$\QQ[x]_{(x)}$ is a ring of fractions of polynomials where denominators does not vanish at $x=0$.
What is $\QQ[x,y]_{xy}$? It is a ring of fractions of polynomials where denominators are powers of $xy$.
Since $\frac{1}{x} = \frac{y}{xy}$ and $\frac{1}{y} = \frac{x}{xy}$,
we notice that $\QQ[x,y]_{xy}$ is really the ring of Laurent polynomials.

Modules can be fractionalized as well.
Let $M$ be an $R$-module.
Choose a multiplicatively closed set $U$ of the base ring $R$,
and define $U^{-1}M$ as the set of all fractions with elements of $M$ in the numerators and elements of $U$ in the denominators.
The addition within $U^{-1}M$ is defined by \eqref{fraction-add}.
$U^{-1}M$ becomes an $U^{-1}R$-module using \eqref{fraction-multiply}.
In fact, {\em $ U^{-1}M \cong (U^{-1}R) \otimes_R M $ as $U^{-1}R$-modules.} Moreover,
\begin{itemize}
\item {\em $U^{-1}M \otimes_{U^{-1}R} U^{-1}N = U^{-1}( M \otimes_R N)$ for any $R$-modules $M,N$.}
\end{itemize}

The {\bf localization}, the process passing to fractional rings and modules,
is so well behaving under nearly all conceivable operations.
Let $A, B \le M$ be submodules.
\begin{itemize}
\item {\em $U^{-1}A$ is a submodule of $U^{-1}M$}, i.e., $U^{-1}A$ injects into $U^{-1}M$.
\item {\em $U^{-1}A \cap U^{-1}B = U^{-1}(A \cap B)$}.
\item {\em $U^{-1}A + U^{-1}B = U^{-1}(A + B)$}.
\item {\em $U^{-1}(M/A) = (U^{-1}M)/(U^{-1}A)$}.
\end{itemize}
They all originate from the following.
A sequence of maps among modules $A,B,C$
\[
 0 \to A \to B \to C \to 0
\]
is a {\bf short exact sequence} if
\begin{itemize}
 \item $\ker( A \to B ) = 0$,
 \item $\im( A \to B ) = \ker ( B \to C)$,
 \item $\im(B \to C ) = C$.
\end{itemize}
{\em The localization preserves short exact sequence}, i.e., if $0 \to A \to B \to C \to 0$ is a short exact sequence, then
\[
 0 \to U^{-1}A \to U^{-1}B \to U^{-1}C \to 0
\]
is also a short exact sequence. For example, the first and the fourth statement above follows from localizing the short exact sequence
\[
 0 \to A \to M \to M/A \to 0.
\]

The localizations at prime ideals are useful because {\em an $R$-module $M$ is zero if and only if
$M_\pp = 0$ for every prime ideal $\pp$ of $R$.}
In fact, $M = 0$ {\em if and only if $M_\mm = 0$ for every maximal ideal $\mm$ of $R$.}
This can be understood as follows.
The {\bf annihilator} denoted by $\ann_R M$ of an $R$-module $M$ is an ideal of $R$ defined by
\begin{equation}
 \ann_R M = \{ r \in R ~|~ r m = 0 \text{ for any } m \in M \}
\end{equation}
It is an ideal because $(a+b)m = am + bm = 0 + 0 = 0$ and $(ra)m = r(am) = 0$ if $a, b \in \ann_R M$.
Note that {\em any $R$-module $M$ becomes a module over $R/(\ann_R M)$. }
What is to be proved for this statement is the well-definedness of the defining operation \eqref{cond-module},
i.e., one needs to check that $[a]\cdot m = [b]\cdot m$ if $[a]=[b] \in R/(\ann_R M)$ for any $m \in M$.
Note that if $R$ is a polynomial ring over a field, there is an algorithm to compute $\ann_R M$
for a finitely presented module $M$ using pull-backs and Gr\"obner basis techniques.
Now, it is easy to see that 
{\em a module $M_\pp$ (the localization of $M$ at a prime $\pp$)
is zero if and only if the annihilator is not contained in $\pp$},
since a module is zero if and only if $1$ or any unit is an annihilator,
but the localized ring at $\pp$ anything outside $\pp$ is a unit.
If $M$ becomes zero at the localization at any maximal ideal,
it means that the annihilator is not contained in any maximal ideal.
The only elements of $R$ that lie outside of any maximal ideal are units.
Therefore, $M = 0$.

Recall that there exists a canonical map $\phi$ from $R$ to a ring of fractions $U^{-1}R$, sending $r$ to $\frac{r}{1}$.
The image under $\phi$ of any ideal $I$ of $R$ is not in general an ideal.
However, we may consider the ideal of $U^{-1}R$ generated by $\phi(I)$.
This ideal is denoted by $I^e \subseteq U^{-1}R$, called an {\bf extended ideal}.
In fact, {\em any ideal of $U^{-1}R$ is an extended ideal.}
To see this, consider any ideal $J$ of $U^{-1}R$, and its {\bf contraction} $\phi^{-1}(J) \subseteq R$.
An element of $J$ is $\frac{x}{s}$. Since $J$ is an ideal, $\frac{s}{1} \frac{x}{s} = \frac{x}{1} \in J$.
That is, $x \in \phi^{-1}(J)$. We have shown $(\phi^{-1}(J))^e = J$.
Note that {\em every prime ideal of $U^{-1}R$ is an extended ideal of a prime ideal of $R$ 
that do not intersect $U$.}
We have seen that the inverse image of a prime ideal $\pp$ is prime;
the contraction $\pp^c = \phi^{-1}(\pp)$ of a prime ideal of $U^{-1}R$ is prime.
Since prime ideal is not $(1)$, it cannot contain any unit.
Hence $\pp^c$ cannot meet the set of units $U$.
Conversely, for any prime ideal $\mathfrak q$ of $R$ that do not intersect $U$,
$\mathfrak q^e$ is a prime ideal of $U^{-1}S$,
since $\frac{ab}{ss'} \in \mathfrak q^e \Longleftrightarrow \frac{ab}{1} \in \mathfrak q^e$.

It is clear that if $\phi^{-1}(J)$ is generated by $n$ elements of $R$,
then $J$ is generated by $n$ elements of $U^{-1}R$.
More importantly, {\em if $R$ is a Noetherian ring, then any localization is a Noetherian ring.}

A {\bf local ring} is a ring with a unique maximal ideal.
Typically, maximal ideals are not unique.
In the ring of integers, every prime number generates a distinct maximal ideal.
Let $R_\pp$ be a localized ring at a prime ideal $\pp$ of $R$.
By definition of $R_\pp$, the set of all denominators do not meet $\pp$.
Hence, $\pp_\pp$ is a prime ideal of $R_\pp$.
If $\mm$ is any ideal of $R_\pp$, we must have $\mm \subseteq \pp_\pp$.
Otherwise, $\mm$ contains an element outside $\pp_\pp$, which is invertible,
and hence $\mm = (1)$.
The localized ring $R_\pp$ at a prime ideal is a local ring with a unique maximal ideal $\pp_\pp$.
One thing to remember about a local ring is that
{\em any element outside the unique maximal ideal is a unit, i.e., it is invertible.}

\subsection*{Geometry}

Why is it called ``localization?''
Here we briefly introduce the algebro-geometric point of view on rings.

The (only) algebraic structure that we should start with for a geometric object is a function space.
By a function we mean a set-theoretical map from the geometric object to numbers.
We can impose certain conditions such as continuity or differentiability
on the function space to study deeper and interesting aspects of the geometric object.
In classical algebraic geometry, the geometric object, called variety,
is defined by one or more polynomial equations,
and the functions are defined by polynomials.

Roughly speaking, there are two interesting types of functions:
globally defined functions and locally defined functions.
The locally defined ones may not be well-defined for some regions far from the point of interest.
In the algebraic setting, the global functions are given by
polynomial expressions without denominators.
There may be several polynomial expressions for the same function.
This happens when two polynomials differ by the equations of the variety.
For example, the global function space, called {\bf coordinate ring}
of a parabola on $\RR^2$ defined by $y = x^2$ is a quotient ring $R = \RR[x,y]/(y-x^2)$.

Locally defined functions at a point $p$, if they are given as a fraction of two polynomials,
are precisely those whose the denominators are not zero at $p$.
For example, the fraction $\frac{1}{x+1}$ is a locally defined function at a point $(0,0)$ of the parabola,
but $\frac{x-1}{y}$ is not.
Collecting all locally defined functions at $\mm = (0,0)$, we obtain
a {\bf local ring} $R_\mm = (\RR[x,y]/(y-x^2))_{(x,y)}$,
which is consisted of fractions whose denominators are not contained in the maximal ideal $(x,y)$.
The global function ring $R$ is \emph{localized at $\mm$ to be a local ring $R_\mm$}.
Note that we have identified a point on the parabola with a maximal ideal $\mm$.
This makes sense because the set of all solutions 
of the equations $f(x,y)=0$ where $f(x,y) \in \mm$ is exactly $(0,0)$.

One might ask what it means geometrically to localize at a prime ideal $\pp$.
A prime ideal defines a subvariety by equations $f = 0$ where $f \in \pp$.
The localization reveals the set of ``functions'' that are locally defined ``around'' the subvariety.
(The notion of neighborhood can be made rigorous by defining a topology using polynomials.)

\section{Determinantal ideal}
\label{sec:det-ideal}

When is a square matrix with integer entries invertible?
A naive answer is that it is when the determinant is nonzero.
It is true when the matrix is viewed as a matrix over $\QQ$.
If we insist that the inverse matrix must be expressed over $\ZZ$,
then the answer is that it is when the determinant is $\pm 1$.
That is, the determinant must be invertible within the ring.
Put differently, it is when the ideal generated by the determinant is the unit ideal.
The determinant has another use. The rank of a matrix is the largest $k$ such that
there exists a $k \times k$ submatrix whose determinant is a nonzero.
Put differently, it is when the ideal generated by determinants of all $k \times k$ submatrices
is nonzero.

These observations motivate us to define a determinantal ideal over an arbitrary commutative ring $R$.
Consider a rectangular matrix $\mathbf M$ with entries in $R$.
A {\bf $k^\text{th}$ minor} is the determinant of a $k \times k$ submatrix.
(There are quite many $k^\text{th}$ minors if $k$ is about half of the matrix size.)
Note that to define the minor we do not have to assume that the full matrix $\mathbf M$ is square.
When $k$ is larger than the number of rows or columns, $k^\text{th}$ minor is, by definition, zero.
{\bf $k^\text{th}$ determinantal ideal $I_k(\mathbf M)$} of $\mathbf M$ is the ideal of $R$ 
generated by all $k^\text{th}$ minors of $\mathbf M$.
By convention, $0^\text{th}$ determinantal ideal is taken as $I_0(\mathbf M) = (1) = R$.
If $\mathbf M$ is $n \times n$, then $n^\text{th}$ determinantal ideal is generated by a single element,
the determinant of $\mathbf M$.
The first determinantal ideal is generated by all entries of $\mathbf M$.
A $(k+1)^\text{st}$ minor is a linear combination of $k^\text{th}$ minors,
as it has a cofactor expansion.
Therefore, the determinantal ideals form a decreasing chain:
\[
 R = I_0(\mathbf M) \supseteq I_1(\mathbf M) \supseteq I_2(\mathbf M) \supseteq \cdots 
    \supseteq I_{\min( \# row, \# col)}( \mathbf M ) \supseteq (0)
\]
The {\bf rank} of $\mathbf M$ is the largest $k$ such that $I_k(\mathbf M)$ is nonzero.

The determinantal ideals are invariants of $\mathbf M$ 
under invertible matrix multiplications on the right or on the left.
\[
 I_k ( \mathbf {M D} ) = I_k( \mathbf{E M} ) = I_k (\mathbf M) \quad 
\quad \text{if $\mathbf{D,~E}$ are square and invertible.}
\]
To understand this, it is enough to see $I_k( \mathbf{E M} ) \subseteq I_k( \mathbf M )$ for arbitrary matrix $E$.
Then, for invertible $\mathbf E$, we would have $I_k(\mathbf{E^{-1} E M }) \subseteq I_k(\mathbf{E M})$.
We need to see that a $k^\text{th}$ minor of $ \mathbf{E M}$ is a linear combination of $k^\text{th}$ minors of $\mathbf M$.
Fix a $k \times k$ submatrix $\mathbf K$ of $\mathbf{EM}$.
\[
 \mathbf K_{il} = \sum_{j} e _{ij} \mathbf M _{j l} 
\]
where $e_{ij} = \mathbf{E}_{ij}$ is written with small letter to emphasize it is merely an element of $R$.
Now the determinant of $\mathbf K$ is vividly expressed by minors of $\mathbf M$,
since the determinant is multilinear over $R$ in rows $K_i$ of $\mathbf K$.
\begin{align*}
 \det \mathbf K 
&= \det( K_1, \ldots, K_k ) 
= \det \left[ \sum_{j_1} e_{1 j_1} M_{j_1}, \ldots, \sum_{j_k} e_{k j_k} M_{j_k} \right]\\
&= \sum_{j_1,\ldots,j_k} e_{1 j_1} \cdots e_{k j_k} ~\det \left[  M_{j_1},\ldots, M_{j_k} \right]
\end{align*}
where $M_j$ is the row $j$ of the submatrix of $\mathbf M$ that is consisted of entries contributing to $\mathbf K$.

Suppose $R$ be a local ring with a maximal ideal $\mm$,
and consider a matrix $\mathbf M$ over $R$ of rank $k$.
If $I_k(\mathbf M) = (1)$, how would $\mathbf M$ look like after row and column operations?
The operations are invertible, and therefore do not change the determinantal ideals.
Since the determinantal ideals form a descending chain,
it follows that $I_1(\mathbf M) = (1)$.
It means that at least one entry must lie outside $\mm$ and is invertible.
Bringing that entry to $(1,1)$ position by permuting rows and columns,
we can eliminate other entries in the first row and column, by row and column additions.
So
\[
 \mathbf M \cong \begin{pmatrix} 1 & 0 \\ 0 & \mathbf M' \end{pmatrix}
\]
where $\cong$ means the equivalence up to invertible matrix multiplication on the left or right.
Now we use the fact that $I_2(\mathbf M) = (1) = I_1( \mathbf M' )$.
By a similar process we can extract $1$ from $\mathbf M'$.
After $k$ steps, we have
\begin{equation}
 \mathbf M \cong \begin{pmatrix} \id_k & 0 \\ 0 & 0 \end{pmatrix}.
\label{free-module-presentation}
\end{equation}
Here, observe that $R$ is not necessarily a field.
It is just a local ring where some nonzero elements may not be invertible.
The smallest nonzero determinantal ideal being a unit ideal
was strong enough to imply the structure of the matrix.

\subsection*{Smith normal form}

Beyond the ``easy'' local ring case,
there is one more easy case of {\bf principal ideal domains},
in which every ideal is generated by a single element.
It is not necessarily local.
{\em The ring of integers and the ring of polynomials in one variable over a field
are principal ideal domains};
given finitely many generators of an ideal,
one can find the greatest common divisor ($\gcd$) by the Euclidean algorithm.

Let $R$ be a principal ideal domain.
We ask what normal forms of matrices we may have, after
invertible matrix multiplication on the left or right.
In other words, we seek for invariants of the matrices.
Observe that the $\gcd$ of two elements%
\footnote{To speak of $\gcd$, it must be first proved that in $R$ any element has a unique factorization.
It is true that in any principal ideal domain, any element is a product of irreducible factors
that are unique up to units. If the latter is satisfied, the ring is called a {\bf unique factorization domain}.
Note that an arbitrary ring may not be a unique factorization domain.
For example, in $\ZZ[\sqrt{-5}]$, we have $6 = 2 \cdot 3 = (1+\sqrt{-5})(1-\sqrt{-5})$,
which are two different factorizations of $6$ into irreducible factors.}
$f$ and $g$ can be expressed as a linear combination of $f$ and $g$
\[
\gcd(f,g) = a f + b g
\]
for some $a,b \in R$. To see this, consider the ideal $(f,g)$.
It must be generated by a single element, say, $d$, i.e., $(d) = (f,g)$.
Then, $f \in (d) \Longleftrightarrow f = d f'$ and $g \in (d) \Longleftrightarrow g = d g'$.
So $d$ is a common divisor of $f$ and $g$. In addition, $(d) = (f,g) \subseteq (\gcd(f,g))$
implies that $\gcd(f,g)$ divides $d$. Thus, $\gcd(f,g)$ and $d$ are the same up to units,
and $\gcd(f,g) \in (f,g)$. We thus obtain the above expression.
Now, consider the following matrix equation
\[
\begin{pmatrix} \gcd(f,g) \\ 0 \end{pmatrix}
= \begin{pmatrix} a & b \\ -g' & f' \end{pmatrix} \begin{pmatrix} f \\ g \end{pmatrix}
\]
The $2 \times 2$ matrix has determinant $af' + bg' = \gcd(f,g)/d$, a unit.
We have transformed a $2 \times 1$ matrix by an invertible left multiplication
such that there is only a single nonzero entry left. 
A similar thing can be done for an $n \times 1$ matrix.
\[
\begin{pmatrix} \gcd(f_1,\ldots,f_n) & 0 & \cdots & 0 \end{pmatrix}^T
= \mathbf E \begin{pmatrix} f_1 & \cdots & f_n \end{pmatrix}^T
\]
where $\mathbf E$ is invertible.

If we are given with a rectangular matrix $\mathbf M$,
we can perform this transformation by looking at the first column of $\mathbf M$.
The transformed matrix $M'$ will have a nonzero entry on the first row in the first column.
Perform a similar transformation by right multiplication, focusing on the first row of $\mathbf M'$.
The new matrix $\mathbf M''$ has unique nonzero entry at $\mathbf M''_{1,1}$ in the first row,
but the first column may be screwed up.
We can iterate these transformations as many times as we want.
Will this process eventually terminate? Yes.
For example, if we are working in the ring of integers,
the absolute value of the entry $\mathbf M''_{1,1}$ is smaller than $|\mathbf M'_{1,1}|$, which is $\le |\mathbf M_{1,1}|$.
Positive integers cannot decrease forever, and the process must terminate.
More generally, the ascending chain condition makes the proof smooth.
The ideals generated by the $(1,1)$-entries will form an increasing sequence of ideals.
Our ring $R$, being a principal ideal domain, is Noetherian.
The sequence must become stationary after finitely many iterations.
If one is not too familiar with Noetherian rings,
one can rely on the fact that there are finitely many factors of a given element.
If $(d_1) \subseteq (d_2) \subseteq \cdots$ is an increasing sequence of ideals,
we have $d_{i+1} | d_{i}$. Since there are only finitely many factors in $d_1$,
one cannot have an infinite and strictly increasing sequence of ideals.

The matrix $\mathbf M^{(n)}$ obtained from $\mathbf M$ by iterating
the $\gcd$-computations $n$-times, has a unique nonzero entry $d$ at $(1,1)$
in its row and column.
\[
\mathbf M^{(n)} = \begin{pmatrix} d & 0 \\ 0 & M_1 \end{pmatrix}
\]
We can claim a bit more. If there is any entry at $(u,v)$ of $M_1$ that is not divisible by $d$,
we can add the row $u$ to the first row and run the $\gcd$-computation again.
The ideal generated by $(1,1)$-entry will become larger.
Iterating, we eventually find the largest possible ideal generated by $(1,1)$-entry.
Therefore, we can obtain $M_1$ whose entries are all divisible by $d$.

The above algorithm applied to $M_1$, and to its submatrix, and so on,
will produce a diagonal matrix. This diagonal matrix is called the {\bf Smith normal form} of $\mathbf M$.
The nonzero diagonal entries $d_1, \ldots, d_k$ have a property that
\[
d_1 | d_2 | \cdots | d_k .
\]
Here, the number $k$ is precisely the rank of the matrix $\mathbf M$.
Moreover, the determinantal ideals are 
\[
I_t ( \mathbf M ) = (d_1 \cdots d_t).
\]
This immediately proves that $d_i$'s are uniquely determined by $\mathbf M$.
The invariants $d_i$'s are called {\bf elementary divisors} of $\mathbf M$.
Since any matrix, not necessarily square, can be brought to the Smith normal form
by invertible transformations,
it follows that the elementary divisors are {\bf complete invariants},
i.e., the elementary divisors are the same for two matrices if and only if two matrices are related by
invertible matrix multiplication on the left and right.

\subsection*{Finitely generated modules and Fitting ideals}

We briefly noted about finite presentations of modules when we discussed free modules.
We say that an $R$-module $M$ is finitely presented by a matrix $\phi : R^m \to R^n$
if $M$ is isomorphic to $R^n / \im \phi$. This is more commonly denoted as $ M = \coker \phi$.

To get some feeling, let us consider a simple case where $R = \ZZ$.
The matrix $\phi: R^m \to R^n$ is an $n \times m$ matrix with integer entries.
Any basis change in $R^n$ amounts to an invertible left multiplication on $\phi$.
Any basis change in $R^m$ amounts to an invertible right multiplication on $\phi$.
We know a very convenient canonical form of $\phi$ --- the Smith normal form.
Let $\phi$ be in the Smith normal form with the diagonal elements $d_1, \ldots, d_k$.
In the simplest case $k = m = n =1$, the module $M$ is $R / \im \phi = \ZZ / (d_1)$.
If $m = n = 2$ and $d_2 = 0$, the module $M$ is $R^2 / \im \phi = \ZZ/(d_1) \oplus \ZZ$.

In fact, we can prove the structure theorem for finitely generated abelian groups very easily.
Let $M$ be a finitely generated abelian group. It can be viewed as a finitely generated
$\ZZ$-module. Let $n$ be the number of generators. We have a module map $\tilde \phi : \ZZ^n \to M$,
which is surjective. The kernel $K$ is also finitely generated since $\ZZ^n$ is Noetherian.
Hence, we have another map $\phi : \ZZ^m \to K \subseteq \ZZ^n$ where $m$ is the number of
generators of $K$. We have constructed a finite presentation of $M$ by $\phi$, i.e., $M = \ZZ^n / \im \phi$.
Bring $\phi$ to the Smith normal form by basis changes in $\ZZ^n$ and $\ZZ^m$.
If $d_1 | \cdots | d_k$ are elementary divisors of $\phi$,
we define $d_{k+1} = \cdots = d_n = 0$.
Then we have an isomorphism of $\ZZ$-modules $M \cong \ZZ/(d_1) \oplus \cdots \oplus \ZZ/(d_n)$.
Viwed as a group, the notation will be
\[
 M = \ZZ/d_1 \ZZ \times \cdots \times \ZZ/d_k \ZZ \times \ZZ^{n-k} .
\]
This is the most general form of a finitely generated abelian group.
Note that the rank $n-k$ and the elementary divisors 
$d_1 | d_2 | \cdots | d_k \neq 0$ are uniquely determined by the group $M$.
Any finite abelian group is obviously finitely generated.

Let $R$ be an arbitrary commutative ring, and $M$ be an $R$-module 
with a finite presentation given by a matrix $A: R^m \to R^n$, i.e., $M \cong R^n/\im A = \coker A$.
We define {\bf Fitting ideals} $F_i(M)$ of $M$ as
\begin{equation}
 F_i(M) = I_{n-i}(A),
\label{Fitting-ideal}
\end{equation}
so
\[
 F_0(M) \subseteq F_1(M) \subseteq \cdots \subseteq F_n(M) = R .
\]
One should ask: $A$ is one of many possible finite presentations.
How can we be sure that $F_i(M)$ is determined by $M$ and is independent of a particular presentation $A$?
We can be sure.
The presentation really consists of two sets of things:
a set of generators $x_1,\ldots,x_n \in M$ and their relations (syzygies) given by columns of $A$
\[
 \sum_{i} x_i A_{ij} = 0.
\]
Thus, we can think of $A$ as the matrix of relations.
The isomorphism $M \cong \coker A$ ensures that any possible relations among $x_i$ is generated by columns of $A$.
That is, for any coefficients $c_i$ such that $\sum_i x_i c_i = 0$ we can express $c_i$ as $c_i = \sum_j  A_{ij} d_j$.
Therefore, if we write a $n \times \infty$ matrix $\tilde A$ by collecting \emph{all} relations among $x_i$,
then the $k^\text{th}$ determinantal ideal of $\tilde A$ is exactly $I_k(A)$.
In other words, $I_k(A)$ is determined by the chosen generators of $M$.
Our question on the well-definedness of the Fitting ideal concerns, in fact, many possible choices of generators for $M$.
Now, let $y_1, \ldots, y_{n'}$ be elements of $M$.
We have $n+n'$ generators $x_i, y_{i'}$ of $M$,
and $M$ is presented as a quotient module of $R^{n+n'}$. 
The relations among $x_i,y_{i'}$ constitute a matrix
\[
W =
\begin{pmatrix}
 A & A'  & \# \\
 0 & \id & \# \\
\end{pmatrix} 
\sim
W' =
\begin{pmatrix}
 A & 0   & A' \\
 0 & \id & 0 \\
\end{pmatrix} 
\sim
W'' =
\begin{pmatrix}
 A & 0   & 0 \\
 0 & \id & 0 \\
\end{pmatrix} 
\]
where $\sim$ means equality up to row or column operations.
The first matrix $W$ can be written as shown because $y_{i'}$ can be written as some combination of $x_i$.
$W''$ is obtained from $W'$ since we know that $A$ generates all relations among $x_i$.
Since determinantal ideals are invariant under invertible matrix multiplications,
we see that $I_{k+n'}(W) = I_{k+n'}(W') = I_{k+n'}(W'') = I_k(A)$.
Another finite presentation $B : R^{m'} \to R^{n'}$ of $M \cong \coker B$
gives another set of generators.
By the above calculation, $I_{n-k}(A) = I_{n+n'-k}(W) = I_{n'-k}(B)$.
Therefore, {\em the Fitting ideals are well-defined.}
We now see clearly why the numbering of the Fitting ideals are given as \eqref{Fitting-ideal}.
(One might notice that the finite presentation is actually too much than 
what is needed to define Fitting ideals. They can be defined for any finitely generated module.)

The first nonvanishing Fitting ideal is important because it tells 
when a finitely generated module becomes free after localization.
For example, suppose an $R$-module $M$ is finitely presented by $\phi : R^m \to R^3$ and $F_1(M)$ is nonzero.
Let $\pp$ be a prime ideal of $R$ such that $F_1(M) \not\subseteq \pp$.
Localizing at $\pp$, we see that $F_1(M) = I_2(\phi)$ becomes the unit ideal,
since anything outside $\pp$ is a unit in $R_\pp$.
Then, we have seen in \eqref{free-module-presentation} that
the matrix $\phi_\pp$ is equivalent to a diagonal matrix with entries $0$ or $1$.
Therefore, $\coker \phi_\pp$ is obviously isomorphic to $R_\pp^3/R_\pp^2 = R_\pp^1$, a free module.
Conversely, if $\coker \phi_\pp$ is free for some prime ideal $\pp$,
then $I_k(\phi_\pp)$ is either $(1)$ or $(0)$.
Therefore, every Fitting ideal of $\coker \phi_\pp$ is either $(0)$ or $(1)$.
In conclusion, {\em the localizaed module $M_\pp$ of $M$ at a prime ideal $\pp$ is free 
if and only if the first nonvanishing Fitting ideal of $M$ is not contained in $\pp$.}

The initial Fitting ideal $F_0(M)$ is also interesting because it approximates the annihilator of $M$.
Let $M$ be generated by $n$ elements. Then,
\begin{equation}
(\ann M)^n \subseteq F_0(M) \subseteq \ann M .
\label{Fitting-annihilator}
\end{equation}
Here, the {\bf ideal power} $(\ann M)^n$ means the ideal generated by all products of $n$ elements from $\ann M$.
By definition, $F_0(M) = I_n(A)$ where $A$ is a matrix of relations among the $n$ generators $x_i$ of $M$.
If $Z$ is an $n \times n$ submatrix of $A$, we have $\sum_i  x_i Z_{ij} = 0$.
Multiplying the adjugate matrix of $Z$, we see $(\det Z) x_i = 0$ for all $1 \le i \le n$.
Hence, $\det Z \in \ann M$. Since $F_0(M)$ is generated by these $\det Z$, we have the second inclusion
of \eqref{Fitting-annihilator}.
If $a_1,\ldots, a_n \in \ann M$, then the diagonal matrix made of $a_i$ expresses relations among $x_i$.
The determinant $a_1\cdots a_n$ therefore belongs to $F_0(M)$.
This proves the first inclusion of \eqref{Fitting-annihilator}.

\section{Finite fields}

A field is a nonzero commutative ring where every nonzero elements are invertible.
There is only one proper ideal, the zero ideal $(0)$.
A {\bf finite field} is a field with finitely many elements.
If there are $q$ elements in the field, we denote the field by $\FF_q$.
The minimum possible number is of course $q=2$ because we need $0 \neq 1$.
In this case it is called {\bf bianary field} $\FF_2$.
It may seem quite modest to require for a field to have finitely many elements,
but, unlike finite groups, finite fields are particularly simple ---
{\em there is a unique field of $q$ elements up to isomorphisms.}
Let us see why.

Since any ring is an additive group, $\FF_q$ is a finite abelian group.
How does it look like as an additive group?
Applying the ``structure theorem for finitely generated abelian group,''
we know $\FF_q \cong \bigoplus_{i=1}^k \ZZ/(d_i)$.
where $1 < d_1 | d_2 | \cdots | d_k \neq 0$.
It follows that for any element $x \in \FF_q$ we have $d_k x = 0$.
Here, $d_k$ is the smallest number such that $d_k x = 0$ for any $x \in \FF_q$.
If $d_k$ is not a prime number, say, $d_k = pp'$ with $p, p' < d_k$,
then $p = p\cdot 1 \neq 0$ in $\FF_q$ is invertible.
Hence, $p' = p' \cdot 1 = 0$, which is a contradiction to the minimality of $d_k$.
Therefore, $d_k$ is prime, and $d_1 = d_2 = \cdots = d_k = p$.
It follows that $| \FF_q | = q = d_1 d_2 \cdots d_k = p^k$.
This conclusion is often phrased as {\em a finite field has {\bf characteristic} $p$ for some prime number $p$
and the total number of elements is a power of $p$.}
Note that the image of $\ZZ \to \FF_q$, as there is a unique ring homomorphism from $\ZZ$ to any ring,
is $\FF_p$, a subfield of $\FF_q$.
It is of great importance that {\em $\FF_q$ is a $k$-dimensional vector space over $\FF_p$.}

Next, how does the multiplicative group of all nonzero elements of $\FF_q$ look like?
It is again a finitely generated abelian group.
Therefore, its group structure is uniquely determined by elementary divisors $1< e_1 | e_2 | \cdots |e_m \neq 0$.
In a similar logic as above, we have $x^{e_m} = 1$ for any nonzero $x \in \FF_q$.
It is an equation in the field, so it can be rewritten as $x^{e_m} -1 = 0$.
How many solutions can an equation have? At most the degree.
It means that $q-1 = e_1 e_2 \cdots e_m \le e_m$.
The only possibility is that $m = 1$ and $e_1 = q-1$.
That is, the multiplicative group $\FF_q^\times$ consisted of all nonzero elements
is isomorphic to the additive group $\ZZ/(q-1)$, a {\bf cyclic group} generated by a single element.
In other words, there exists an element $x$ in $\FF_q$ such that
$\{ x^n | n \in \ZZ \}$ is equal to the set of all nonzero elements.
Such an element $x$ is called {\bf primitive element} of $\FF_q$.
There could be many primitive elements in $\FF_q$.

Finally, let us mix the two operations, the addition and multiplication, and consider the ring structure of $\FF_q$.
Before we start analyzing $\FF_q$,  remark that for
any field $\FF$, not necessarily finite, and an irreducible polynomial $f(x)$ over $\FF$,
we can construct a larger field $\FF[x]/(f(x))$ that contains $\FF$.
A larger field is called an {\bf extension field} over the smaller field.

We noted above that $\FF_q$ is a $k$-dimensional vector space over $\FF_p$.
Fix a primitive element $\alpha$ of $\FF_q$.
Then, the set of ``vectors''  $\{1,\alpha,\alpha^2,\ldots,\alpha^k\}$ contains $k+1$ vectors,
so it cannot be linearly independent over $\FF_p$.
Consider the set $I_\alpha$ of all polynomials $f(t) \in \FF_p[t] \subseteq \FF_q[t]$ such that $f(\alpha) = 0$.
It is clearly a nonzero ideal, and therefore is generated by a single element $f_\alpha(t)$,
which divides any element in $I_\alpha$.
$f_\alpha(t)$ is unique if we demand the leading coefficient to be $1$.
A polynomial with the leading coefficient $1$ is called {\bf monic}.
The minimality of $f_\alpha(t)$ forces it to be irreducible.
$f_\alpha(t)$ is uniquely determined by $\alpha$, 
called the {\bf minimal polynomial} of $\alpha \in \FF_q$ over $\FF_p$.
We have established a ring homorphism $\FF_p[t] / (f_\alpha(t)) \to \FF_q$ such that $t \mapsto \alpha$.
Since $\alpha$ is a primitive element, this map is surjective.
Moreover, it is injective because of the choice of $f_\alpha(t)$.
It follows that the degree of $f(t)$ is actually equal to $k$,
the dimension of $\FF_q$ as an $\FF_p$-vector space,
and {\em the ring $\FF_q$ is isomorphic to a quotient ring of the polynomial ring $\FF_p[t]$}.

Recall the polynomial $h(x) = x^q -x$ becomes zero at any element of $\FF_q$.
It means that in the ring $\FF_q[x]$,
the polynomial $h(x)$ factorized into linear factors as
\begin{equation}
h(x) = x(x-\alpha)(x-\alpha^2)(x-\alpha^3) \cdots (x- \alpha^{q-1}) .
\label{h-factors}
\end{equation}
Since $x^q-x \in I_\alpha$, it follows that $f_\alpha(x) \in \FF_p[x]$ divides $x^q-x \in \FF_p[x]$.
Therefore, $f_\alpha(x) \in \FF_q[x]$ factorizes into linear factors, too.
That is, any root of $f_\alpha(x)$ can be found in $\FF_q$, not only $\alpha$.
Furthermore, {\em any root of $f_\alpha(x)$ is a primitive element.}
Let $\beta \in \FF_q$ be any root of $f_\alpha(x)$.
The ideal $I_\beta \subseteq \FF_p[x]$ contains $f_\alpha(x)$.
Hence, the minimal polynomial $f_1(x) \in I_\beta$ divides $f_\alpha(x)$.
Since $f_\alpha(x)$ is irreducible in $\FF_p[x]$, we must have $f_\alpha(x) = f_\beta(x)$.
It means that $\FF_q \ni \alpha \mapsto \beta \in \FF_q$ is an isomorphism,
guaranteeing that $\beta$ is a primitive element.

The quotient ring presentation of $\FF_q$ is not unique because it depends on $f_\alpha(t)$ which
is determined by a primitive element. 
It is thus a relevant question whether two finite fields with the same cardinality $\FF_q$ and $\FF_q'$ 
are isomorphic as rings.
Let $\alpha$ be a primitive element of $\FF_q$ with the minimal polynomial $f(x)$ over $\FF_p$,
and $\beta$ be a  primitive element of $\FF_q'$ with the minimal polynomial $g(x)$ over $\FF_p$.
We know that $f(x) | h(x)$ and $g(x) | h(x)$ where $h(x) = x^q -x$.
It follows from \eqref{h-factors} that $f(\beta^n) = 0$ in $\FF_q'$ for some $n$.
Given such $n$, we can define a ring homomorphism $\FF_p[x]/(f(x)) \ni x \mapsto y^n \in \FF_p[y]/(g(y))$.
It amounts to a ring homomorphism $\FF_q \ni \alpha \mapsto \beta^n \in \FF_q'$.
Note that {\em any nonzero ring homomorphism from a field is injective}
simply because the kernel is a proper ideal.
Since both the domain and the target are finite dimensional vector spaces,
our homomorphism is bijective, and we obtain a field-isomorphism.
In conclusion, {\em any finite field is uniquely determined up to isomorphisms by its cardinality.}

Let $f(x) \in \FF_q[x]$ be any irreducible polynomial over $\FF_q$.
We may consider an extension field $\mathbb E = \FF_q[x]/(f(x))$ over $\FF_q$.
It is still a finite field for being a finite dimensional vector space over a finite field.
In $\mathbb E$, $f(x)$ factorizes into linear factors by the same reasoning as above.
More generally, since any polynomial $g(x)$ is a product of irreducible polynomials,
by extending the field, one can factor further some of the irreducible factors.
Since there are finitely many irreducible factors, one eventually reaches an extension field
where $g(x)$ factorizes into linear factors completely, after finitely many extensions.
If we started with a finite field, then the ultimate field will still be finite of cardinality, say, $q'$.
It follows that $g(x)$ is a factor of $(x^{q'}-x)^n$ for some $n$,
where $n$ is to take care of potential multiplicity in $g(x)$.
$n$ can be chosen to be large. If one wishes, it may be of form $p^m$ 
where $p$ is the characteristic of the field.
Then, $(x^q -x)^{p^m} = x^{q' p^m} - x^{p^m}$.
Summarizing, {\em any nonzero polynomial $g(x)$ over a finite field divides $x^{p^{m'}} - x^{p^m}$
for some $m' > m$.}

\end{document}